\documentclass[12pt]{article}
\usepackage[utf8]{inputenc}
\usepackage{amsmath}
\usepackage{amsthm}
\usepackage{amsfonts}

\usepackage{algorithm}
\usepackage[noend]{algpseudocode}
\usepackage{color}  
\usepackage{lineno}
\usepackage{graphicx}
\usepackage{hyperref,natbib}
\usepackage{authblk}
\usepackage{subcaption}
\usepackage{multirow}
\usepackage[doublespacing]{setspace}
\usepackage{esvect}
   
\newtheorem{theorem}{Theorem}[section]

\newtheorem{proposition}[theorem]{Proposition}

\newcommand{\blind}{0}

\title{Bayesian modeling of nearly mutually orthogonal processes}
\author[1]{James Matuk} 
\author[2,3,4]{Amy H. Herring}
\author[2,5]{David B. Dunson}
\affil[1]{Department of Epidemiology, University of Pittsburgh, Pittsburgh, PA, USA}
\affil[2]{Department of Statistical Science, Duke University, Durham, NC, USA}
\affil[3]{Duke Global Health Institute, Duke University, Durham, NC, USA}
\affil[4]{Department of Biostatistics and Bioinformatics, Duke University, Durham, NC, USA}
\affil[5]{Department of Mathematics, Duke University, Durham, NC, USA}
\date{ }

\begin{document}

\linespread{1}

\if1\blind
{
  \bigskip
  \bigskip
  \bigskip
  \begin{center}
    {\LARGE\bf Bayesian modeling of nearly mutually orthogonal processes}
\end{center}
  \medskip
}\fi

\if0\blind
{
\maketitle
}\fi

\begin{abstract}
Functional factor analysis is an important dimension reduction method for functional and longitudinal data. Factor loadings give insight into patterns of variability of the observations, while latent factors provide a low-dimensional representation of the data that is useful for inferential tasks. Constraining the functional factor loadings to be mutually orthogonal is desirable for model parsimony but is computationally challenging. In this work, we introduce nearly mutually orthogonal processes, which can be used to effectively enforce mutual orthogonality of factor loadings while maintaining computational simplicity and efficiency. The joint distribution is governed by a penalty parameter that determines the degree to which the processes are mutually orthogonal and is related to ease of posterior computation. We demonstrate that our approach can be used for flexible and interpretable inference in an application to studying the effects of breastfeeding status, illness, and demographic factors on weight dynamics in early childhood. Code is available at \url{https://github.com/jamesmatuk/NeMO-FFA}.
\end{abstract}
    
\clearpage

\linespread{2}

\section{Introduction}\label{sec:intro}

A fundamental problem in functional data analysis is that observations are generated from infinite-dimensional stochastic processes \citep{ramsay2005,ferraty2006}. Consequently, dimension reduction is an important modeling tool. Low-dimensional representations of functional data can retain most of the information in the original observations when the modes of variability are assumed to be smooth. Functional factor analysis provides a generative model for functional data that decomposes observations into latent factors and functional factor loadings. The loadings are useful for understanding patterns of variability in the observations, while the latent factors serve as a low-dimensional representation of the data. Bayesian hierarchical models are successful in using latent factors for a variety of inferential tasks, including regression \citep{montagna2012}, joint modeling \citep{moran2021}, clustering \citep{marco2022}, and modeling multivariate functional data \citep{kowal2017,kowal2021}.

The most popular dimension reduction tool for functional data is functional principal component analysis \citep{james2000,yao2005,peng2009}. When viewed through a probabilistic framework \citep{tipping1999}, functional principal components are functional factor loadings, while projections of the data onto the components are latent factors. For relevant Bayesian articles, refer to \cite{behseta2005,vdLinde2008,saurez2017,jiang2020,nolan2021}. Within this context, a mutual orthogonality constraint is imposed on the factor loadings to aid interpretability of latent factors, especially for subsequent inferential tasks. Exact orthogonality can be imposed through an appropriate prior \citep{jauch2021,kowal2017,Sartini2026}, however, strictly enforcing mutual orthogonality can be computationally challenging \citep{shamshoian2022}.

We effectively enforce this constraint by defining the nearly mutually orthogonal ($\small{\mbox{NeMO}}$) process prior to shrink factor loadings towards mutual orthogonality. Our approach represents a fundamentally new strategy for functional factor analysis that relaxes strict orthogonality constraints while preserving their advantages in terms of interpretability and model parsimony. The joint distribution induced by $\small{\mbox{NeMO}}$ processes is governed by a penalty parameter controlling departures from mutual orthogonality, similar in spirit to the Bayesian constraint relaxation framework of \cite{duan2019}. This allows $\small{\mbox{NeMO}}$ to balance the extremes of an unconstrained model and a completely constrained model in a computationally efficient framework.

Our approach fits within the broader literature of using Bayesian nonparametric models to encode constraints, including monotonicity \citep{shively2009,lin2014}, convexity or concavity \citep{shively2011,hannah2013,lenk2017}, and other shape constraints \citep{wheeler2017,dasgupta2021,yu2022}.  In principle, $\small{\mbox{NeMO}}$ could be applied to a diverse range of problems where orthogonality plays an important role, such as semiparametric inference \citep{plumlee2018,kowal2022} and modeling spatial random effects \citep{hodges2010,khan2022}. However, we illustrate $\small{\mbox{NeMO}}$ processes through applications to factor analysis of functional data. 

The remainder of the paper is organized as follows. In section \ref{sec:nemo}, we define $\small{\mbox{NeMO}}$ processes and discuss important properties. In section \ref{sec:fpca}, we discuss using $\small{\mbox{NeMO}}$ processes as a prior model for functional factor analysis models within a wide range of application settings, including sparse and irregular observations and generalized functional data. In section \ref{sec:simulations}, we present simulations that highlight the computational tractability of $\small{\mbox{NeMO}}$ processes. In section \ref{sec:cebu}, we present a detailed application that illustrates how $\small{\mbox{NeMO}}$ processes can be used as a layer in a hierarchical model that provides flexible and interpretable inference with uncertainty quantification. In section \ref{sec:discussion}, we include a discussion of our approach and identify future directions of work. In extensive Supplemental Materials, we present proofs for the propositions in the main paper, full implementation details for the proposed models, additional simulation examples, and extensive details related to the Cebu data analysis.

\section{Nearly Mutually Orthogonal Processes}\label{sec:nemo}
\subsection{Definition of Nearly Mutually Orthogonal Processes}\label{sec:remo}

We define $\small{\mbox{NeMO}}$ processes with the goal of developing a prior model that shrinks the pairwise inner products of a set of random functions, $\{\lambda_k:\mathcal{T}\rightarrow\mathbb{R}\}_{k = 1}^K$. Our prior enforces a relaxed version of the following mutual orthogonality constraint:
\begin{equation}\label{eq:mutualOrtho}
    \langle \lambda_j,\lambda_k \rangle :=  \int_{\mathcal{T}} \lambda_j(s)\lambda_k(s)ds = 0, \quad j\neq k,
\end{equation}
where $\langle \cdot ,\cdot \rangle$ denotes the $\mathbb{L}^2$ inner product between two functions. Let $\mathcal{G}$ denote the product measure induced by $K$ independent zero-mean Gaussian processes with covariance kernels $C_k(\cdot,\cdot),\ k=1,\ldots,K$, and $\mathcal{H}$ denote the product of reproducing kernel Hilbert spaces defined by the covariance kernels \citep{van2008}. We refer to these as {\em parent Gaussian processes}.  We assume the measure induced by $\small{\mbox{NeMO}}$ processes, denoted by $\mathcal{N}$, is absolutely continuous with respect to $\mathcal{G}$. We define $\mathcal{N}$ through the Radon-Nikodym derivative, 
\begin{equation}\label{eq:jointPrior}
    \frac{\partial \mathcal{N}}{\partial \mathcal{G}}(\lambda_1,\ldots,\lambda_K) \propto \exp\bigg(-\frac{1}{2\nu_\lambda}\sum_{k = 1}^K\sum_{j<k}\langle \lambda_j,\lambda_k\rangle^2\bigg).
\end{equation}
The degree to which the mutual orthogonality constraint is met is governed by the penalty parameter, $\nu_\lambda$. 

To study $\mathcal{N}$ with fixed $\nu_\lambda >0$, we define a set of all functions $\{\lambda_1,\ldots,\lambda_K\}$ that are 
$\omega$-far from being mutually orthogonal: $E_\omega :=  \{\lambda_1,\ldots,\lambda_K\in\mathcal{H}\ \mid \ \sum_{ 1}^K\sum_{j<k}\langle\lambda_j,\lambda_k \rangle ^2> \omega \}$. For increasing $\omega$, $E_\omega$ will not contain any functions that are close to mutually orthogonal, quantified through the pairwise sum of squared inner products. Proposition \ref{propNemo3} formalizes the behavior of $\mathcal{N}$ for measurable subsets of $E_\omega$.

\begin{proposition}\label{propNemo3}
For any $\nu_\lambda > 0$ and measurable subset $E_\omega'\subseteq E_\omega$ with non-zero $\mathcal{G}(E_\omega')$, \begin{equation}
    \frac{\mathcal{N}(E_\omega')}{\mathcal{G}(E_\omega')} \leq F\exp\bigg(-\frac{\omega}{2\nu_\lambda}\bigg), \nonumber
\end{equation}
where $F = \big[\mathbb{E}_{\mathcal{G}}\{\exp{(-\frac{1}{2\nu_\lambda}\sum_{ 1}^K\sum_{j<k}\langle \lambda_j,\lambda_k\rangle^2)}\}\big]^{-1}$ is a normalizing constant.
\end{proposition}

Relative to the Gaussian measure, $\mathcal{N}$ assigns probability to $E_\omega'$ that decays exponentially fast as $\omega$ increases, with rate governed by $\frac{1}{2\nu_\lambda}$. This implies that, relative to drawing $\lambda_1,\ldots,\lambda_K$ independently from Gaussian processes, sampling from a $\small{\mbox{NeMO}}$ process will place higher probability on functions that are close to mutually orthogonal.

In propositions \ref{propNemo1} \& \ref{propNemo2}, we present the limiting properties of $\mathcal{N}$ with respect to $\nu_\lambda$.

\begin{proposition}\label{propNemo1}
Let $E :=  \{\lambda_1,\ldots,\lambda_K\in\mathcal{H}\ \mid \ \langle\lambda_j,\lambda_k \rangle \neq 0, \text{ for some } j\neq k \}$.
For any measurable subset $E'\subseteq E$, $\underset{\nu_\lambda\rightarrow 0}{\lim}\mathcal{N}(E') = 0$.
\end{proposition}

\begin{proposition}\label{propNemo2}
For any measurable set $A\subseteq \mathcal{H}$, $\underset{\nu_\lambda\rightarrow \infty}{\lim}\mathcal{N}(A) = \mathcal{G}(A)$.
\end{proposition}

For measurable subsets of non-mutually orthogonal functions, $E'$, $\mathcal{N}$ has measure zero, as $\nu_\lambda\rightarrow0$. On the other hand, for any measurable set, $A$, $\mathcal{N}$ is equal in measure to $\mathcal{G}$, as $\nu_\lambda \rightarrow \infty$. Consequently, for fixed $\nu_\lambda >0$, $\mathcal{N}$ balances Gaussian process-like behavior and mutual orthogonality. As a result, $\small{\mbox{NeMO}}$ processes inherit properties from the covariance functions used to define the parent Gaussian processes, such as magnitude, smoothness, and periodicity. The role of these covariance functions is further discussed in Section \ref{sec:cov}.

Proofs for Propositions \ref{propNemo3}, \ref{propNemo1}, and \ref{propNemo2} are presented in Appendix 1. 

\subsection{Conditional Distribution of $\lambda_k$ given $\{\lambda_j\}_{j\neq k}$} \label{sec:priorLimit}

Assigning  functional factor loadings a $\small{\mbox{NeMO}}$ process prior, $\{\lambda_1,\ldots,\lambda_K\} \sim \mathcal{N}$, we develop a broad class of methods for functional factor analysis in Section \ref{sec:fpca}. A concern is that posterior computation may be difficult due to the analytically intractable normalizing constant of the measure $\mathcal{N}$ defined in Equation \eqref{eq:jointPrior}. However, we bypass this problem by exploiting a simple form for the conditional distribution of $\lambda_k$ given $\{\lambda_j\}_{j\neq k}$ under the joint prior 
$\{\lambda_1,\ldots,\lambda_K\} \sim \mathcal{N}.$
This conditional distribution is stated in Proposition \ref{propNemocond1}. 

\begin{proposition}\label{propNemocond1}
Given $\{\lambda_j\}_{j\neq k}$, $\lambda_k$ is a zero-mean Gaussian process with covariance
\begin{align}
        & C^{\nu_\lambda}_k(s,t) = C_k(s,t) - h_{\Lambda_{(-k)}}(s)^\top\{\nu_\lambda I_{K-1} + H_{\Lambda_{(-k)}}\}^{-1}h_{\Lambda_{(-k)}}(t),\text{ where} \label{eq:relaxedcov} \\
    & h_{\Lambda_{(-k)}}(t) = \int_{\mathcal{T}} C_k(s,t)\Lambda_{(-k)}(s)ds,\ H_{\Lambda_{(-k)}} = \int_{\mathcal{T}}\int_{\mathcal{T}} C_k(s,s')\Lambda_{(-k)}(s)\Lambda_{(-k)}(s')^\top ds ds' \label{eq:integrals}
\end{align}
with $\Lambda_{(-k)}(s) = \{\lambda_1(s),\ldots,\lambda_{k-1}(s),\lambda_{k+1}(s),\ldots,\lambda_K(s)\}^{\top}$.
\end{proposition}

Proposition \ref{propNemocond1} implies that the finite-dimensional conditional distribution of $\lambda_k$ evaluated at a fixed set of points, given $\{\lambda_j\}_{j\neq k}$, is multivariate Gaussian, which is conditionally conjugate to the Gaussian likelihood of discretely observed functional data assumed in the model in Section \ref{sec:fpca}. Conditional conjugacy enables posterior inference through a computationally simple Gibbs sampler, outlined in Section \ref{sec:computation}. The behavior of $\lambda_k\mid \{\lambda_j\}_{j\neq k} \sim \mbox{GP}\{0,C^{\nu_\lambda}_k(\cdot.\cdot)\}$ is governed by the covariance function $C^{\nu_\lambda}_k(s,t) = \text{cov}\{\lambda_k(s),\lambda_k(t)\mid \{\lambda_j\}_{j\neq k}\}$, illustrated in Section \ref{sec:cov}. 

Conditionally, $\small{\mbox{NeMO}}$ processes are related to the orthogonal Gaussian processes of \cite{plumlee2018}, introduced in the context of the design and analysis of computer experiments. Their goal was to develop a non-parametric model for a single function $f(t)$, imposing $f(t)$ is orthogonal to a pre-specified set of functions $g(t) = \{g_1(t),\ldots,g_K(t)\}^\top$, which typically have some physical interpretability. To this end, the authors define orthogonal Gaussian processes through an orthogonal covariance function, 
$    C^\perp_{g}(s,t) = C(s,t) - h_{g}(s)^\top H_{g}^{-1}h_{g}(t).$ 
The authors show $f\sim \mbox{GP}\{0,C^\perp_g(\cdot,\cdot)\}$ is orthogonal to each $g_1(t),\ldots,g_K(t)$ with probability 1. The relationship between $\small{\mbox{NeMO}}$ processes and orthogonal Gaussian processes is stated in Proposition \ref{propNemocond2}.  

\begin{proposition}\label{propNemocond2}
Viewing $\Lambda_{(-k)}=\{\lambda_j \}_{j \neq k}$ as fixed, let $\lambda_k^\perp\mid \{\lambda_j\}_{j\neq k} \sim \mbox{GP}\{0,C^\perp_{\Lambda_{(-k)}}\}$ be drawn from a zero-mean Gaussian process orthogonal to $\Lambda_{(-k)}$. Also, let $\{\lambda_1,\ldots,\lambda_K\} \sim \mathcal{N}$ with $\mathcal{N}$ defined in 
\eqref{eq:jointPrior}. Then, 
    $\lambda_k\mid \{\lambda_j\}_{j\neq k}$ converges in distribution to   $\lambda_k^\perp\mid \{\lambda_j\}_{j\neq k}$ as $\nu_\lambda\rightarrow 0$.
\end{proposition}

While an orthogonal Gaussian process is a prior for a single function orthogonal to a fixed set of functions, $\small{\mbox{NeMO}}$ processes provide a prior for a collection of nearly mutually orthogonal functions. \cite{plumlee2018} restrict attention to orthogonal Gaussian processes with continuous sample paths. Proposition \ref{propNemocond3} follows under an analogous assumption that the parent Gaussian processes used to define $\small{\mbox{NeMO}}$ processes have continuous sample paths. This verifies the desired property that for any $j \neq k$ the inner product between $\lambda_j$ and $\lambda_k$ converges in probability to zero.

\begin{proposition}\label{propNemocond3}
Assuming $\{\lambda_1,\ldots,\lambda_K\} \sim \mathcal{N}$ with $\mathcal{N}$ defined in 
\eqref{eq:jointPrior}, 
    $\langle \lambda_j,\lambda_k\rangle\mid \{\lambda_j\}_{j\neq k}$ converges in probability to $0$ for any $j\neq k$ as $\nu_\lambda\rightarrow 0$.
\end{proposition}

Proofs for Propositions \ref{propNemocond1}, \ref{propNemocond2}, and \ref{propNemocond3} are presented in Appendix 1. 

\subsection{Role of Covariance Parameters}\label{sec:cov}

\begin{figure}[t!]
\begin{center}
 \begin{tabular}{c}
\includegraphics[width = 1.5 in]{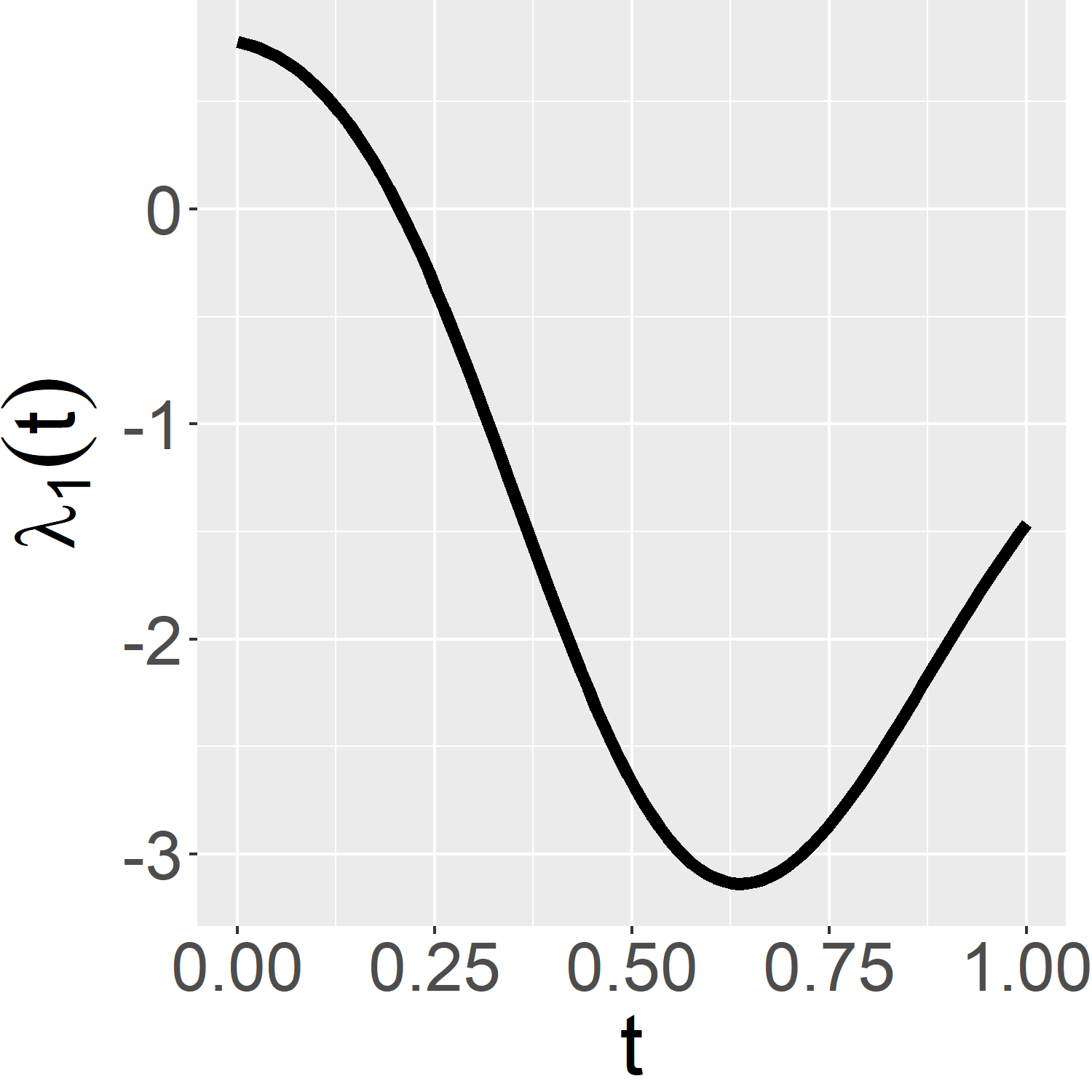} \\
(a)  \\

\includegraphics[width =  1.5in]{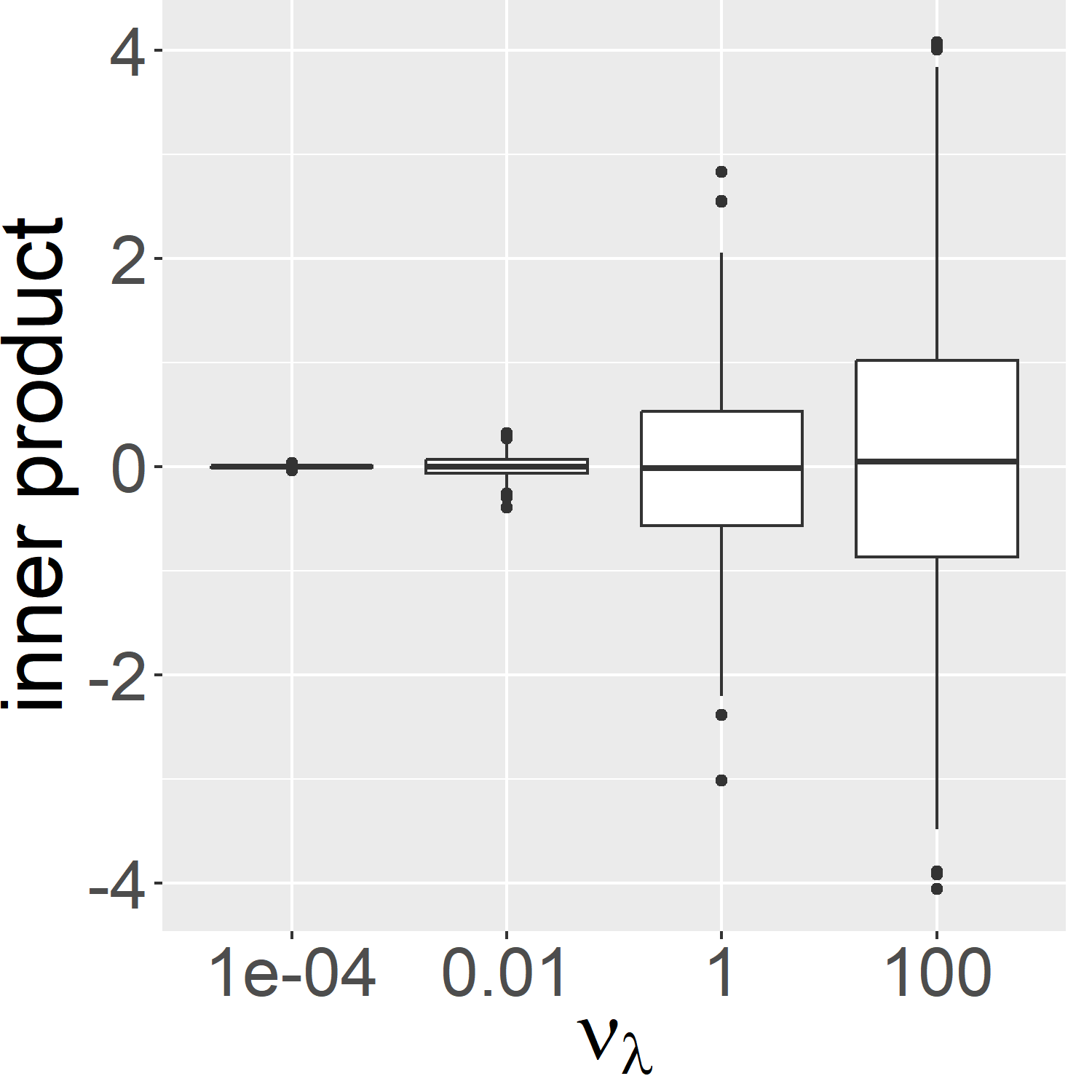}\\
(f)  \\
\end{tabular}
 \begin{tabular}{|cc}
\includegraphics[width = 1.5in]{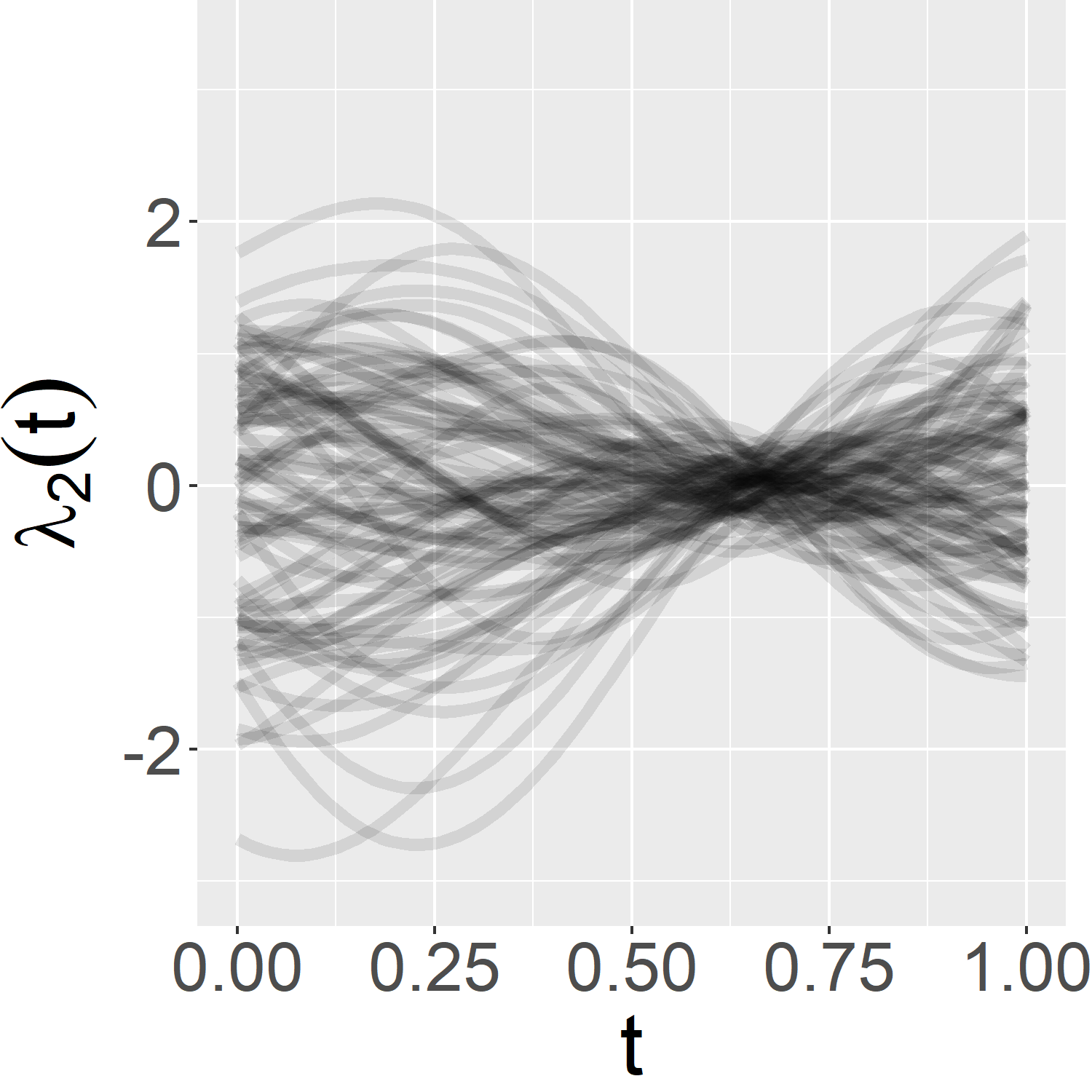} & \includegraphics[width = 1.5in]{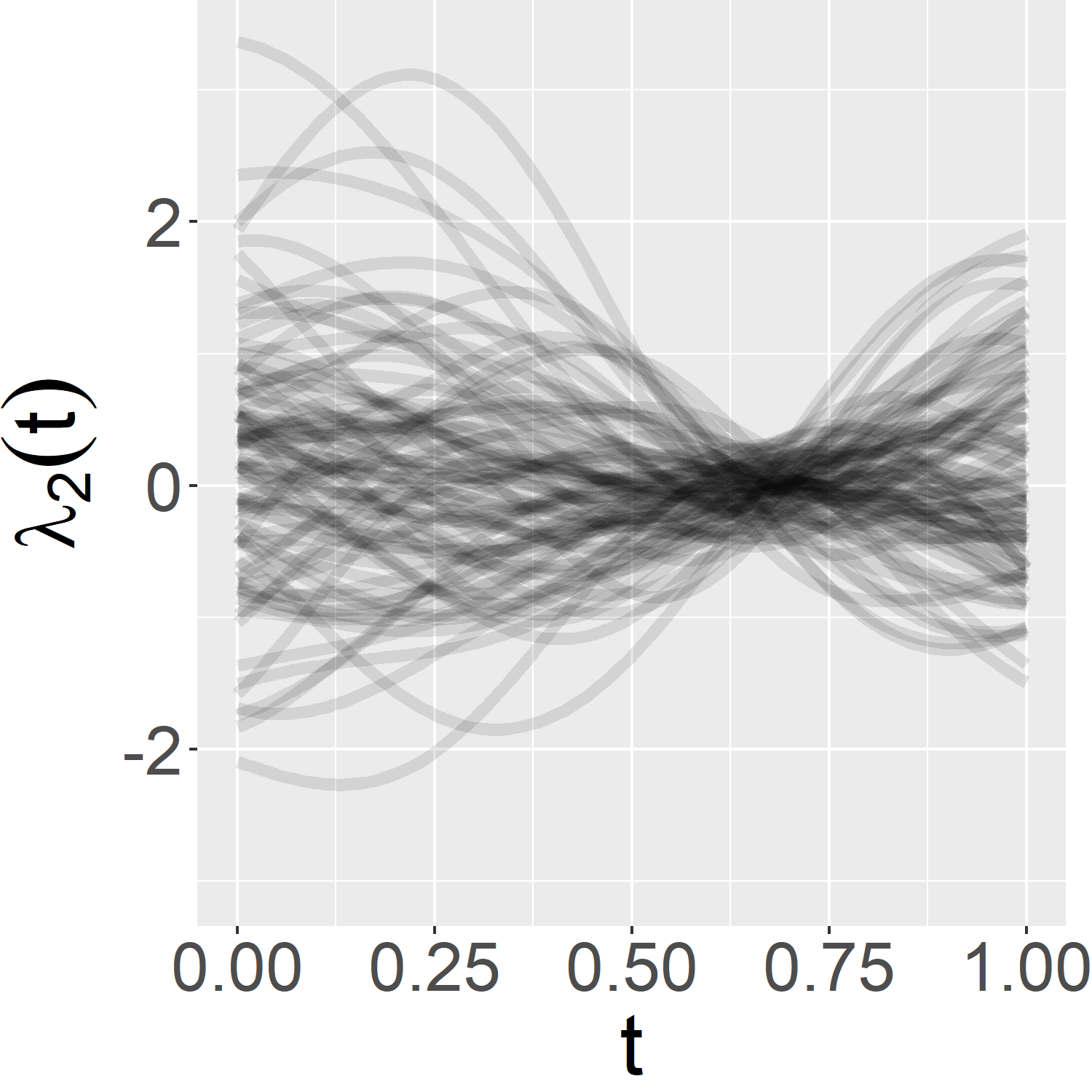} \\
(b) & (c) \\
\includegraphics[width = 1.5in]{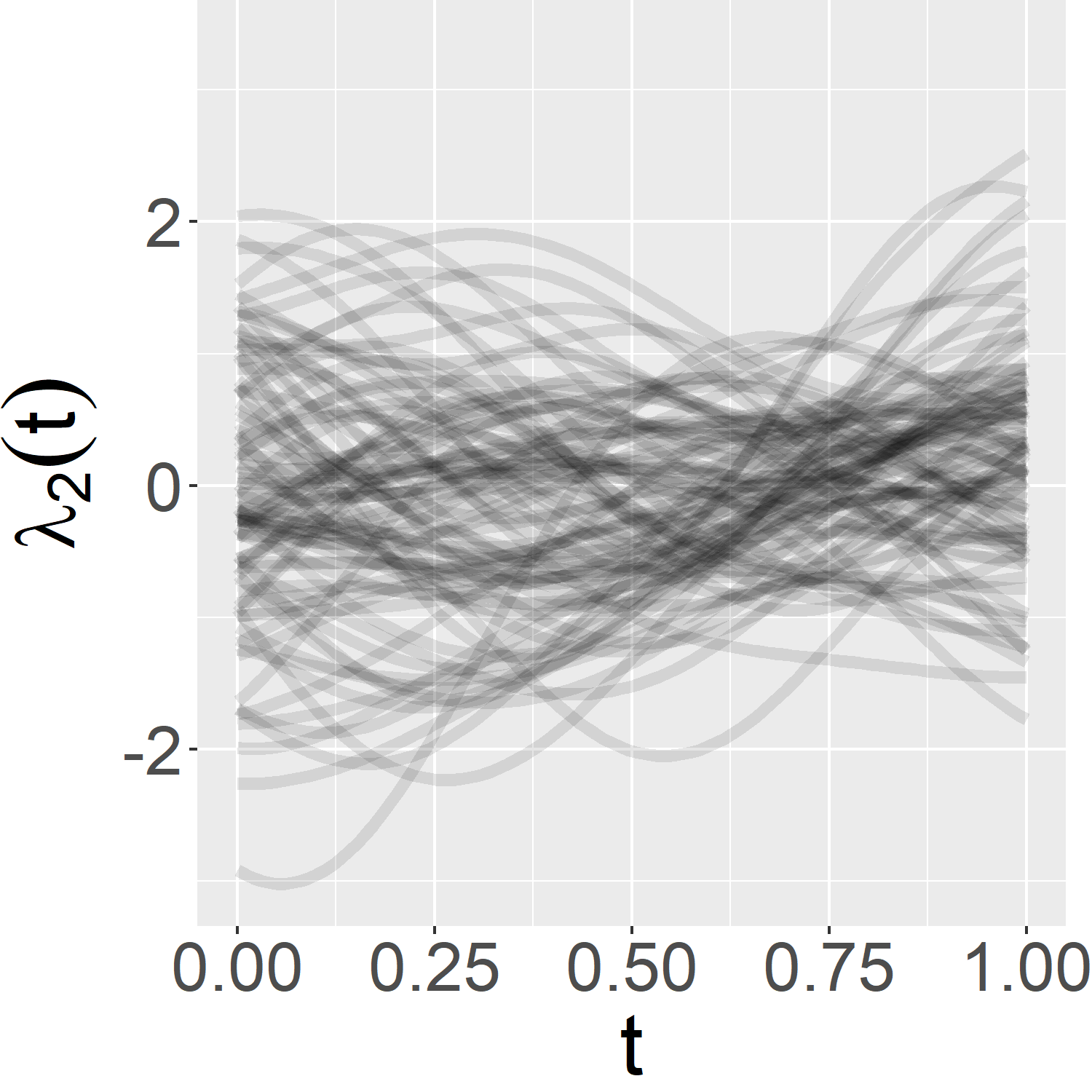} & \includegraphics[width = 1.5in]{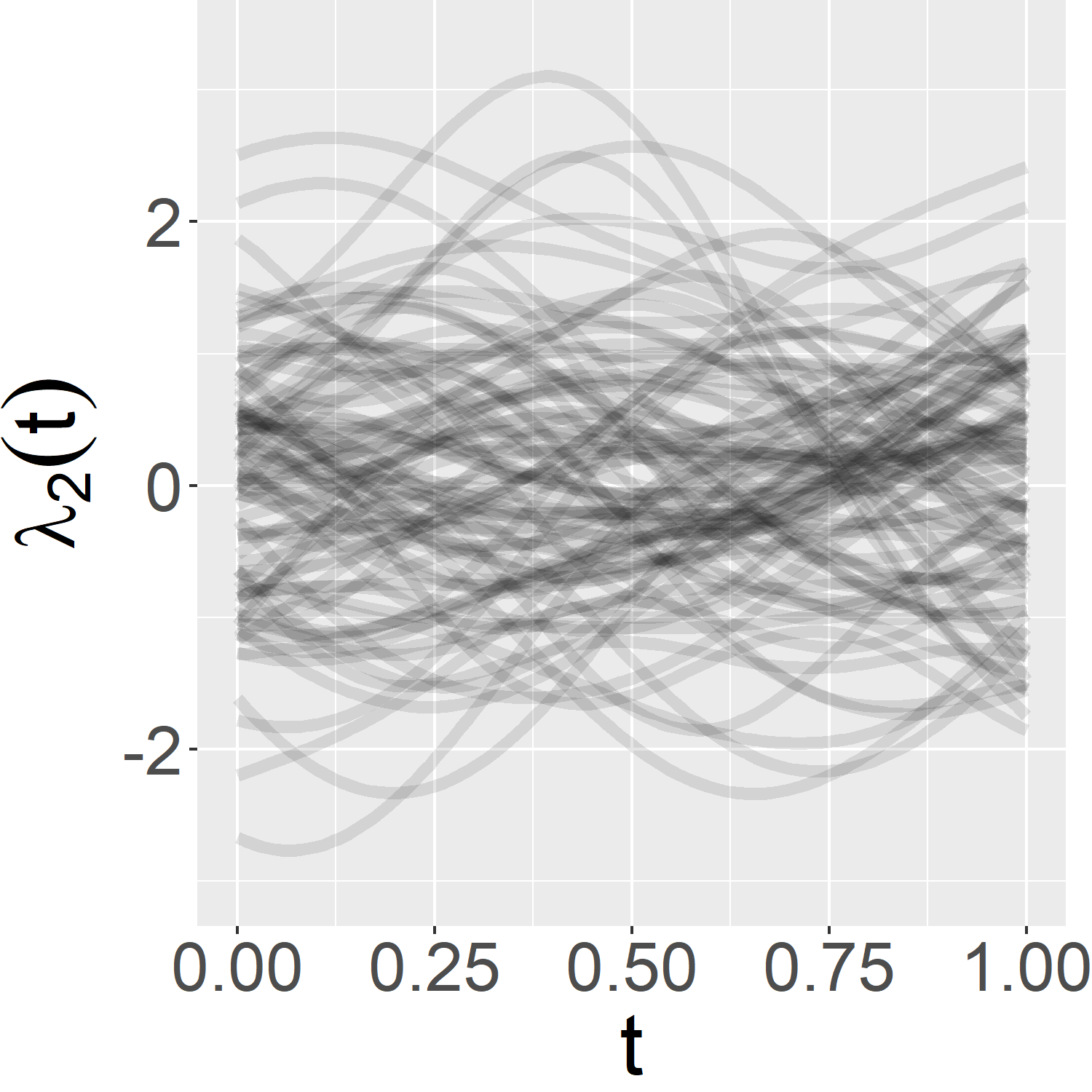}\\
 (d) & (e) \\
\end{tabular}

\caption{(a) Value of $\lambda_1$ on which $\lambda_2$ is conditioned. Realizations of $\lambda_2\mid \lambda_1 \sim \mbox{GP}\{0,C_2^{\nu_\lambda}(\cdot,\cdot)\}$ with  $l_2^2 = .1, \tau^2_2 = 1$ fixed, while (b) $\nu_\lambda = 0.0001$, (c) $\nu_\lambda = 0.01$, (d) $\nu_\lambda = 1$, (e) $\nu_\lambda = 100$. (f) Boxplots of 1000 realizations of $\langle\lambda_1,\lambda_2\rangle$, with $\lambda_2\mid \lambda_1 \sim \mbox{GP}\{0,C_2^{\nu_\lambda}(\cdot,\cdot)\}$ with fixed $l^2_2 = 0.1,\ \tau^2_2 = 1$ and $\nu_\lambda = 0.0001,0.01,1,100$. }\label{fig:simu_inner_prod}
    \end{center}
\end{figure}

While the joint distribution defined in Equation \eqref{eq:jointPrior} is well defined for any choice of covariance kernels $C_k(\cdot,\cdot),\  1,\ldots,K$, we illustrate the model based on the squared exponential: 
\begin{equation}\label{eq:sqExpKernel}
    C_k(s,t)= \tau_k^2\exp\bigg\{-\frac{1}{2l^2_k}(s - t)^2\bigg\},\quad s,t\in\mathcal{T}\times\mathcal{T}.
\end{equation}
In a typical Gaussian process, the scale parameters $\tau_k^2$ determine the spread of realized processes about the mean, while the length-scale parameters $l_k$ determine the smoothness of the realizations. 

For a simple setting when $K=2$, we illustrate the conditional prior of Proposition \ref{propNemocond1} when $\lambda_1$ is sampled from a Gaussian process, shown in panel (a) of Figure \ref{fig:simu_inner_prod}, and $C_2$ is defined according to Equation \eqref{eq:sqExpKernel}. Panels (b)-(e) show realizations from the prior with $l_2^2 = .1, \tau^2_2 = 1$ fixed, and $\nu_\lambda = 0.0001,0.01,1,\ 100$ varying. As discussed in Section \ref{sec:priorLimit}, $\nu_\lambda$ has an important role in enforcing orthogonality. As $\nu_\lambda$ increases, the realizations have typical Gaussian process behavior. However, as $\nu_\lambda$ decreases, the realizations become close to orthogonal to $\lambda_1$. When $\nu_\lambda$ is small, the variance shrinks around $t \approx .65$. This is due to the prior's need to offset the extrema of $\lambda_1$ to enforce orthogonality. Panel (f) shows a boxplot of the inner product between $\lambda_1$ and 1000 realizations of $\lambda_2\mid \lambda_1$ for different values of the penalty parameter $\nu_\lambda$. As expected, varying $\nu_\lambda$ has a drastic effect on the inner product between $\lambda_1$ and the realizations of $\lambda_2\mid \lambda_1$. The inner product is negligible when $\nu = 0.0001$ compared to when $\nu = 100$.

As discussed in Appendix 2, parameters $l_2^2$ and $\tau_2^2$ retain their meaning from the typical Gaussian process setting  of governing the smoothness and spread of the prior, respectively. Changing $l_2^2$ and $\tau_2^2$ with $\nu_\lambda$ fixed has negligible impact on the inner product between the two functions. In our implementations, we infer $l_k^2$ and $\tau_k^2$, while fixing $\nu_\lambda$ as a small value. 

\section{Functional Data Analysis using Nearly Mutually Orthogonal Processes}\label{sec:fpca}

\subsection{Observation Model}\label{sec:obsModel}
Letting $y_i:\mathcal{T}\rightarrow \mathbb{R},\ i = 1,\ldots,n$ denote functional observations, we assume 
\begin{equation}\label{eq:fpca}
    y_i(t) = \mu(t) + \{\lambda_1(t),\ldots,\lambda_K(t)\}^\top\eta_i + \epsilon_i(t),\quad i = 1,\ldots,n.
\end{equation}  
where $\mu(t)$ is a mean process common to all observations, $\lambda_1(t), \ldots, \lambda_K(t)$ are functional factor loadings, $\eta_i\in \Re^K$ are latent factors, and $\epsilon_i(t)$ is a subject-specific error term. Near mutual orthogonality is imposed on $\lambda_1(t), \ldots, \lambda_K(t)$, and the factors $\eta_i\in\Re^K$ are constrained to have mean close to zero across observations as detailed in Section \ref{sec:priorFormulation}.

With sparse and irregular observations, each $y_i$ is observed on an observation-specific grid $\vv{t}_i = (t_{i,1},\ldots,t_{i,m_i})^\top$. Equation \eqref{eq:fpca} can be altered to accommodate this case, 
\begin{equation}\label{eq:fpca_sparse}
    y_i(\vv{t}_i) = O_i\mu(t) + O_i\{\lambda_1(\vv{t}_i),\ldots,\lambda_K(\vv{t}_i)\}^\top\eta_i + \epsilon_i(\vv{t}_i),\quad i = 1,\ldots,n,
\end{equation}
with $\vv{t}$ a common grid of size $m$ merging 
all observation-specific grid points. Use $f(\vv{t}) = \{f(t_1),\ldots,f(t_m)\}^\top$ to denote a vector of function evaluations. The matrix $O_i\in\{0,1\}^{m_i \times m}$ relates the common and observation-specific grids, so that $O_if(t) = f(\vv{t}_i)$. The $j$\textsuperscript{th} row of $O_i$ is a row vector of $0$'s with a $1$ in the index corresponding to the $t_{i,j}$ element of 
$\vv{t}$.

We assume independent Gaussian white noise error processes, $\epsilon_i \sim \mbox{GP}\{0,C_\epsilon(\cdot,\cdot)\}$. For $s,t \in \mathcal{T}$, $C_\epsilon(s,t) = \sigma^2\delta(s-t)$, where $\delta(s - t)$ is the Dirac delta function. For the discretized version of Equation \eqref{eq:fpca_sparse}, we obtain a Gaussian likelihood,
\begin{equation}
    y_i\mid \mu(\vv{t}),\lambda_1(\vv{t}),\ldots,\lambda_K(\vv{t}),\eta_i \sim N_{m_i}\big[ O_i\mu(\vv{t}) + O_i\{\lambda_1(\vv{t}),\ldots,\lambda_K(\vv{t})\}^\top\eta_i,\sigma^2 I_{m_i}\big]. \nonumber
\end{equation}

\subsection{Prior Specification}\label{sec:priorFormulation}

We assume a Gaussian process prior for the mean in Equation \eqref{eq:fpca}, $\mu \sim \mbox{GP}\{0,C_\mu(\cdot,\cdot)\}$, with  
$C_\mu(\cdot,\cdot)$ squared exponential with 
smoothness $l_\mu$ and scale $\tau^2_\mu$. For the error variance, we assume $(\sigma^2)^{-1}\sim\text{gamma}(\alpha_\sigma,\beta_\sigma)$.
We further assume $\{\lambda_1,\ldots,\lambda_K\} \sim \mathcal{N}$. We provide a discussion of the role of $\nu_\lambda$ for posterior inference in Section \ref{sec:nu_sensitivity} for a simulated dataset. While $\nu_\lambda$ should be chosen to be small enough so that the loadings are nearly mutually orthogonal, extremely small values can lead to slow Markov chain Monte Carlo (MCMC) mixing. The value $\nu_\lambda = 1\times 10^{-4}$ reflects a reasonable balance in our experience, and is implemented in Section \ref{sec:cebu}. 

Latent factors are often constrained to have zero mean \citep{ramsay2005,ferraty2006}. Let $\eta_{\cdot k} = (\eta_{1,k},\ldots,\eta_{n,k})^\top$ denote the vector of the $k$\textsuperscript{th} factors for all observations. For equation \eqref{eq:fpca}, we impose a relaxed version of the sum-to-zero constraint, 
    $\sum_{i = 1}^n \eta_{i,k} = \eta_{\cdot k}^\top 1_n = 0,$ 
    for $k=1,\ldots,K$
via a prior that favors $\eta_{\cdot k}$ values with small sums,
\begin{equation}\label{eq:etaPrior}
    \pi(\eta_{\cdot k}) \propto \exp\bigg(-\frac{1}{2}\eta_{\cdot k}^\top\eta_{\cdot k}\bigg)\exp\bigg\{-\frac{1}{2\nu_\eta}(1_n^\top\eta_{\cdot k})^2\bigg\}, 
\end{equation}
independently across $k$. The first term in \eqref{eq:etaPrior} is the kernel of a multivariate Gaussian distribution with identity covariance, while the second term enforces the relaxed sum-to-zero constraint.  The sum-to-zero constraint for the latent factors is specified so that the mean process is able to be identified. The prior allows for us to infer the mean and does not require mean centering prior to data analysis.

Through rearranging the terms in \eqref{eq:etaPrior}, the prior can be recognized as a zero-mean multivariate Gaussian distribution with covariance,
\begin{equation}
    \text{var}(\eta_{\cdot k}) = \bigg(I_n + \frac{1}{\nu_\eta}1_n1_n^\top\bigg)^{-1}  = I_n - \frac{1}{\nu_\eta + n}1_n1_n^\top. \nonumber
\end{equation}
The matrix inversion is computed via Woodbury's formula \citep{harville1998}. The limiting distribution of $\eta_{\cdot k}$, as $\nu_\eta \rightarrow 0$, is stated in Proposition \ref{prop3}. 
\begin{proposition}\label{prop3}
Under equation \eqref{eq:etaPrior}, $\eta_{\cdot k}$ converges in distribution to \\ $(I_n - \frac{1}{n}1_n1_n^\top)N(0,I_n)$ as $\nu_\eta \rightarrow 0$.
\end{proposition}
\noindent Here $(I_n - \frac{1}{n}1_n1_n^\top)$ is the matrix that projects $n$-vectors onto the subspace of vectors that sum to zero. Consequently, when $\nu_\eta$ is small, the prior enforces a close to sum-to-zero constraint. A proof for this Proposition is presented in Appendix 1. 

The limiting case for the prior of $\eta_{\cdot,k}$ defined in Equation \eqref{eq:etaPrior} is a degenerate multivariate Gaussian distribution, since the limiting covariance matrix is singular. Constraint relaxation for this term defines a prior that results in a non-degenerate multivariate Gaussian distribution in an $n-$dimensional Euclidean space. The condition number for the covariance matrix from the relaxed prior is $\frac{\nu_\eta + n}{\nu_\eta}$, which converges to $\infty$ as $\nu_\eta \rightarrow 0$ and 1 as $\nu_\eta \rightarrow \infty$. Consequently, sum-to-zero constraint relaxation in this context provides a balance between computational tractability and constraint enforcement.

\subsection{Generalized Functional Factor Analysis} \label{sec:gfpca}

The above methodology can be extended to allow non-Gaussian observations, such as binary or count data. \cite{hall2008} consider such an extension of the functional data methods in \cite{yao2005} by using latent Gaussian processes.
Bayesian \citep{vdLinde2009} and frequentist \citep{goldsmith2015,wrobel2019,zhong2021} extensions of \cite{james2000} to allow similar generality have also been proposed.

To modify \eqref{eq:fpca} to allow non-Gaussian $y_i(t)$, let 
\begin{equation}\label{eq:gfpca}
    \mathbb{E}\{ y_i(t)\mid \mu,\lambda_1,\ldots,\lambda_K,\eta_i\} = h\big[\mu(t) + \{\lambda_1(t),\ldots,\lambda_K(t)\}^\top\eta_i\big],\quad i = 1,\ldots,n,
\end{equation}
where $h(\cdot)$ is a monotone one-to-one and differentiable link function. Model \eqref{eq:gfpca} implies the following model for irregular observed data, 
\begin{equation}
    \mathbb{E}\{y_i(t_i)\mid \mu,\lambda_1,\ldots,\lambda_K,\eta_i\} = h\big[ O_i\mu(\vv{t}) + O_i\{\lambda_1(\vv{t}),\ldots,\lambda_K(\vv{t})\}^\top\eta_i\big],\quad i = 1,\ldots,n. \nonumber
\end{equation}
In the application in Section \ref{sec:cebu}, we consider binary functional data with $h(\cdot)$ the standard normal cumulative distribution function, $\Phi$.

\subsection{Notes on Computation} \label{sec:computation}

We present an adaptive Metropolis-within-Gibbs algorithm to draw $N$ samples from the joint posterior in Algorithm 1 in Appendix 3. Gibbs steps are available for the majority of parameters, and we use adaptive Metropolis steps with proposal variances that are tuned to meet a target acceptance rate \citep[Algorithm 4]{andrieu2008} for parameters that do not have conditionally conjugate full posterior distributions. In Algorithm 2 of Appendix 3, we present an extension of Algorithm 1 for the case of sparse and irregularly sampled binary functional data.

When evaluating the conditional variance $C_k^{\nu_\lambda}(\cdot,\cdot)$ in Equation \eqref{eq:relaxedcov}, we approximate integrals in Equation \ref{eq:integrals} using Riemann sums. This choice of approximation is convenient for the Gibbs sampler that is implemented, although other choices could be used in different settings. Specifically, we use the approximation $\int_{T}f(s)ds\approx \sum_{l = 1}^mw_lf(t_l)$ relying on the grid $\vv{t} = (t_1,\ldots,t_m)^\top$. The weights $w_l, l=1,\ldots,m,$ are scalars accounting for the spacing of the grid. Letting 
$W = \mbox{diag}(w_1,\ldots,w_m)$, 
we choose $w_1 = \frac{t_2 - t_1}{2},\ w_l = \frac{t_{l + 1} - t_{l-1}}{2},\ l = 2,3,\ldots,m-1, w_m = \frac{t_m - t_{m-1}}{2}$. The approximation error, $|\int_{T}f(s)ds - \sum_{l = 1}^mw_lf(t_l)|$, is bounded above by a value proportional to $\frac{1}{m^2}$. 

The full conditional posterior distribution of $\lambda_k(\vv{t})$ is given by
\begin{eqnarray} 
   \lambda_k(\vv{t})\mid - & \sim & \nonumber N\big[\mathbb{E}\{\lambda_k(\vv{t})\mid -\},\mathbb{V}\{\lambda_k(\vv{t})\mid -\}\big], \\
   \mathbb{V}\{\lambda_k(\vv{t})\mid -\} & = & \bigg\{C_k^{-1}(\vv{t},\vv{t}) + \frac{1}{\nu_\lambda}W\Lambda_{(-k)}(\vv{t})\Lambda_{(-k)}(\vv{t})^\top W + \frac{1}{\sigma^2}\sum_{i = 1}^n\eta_{i,k}^2O_i^\top O_i\bigg\}^{-1}\nonumber \\
   \mathbb{E}\{\lambda_k(\vv{t})\mid -\} & = & \nonumber \frac{1}{\sigma^2}\mathbb{V}\{\lambda_k(\vv{t})\mid -\}\sum_{i = 1}^n\bigg[\eta_{i,k}O_i^\top\Big\{y_i(\vv{t}_i) - O_i\mu(\vv{t})- O_i\sum_{j \neq k}\eta_{i,j}\lambda_j(\vv{t})\Big\}\bigg].
\end{eqnarray}
Hence, the Gibbs step for the loadings is simple. This step is easily modified to generalized functional factor analysis with binary or categorical observations by using data augmentation \citep{albert1993,polson2013bayesian}.

Code to implement functional factor analysis and generalizations on simulated data is available on GitHub: \url{https://github.com/jamesmatuk/NeMO-FFA}. This repository also contains implementations of all examples presented in this paper. The sign and label ambiguity typical of factor analysis is resolved using 
\cite{poworoznek2021}.

\subsection{Selecting the Number of Latent Factors} \label{sec:select_K}

We  take an over-fitted factor modeling approach with  latent space dimension $K_\text{max}$ providing an upper bound on the number of factors, and select a number of latent factors, $K < K_\text{max}$, based on posterior evidence.

Our approach relies on family-wise $(1 - \alpha)$\% simultaneous credible bands for factor loadings using the approach in \cite{ruppert2003}, Chapter 6.5. The simultaneous credible band for any factor loading, $\lambda_k$ is constructed using MCMC samples evaluated on a grid, $\lambda^{[iter]}(\vv{t}),\ iter= 1,\ldots,N$, as follows. First, we compute the pointwise posterior mean and standard deviation of the factor loading, $\widehat{\text{mean}}[\lambda_k(\vv{t})|-] = \frac{1}{N}\sum_{iter = 1}^N\lambda_k^{[iter]}(\vv{t})$ and $\widehat{\text{sd}}[\lambda_k(\vv{t})|-] = \sqrt{\frac{1}{N -1}\sum_{iter = 1}^N(\lambda_k^{[iter]}(\vv{t}) - \widehat{\text{mean}}[\lambda_k(\vv{t})|-] )^2}$, respectively. Based on these posterior summaries, we compute the multiplier for the credible interval, $q_{1-\alpha}$, which is the $1-\alpha$ quantile of 
$$
\max_{t \in \vv{t}}|\frac{\lambda_k^{[iter]}(\vv{t}) - \widehat{\text{mean}}[\lambda_k(\vv{t})|-]}{\widehat{\text{sd}}[\lambda_k(\vv{t})|-] }|
$$
across MCMC samples. Finally, we compute the simultaneous credible interval as 
\begin{equation}
    \widehat{\text{mean}}[\lambda_k(\vv{t})|-] \pm q_{1-\alpha}\widehat{\text{sd}}[\lambda_k(\vv{t})|-]
\end{equation}

In the application presented in Section \ref{sec:cebu}, we adopt a two-stage strategy: i) fitting an initial model to determine the number of factors whose credible intervals exclude the zero function, and ii) refitting the model for parsimony. This approach simplifies the final interpretation, though one could draw conclusions directly from the initial over-specified model, in which several factors loadings are concentrated around the zero function.

To select the number of latent factors, one could alternatively adopt frameworks from the over-fitted factor models literature \citep{bhattacharya2011,Schiavon2022}, 
in which one specifies an upper bound on the number of factors and uses carefully-defined shrinkage priors to effectively delete unnecessary factors.  $\small{\mbox{NeMO}}$ employs a similar strategy, except shrinkage of redundant factors towards zero is encouraged through the near mutual orthogonality constraint rather than through direct shrinking of the magnitude of factor loadings. Information criteria, such as Watanabe-Akaike information criteria \citep{vehtari2017}, offer an additional alternative approach, which requires refitting the model for each choice of the number of factors.

\section{Simulation examples}\label{sec:simulations}

\begin{figure}[!t]
\begin{center}
 \begin{tabular}{ccc}
\includegraphics[width = 1.5 in]{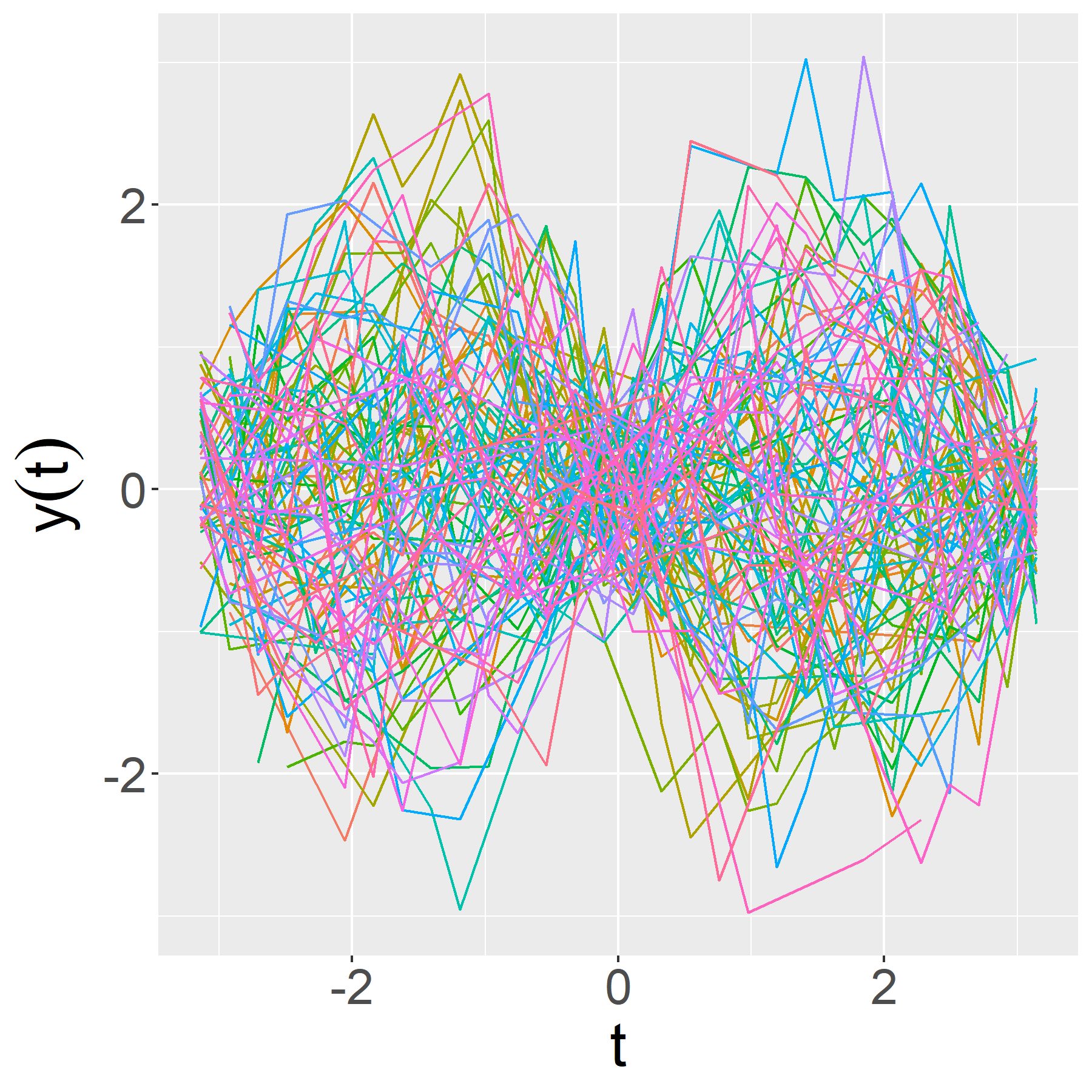} & \includegraphics[width = 1.5 in]{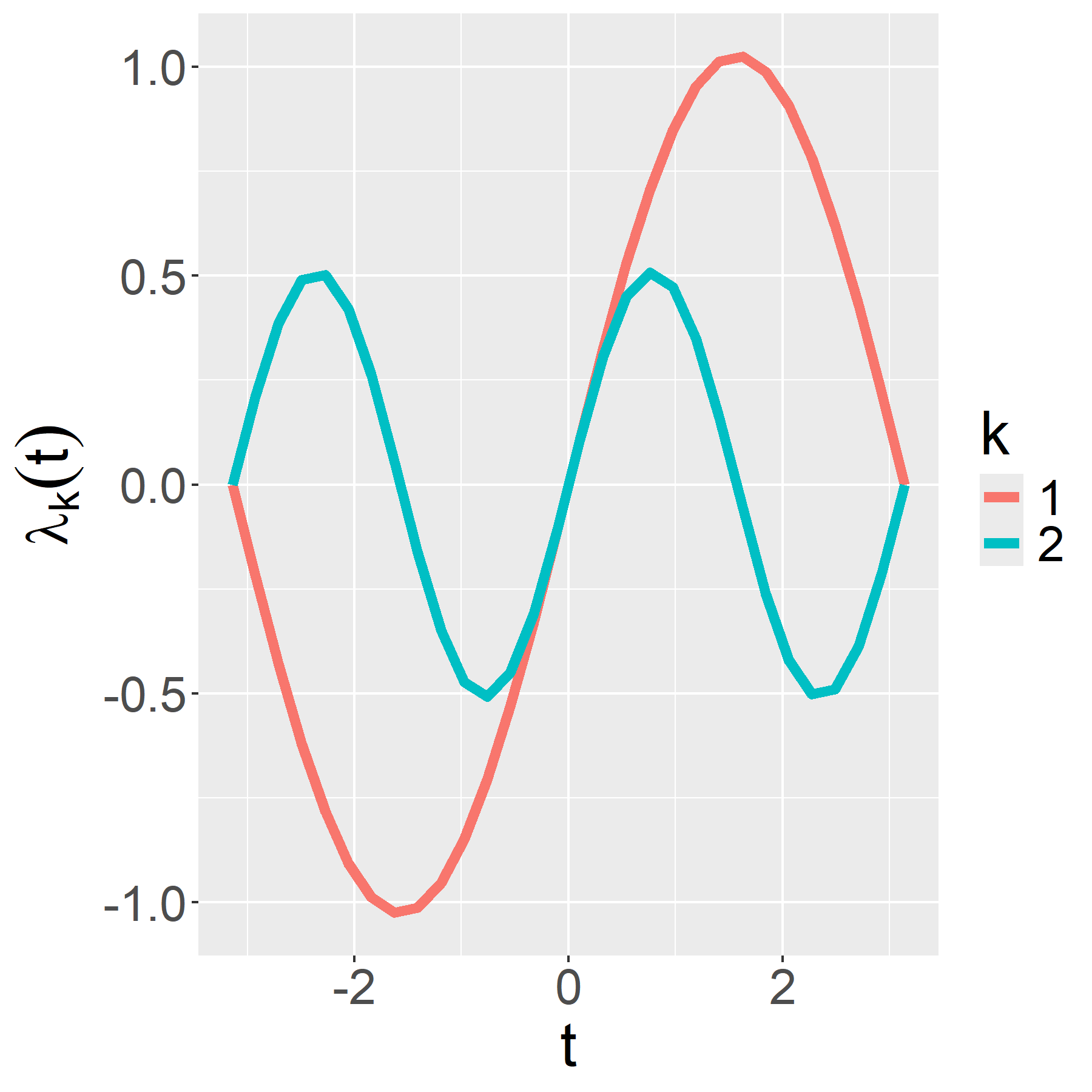}  & \includegraphics[width = 1.5 in]{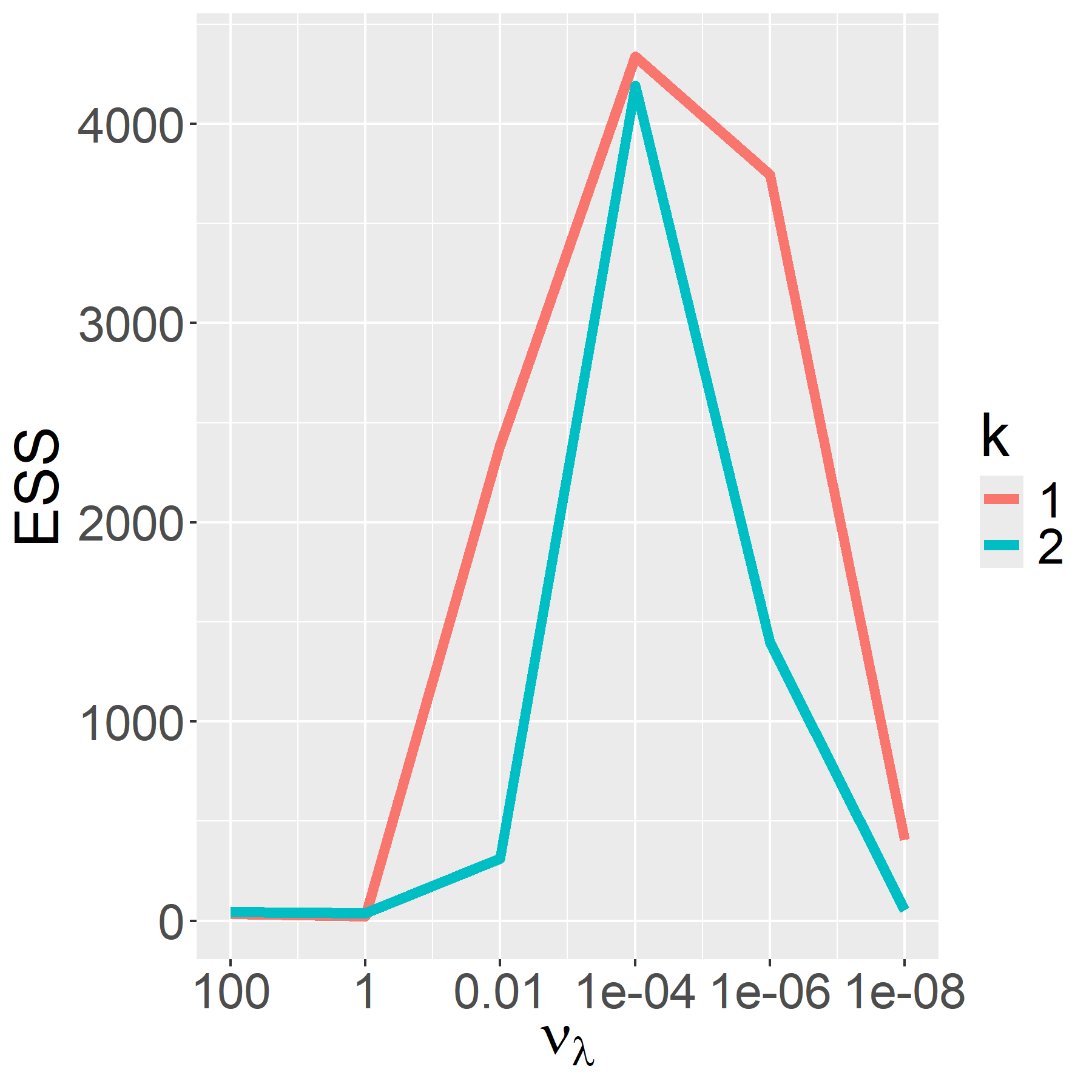}\\
(a) & (b) & (c) \\
\includegraphics[width = 1.5 in]{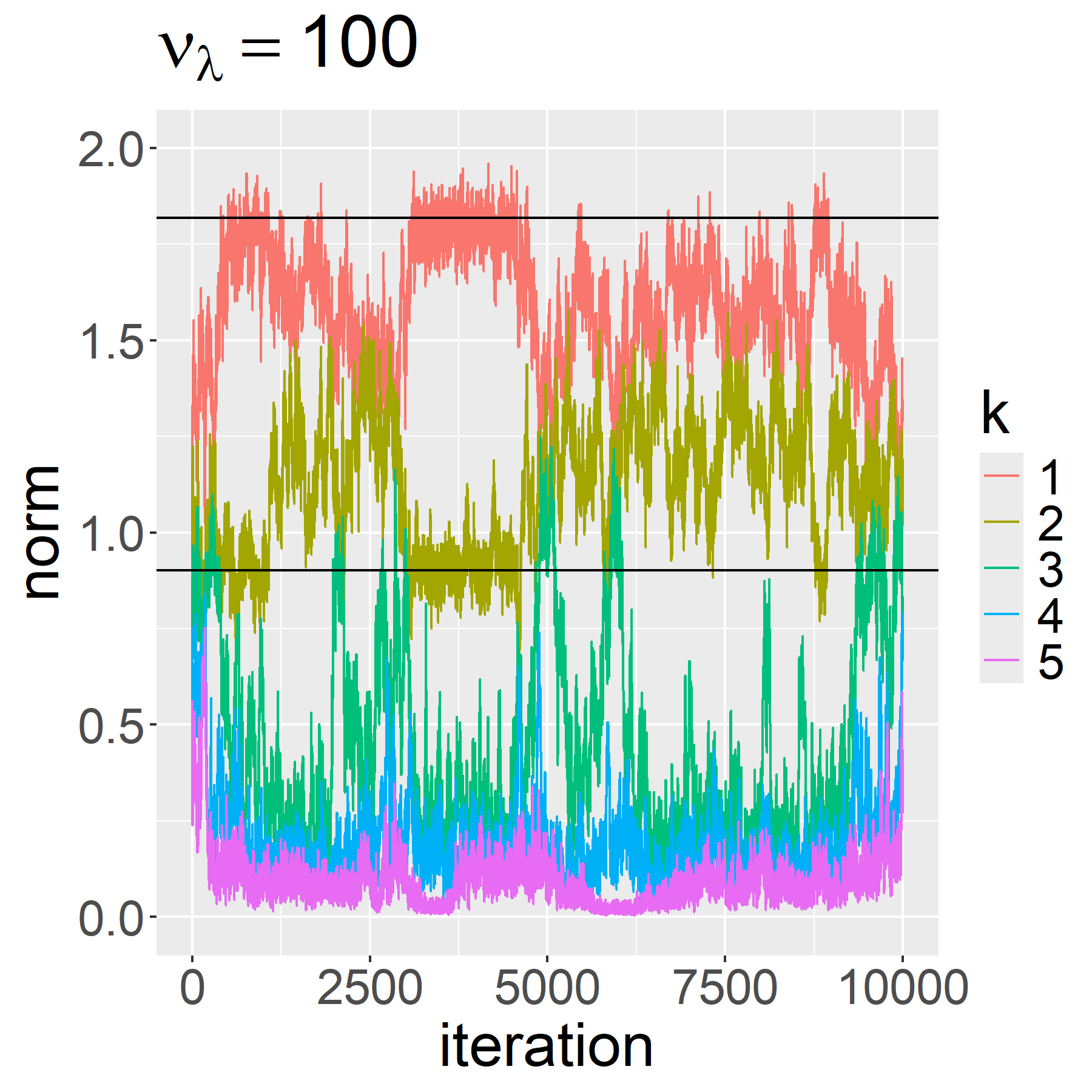} & \includegraphics[width = 1.5 in]{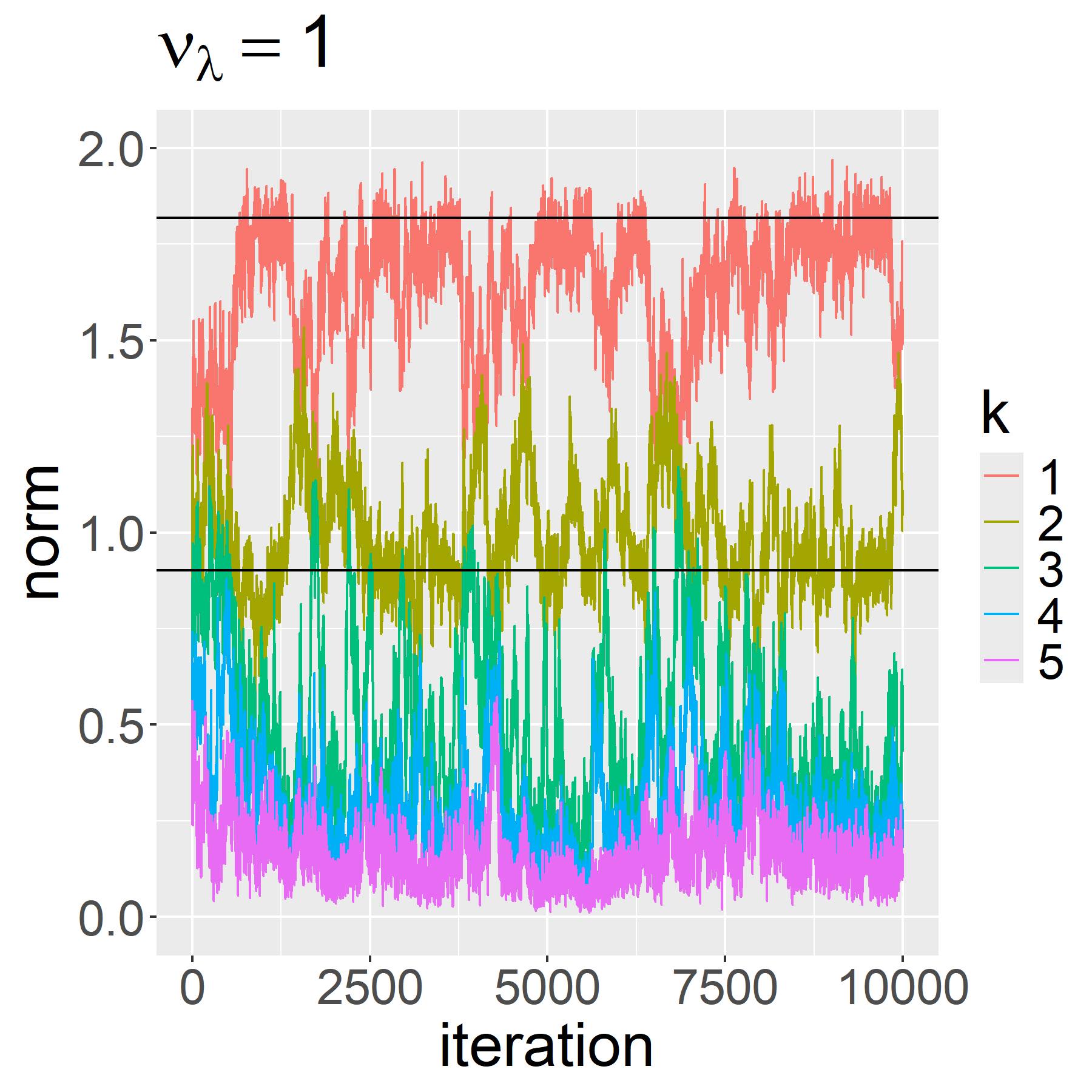}  & \includegraphics[width = 1.5 in]{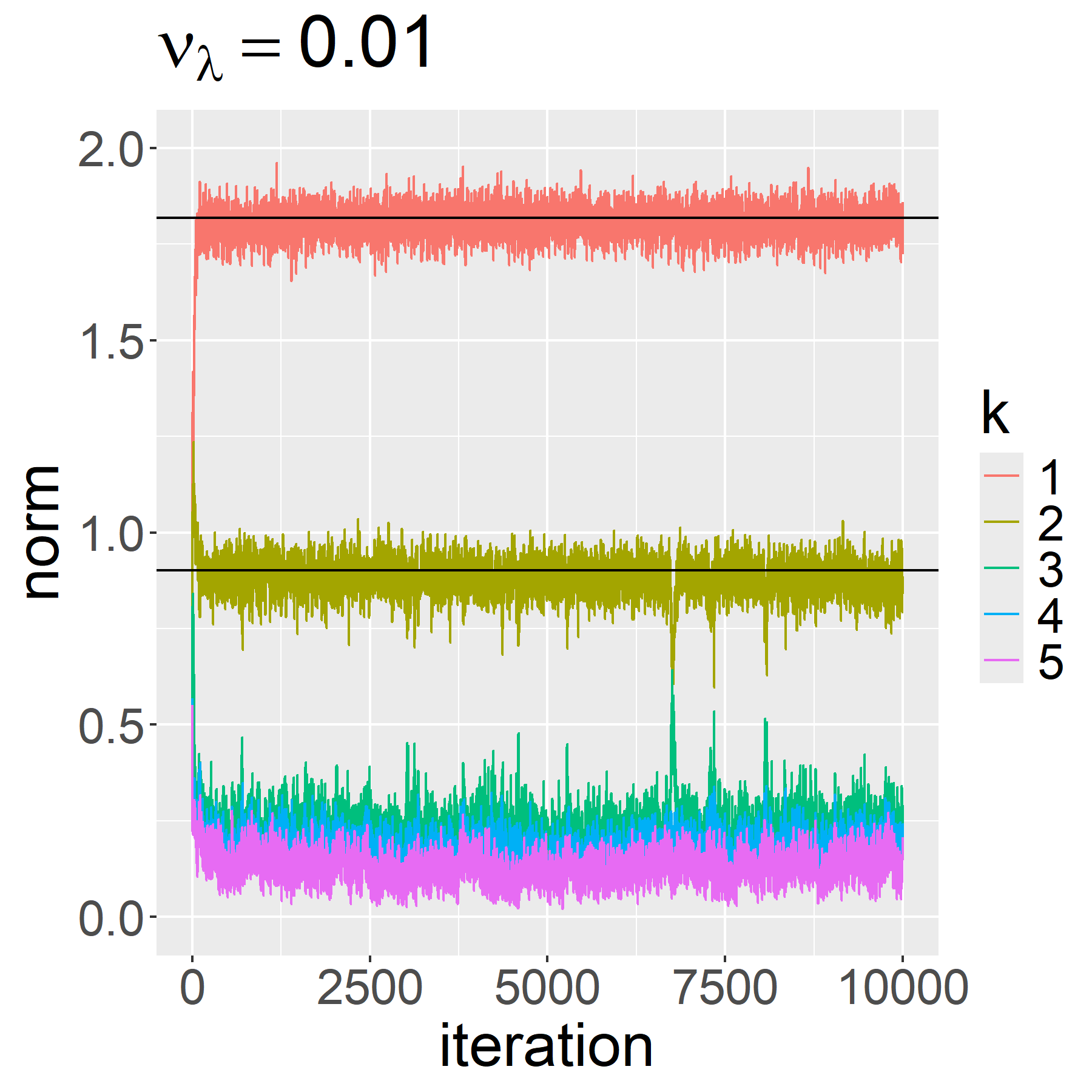}\\
(d) & (e) & (f) \\
\includegraphics[width = 1.5 in]{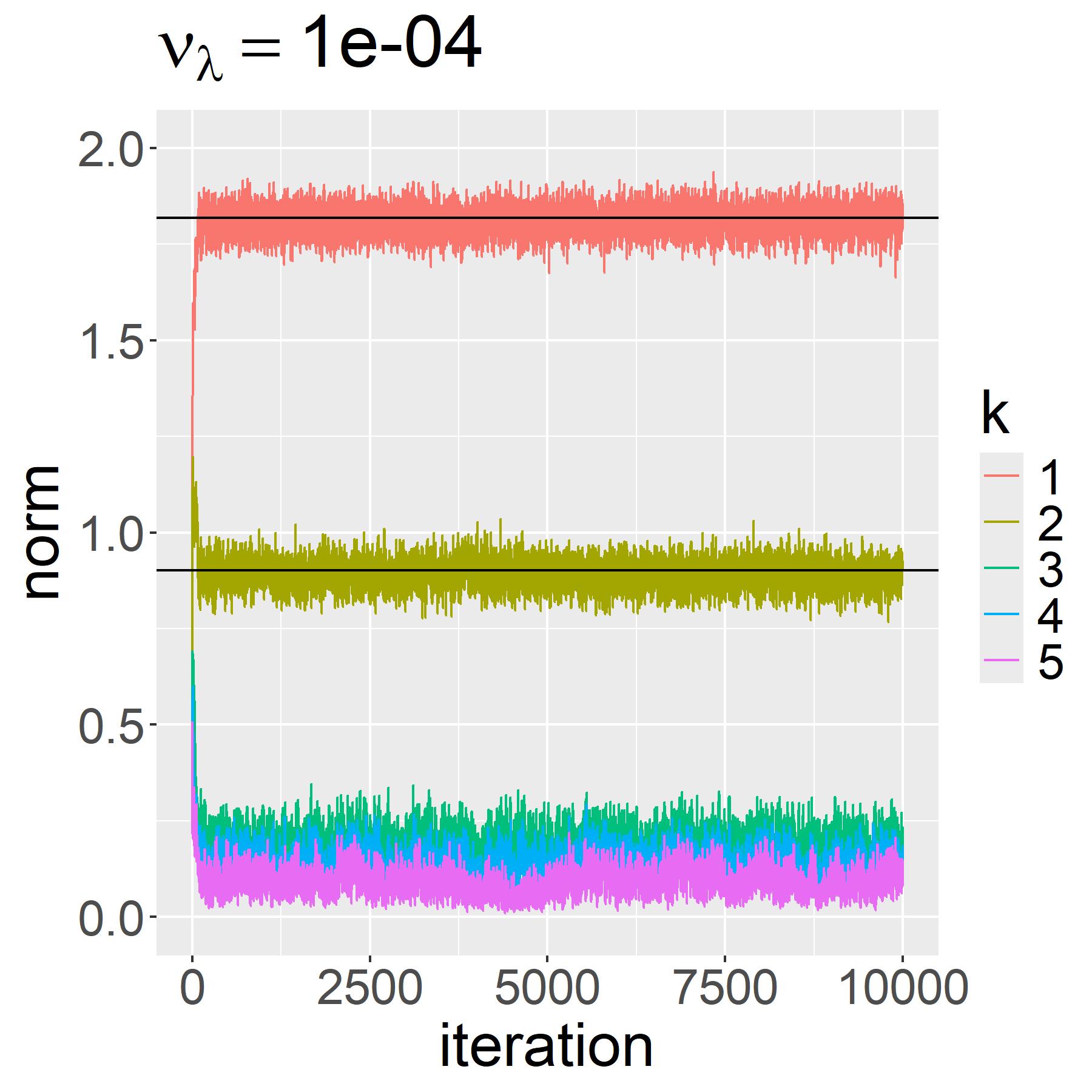} & \includegraphics[width = 1.5 in]{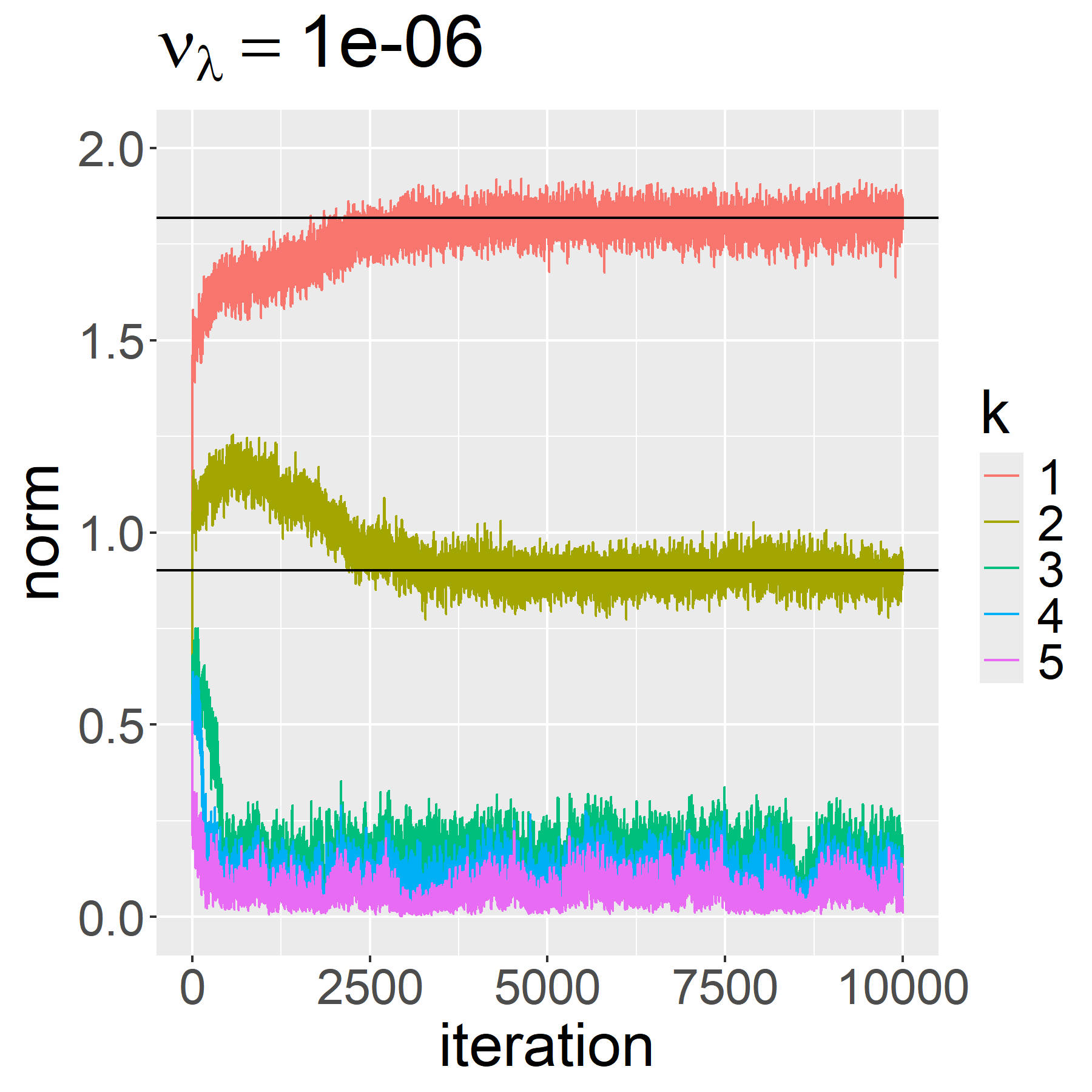}  & \includegraphics[width = 1.5 in]{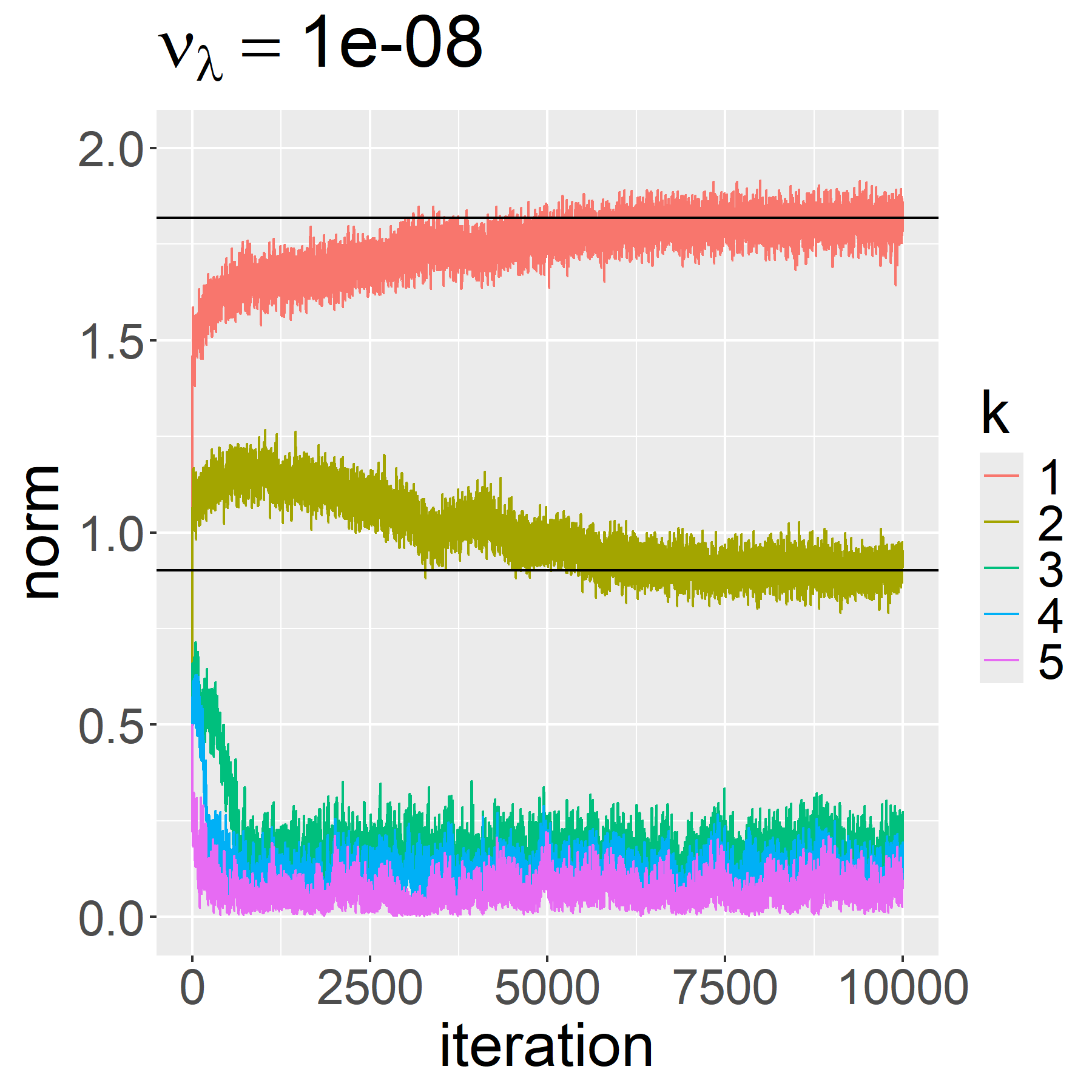}\\
(g) & (h) & (i) \\
\end{tabular}
\caption{(a) Simulated observations. (b) Functional factor loadings that underlie the observations. (c) Effective sample size (ESS) of $\|\lambda_k\|_2,\ k = 1,2$ for different values of $\nu_\lambda$. Trace plots of $\|\lambda_k\|_2,\ k = 1,\ldots,5$ for (d) $\nu_\lambda = 100$, (e) $\nu_\lambda = 1$, (f) $\nu_\lambda = 0.01$, (g) $\nu_\lambda = 1e-04$, (h) $\nu_\lambda = 1e-06$, (i) $\nu_\lambda = 1e-08$.}\label{fig:nu_sensitivity}
    \end{center}
\end{figure}

In this section, we present examples that illustrate the importance of the mutual orthogonality constraint in functional factor analysis models for identifying factor loadings, and the computational advantages of relaxing the constraint through the $\nu_\lambda$ hyperparameter of the $\small{\mbox{NeMO}}$ prior in terms of MCMC mixing properties. 

\subsection{Posterior sensitivity to specification of $\nu_\lambda$}\label{sec:nu_sensitivity}

In this subsection, we use a simple setting to investigate posterior sensitivity for different values of $\nu_\lambda$. Figure \ref{fig:nu_sensitivity} panel (a) shows a simulated dataset that was generated according to the functional factor analysis observation model presented in \eqref{eq:fpca_sparse}, with $K = 2$. Panel (b) shows the $2$ functional factor loadings that underlie the observations in panel (a). These are specified as the orthogonal functions $\lambda_1(t) := \sin(t)$ and $\lambda_2(t) := \frac{1}{2}\sin(2t)$. 

To study the effects of $\nu_\lambda$ on posterior inference, we run several MCMC implementations of our model at different values of $\nu_\lambda = 100,1,0.01,1e-04,1e-06,1e-08$. During these different implementations we fix $K = 5$, and the initialization and the seed for the pseudo-random number generator to simulate MCMC samples from the target posterior. Panels (d)-(i) of Figure \ref{fig:nu_sensitivity} show trace plots of $10,000$ posterior samples of $\|\lambda_k\|_2,\ k = 1,\ldots,5$, with the true values of $\|\lambda_k\|_2,\ k = 1,2$ shown as black horizontal lines. Panels (d)\&(e) display results when $\nu_\lambda$ is relatively large. In this setting, the model does not effectively enforce the mutual orthogonality constraint. This results in the inferred factor loadings splitting the true factors as indicated by the norms of the MCMC samples failing to converge. Panels (h)\&(i) display results when $\nu_\lambda$ is close to zero. In this setting, the model strictly enforces the mutual orthogonality constraint. The posterior samples of the norm of the factor loadings slowly converge to the values used to generate the data while shrinking the other factor loadings close to zero. Panel (c)\&(d) display results for moderate values of $\nu_\lambda$ that balances these two extremes. The trace plots converge relatively quickly to values used to generate the data, as the model shrinks redundant factor loadings to zero while exhibiting trace plots that indicate adequate convergence. 

The benefit of relaxing the mutual orthogonality constraint is further evidenced by Panel (c), which shows the effective sample size of the posterior samples of $\|\lambda_k\|_2,\ k = 1,2$ after the first $5,000$ MCMC iterations. The effective sample size of $\|\lambda_k\|_2,\ k = 1,2$  is small for relatively large values of $\nu_\lambda$ because of factor splitting, while the effective sample size for small values of $\nu_\lambda$ is small because of slow mixing. Moderate values of $\nu_\lambda$ have relatively large effective sample size. In practice, a value of $\nu_\lambda$ could be tuned for a single dataset based on mixing properties in a burn-in phase during sampling.

\subsection{Comparison with Existing Bayesian FFA Approaches}

\begin{table}[]
\resizebox{\textwidth}{!}{
\begin{tabular}{p{2cm}llll}
                            & \mbox{NeMO}                          & FDLM                       & FAST                            & VMP                        \\
\textbf{Run Time (Seconds)} & \textbf{8.8 (8.6, 9.0)}          & 18.9 (18.6, 19.1)          & 32.4 (30.6, 34.9)               & 62.7 (33.2, 153.4)         \\
\multicolumn{3}{l}{\textbf{ESS per 1000 iterations}}                                        &                                 &                            \\
$\mu$                       & \textbf{1146.9 (1134.9, 1168.1)} &                            & 1014 (966.6, 1084.5)            & \multirow{6}{*}{}          \\
$\lambda_1$                 & 650.9 (641.6, 659.7)             & 380 (357.7, 407.4)         & \textbf{1022.9 (958.6, 1082.6)} &                            \\
$\lambda_2$                 & 676.9 (666.5, 694)               & 388.9 (364.7, 415.3)       & \textbf{1009.6 (958.9, 1088.6)} &                            \\
$\lambda_3$                 & 641.4 (627.7, 651.1)             & 392.4 (359.6, 413.2)       & \textbf{1010.3 (954.2, 1064.3)} &                            \\
$\lambda_4$                 & 623.2 (610.8, 638.4)             & 382.1 (357.2, 408.6)       & \textbf{1001.1 (958.9, 1055.3)} &                            \\
$\lambda_5$                 & 603.5 (588.4, 627.8)             & 352.9 (328.9, 374.5)       & \textbf{1014.3 (937.7, 1095.1)} &                            \\
\multicolumn{3}{l}{\textbf{MISE}}                                                           &                                 &                            \\
$\mu$                       & 0.52 (0.36, 0.72)                &                            & 0.66 (0.45, 0.94)               & \textbf{0.50 (0.34, 0.64)} \\
$\lambda_1$                 & \textbf{0.10 (0.07, 0.15)}       & \textbf{0.10 (0.06, 0.14)} & \textbf{0.10 (0.07, 0.13)}      & 0.11 (0.07, 0.16)          \\
$\lambda_2$                 & 0.18 (0.12, 0.23)                & \textbf{0.15 (0.12, 0.22)} & 0.16 (0.13, 0.21)               & 0.17 (0.14, 0.24)          \\
$\lambda_3$                 & 0.19 (0.13, 0.24)                & \textbf{0.16 (0.12, 0.21)} & 0.18 (0.15, 0.22)               & 0.18 (0.14, 0.24)          \\
$\lambda_4$                 & 0.19 (0.16, 0.23)                & \textbf{0.17 (0.14, 0.22)} & 0.25 (0.22, 0.28)               & 0.2 (0.16, 0.4)            \\
$\lambda_5$                 & 0.31 (0.27, 0.36)                & \textbf{0.3 (0.24, 0.35)}  & 0.44 (0.39, 0.53)               & 0.35 (0.28, 0.54)         
\end{tabular}
}
\caption{Summaries of MISE, mESS and run time for the dense example. Each cell displays median (25\textsuperscript{th} percentile, 75\textsuperscript{th} percentile) across 100 replicates.}\label{table:dense_sim}
\end{table}

In this subsection, estimation and MCMC mixing properties are compared with existing Bayesian FFA approaches. We compare against three different methods: i) \cite{kowal2021} (FDLM) specifies a spline-based representation of loadings, where orthogonality is enforced through a conditionally defined prior. Model fitting is achieved through a Gibbs sampler. ii) \cite{Sartini2026} (FAST) uses polar expansion to enforce orthogonality through spline coefficients, while relying on Hamiltonian Monte Carlo using Stan \citep{carpenter2017stan}. iii) The `BayesFPCA' R package \citep{bayesFPCA2023} implements a variational-Bayes approach (VMP). The model does not enforce constraints, but post-processing techniques are used to identify latent factors and factor loadings.

Data for this simulation are generated using Equation \eqref{eq:fpca} with $n = 100$. The factor loadings are five orthonormalized B-spline basis functions, and the latent factors have mean zero and standard deviations, 5, 3, 2, 1, 0.5. A non-linear mean function is set to $\mu(t) = 12 + 5\log(t + 0.1) + 2t^2,\ t\in [0,1]$. Gaussian noise is added to the smooth observations with mean zero and noise standard deviation consistent with a root signal-to-noise ratio of 3. We present results under dense and sparse settings, assuming the number of latent factors is $K=8$ for all models. In the dense setting, each functional observation is sampled on the same grid with 30 time points. In the sparse setting, these functions are evaluated at 10 uniformly distributed locations. In the sparse setting only $\mbox{NeMO}$ and FDLM are compared, as FAST and VMP are not readily applicable in this setting. FDLM does not directly infer a mean function, so a pointwise mean of the raw data is subtracted from observations prior to model fitting. The competing approaches do not provide an estimate for the number of active latent factors, so we focus on $\mbox{NeMO}$'s ability to exclude redundant factors using the framework detailed in Section \ref{sec:select_K}. For each setting, the simulation was repeated 100 times.

Comparisons focus on MCMC efficiency and estimation accuracy. Estimation accuracy is quantified using mean integrated squared error (MISE) calculated as the integral of the squared difference between an estimate and the value of a functional parameter used for data generation. To enhance comparability, estimates across all approaches were normalized. MCMC efficiency is quantified using multivariate effective sample size (mESS), defined in \cite{vats2019}, and total algorithm run time. All MCMC based approaches were run for 2,000 total iterations, while discarding the first 1,000 iterations. All implementations were run on the same Windows desktop with 32.0 GB RAM and $12^\text{th}$ Gen Intel(R) Core(TM) i9-12900 processor.

\begin{table}[]
\begin{tabular}{lll}
                            & \mbox{NeMO}                           & FDLM                       \\
\textbf{Run Time (Seconds)} & \textbf{8.3 (8.2, 8.4)}         & 19.8 (19.5, 20)            \\
\multicolumn{3}{l}{\textbf{ESS per 1000 iterations}}                                       \\
$\mu$                       & \textbf{703.7   (684.4, 733.3)} &                            \\
$\lambda_1$                 & \textbf{469.2 (450.9, 484.9)}   & 157.2 (145.4, 168.3)       \\
$\lambda_2$                 & \textbf{475.1 (451.4, 491.7)}   & 149.4 (139.7, 159.2)       \\
$\lambda_3$                 & \textbf{479.8 (457.3, 498.3)}   & 151.7 (138.8, 160.4)       \\
$\lambda_4$                 & \textbf{509.7 (477.2, 561.5)}   & 148.9 (138.5, 157.9)       \\
$\lambda_5$                 & \textbf{616.6 (571.3, 657.6)}            & 143.6 (133.7, 152.8)       \\
\multicolumn{3}{l}{\textbf{MISE}}                                                          \\
$\mu$                       & \textbf{0.61 (0.49, 0.74)}      &                            \\
$\lambda_1$                 & 0.43 (0.25, 0.61)               & \textbf{0.21 (0.14, 0.41)} \\
$\lambda_2$                 & \textbf{0.55 (0.36, 0.7)}       & 0.65 (0.51, 0.75)          \\
$\lambda_3$                 & \textbf{0.54 (0.33, 0.69)}      & 0.73 (0.65, 0.78)          \\
$\lambda_4$                 & \textbf{0.46 (0.36, 0.61)}      & 0.85 (0.75, 0.95)          \\
$\lambda_5$                 & 1.06 (0.84, 1.19)               & \textbf{0.97 (0.81, 1.07)}
\end{tabular}
\caption{Summaries of MISE, mESS and run time for the sparse example. Each cell displays median (25\textsuperscript{th} percentile, 75\textsuperscript{th} percentile) across 100 replicates.}\label{table:sparse_sim}
\end{table}

Tables \ref{table:dense_sim} and \ref{table:sparse_sim} summarize the simulation results for the dense and sparse setting, respectively. Across both settings, estimation results are generally comparable given the quartiles of MISE for each of the functional parameters. In the dense setting, FDLM has slightly lower MISE for the factor loadings relative to competitors, while in the sparse setting there are mixed results between \mbox{NeMO} and FDLM. In terms of mean estimation, \mbox{NeMO} and VMP have slightly lower MISE. This is potentially due to the fact that these approaches impose additional constraints on latent factors to aid mean identifiability. In terms of computational efficiency, \mbox{NeMO} had the fastest total run time across both settings while attaining adequate mESS. In the dense setting \mbox{NeMO} identified 5 latent factors in 64\% and 4 latent factors in 36\% of the replicated simulations. In the sparse setting, \mbox{NeMO} generally identified a more conservative number of latent factors, selecting 3, 4, and 5 latent factors 63\%, 36\% and 1\% of the replicated simulations, respectively. $\small{\mbox{NeMO}}$ may underselect latent factors that describe relatively low levels of variability, as the model infers there is insufficient posterior evidence that they are non-zero. In both the dense and sparse setting, $\small{\mbox{NeMO}}$ has favorable performance in terms of ratio of mESS to run time, while maintaining comparable estimation accuracy. These simulations highlight the computational efficiency of our proposed approach.

In Appendix 4, we present additional simulated examples. For the remainder of the paper, we present an in-depth application that leverages a Bayesian hierarchical model using functional factor analysis with a $\mbox{NeMO}$ prior. 

\section{Cebu Longitudinal Health and Nutrition Survey}  \label{sec:cebu}

\subsection{Study design and motivation}

The Cebu Longitudinal Health and Nutrition Survey is a multi-generational study. The focus is on maternal, child, and environmental factors and their connection to health, social, and economic outcomes for mother-child pairs in the metropolitan area of Cebu, Philippines. Refer to \url{https://cebu.cpc.unc.edu/}. In this work, we focus on studying the association between weight dynamics in early childhood and environmental and health related covariates.  Previous work \citep{adair2011} found that scalar covariates, such as height of the mother, sex of the child, area (urban or rural) where the family lives, season of birth (rainy or dry), and longitudinal covariates, such as breastfeeding status and illness, are associated with weight dynamics in early childhood. We base inference on $n = 2898$ mother-child pairs of the $3327$ pairs originally enrolled in the study, where subjects were excluded due to child death, the family moving out of the Cebu metropolitan area, refusal to participate after the baseline survey, or if there was no covariate or response data recorded. Our goal is to specify a Bayesian hierarchical model using $\mbox{NeMO}$ process priors that overcomes the analytic challenges of flexibly capturing non-linear weight dynamics, capturing the regression relationships between weight and the covariates, while accounting for missingness in the response and longitudinal covariates. 

\begin{figure}[t!]
 \begin{tabular}{c}
\includegraphics[width = 1.5 in]{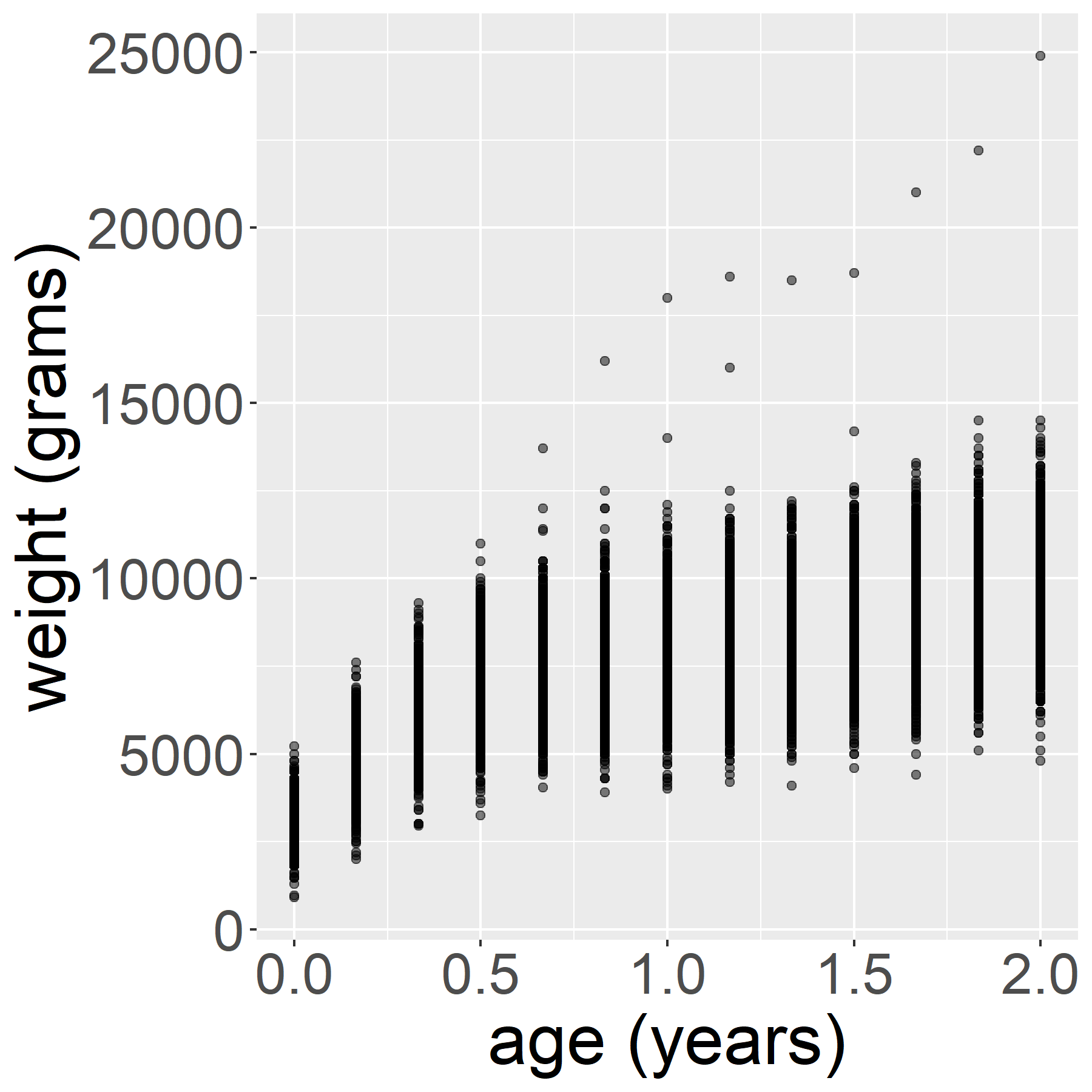} \\ (a) \\\includegraphics[width = 1.5 in]{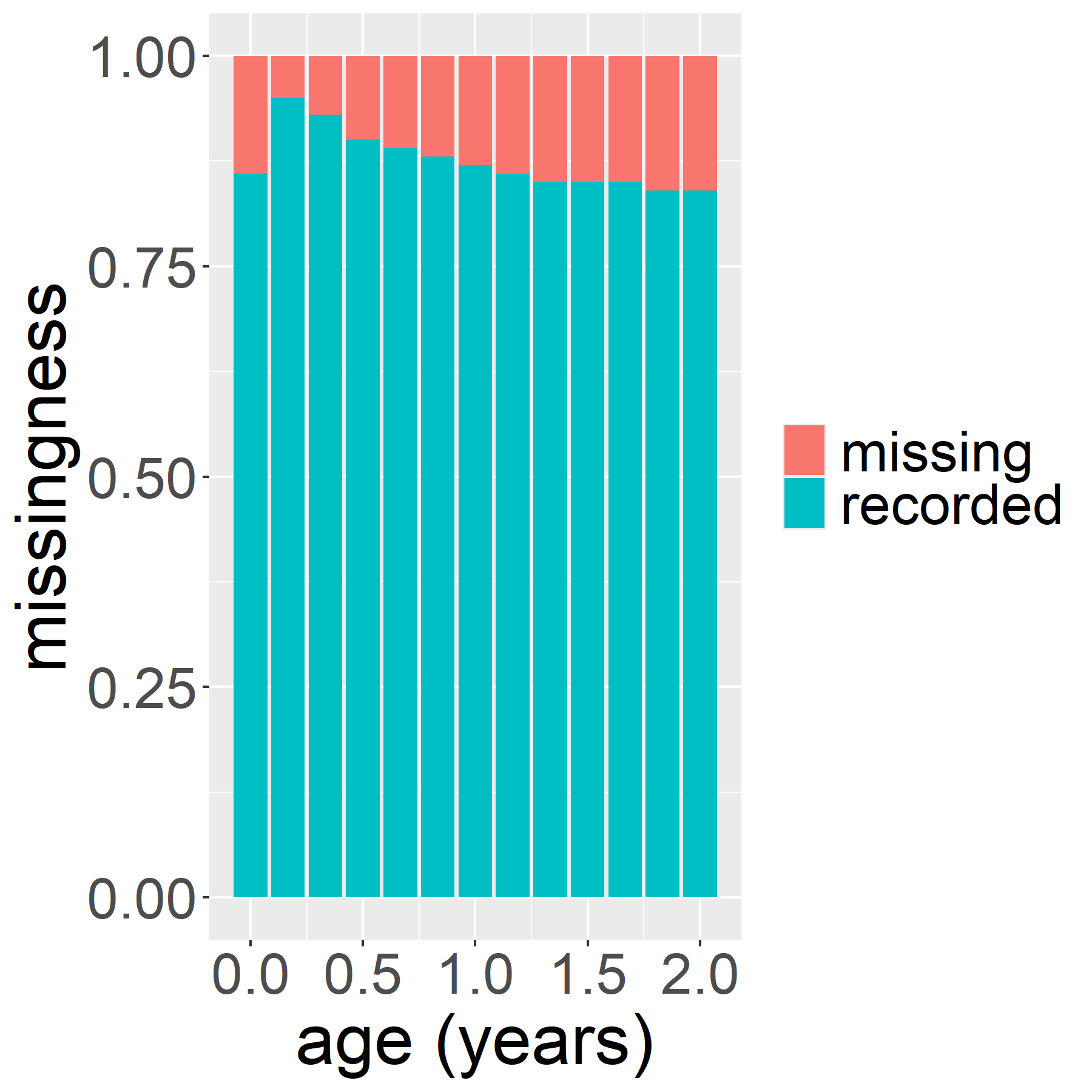} \\
(b) 
\end{tabular}
 \begin{tabular}{|cc}
\includegraphics[width = 1.5 in]{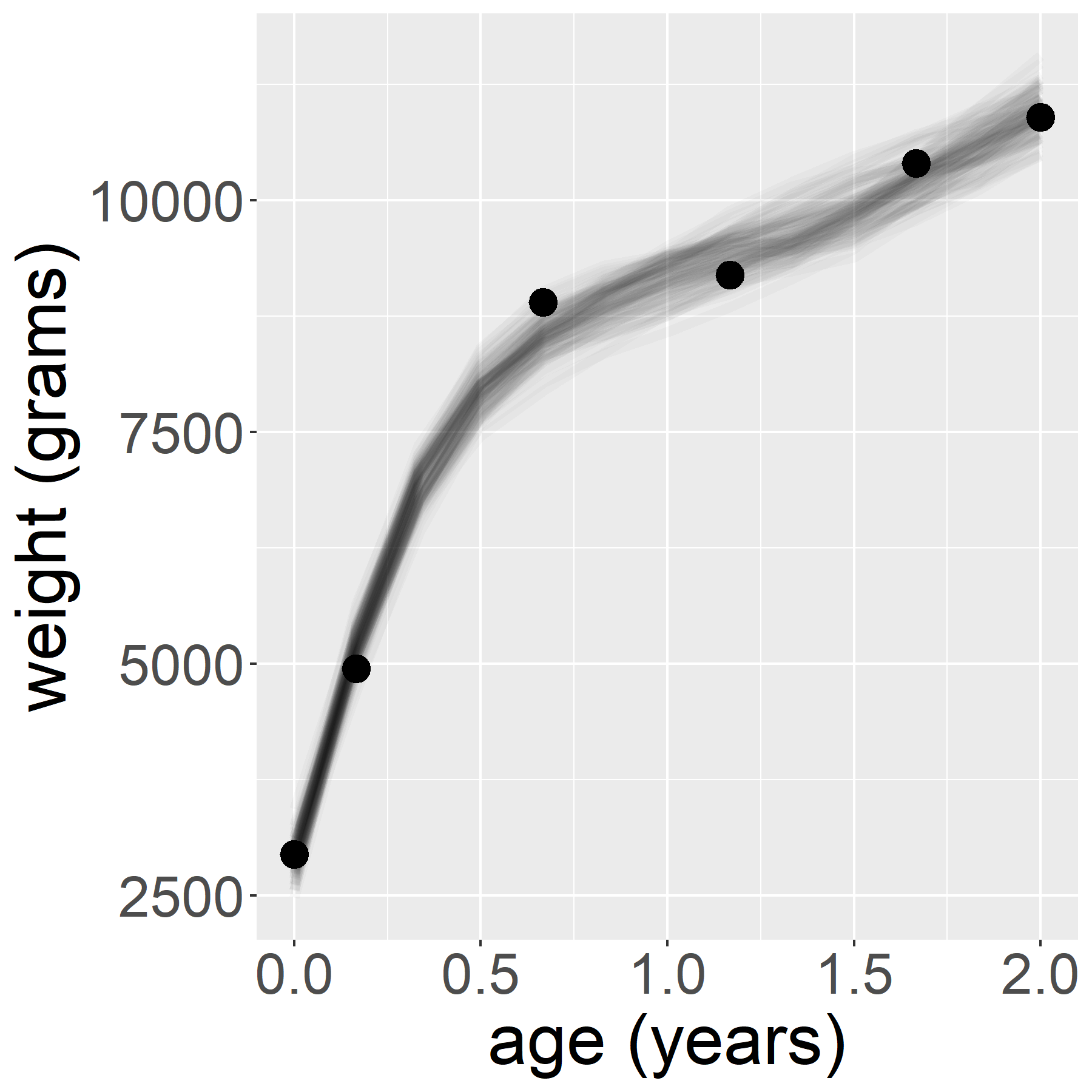} & \includegraphics[width = 1.5 in]{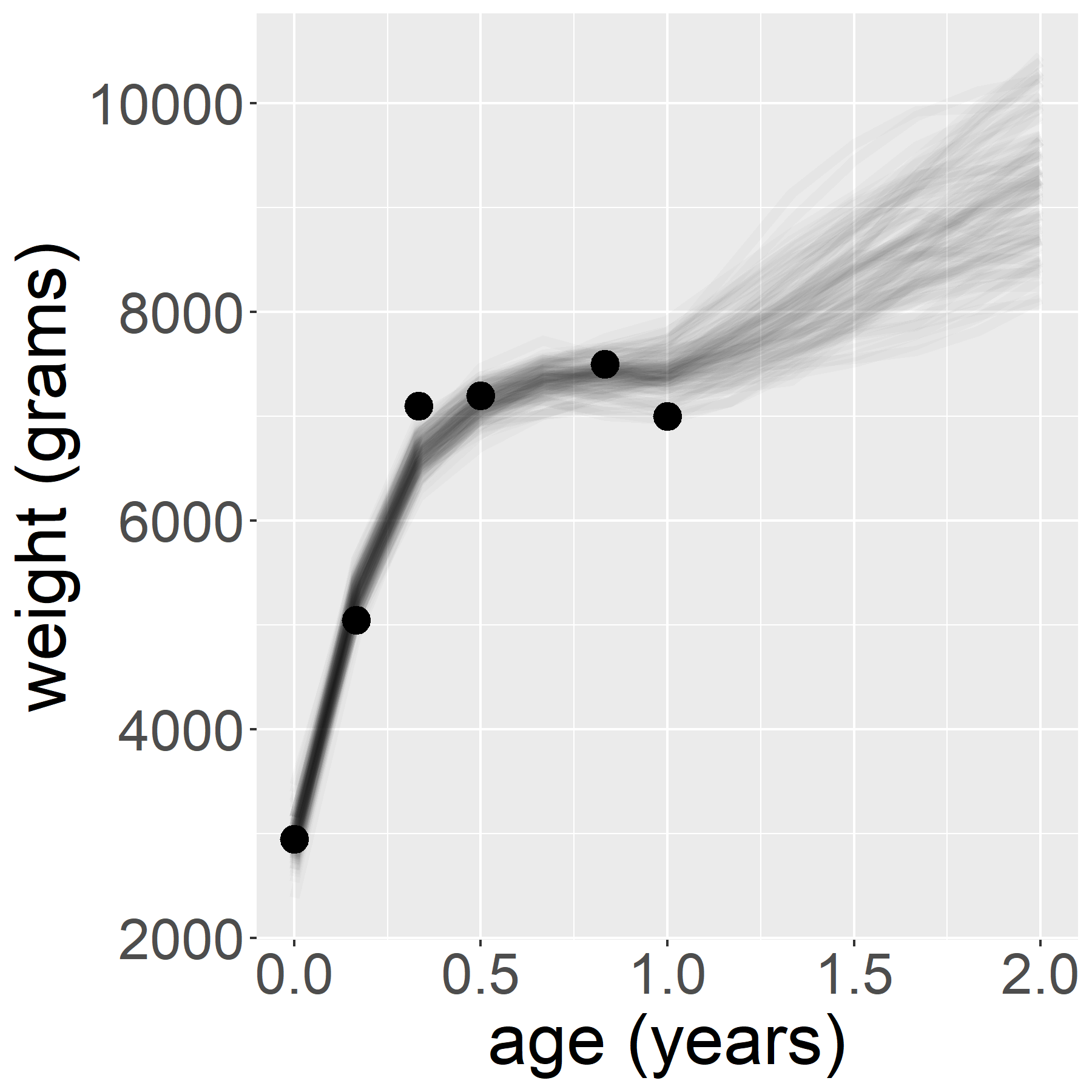}  \\ (c) & (d) \\ \includegraphics[width = 1.5 in]{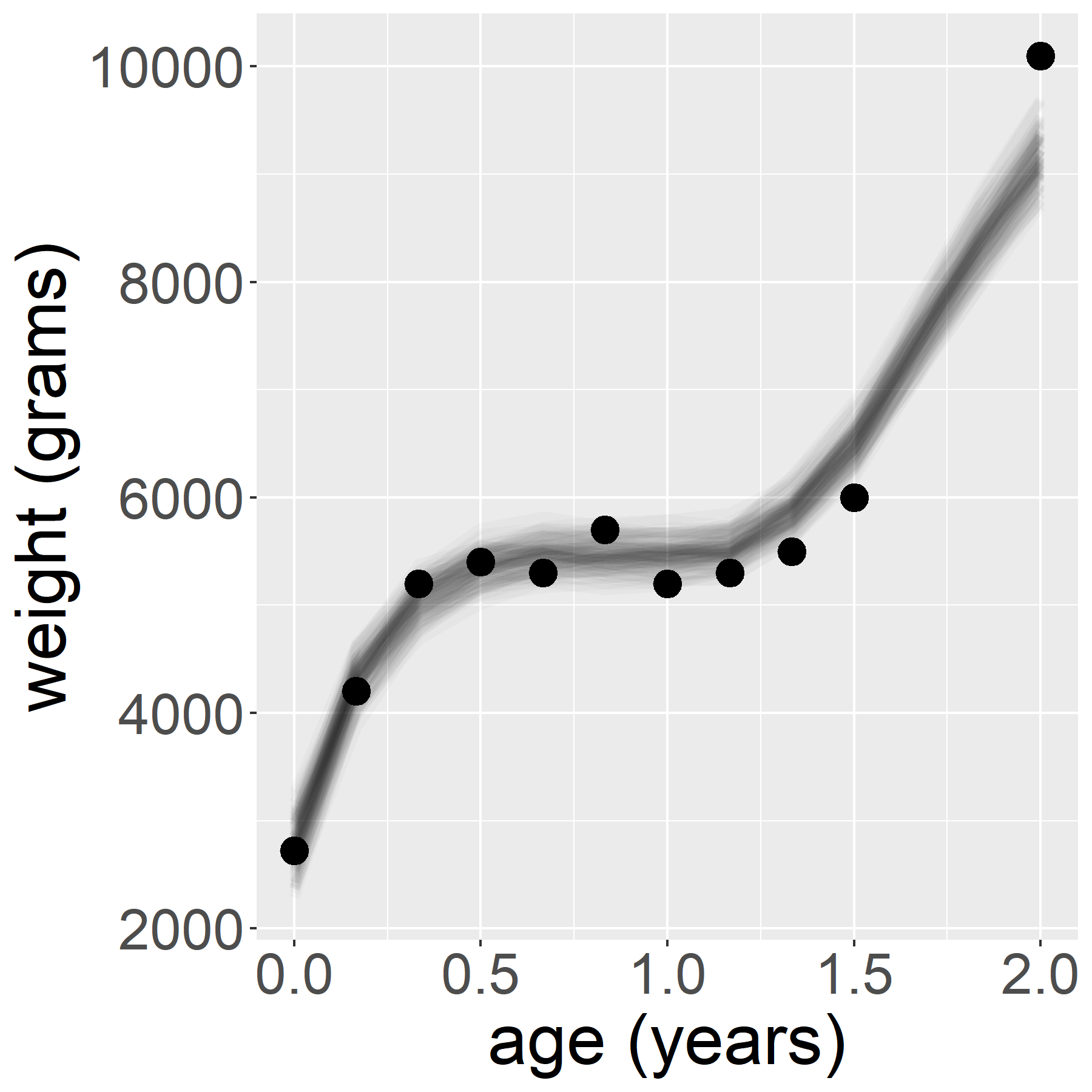} & \includegraphics[width = 1.5 in]{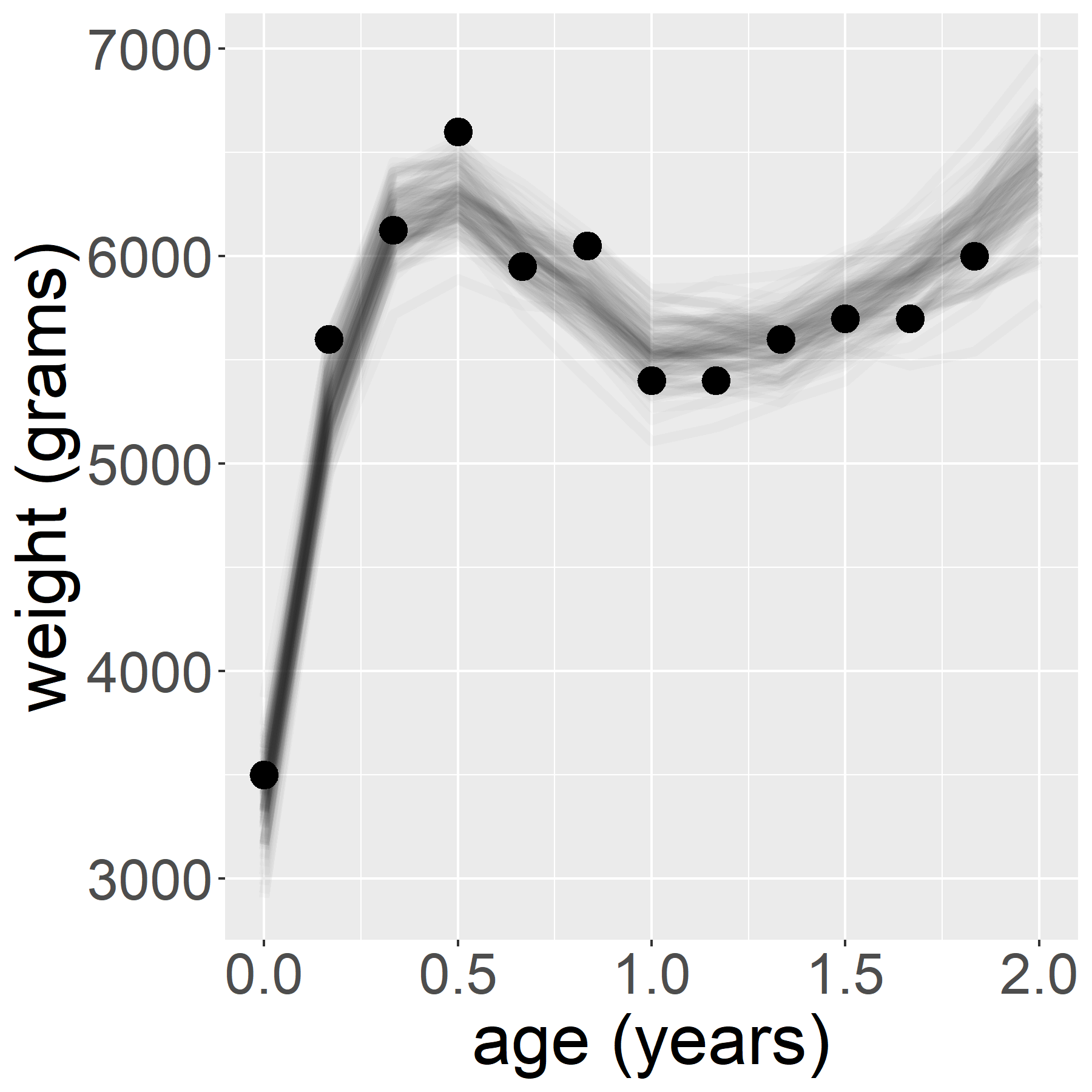} \\
 (e) & (f)\\  
\end{tabular}
\caption{(a) Weight by age for $n = 2898$ children in the Cebu Longitudinal Health and Nutrition Survey. (b) Proportion of missing weight data by age. (c)-(f) Posterior samples of fitted weight trajectories, $\mu(\protect\vv{t}) + \{\lambda_1(\protect\vv{t}),\ldots,\lambda_5(\protect\vv{t})\}^\top\eta_i + \sum_{j = 1}^4\int_{T}\beta_j(s,\protect\vv{t})z_{i,j}(s)ds$ overlaid on $y_i(\protect\vv{t}_i)$ for $i = 421,\ 1626,\ 2205,\ 2601$.
}\label{fig:cebu_data}
\end{figure}

\begin{figure}[t]
\begin{center}
 \begin{tabular}{ccc}
\includegraphics[width = 1.5 in]{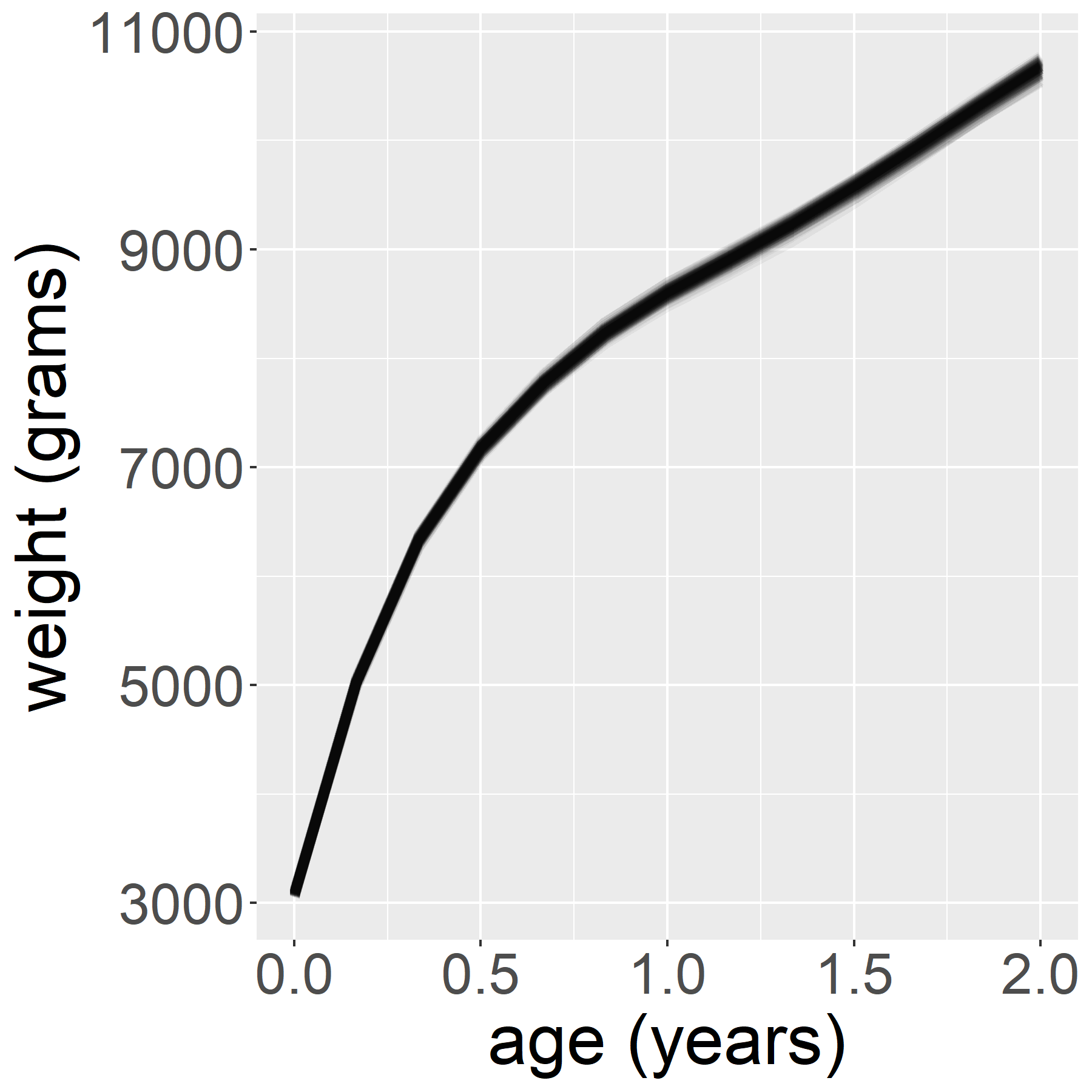} & \includegraphics[width = 1.5 in]{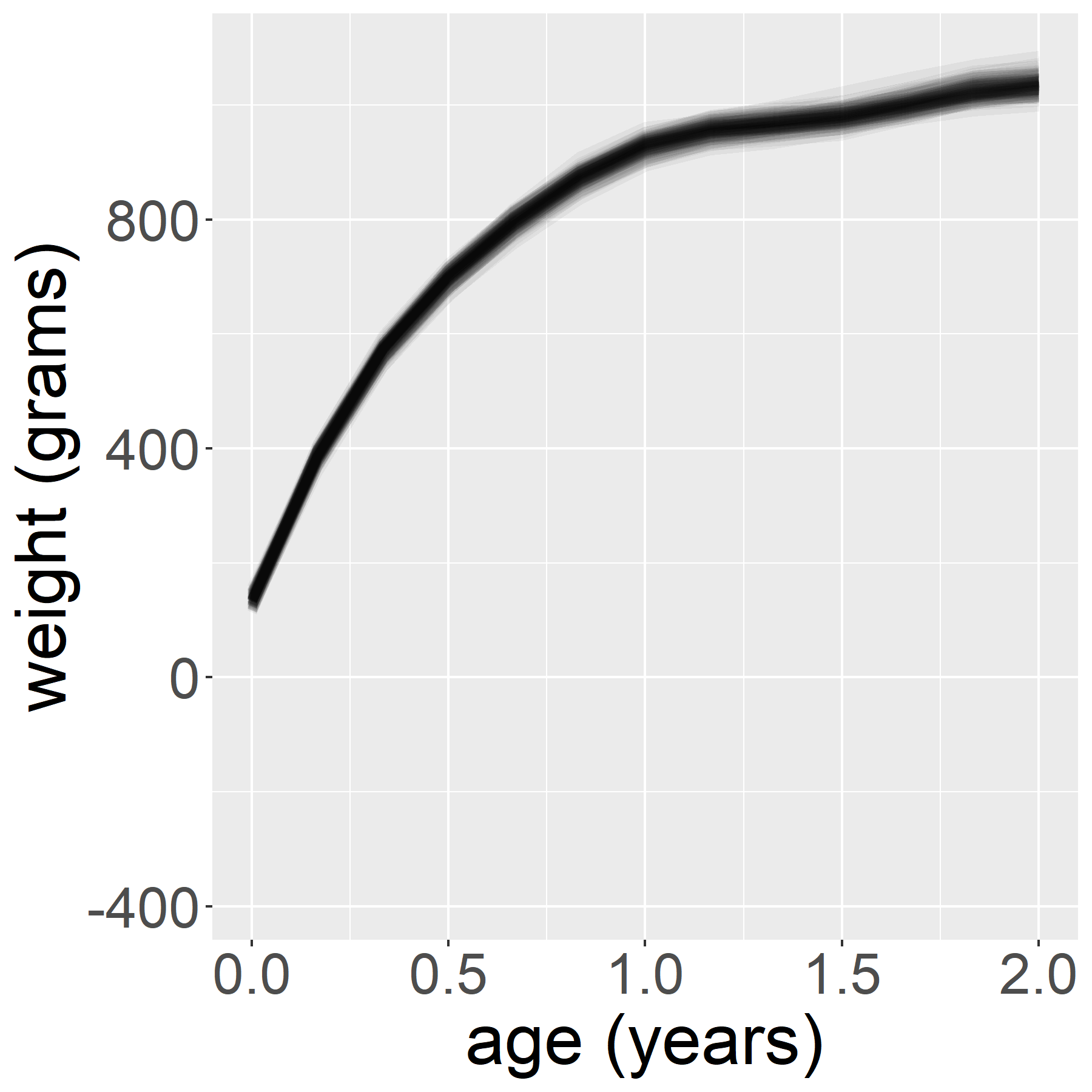}& \includegraphics[width = 1.5 in]{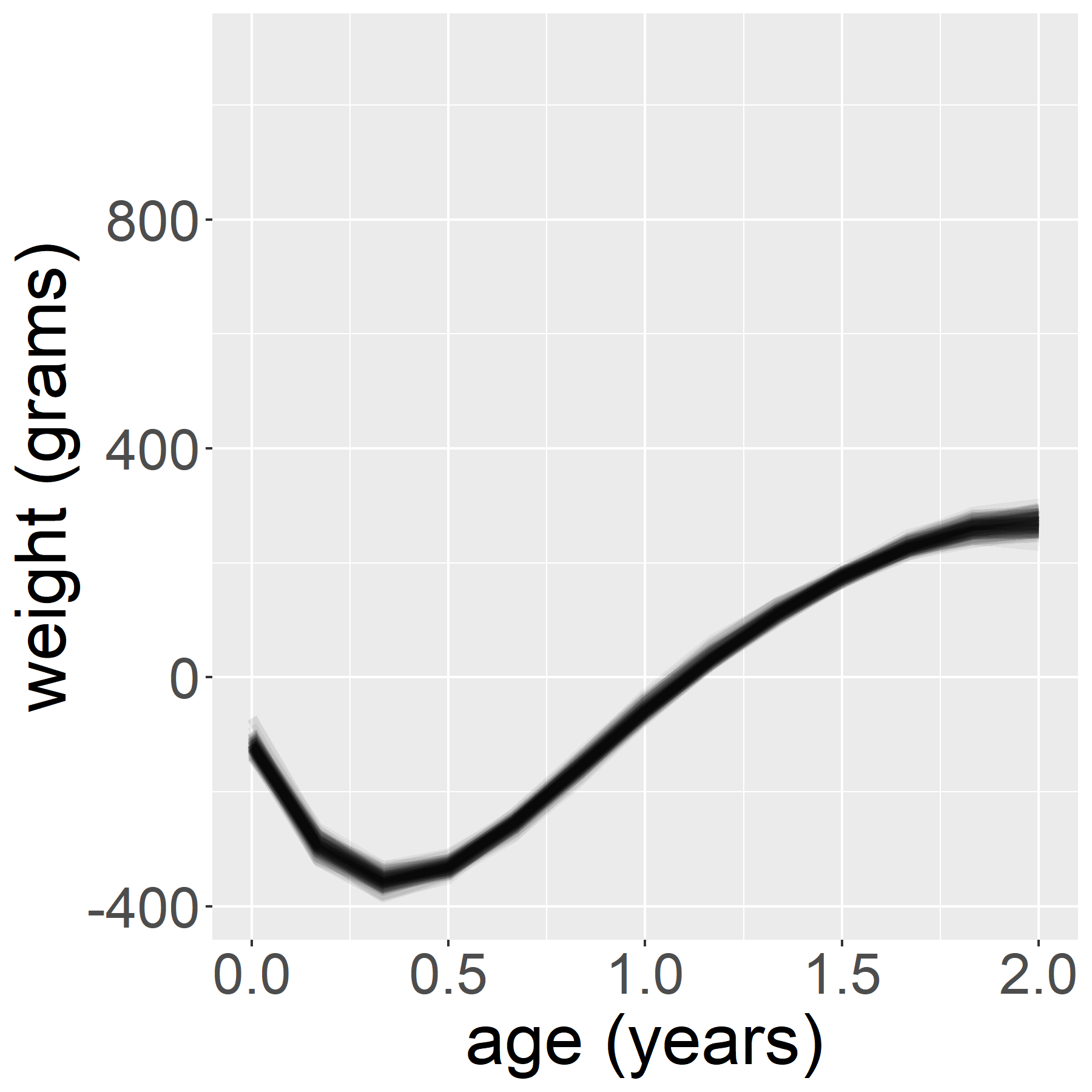} \\
(a) & (b) & (c)\\
\includegraphics[width = 1.5 in]{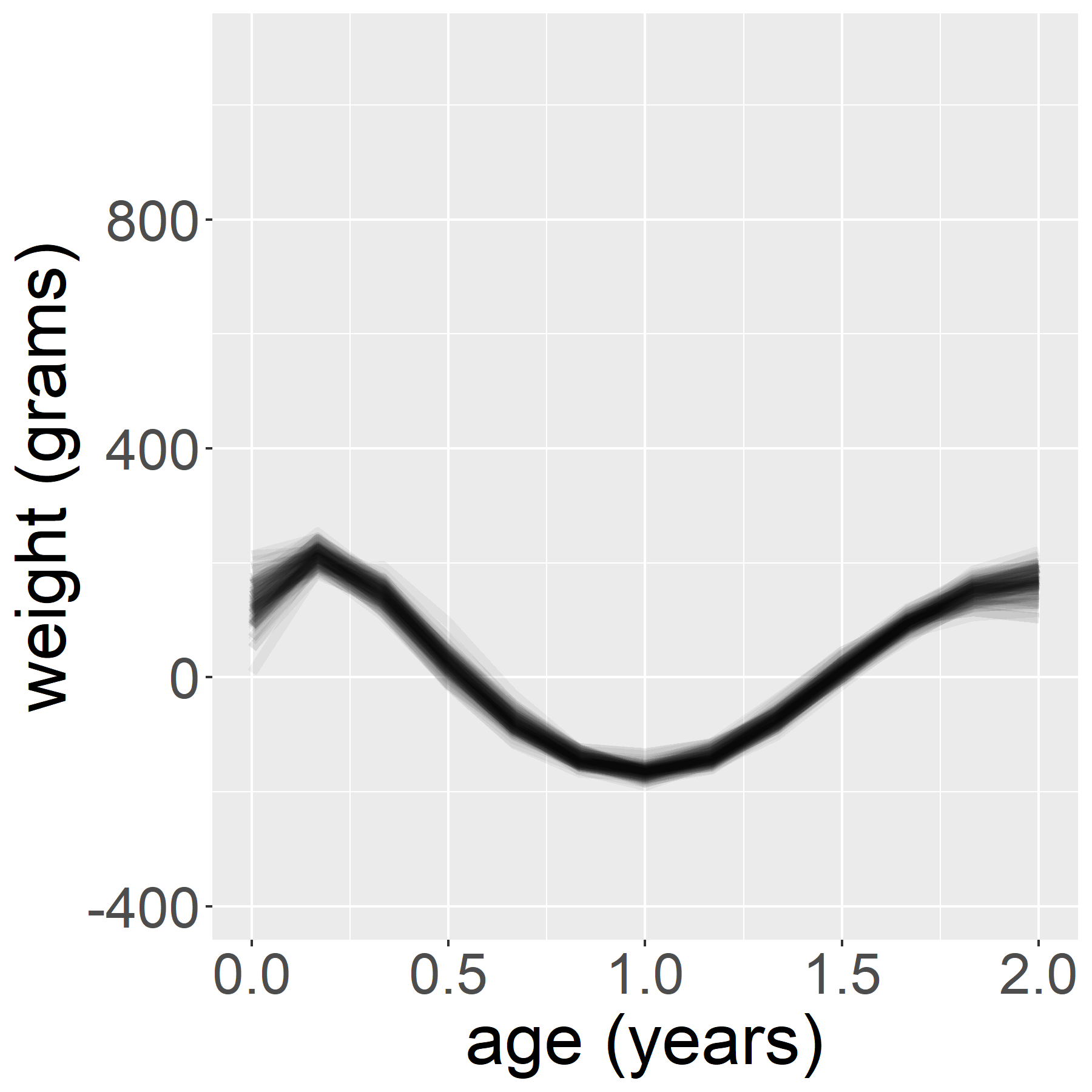} & \includegraphics[width = 1.5 in]{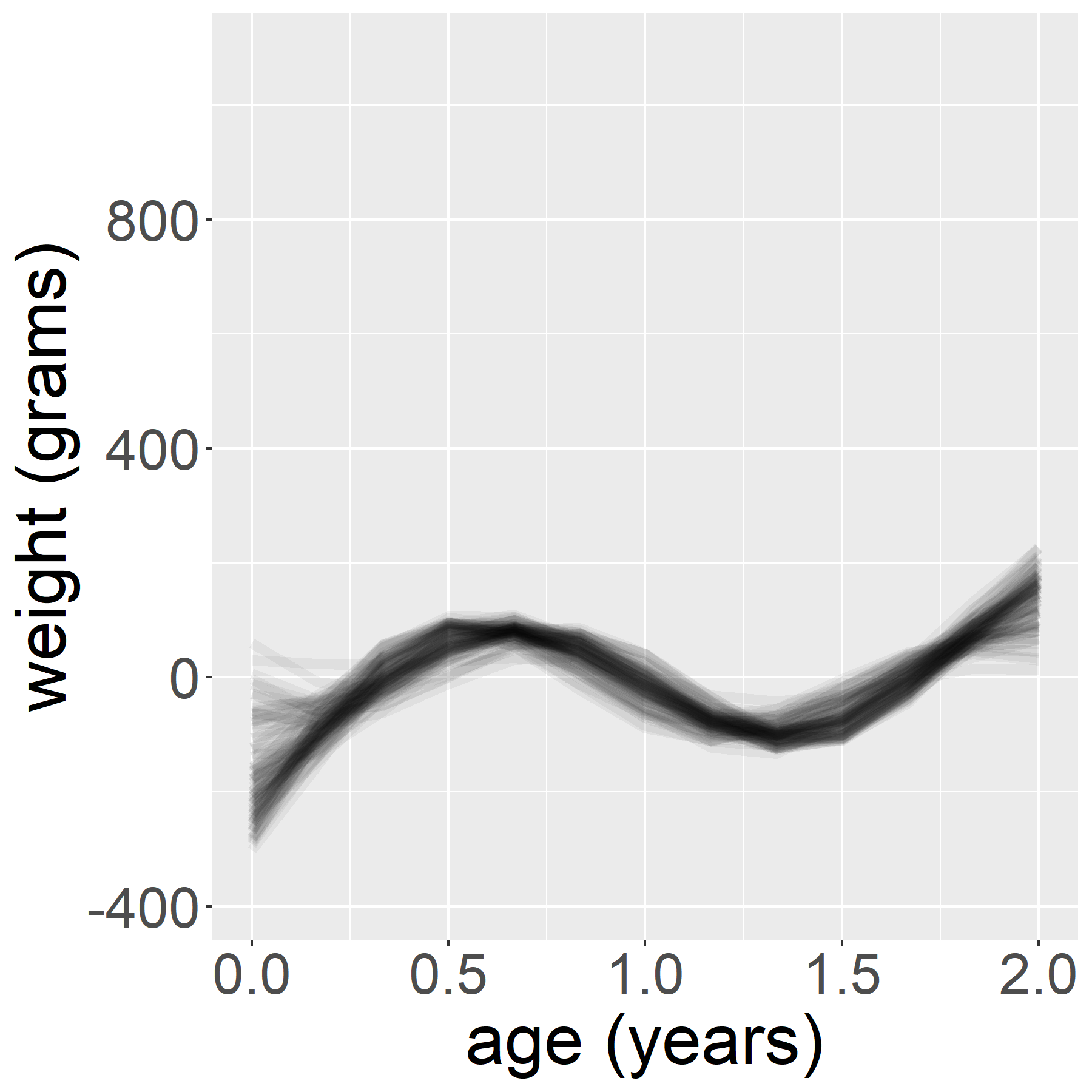}& \includegraphics[width = 1.5 in]{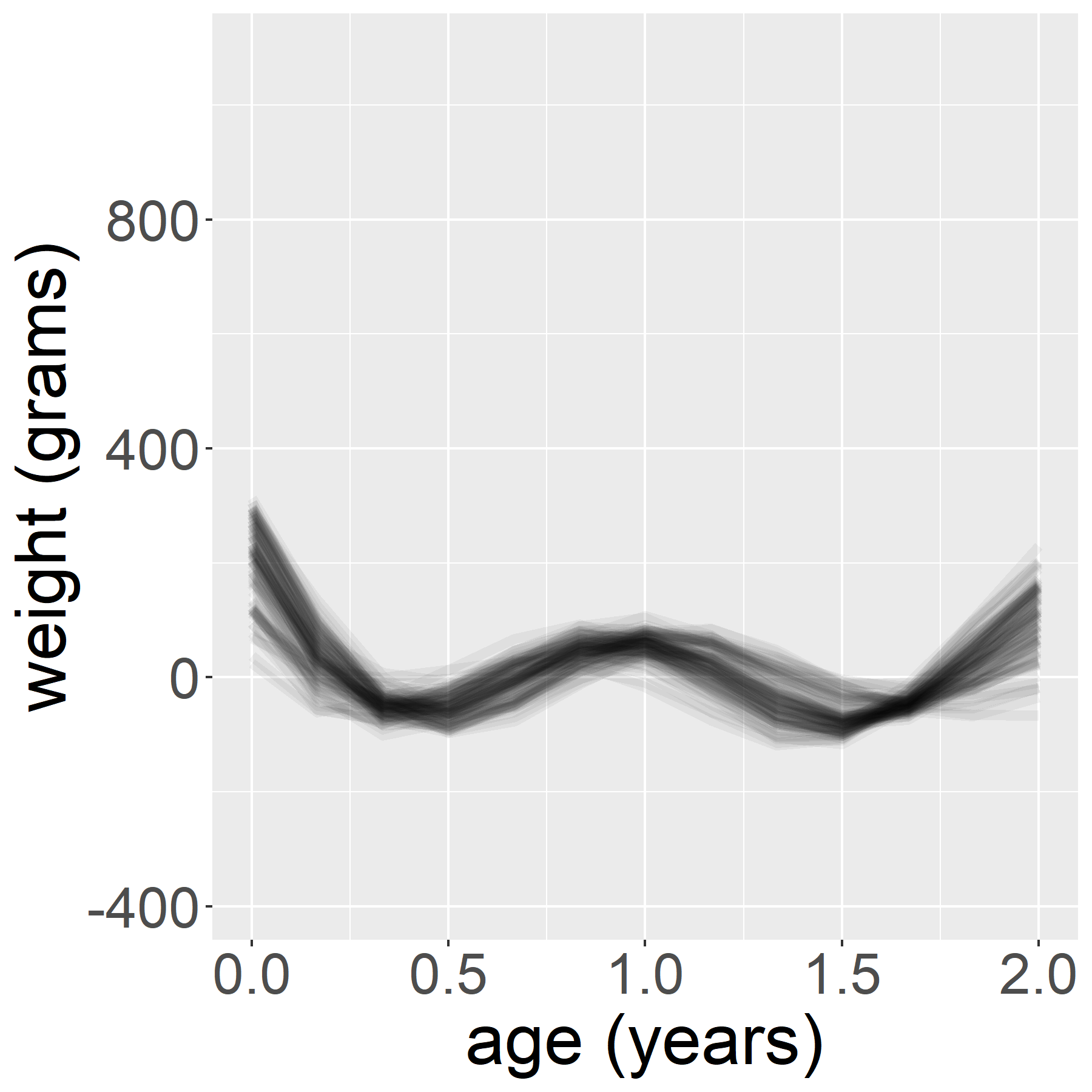} \\
(d) & (e) & (f)\\
\end{tabular}
\caption{(a) Posterior samples of the mean process, $\mu(\protect\vv{t})$, and (b)-(f) functional factor loadings, $\lambda_1(\protect\vv{t}),\ldots,\lambda_5(\protect\vv{t})$, for weight measurements of children recorded in the Cebu Longitudinal Health and Nutrition Survey. }\label{fig:cebu_FPCs}
    \end{center}
\end{figure}

\subsection{Bayesian hierarchical model for weight dynamics}

We implement the following flexible model for the weight dynamics, 
\begin{eqnarray}
y_i(\vv{t}_i) & = & \mu(\vv{t}_i) + \{\lambda_1(\vv{t}_i),\ldots,\lambda_K(\vv{t}_i)\}^\top\eta_i + O_i\sum_{j = 1}^4\int_{T}\beta_j(s,\vv{t})z_{i,j}(s)ds + \epsilon_i(\vv{t}_i) \nonumber \\
 \eta_i & = & \Theta x_i + \xi_i  \nonumber 
\end{eqnarray}
The longitudinal records of weight of the $i$\textsuperscript{th} child are denoted by $y_i(\vv{t}_i)$, where the observation-specific grid points $\vv{t}_i$ are a subset of $\vv{t} = (0,\frac{1}{2},\ldots,\frac{24}{2})^\top$, representing the two month intervals on which the measurements are recorded. We use the functional factor analysis approach developed in Section \ref{sec:fpca} for $\mu,\ \lambda_1,\ldots, \lambda_K$ and $\xi_i$. 

The response can vary with the scalar covariates, $x_i\in\mathbb{R}^4$, through a latent factor regression setup \citep{montagna2012}. The coefficient matrix $\Theta\in\mathbb{R}^{K\times4}$ models the relationship between the latent factors and the scalar covariates. To account for longitudinal covariates, $z_{i,j}(s),\ j=1,\ldots,4$, we use a functional linear model \citep{ramsay1991}. These covariates track if a child is fed breast milk and 3 indicators of illness (diarrhea, fever, and cough). We use our generalized functional factor analysis approach for model-based imputation of these covariates when not observed.  For the coefficient surface $\beta_1(s,t)$ corresponding to the breast milk indicator, we specify $\beta_1(s,t) = 0$ for $s>t$ in a historical linear model approach \citep{malfait2003}. Through this setup, weight at time $t$ is only affected by breastfeeding status at earlier ages. For the coefficient surface of the three illness indicators $\beta_2(s,t),\ \beta_3(s,t),\ \beta_4(s,t)$, we specify $\beta_j(s,t) = 0$ for $s\neq t,\ j = 2,3,4$ in a concurrent linear model \citep{hastie1993}. Through this setup the illness status of a child at time $t$ only effects the weight at time $t$. In Appendix 5, we provide full details of the hierarchical model specification and present additional supportive figures that include trace plot diagnostics for the MCMC used for posterior inference.

\subsection{Inferential results for weight dynamics}

Figure \ref{fig:cebu_data} panel (a) displays the weight of all children measured every two months from birth to age 2 and panel (b) displays the proportion of missing observations by age. In panels (c)-(f), we plot weight measurements, $y_i(\vv{t}_i)$, for a few children along with posterior samples of $\mu(\vv{t}) + \{\lambda_1(\vv{t}),\ldots,\lambda_5(\vv{t})\}^\top\eta_i + \sum_{j = 1}^4\int_{T}\beta_j(s,\vv{t})z_{i,j}(s)ds$. The model is able to fit a variety of shapes of underlying functions, with greater uncertainty at ages where the child's weight was not observed, such as panel (d). The functional factor loadings, $\lambda_1,\ldots,\lambda_K$, are able to describe modes of variability in weight dynamics. We implement our approach described in Section \ref{sec:select_K} to select the number of latent factors. Starting with an over-fitted factor model, our approach forces the loadings for extra unnecessary factors to be close to zero. 

Posterior samples of the mean process, $\mu$, and functional factor loadings, $\lambda_1,\ldots,\lambda_5$ are displayed in Figure \ref{fig:cebu_FPCs} in panels (a) and (b)-(f), respectively. The mean process increases rapidly from age 0 to 6 months, then appears to behave nearly linearly from age 6 months to 2 years. The first loading shown in panel (b) captures subject specific deviations having roughly a similar shape as the mean. The latter loadings in panels (c)-(f) capture more intricate variability that is smaller in magnitude. 

\subsection{Inferential results for the associations between weight dynamics and scalar covariates}

\begin{figure}[t]
\begin{center}
 \begin{tabular}{cc}
\includegraphics[width = 1.7 in]{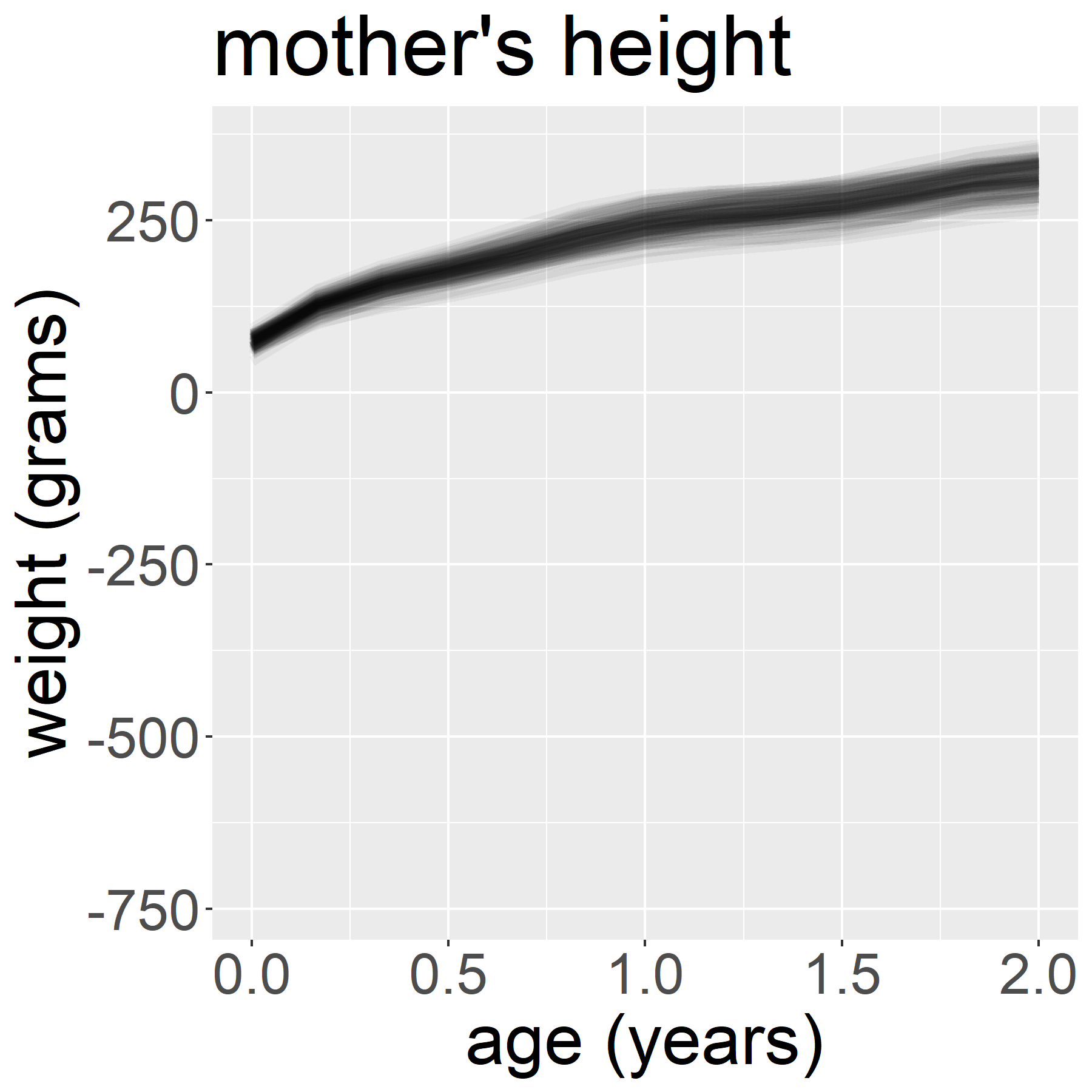} & \includegraphics[width = 1.7 in]{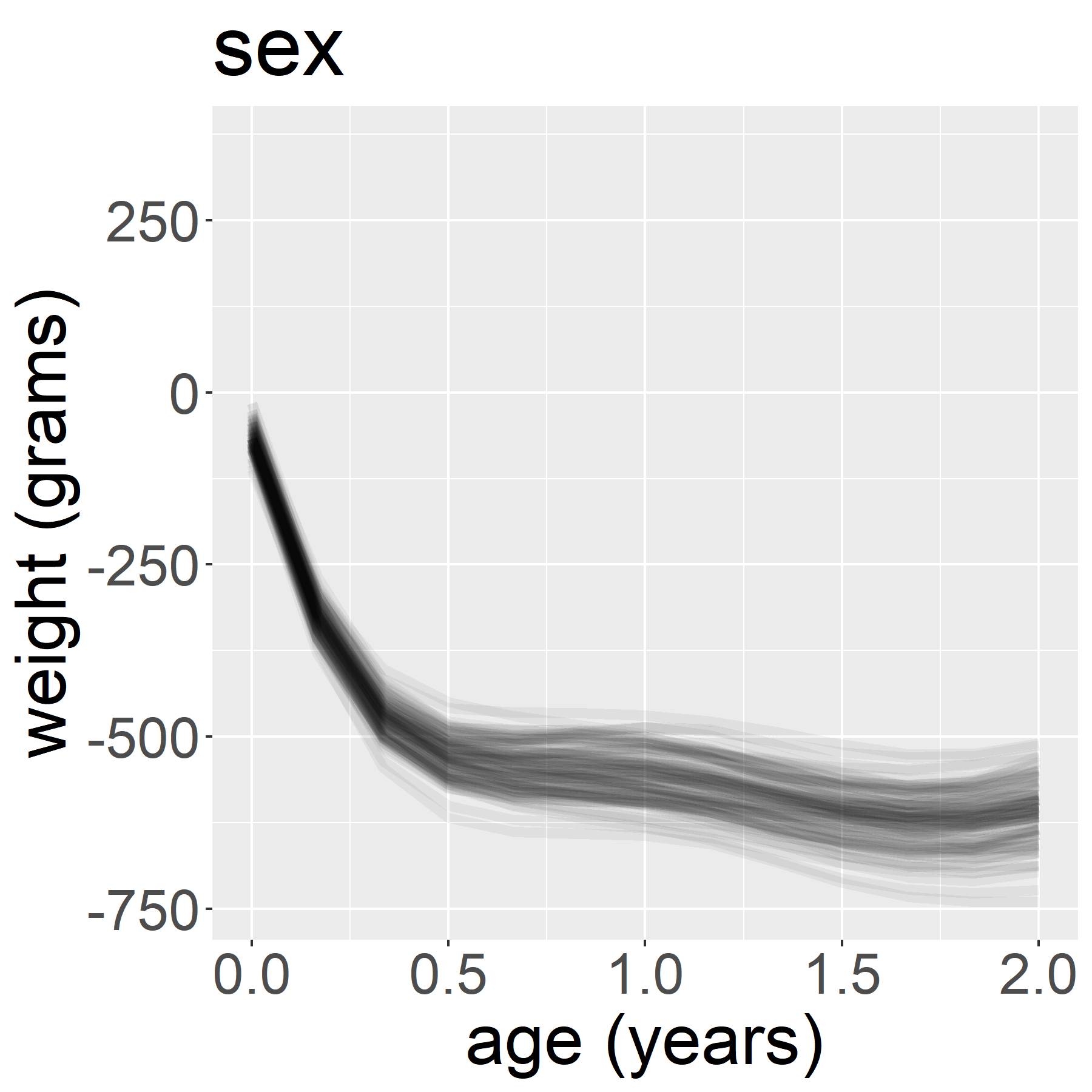}  \\
(a) & (b) \\
\includegraphics[width = 1.7 in]{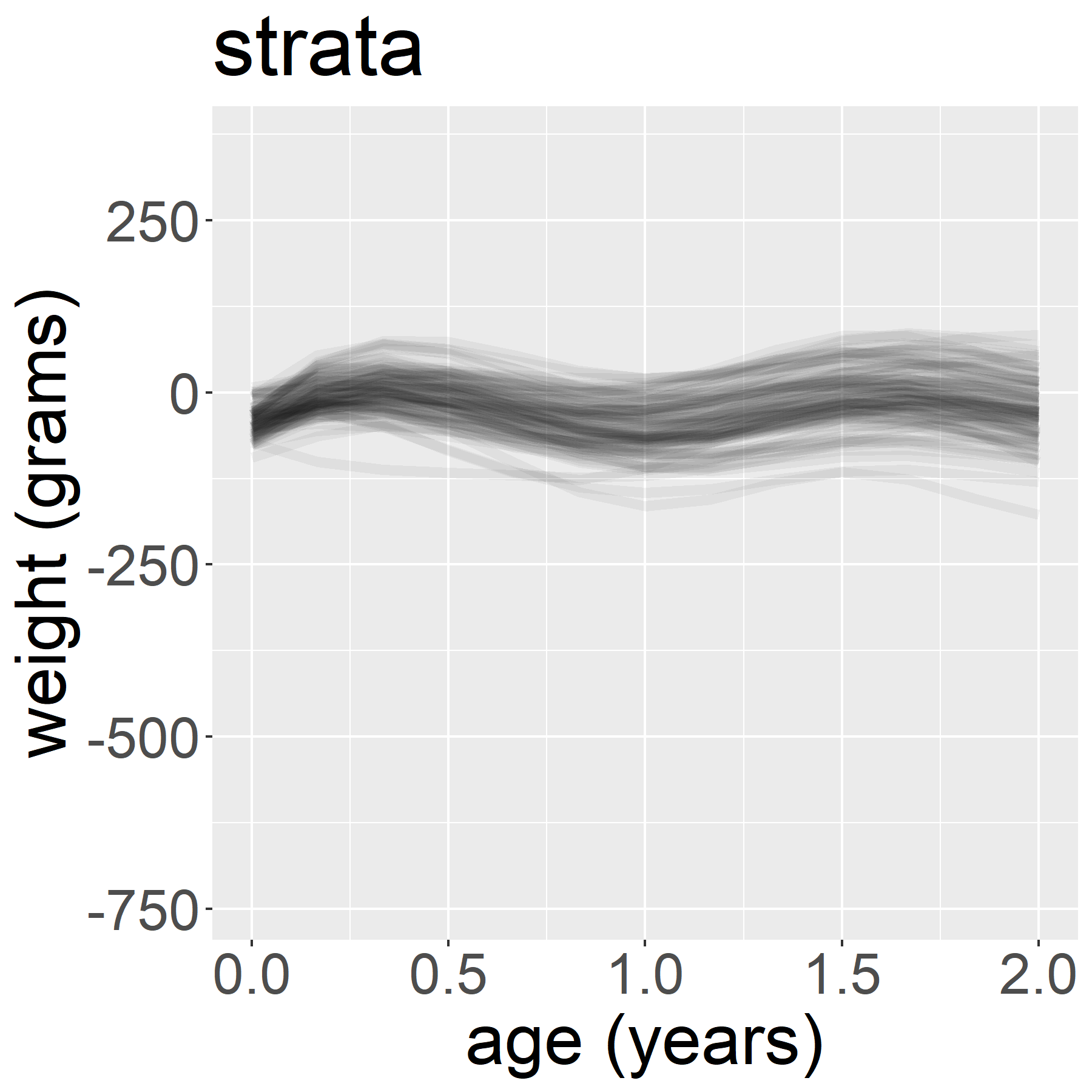}& \includegraphics[width = 1.7 in]{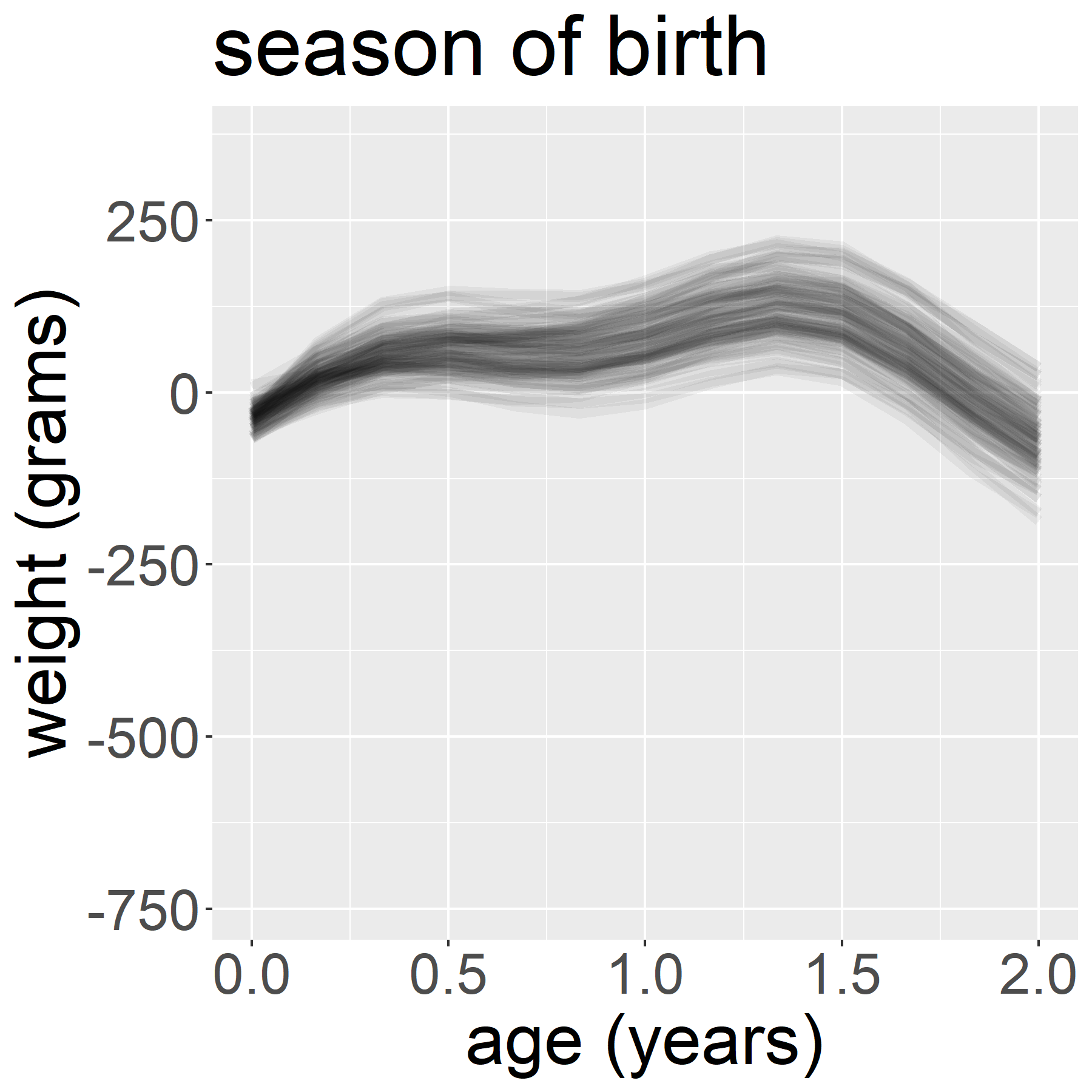}   \\
 (c) & (d)
\end{tabular}
\caption{(a)-(d) Posterior samples of $(\lambda_1(\protect\vv{t}),\ldots,\lambda_5(\protect\vv{t}))^\top \Theta_{q}$ for $q = 1,\ldots,4$, representing the effect of the scalar covariates (height of the mother, sex of the child, area (urban or rural) where the family lives, season of birth (rainy or dry)) on weight. Mother's height and sex of child are clearly associated with weight dynamics. Strata and season of birth appear to have near zero effects with relatively high uncertainty. \label{fig:cebu_b_inference}}
    \end{center}
\end{figure}

The effect of the scalar covariates on the response is modeled through the latent factors, $\eta_i = \Theta x_i + \xi_i$. With all other variables held constant, the effect of a unit increase of a single covariate on the response is given by $\{\lambda_1(\vv{t}),\ldots,\lambda_5(\vv{t})\}^\top \Theta_{q}$,  where $\Theta_{q}$ denotes the $q$\textsuperscript{th} column of the matrix $\Theta$, for $q = 1,\ldots,4$. Posterior samples of these functions are shown in Figure \ref{fig:cebu_b_inference}. Panel (a) shows the inferred effect of a unit increase in mother's height on child weight. The model infers that children with taller mothers generally weigh more. Panel (b) shows the inferred difference in weight between male and female children, where $x_{i,2} = 0$ if the $i$\textsuperscript{th} subject is male. The model infers that female children generally weigh less than male children. The discrepancy in weight is small at birth, but changes rapidly from birth to 6 months. Panel (c) shows the inferred difference between children born in urban or rural strata, where $x_{i,3} = 0$ if the $i$\textsuperscript{th} subject was born in an urban area. Panel (d) shows the inferred difference between children born in the rainy or dry season, where $x_{i,4} = 0$ if the $i$\textsuperscript{th} subject was born during the rainy season. Compared to mother's height or sex, these covariate effects are smaller in magnitude with greater uncertainty.

\subsection{Inferential results for the associations between weight dynamics and longitudinal covariates}

\begin{figure}[t]
\begin{center}
 \begin{tabular}{ccc}
\includegraphics[width = 1.5 in]{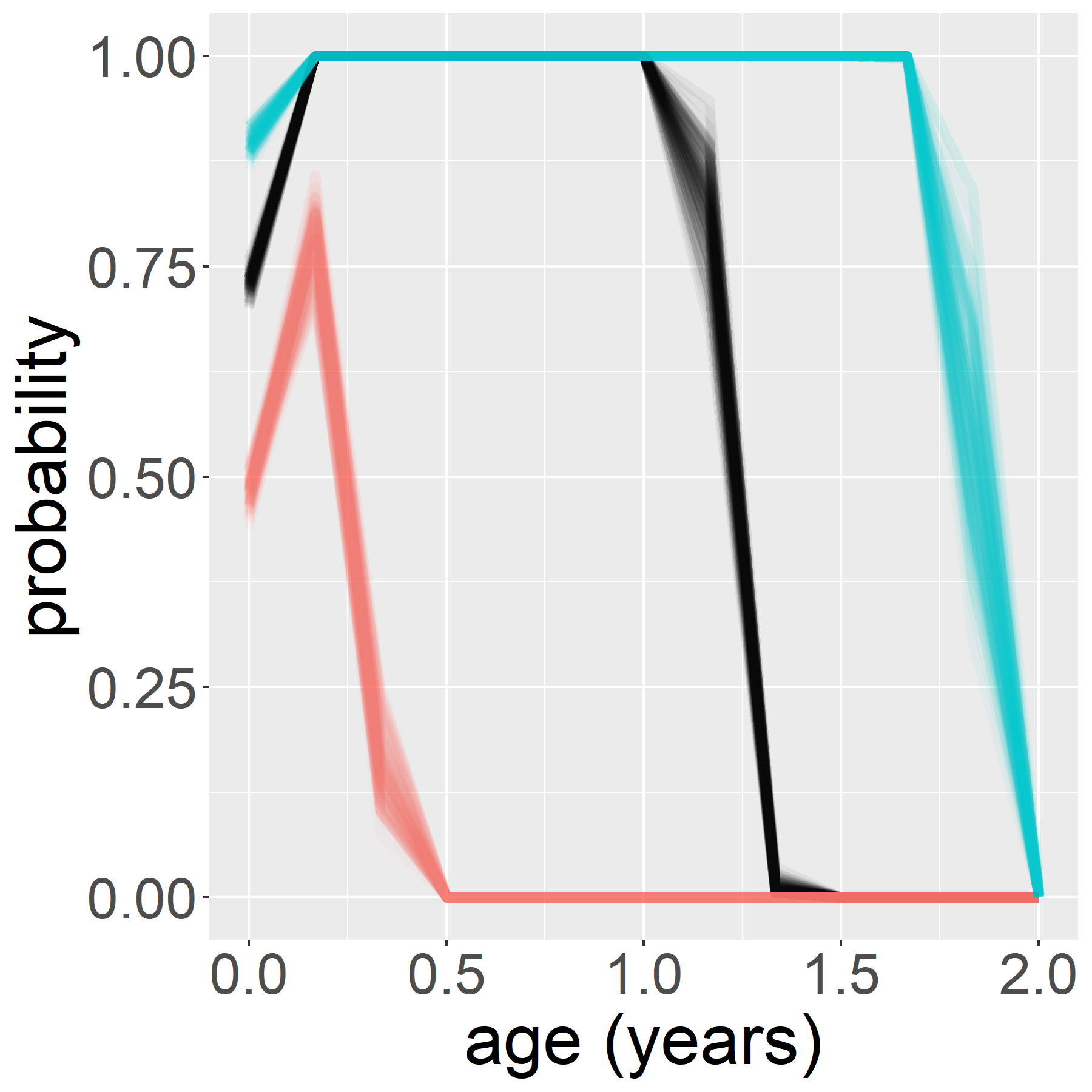} & \includegraphics[width = 1.5 in]{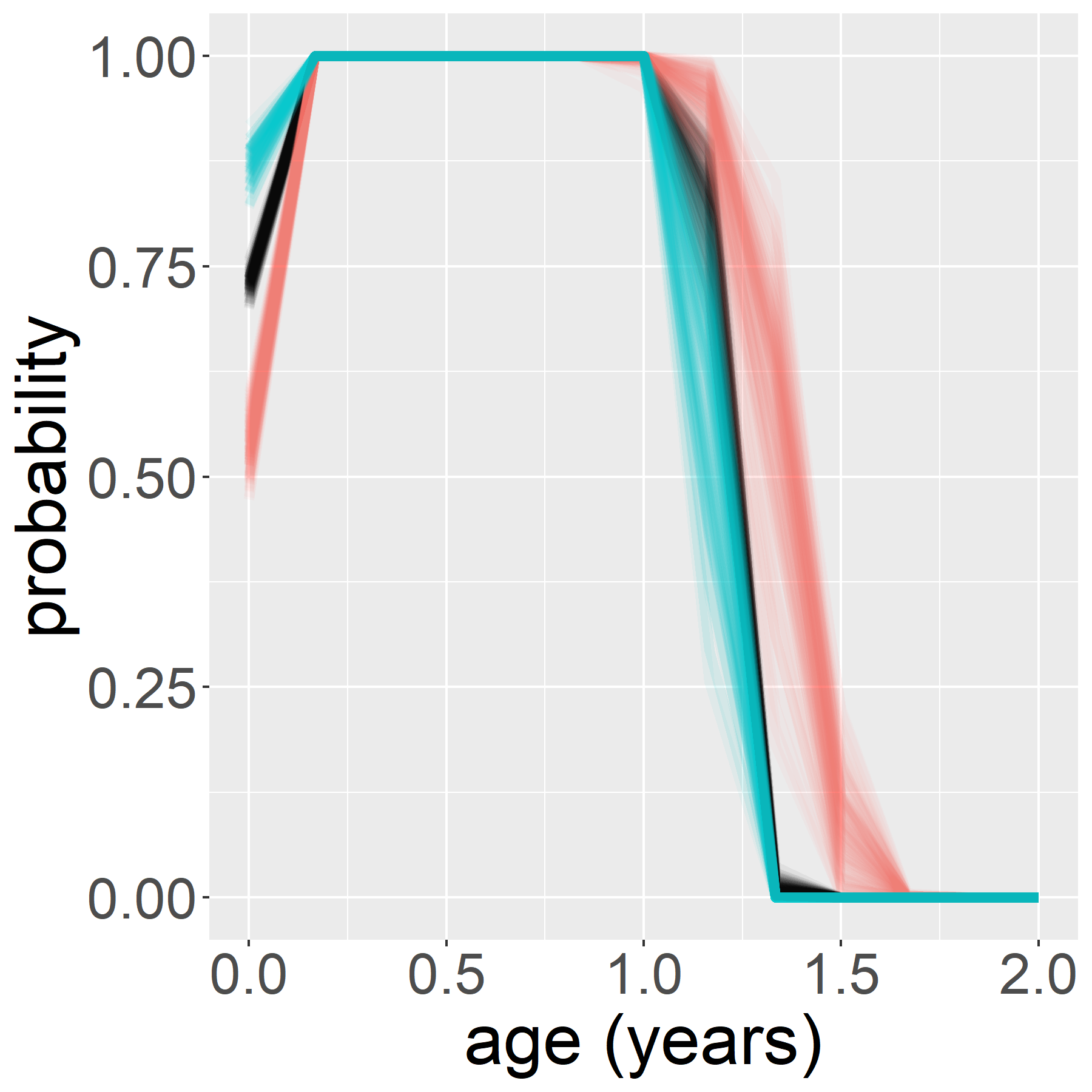}  & \includegraphics[width = 1.5 in]{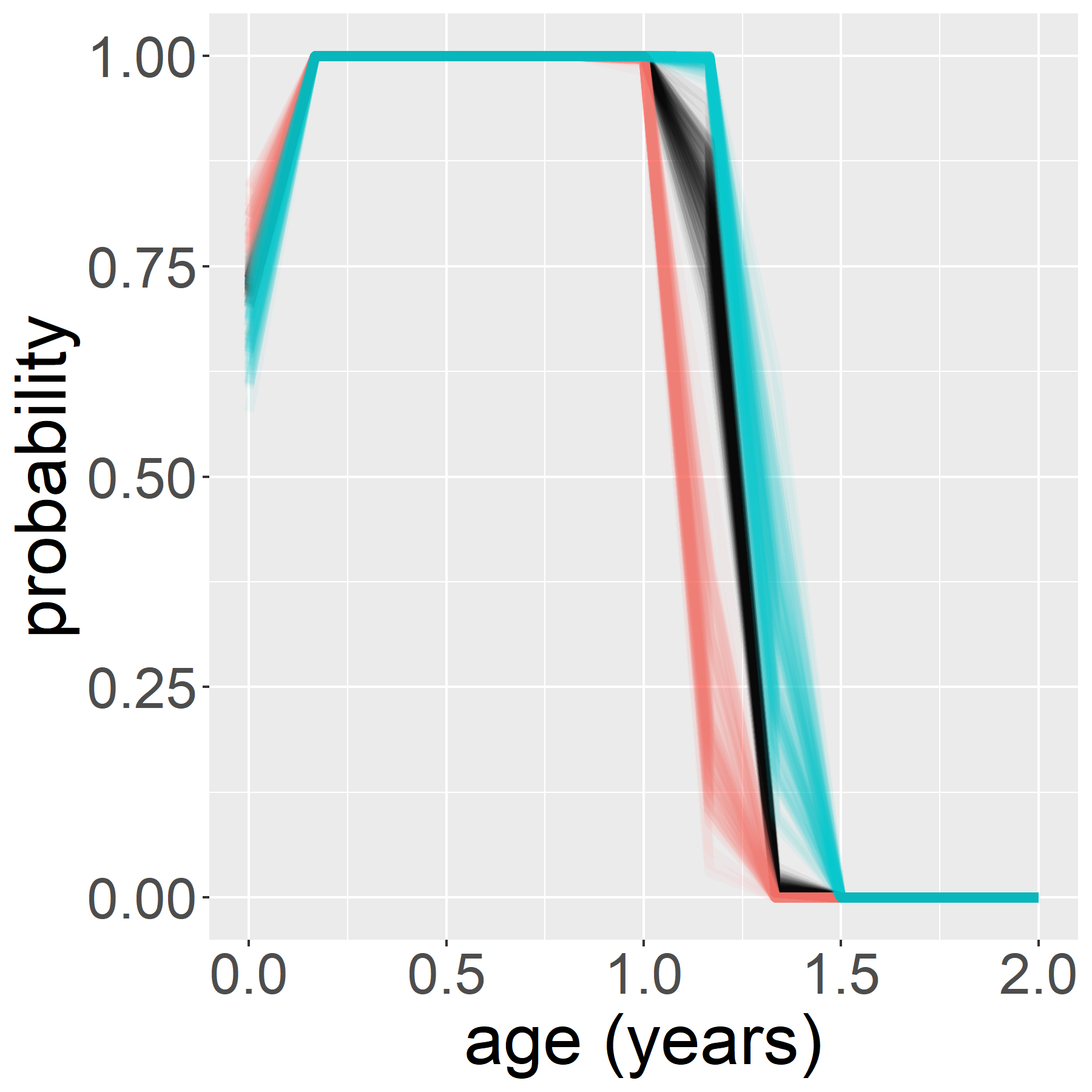} \\
(a) & (b) & (c)\\
\includegraphics[width = 1.5 in]{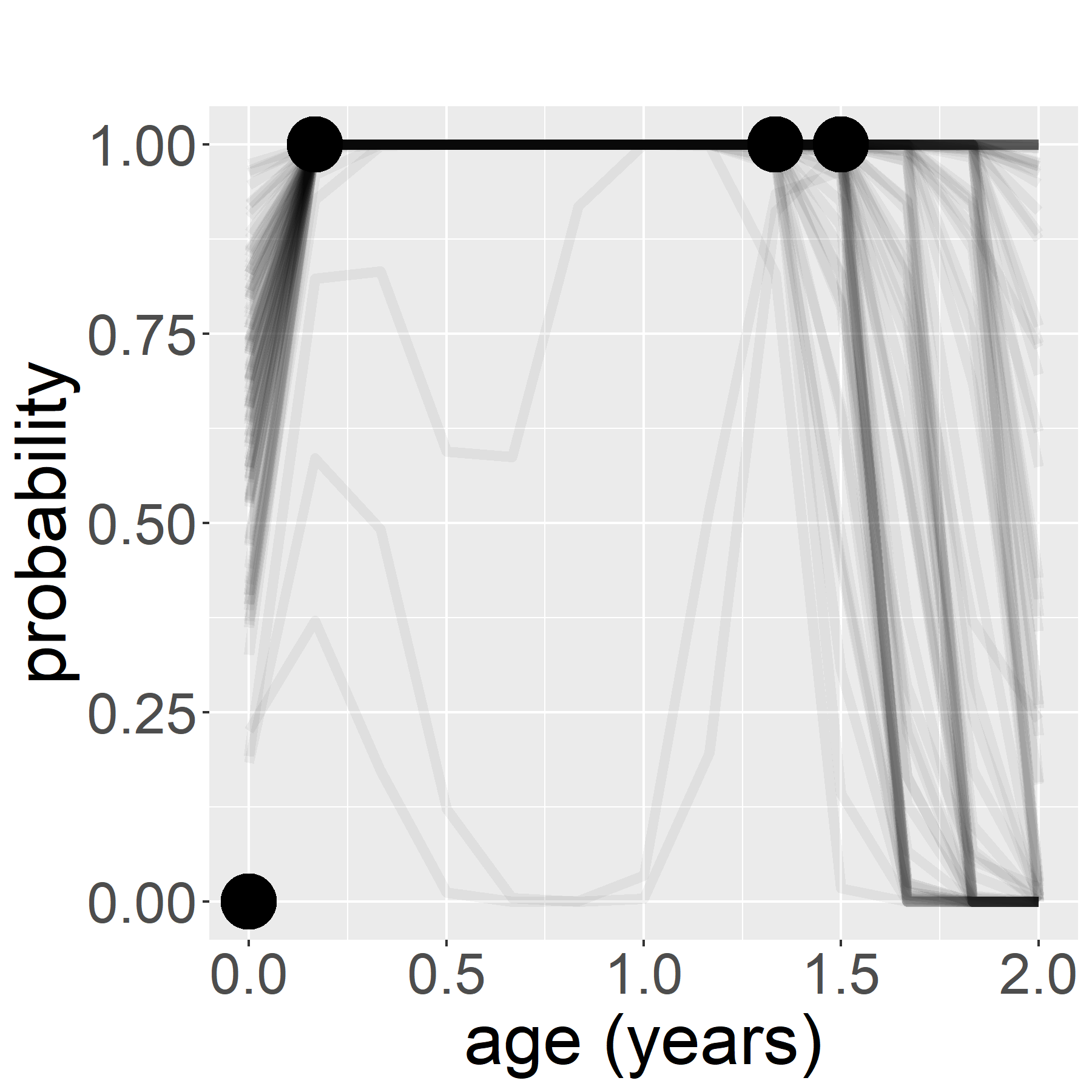} & \includegraphics[width = 1.5 in]{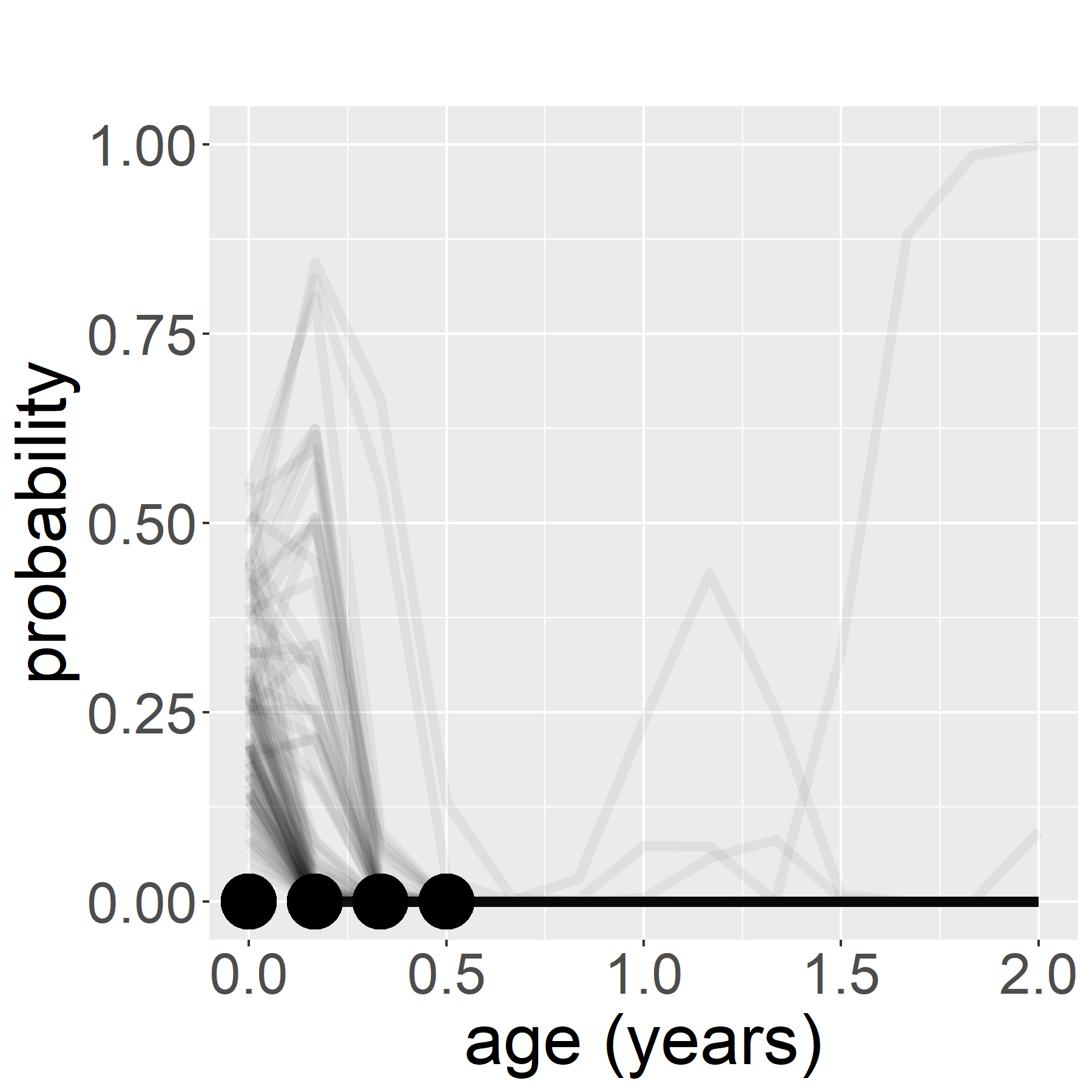}  & \includegraphics[width = 1.5 in]{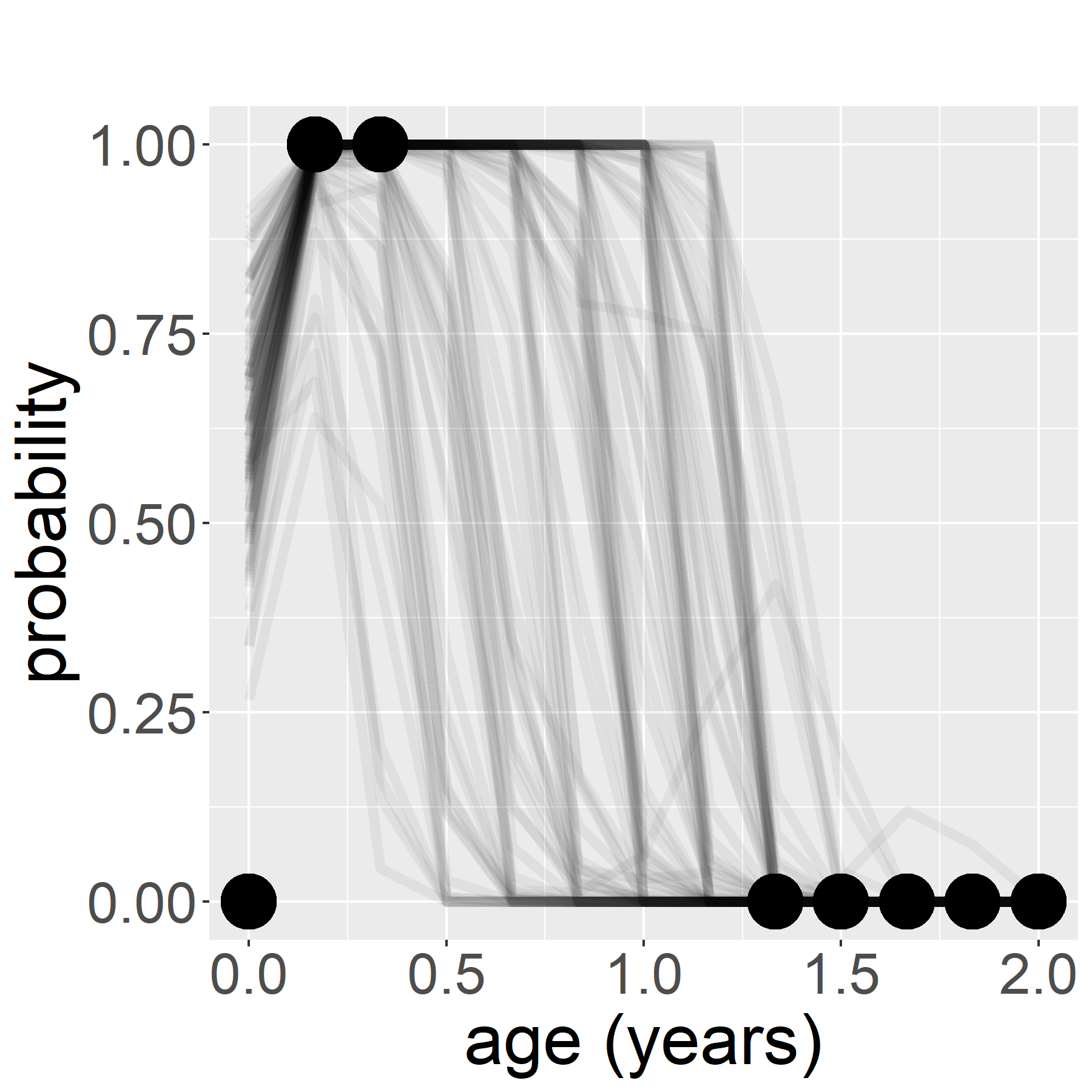} \\
(d) & (e) & (f)
\end{tabular}
\caption{(a)-(c) Posterior samples of $\Phi(\mu^{z_1}(\protect\vv{t}))$ (black) and $\Phi(\mu^{z_1}(\protect\vv{t}) \pm \text{sd}(\eta^z_{\cdot k})\lambda^z_{k}(\protect\vv{t}))$ (blue and red) for $k=1,2,3$, representing the variability described by the loadings of breastfeeding trajectories. (d)-(f) Posterior samples fitted breastfeeding trajectories, $\Phi(\mu^{z_1}(\protect\vv{t}) + (\lambda^z_1(\protect\vv{t}),\lambda^z_2(\protect\vv{t}),\lambda^z_3(\protect\vv{t}))^\top\eta^z_i)$ overlaid on $z_{i,1}(\protect\vv{t}_i)$ for $i = 250,\ 1079,\ 2087$.}\label{fig:cebu_bf_gfpca}
\end{center}
\end{figure}

\begin{figure}[t]
\begin{center}
 \begin{tabular}{cc}
\includegraphics[width = 1.5 in]{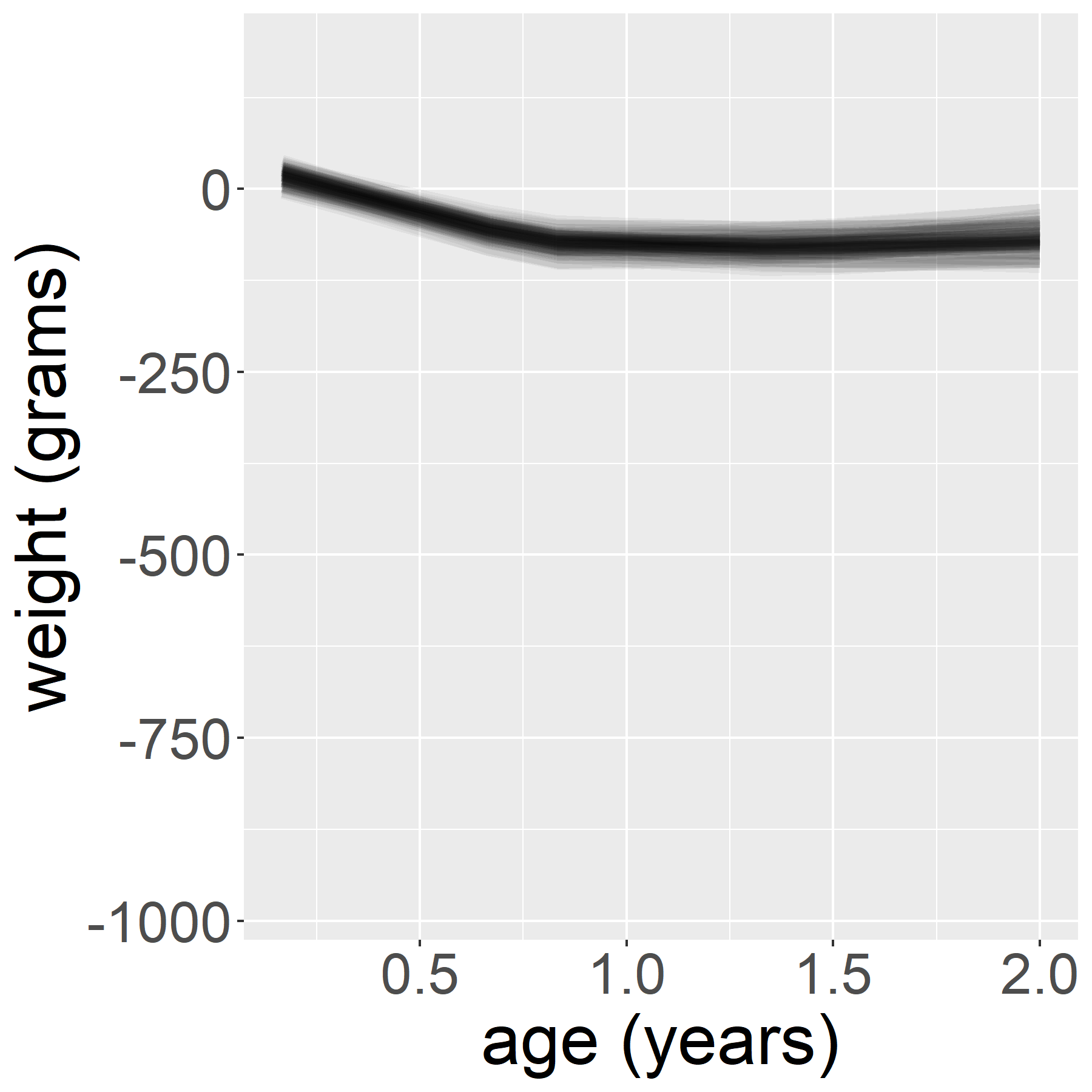} & \includegraphics[width = 1.5 in]{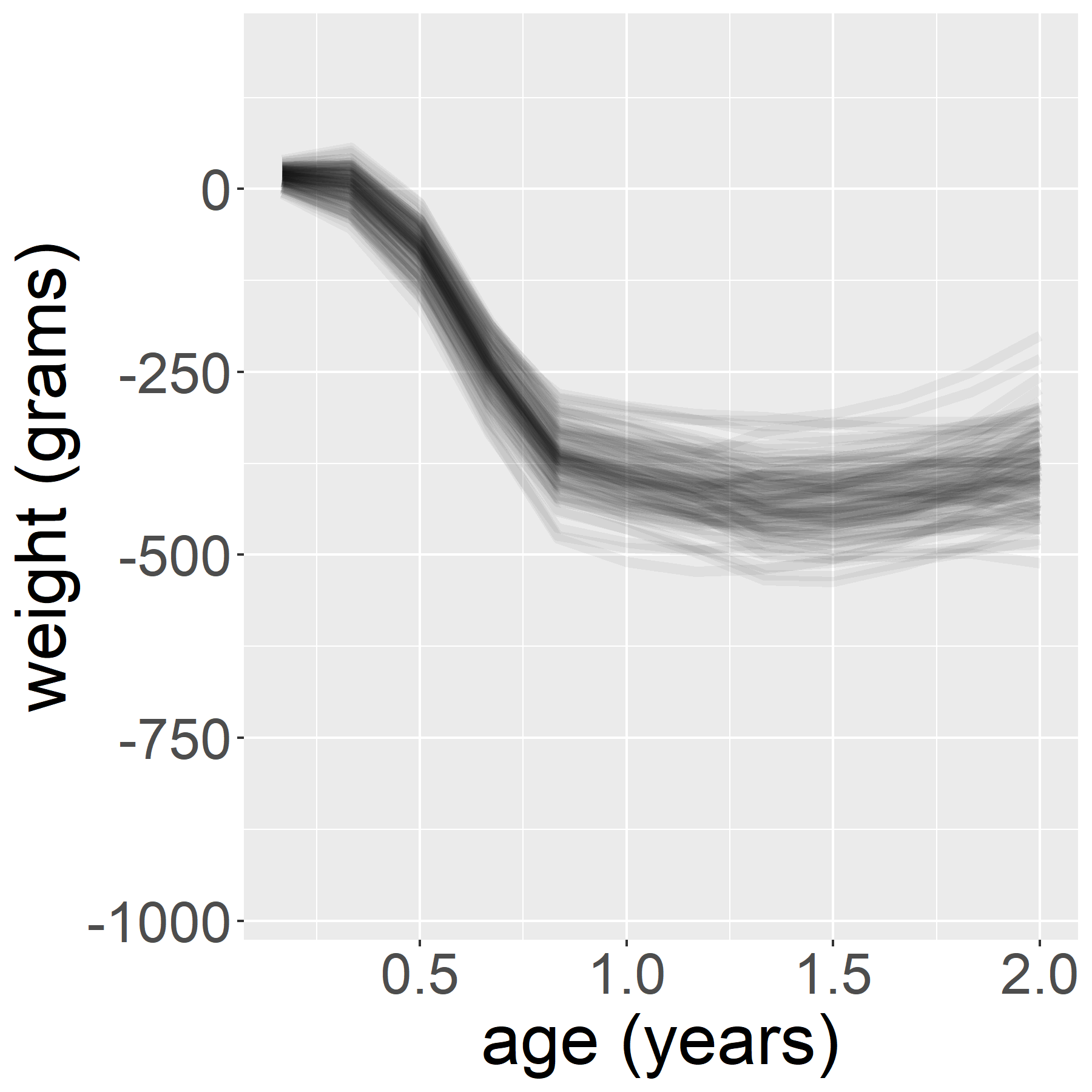} 
 \\ 
 (a) & (b) \\
 \includegraphics[width = 1.5 in]{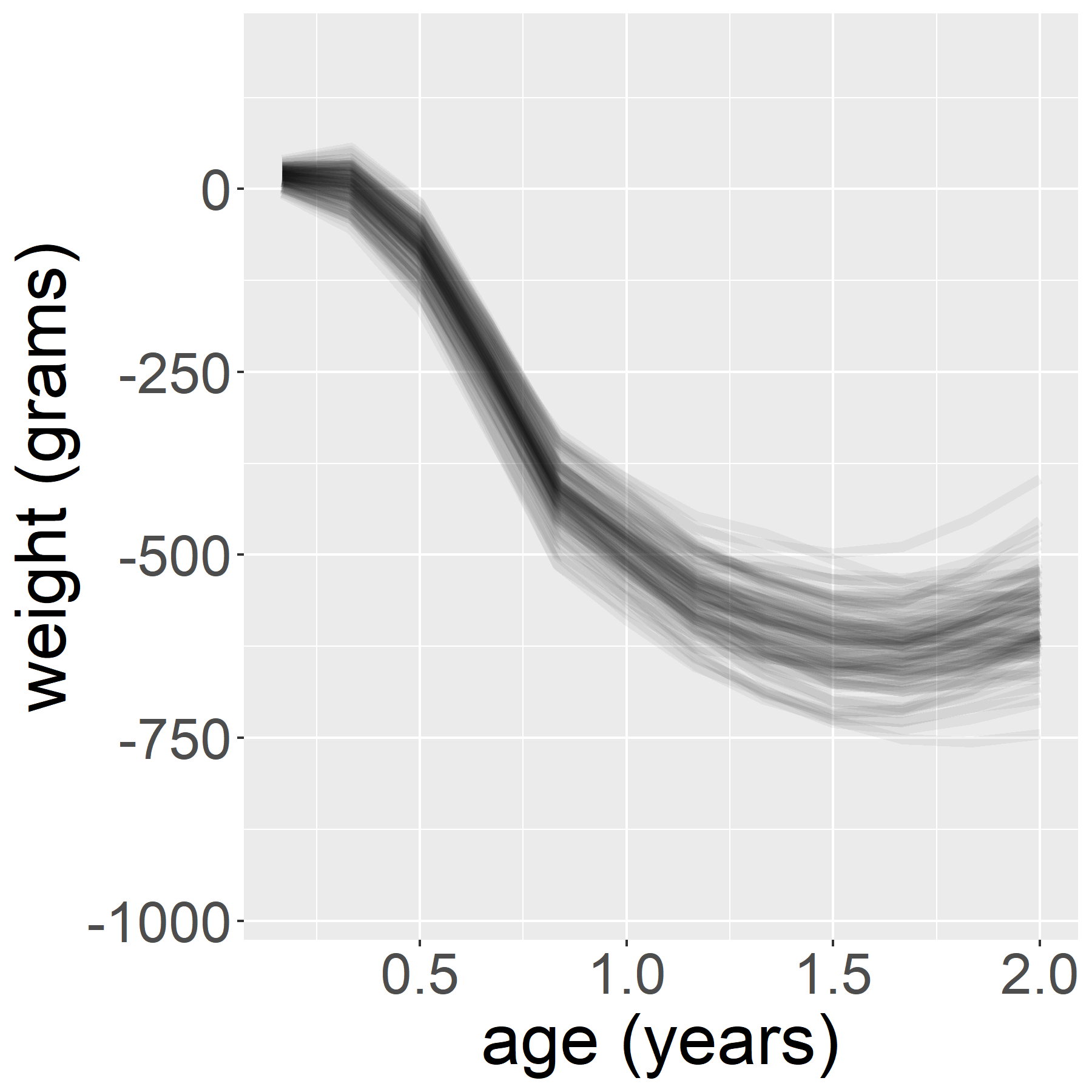} &\includegraphics[width = 1.5in]{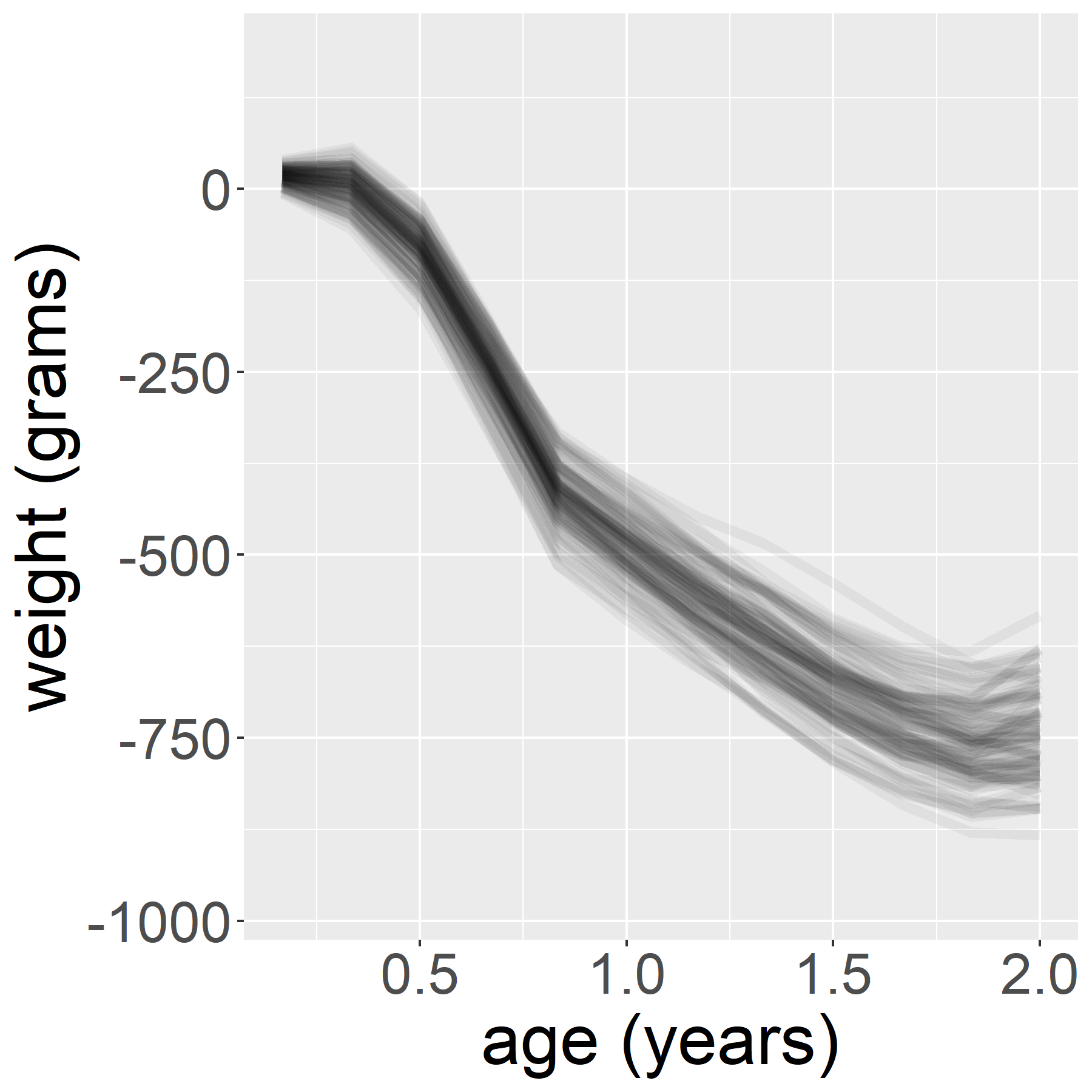} \\
(c) & (d) \\
\end{tabular}
 \begin{tabular}{c}
\includegraphics[width = 1.7 in]{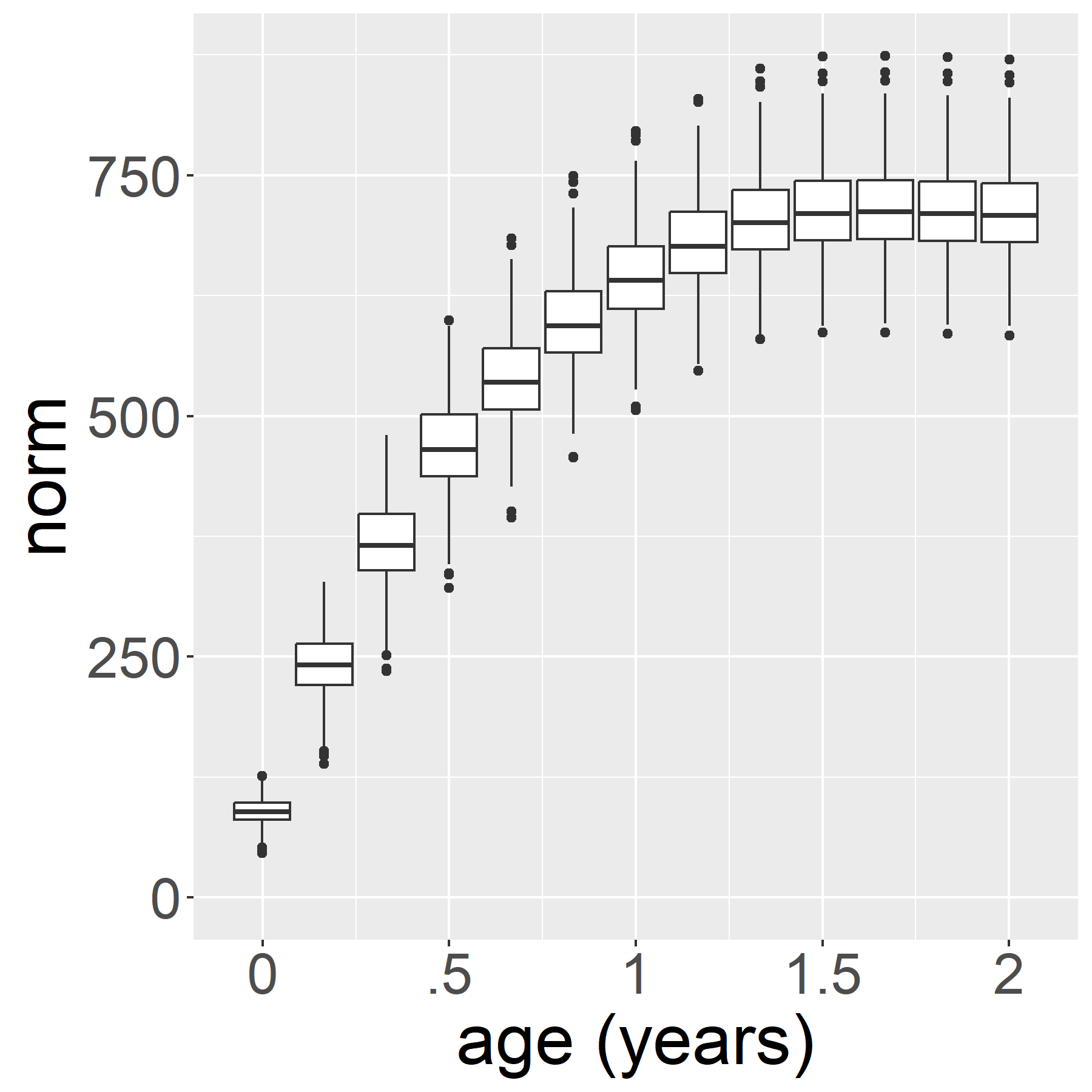} \\
(e)
\end{tabular}
\caption{(a)-(d) Posterior samples of $\int_{\mathcal{T}} \beta_1 1_{s\leq s'}ds$ for $s' = 0,\ 0.5,\ 1,\ 2$ years, representing the effect of breastfeeding for different durations on weight. (e) Boxplots of posterior samples of $\|\int_{\mathcal{T}} \beta_1 1_{s\leq s'}ds\|_2$ for different values of $s'$, representing the magnitude of the effect of breastfeeding for different durations. Children that are fed breast milk tend to weigh less than non-breast milk fed counterparts. The effect of breastfeeding on weight dynamics tends to level off after 1 year.}\label{fig:cebu_bf_reg}
    \end{center}
\end{figure}

\begin{figure}[t]
\begin{center}
 \begin{tabular}{ccc}
\includegraphics[width = 1.5 in]{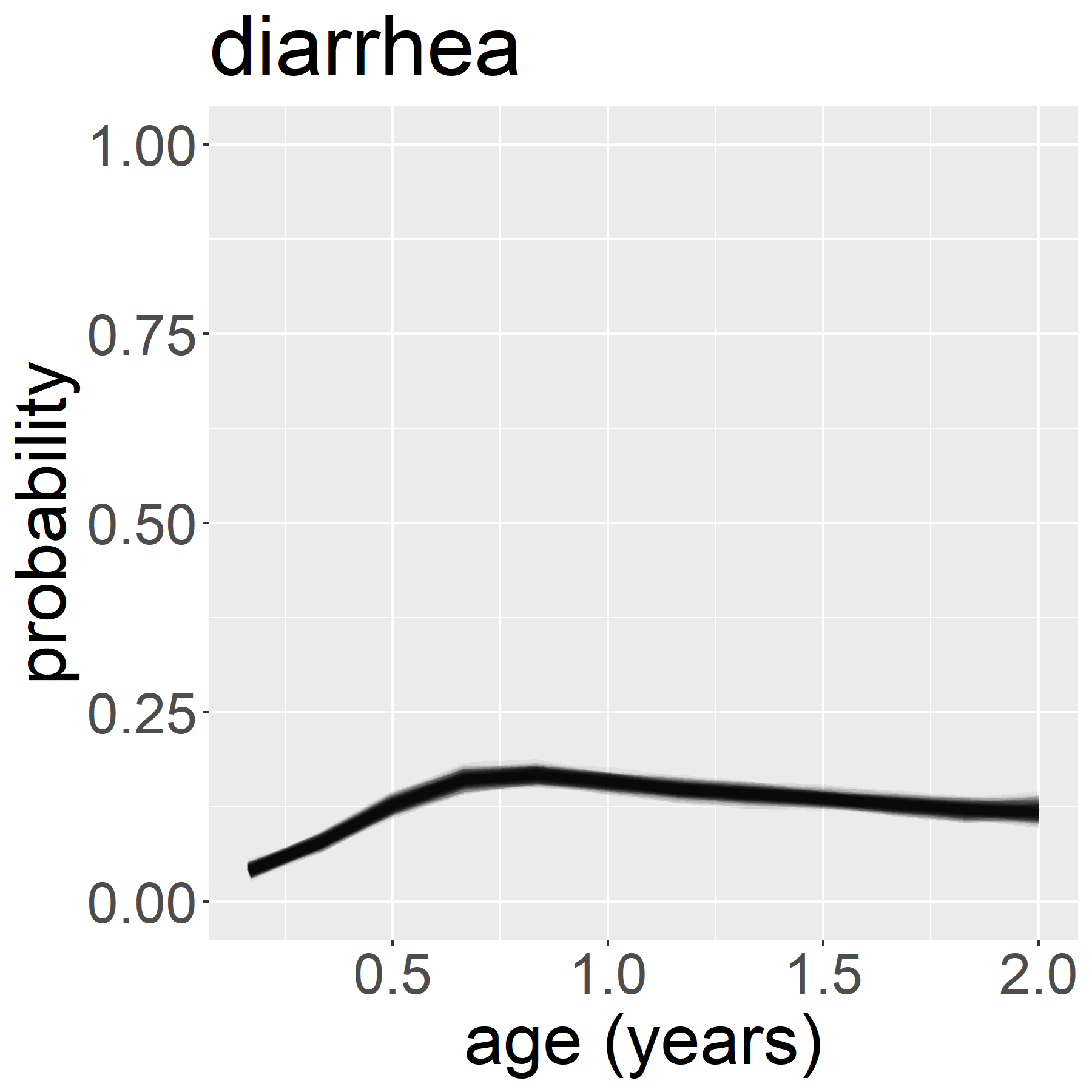} & \includegraphics[width = 1.5 in]{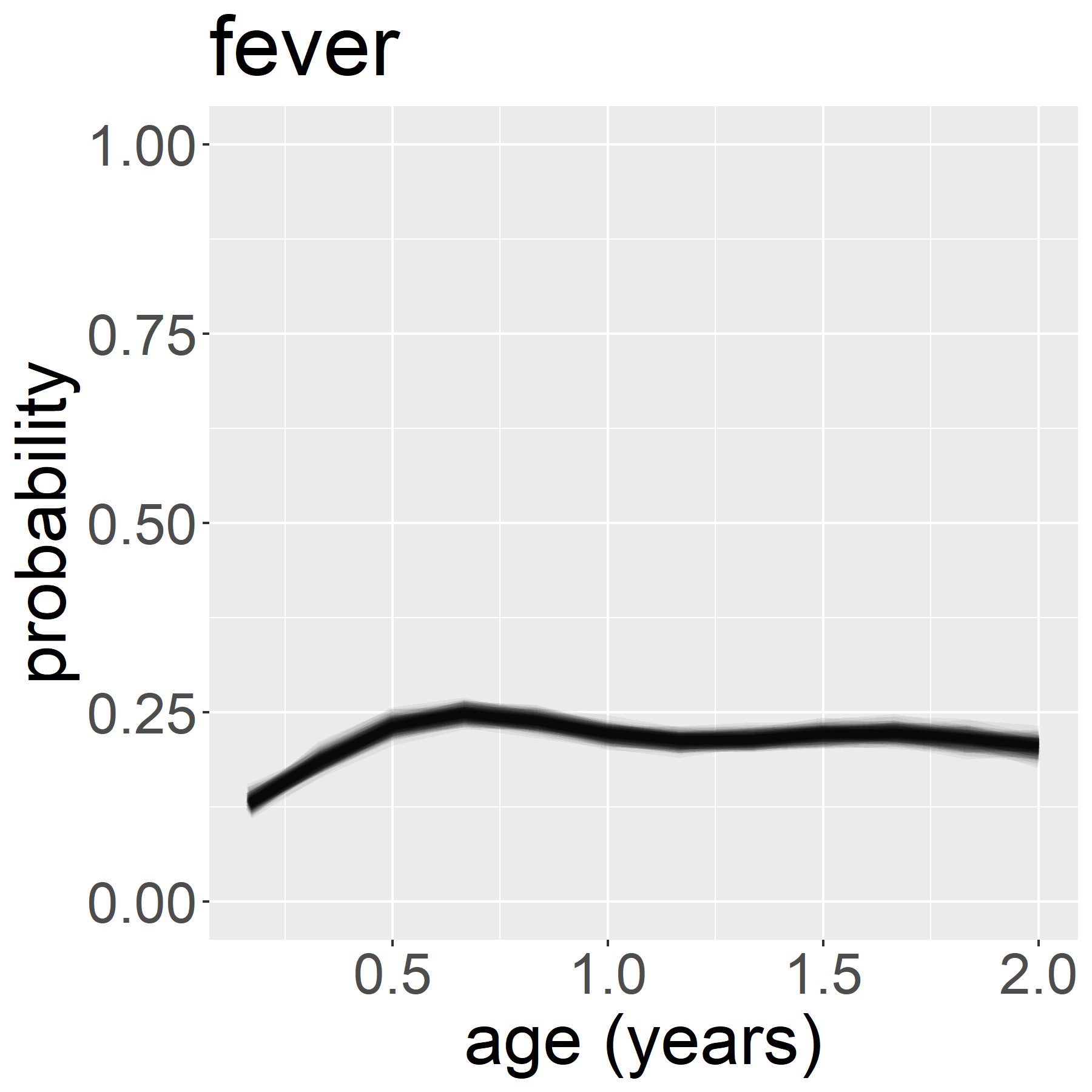}  & \includegraphics[width = 1.5 in]{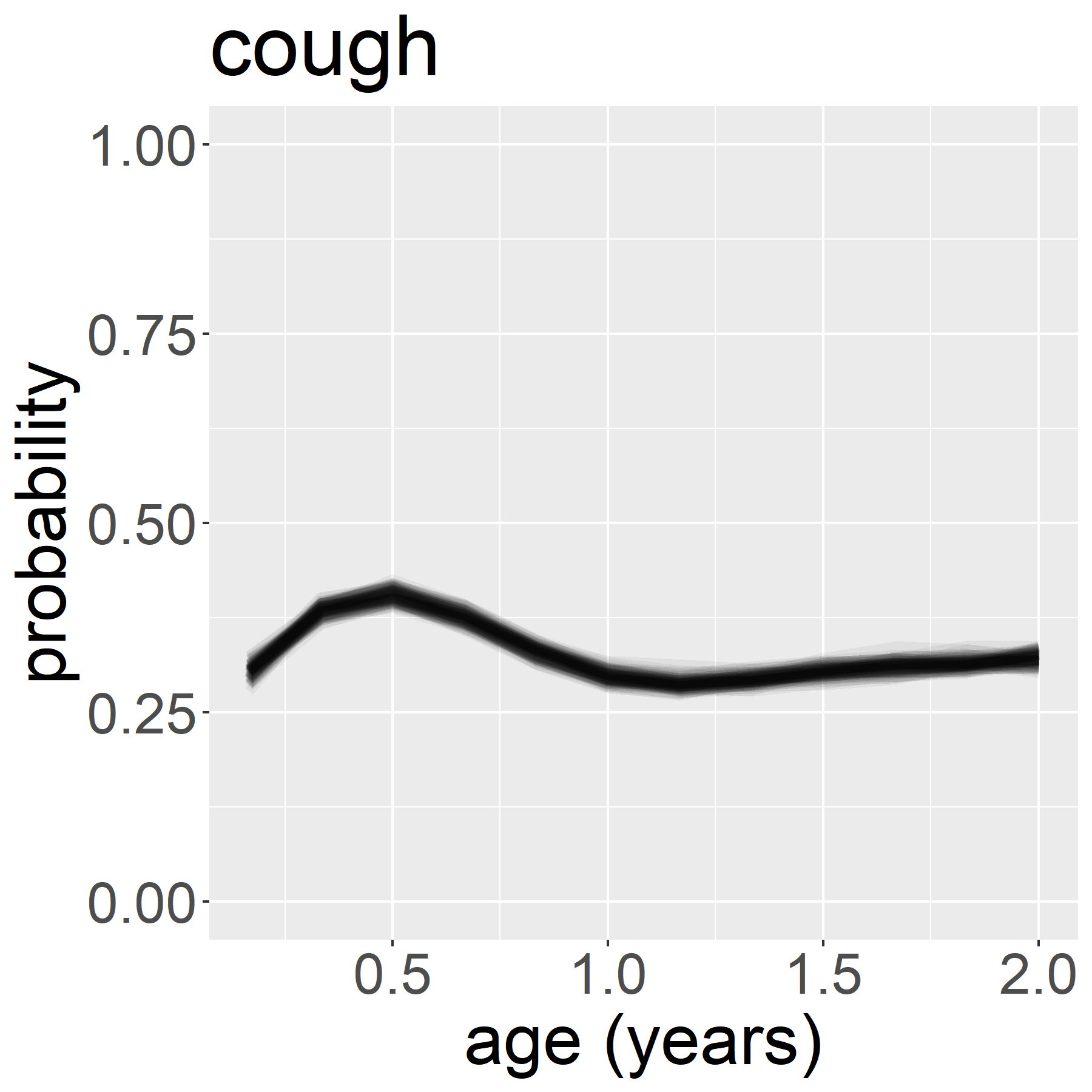} \\
(a) & (b) & (c)\\
\includegraphics[width = 1.5 in]{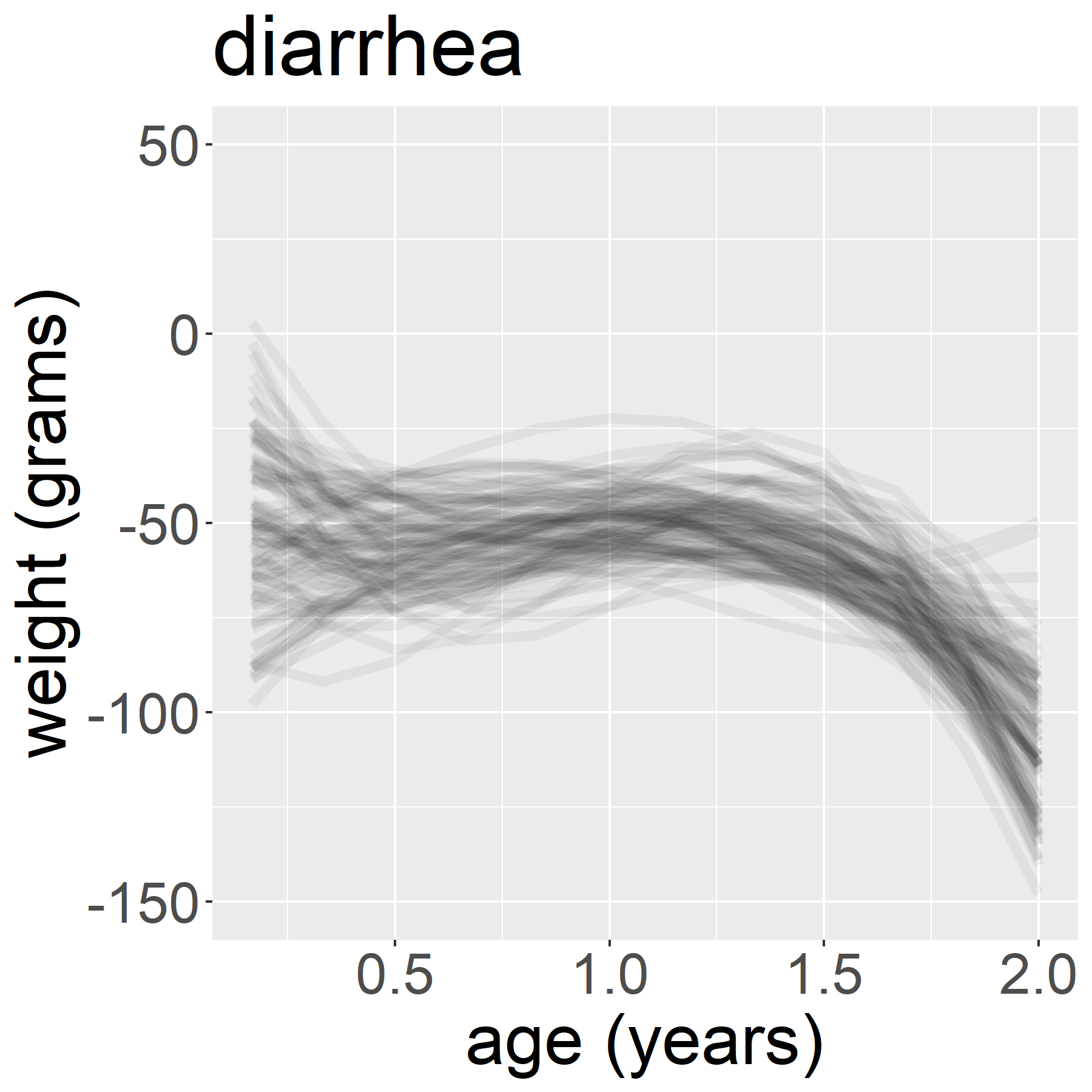} & \includegraphics[width = 1.5 in]{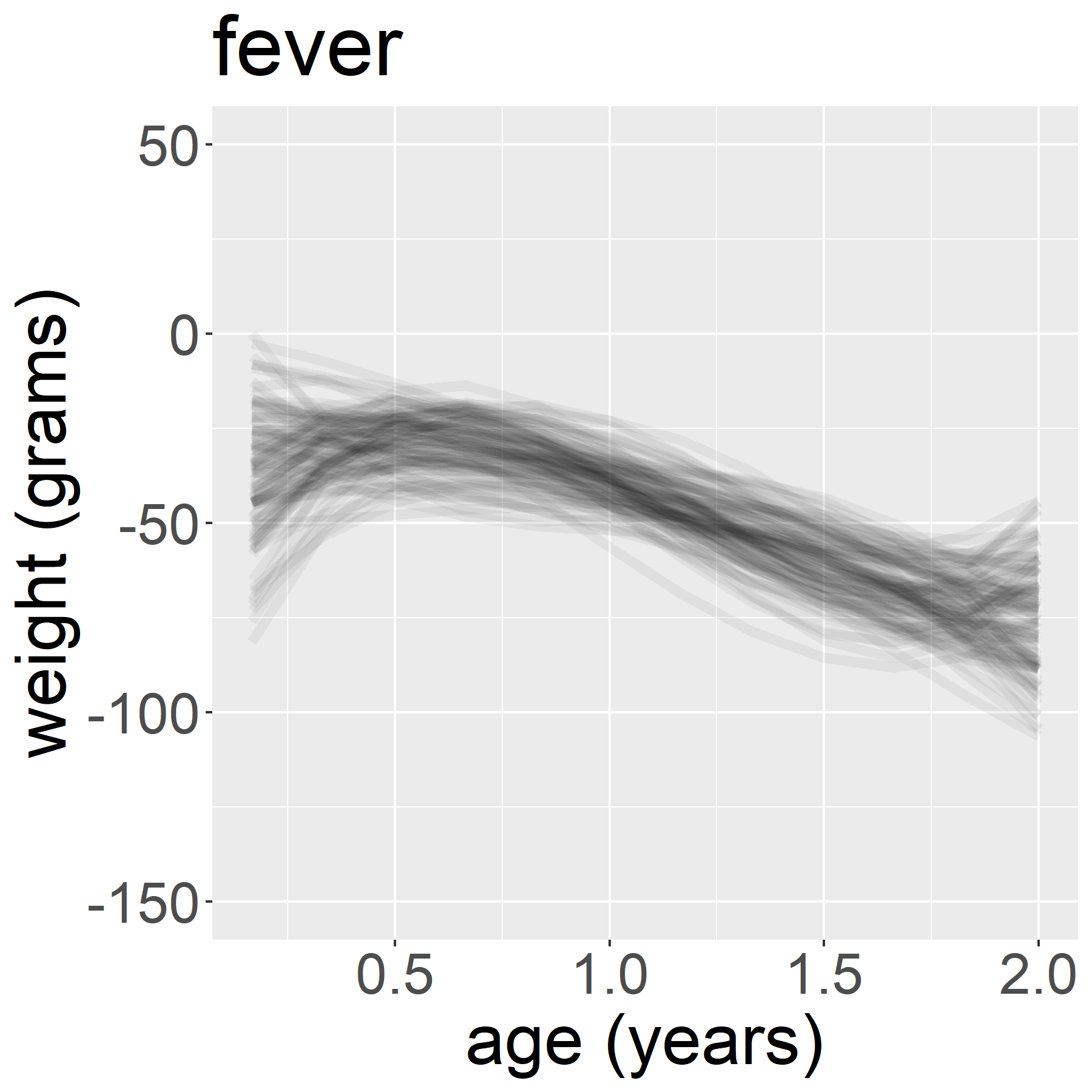}  & \includegraphics[width = 1.5 in]{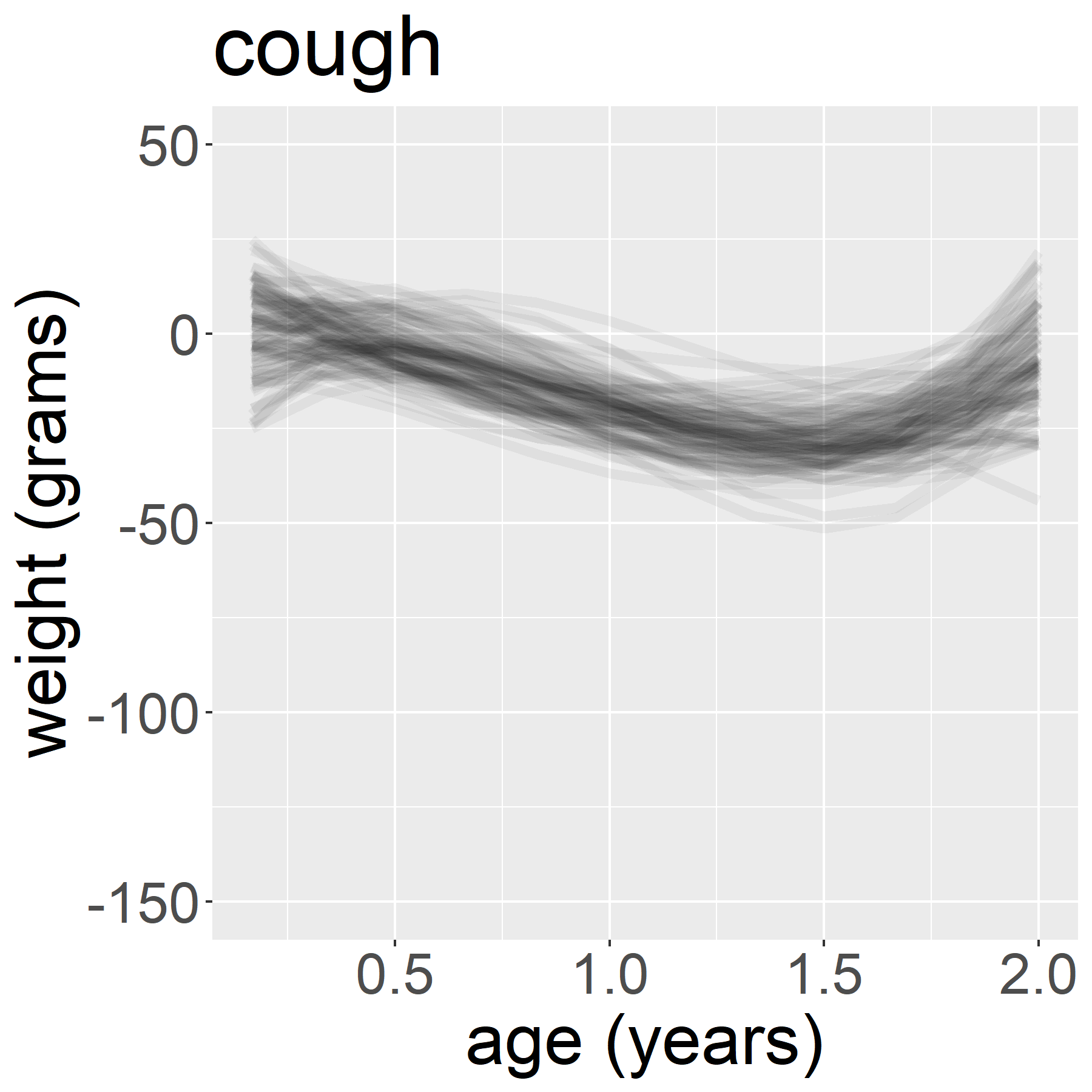} \\
(d) & (e) & (f)
\end{tabular}
\caption{ (a)-(c) Posterior samples of $\Phi(\mu^{z_j}(\protect\vv{t}))$ for $j=2,3,4$, representing the proportion of children in the study that experienced an illness by age.(d)-(f) Posterior samples of $\beta_j(\protect\vv{t})$, $j = 2,3,4$, representing the effect of experiencing an illness on weight. Children tend to experience cough more often than fever and fever more often than diarrhea. Experiencing diarrhea or fever tend to have a larger effect on weight than experiencing cough.}\label{fig:cebu_ill}
    \end{center}
\end{figure}

In order to compute the integral $\int_{T}\beta_1(s,\vv{t})z_{i,1}(s)ds$ related to the functional linear model relating breastfeeding status to weight, we rely on generalized functional factor analysis for model based imputation of the binary functional data $z_{i,1},\ i =1,\ldots,n$. Figure \ref{fig:cebu_bf_gfpca} panels (a)-(c) visualize the variability in breastfeeding dynamics captured through our generalized functional factor analysis model by showing posterior samples of plus and minus one standard deviation in the direction of the generalized loadings interpreted on the probability scale through the link function $\Phi(\cdot)$. The loadings mostly capture whether or not a child was breastfed at all and, if so, when breastfeeding was discontinued. In the posterior samples of the mean, the values at birth are lower than at 2 months. This is due to the fact that breastfeeding status at birth was measured within the first 2 days of birth, and it could take longer than 2 days for mothers to express breast milk for the first time after childbirth \citep{casey1986}. Panels (d)-(f) of Figure \ref{fig:cebu_bf_gfpca} shows posterior samples of underlying processes describing the probability that a child is fed breast milk by age for different subjects. In panel (d), breastfeeding status was only recorded at 4 ages (black dots). Between the ages when the child was reported to have breast milk, the model infers that the child was fed breast milk. After the last recording, the model is uncertain if and when the child stopped. In panel (e), initially the child was reported to not have breast milk, so the model infers that the child was not given breast milk after the age of the last recording. In panel (f), the child was reported to have breastfed early on, and stopped at some age after age 1. The model is uncertain when breastfeeding was discontinued. 

In order to understand the effects of breastfeeding on child weight, we study posterior samples of the functions $\int_{\mathcal{T}} \beta_1(s,s') 1_{s\leq s'}ds$, where $1_{s\leq s'}$ denotes the indicator function. These functions capture the effect of breastfeeding on weight at $s'$ years with other variables held constant. They can be compared to the effect of never breastfeeding captured by the zero function. Figure \ref{fig:cebu_bf_reg} shows the effect of breastfeeding for  $s' = 0,\ 0.5,\ 1,\ 2$ years. If a child is breastfed only at birth, the effect is close to zero (panel (a)). While breastfed and non-breastfed children are similar in weight at a younger age, their discrepancy is pronounced at larger ages, where breastfed children tend to weigh less. This difference is exaggerated the longer that a child is breastfed (panels (b) through (d)). The effect of breastfeeding tends to level off, as shown through $\|\int_{\mathcal{T}} \beta_1(s,s') 1_{s\leq s'}ds\|_2$ for different values of $s'$ shown in panel (e). Continuing breastfeeding for more months in the first year has a significant impact on weight, while continuing longer after age 1 has a relatively small impact. These inferences corroborate past work devoted to studying the effects of breastfeeding. \cite{dewey1998} conduct a meta-analysis studying effects of breastfeeding on child growth outcomes, including weight, height and  head circumference. They report that breastfed infants tend to self-regulate food intake, have lower observable metabolic rates, and are leaner than their non-breastfed counterparts. The differences in weight between breastfed and non-breastfed children tend to be noticeable after 6 months of age. 

In contrast with the breastfeeding covariates, there does not appear to be structured variability in timing of illness, with children experiencing diarrhea, fever and cough sporadically. Consequently, model based imputation is based only on a mean process interpreted on the probability scale through the link function $\Phi(\cdot)$. Figure \ref{fig:cebu_ill} panels (a)-(c) show posterior samples of the average probability that a child experiences diarrhea, fever, or cough by age. Generally, children experienced fever more often than diarrhea, and cough more often than fever. In all of these mean processes, there appear to be local maxima before year 1. For the illness longitudinal covariates we assume a concurrent regression model, which simplifies the integrals, $\int_{T}\beta_j(s,t)z_{i,j}(s)ds = \beta_j(t)z_{i,j}(t)$, for all $j = 2,3,4$, $i = 1\ldots,n$.  The effects on weight of experiencing these illnesses with all other variables held constant are shown in panels (d)-(f). Diarrhea tends to affect weight more than fever, and fever tends to affect weight more than cough. These regression functions tend to be more negative at higher ages. We speculate that the effects of diarrhea and fever affect weight more than cough because they can lead to dehydration.

In Appendix 5.3, we present an analysis of these data using \cite{crainiceanu2010}. The inferential results are generally consistent with our $\small{\mbox{NeMO}}$-based analysis with a few marked differences. In the competing model, there is no ability to propagate uncertainty in the estimated mean process and factor loadings while inferring covariate effects on weight. Consequently, the empirical Bayes setup of the competing model can distort the inferences made about effect sizes. There are discrepancies between the competing model and the $\small{\mbox{NeMO}}$ model in the magnitude of effects of season of birth and breastfeeding. In the competing model, season of birth has a larger effect on weight, and breastfeeding has a smaller effect on weight. To further compare the two models, we use the widely applicable information criterion to assess predictive performance \citep{vehtari2017}, with preference for the $\small{\mbox{NeMO}}$ model (69709.10) compared to the model of \cite{crainiceanu2010} (69905.07). Further, we assess fit of the $\small{\mbox{NeMO}}$ model by residual analysis, performing posterior predictive checks and considering reduced models omitting covariates with seemingly small effects. We find that the model is adequate in describing individual- and population-level features of the observations, and the full model with all covariates considered is preferred over reduced models according to the widely applicable information criterion. For this dataset, we conclude that the model provides reasonable fit to the observations while enabling interpretable inference with uncertainty quantification.

\section{Discussion} \label{sec:discussion}

In this work, we have presented a novel approach for Bayesian functional factor analysis through the introduction of $\mbox{NeMO}$ processes. As illustrated by the data analysis of the Cebu Longitudinal Health and Nutrition Survey, our approach can flexibly model functional data, while enabling interpretable inference with coherent uncertainty quantification. We have identified several directions for future work.

Our approach shrinks towards the orthogonality constraint needed to make meaningful inference about factor loadings, while lending itself to a computationally simple Gibbs sampler. Future work could explore more complex sampling regimes. For example, MCMC convergence properties during a burn-in period could be monitored to adaptively select the penalty parameters for $\mbox{NeMO}$ processes. Additionally, it will be interesting to further study this approach using gradient-based MCMC to simultaneously update all model parameters.

There is a rich literature on Bayesian factor analysis and the problem of how to best choose a prior for high-dimensional loadings. It is particularly popular to take an over-fitted factor modeling approach, starting with an upper bound on the number of factors $K$ and then choosing a shrinkage prior that favors setting columns of the loadings matrix for unnecessary factors equal to zero \citep{bhattacharya2011,Schiavon2022}. Such priors must increase shrinkage aggressively enough to protect against factor splitting but not too much so that important factors may have their loadings over-shrunk; achieving the appropriate tradeoff in practice is a subtle problem lacking associated theory. The use of mutual orthogonality to encourage factor sparsity represents an alternative to shrinkage methods, which deserves further theoretical and empirical study.


One aspect of functional data that we have not considered is phase variability, corresponding to the variability of the relative timing of functional features, such as extrema. \cite{marron2015} discuss this concept and state that models that explicitly account for this variability can be more parsimonious. In future work, we plan to extend our model to account for phase variability, along the lines of \cite{kneip2008,earls2017}. One challenge is developing a prior model for phase that is computationally efficient and can scale to large samples.

Our implementation utilizes independent parent Gaussian processes with squared exponential covariance kernels, though this framework is naturally extensible to alternative parent distributions. For example, it is automatic to consider multivariate Gaussian parent distributions, obtaining new types of methods for Bayesian factor analysis that can shrink towards the distribution discussed in \cite{jauch2021}. It is also interesting to consider parent Gaussian processes that assume different covariance kernels or more intricate parent processes, such as L\'evy processes.

\section*{Acknowledgements}
This work was supported by NIH grants RO1ES027489 and R01ES028804. 

\clearpage

  \bigskip
  \bigskip
  \bigskip
  \begin{center}
    {\LARGE Supplemental Material: Bayesian modeling of nearly mutually orthogonal processes}
\end{center}
  \medskip
  
\linespread{1}

\begin{abstract}
This is the supplemental material for the paper entitled `Bayesian modeling of nearly mutually orthogonal processes'. In Section \ref{sec:proofs}, we present proofs for the propositions stated in the main paper. In Section \ref{sec:prior_illustration}, we present figures that illustrate the role of covariance hyperparameters in the prior. In Section \ref{sec:MCMC}, we provide full implementation details for the functional factor analysis models and generalized functional factor analysis model for binary functional data. In Section \ref{sec:additionalSim}, we present additional simulation examples that compare our approach to other state-of-the-art methods. In Section \ref{sec:cebuAdd}, we provide additional details related to the data analysis of the Cebu Longitudinal Health and Nutrition Survey.
\end{abstract}

\linespread{2}

\appendix

\section{Proofs}\label{sec:proofs}

\begin{proposition}
Let $E_\omega :=  \{\lambda_1,\ldots,\lambda_K\in\mathcal{H}\ \mid\ \sum_{k = 1}^K\sum_{j<k}\langle\lambda_j,\lambda_k \rangle^2 > \omega \}$. For any $\nu_\lambda > 0$ and measurable subset $E_\omega'\subseteq E_\omega$ with non-zero $\mathcal{G}(E_\omega')$, \begin{equation}
    \frac{\mathcal{N}(E_\omega')}{\mathcal{G}(E_\omega')} \leq F\exp\big(-\frac{\omega}{2\nu_\lambda}\big),\nonumber
\end{equation}
where $F = \big(\mathbb{E}_{\mathcal{G}}[\exp{\{-\frac{1}{2\nu_\lambda}\sum_{k = 1}^K\sum_{j<k}\langle \lambda_j,\lambda_k\rangle^2\}}]\big)^{-1}$ is a normalizing constant.
\end{proposition}

\begin{proof}
Note that for any $\lambda_1,\ldots,\lambda_K\in E_{\omega}'$, $\exp\big\{-\frac{1}{2\nu_\lambda} \sum_{k = 1}^K\sum_{j<k}\langle\lambda_j,\lambda_k \rangle^2\big\} \leq \exp\big(-\frac{\omega}{2\nu_\lambda}\big)$.

By the definition of $\mathcal{N}$,
\begin{align}
\mathcal{N}(E_\omega')& = F \int_{E_\omega'} \exp\big\{-\frac{1}{2\nu_\lambda} \sum_{k = 1}^K\sum_{j<k}\langle\lambda_j,\lambda_k \rangle^2\big\} d\mathcal{G}(\lambda_1,\ldots,\lambda_K)\nonumber\\
& =  F\int_{\mathcal{H}} 1_{\{\lambda_1,\ldots,\lambda_K \in E_\omega'\}}\exp\big\{-\frac{1}{2\nu_\lambda} \sum_{k = 1}^K\sum_{j<k}\langle\lambda_j,\lambda_k \rangle^2\big\}d\mathcal{G}(\lambda_1,\ldots,\lambda_K)\nonumber\\
& \leq F\exp\big(-\frac{\omega}{2\nu_\lambda}\big) \int_{E_\omega'} d\mathcal{G}(\lambda_1,\ldots,\lambda_K)  = F\exp\big(-\frac{\omega}{2\nu_\lambda}\big) \mathcal{G}(E_\omega')\nonumber
\end{align}
Since we assumed $\mathcal{G}(E_\omega')>0$, it follows
\begin{equation}
    \frac{\mathcal{N}(E_\omega')}{\mathcal{G}(E_\omega')} \leq F\exp\big(-\frac{\omega}{2\nu_\lambda}\big).\nonumber
\end{equation}
\end{proof}

\begin{proposition}
Let $E :=  \{\lambda_1,\ldots,\lambda_K\in\mathcal{H}\ \mid\ \langle\lambda_j,\lambda_k \rangle \neq 0, \text{ for some } j\neq k \}$.
For any measurable subset $E'\subseteq E$, $\underset{\nu_\lambda\rightarrow 0}{\lim}\mathcal{N}(E') = 0$.
\end{proposition}

\begin{proof}
For any $(\lambda_1,\ldots,\lambda_k) \in E'$, $\underset{\nu_\lambda\rightarrow 0}{\lim}\exp{\big\{-\frac{1}{2\nu_\lambda}\sum_{k = 1}^K\sum_{j<k}\langle \lambda_j,\lambda_k\rangle^2\big\} = 0}$.

By definition of $\mathcal{N}$,
\begin{align}
    \underset{\nu_\lambda\rightarrow 0}{\lim}\mathcal{N}(E') & \propto \underset{\nu_\lambda\rightarrow 0}{\lim}\int_{E'} \exp{\{-\frac{1}{2\nu_\lambda}\sum_{k = 1}^K\sum_{j<k}\langle \lambda_j,\lambda_k\rangle^2\big\}} d\mathcal{G}(\lambda_1,\ldots,\lambda_K) \label{prop_1_e1}\\
    & = \int_{\mathcal{H}} \underset{\nu_\lambda\rightarrow 0}{\lim} 1_{\{\lambda_1,\ldots,\lambda_K\in 
    E'\}}\exp{\big\{-\frac{1}{2\nu_\lambda}\sum_{k = 1}^K\sum_{j<k}\langle \lambda_j,\lambda_k\rangle^2\big\}} d\mathcal{G}(\lambda_1,\ldots,\lambda_K) \label{prop_1_e2}\\
    & = \int_{\mathcal{H}} 0\, d\mathcal{G}(\lambda_1,\ldots,\lambda_K) = 0.\nonumber
\end{align}
The equality from Equation \eqref{prop_1_e1} to \eqref{prop_1_e2} follows from the Bounded Convergence Theorem, since $\exp{\big\{-\frac{1}{2\nu_\lambda}\sum_{k = 1}^K\sum_{j<k}\langle \lambda_j,\lambda_k\rangle^2\big\}}\leq 1$.

\end{proof}

\begin{proposition}
For any measurable set $A\subseteq \mathcal{H}$, $\underset{\nu_\lambda\rightarrow \infty}{\lim}\mathcal{N}(A) = \mathcal{G}(A)$.
\end{proposition}

\begin{proof}

For any measurable subset $A\subset \mathcal{H}$, $\underset{\nu_\lambda\rightarrow \infty}{\lim}\exp{\big\{-\frac{1}{2\nu_\lambda}\sum_{k = 1}^K\sum_{j<k}\langle \lambda_j,\lambda_k\rangle^2\big\}} = 1$ for any $\lambda_1,\ldots,\lambda_K\in A$. 

By definition of $\mathcal{N}$,
\begin{align}
\underset{\nu_\lambda\rightarrow \infty}{\lim}\mathcal{N}(A) & \propto \underset{\nu_\lambda\rightarrow \infty}{\lim}\int_{A} \exp{\big\{-\frac{1}{2\nu_\lambda}\sum_{k = 1}^K\sum_{j<k}\langle \lambda_j,\lambda_k\rangle^2\big\}} d\mathcal{G}(\lambda_1,\ldots,\lambda_K) \label{prop_2_e1}\\
& = \int_{\mathcal{H}} \underset{\nu_\lambda\rightarrow \infty}{\lim} 1_{\{\lambda_1,\ldots,\lambda_K\in 
A\}}\exp{\big\{-\frac{1}{2\nu_\lambda}\sum_{k = 1}^K\sum_{j<k}\langle \lambda_j,\lambda_k\rangle^2\big\}} d\mathcal{G}(\lambda_1,\ldots,\lambda_K) \label{prop_2_e2}\\
& = \int_{\mathcal{H}} 1_{\{\lambda_1,\ldots,\lambda_K\in 
A\}} d\mathcal{G}(\lambda_1,\ldots,\lambda_K) = \int_{A}d\mathcal{G}(\lambda_1,\ldots,\lambda_K) = \mathcal{G}(A).\nonumber
\end{align}
The equality from Equation \eqref{prop_2_e1} to \eqref{prop_2_e2} follows from the Bounded Convergence Theorem, since $\exp{\big\{-\frac{1}{2\nu_\lambda}\sum_{k = 1}^K\sum_{j<k}\langle \lambda_j,\lambda_k\rangle^2\big\}}\leq 1$. Note that in Equation \eqref{prop_2_e1}, the constant of proportionality is 

$$\big(\underset{\nu_\lambda\rightarrow \infty}{\lim} \int_{\mathcal{H}} \exp{\big\{-\frac{1}{2\nu_\lambda}\sum_{k = 1}^K\sum_{j<k}\langle \lambda_j,\lambda_k\rangle^2\big\}} d\mathcal{G}(\lambda_1,\ldots,\lambda_K)\big)^{-1} = \mathcal{G}(\mathcal{H})^{-1}= 1$$.
\end{proof}

\begin{proposition}
Given $\{\lambda_j\}_{j\neq k}$, $\lambda_k$ is a zero-mean Gaussian process with covariance
\begin{align}
        & C^{\nu_\lambda}_k(s,t) = C_k(s,t) - h_{\Lambda_{(-k)}}(s)^\top\{\nu_\lambda I_{K-1} + H_{\Lambda_{(-k)}}\}^{-1}h_{\Lambda_{(-k)}}(t),\text{ where} \label{eq:relaxedcov} \\
    & h_{\Lambda_{(-k)}}(t) = \int_{\mathcal{T}} C_k(s,t)\Lambda_{(-k)}(s)ds,\ H_{\Lambda_{(-k)}} = \int_{\mathcal{T}}\int_{\mathcal{T}} C_k(s,s')\Lambda_{(-k)}(s)\Lambda_{(-k)}(s')^\top ds ds' \label{eq:integrals}
\end{align}
with $\Lambda_{(-k)}(s) = \{\lambda_1(s),\ldots,\lambda_{k-1}(s),\lambda_{k+1}(s),\ldots,\lambda_K(s)\}^{\top}$.
\end{proposition}

\begin{proof}
Let $H_k$ denote the reproducing kernel Hilbert space defined by the $k$\textsuperscript{th} parent Gaussian process, with Gaussian measure $\mathcal{G}_k$. Similarly, let $H_{(-k)}$ denote the product reproducing kernel Hilbert spaces defined by all parent Gaussian Processes except the $k$\textsuperscript{th}, with Gaussian product measure $\mathcal{G}_{(-k)}$. Then, for any measurable $A\subset H_k$ and $B\subset H_{(-k)}$,
\begin{align}
    P_\mathcal{N}(\lambda_k\in A,\{\lambda_j\}_{j\neq k}\in B)& \propto \nonumber \\ & \hspace{-.5in} \int_\mathcal{H}1_{\big\{\lambda_k \in A,\{\lambda_j\}_{j\neq k}\in B\big\}} \exp\big\{-\frac{1}{2\nu_\lambda} \sum_{k = 1}^K\sum_{j<k}\langle\lambda_j,\lambda_k \rangle^2\big\} d\mathcal{G}(\lambda_1,\ldots,\lambda_K) \nonumber\\
    &  \hspace{-.5in}= \int_{\mathcal{H}_{(-k)}}\int_{\mathcal{H}_{k}}1_{\{\lambda_k \in A\}} \exp\big\{-\frac{1}{2\nu_\lambda} \sum_{j \neq k}\langle\lambda_j,\lambda_k \rangle^2\big\} d\mathcal{G}_k(\lambda_k) \nonumber\\
    & \hspace{-.45in}\times 1_{\big\{\{\lambda_j\}_{j\neq k}\in B\big\}} \exp\big\{-\frac{1}{2\nu_\lambda} \sum_{j \neq k }\sum_{j'<j}\langle\lambda_j,\lambda_{j'} \rangle^2\big\}d\mathcal{G}_{(-k)}(\{\lambda_j\}_{j\neq k}), \nonumber\\
    & \hspace{-.5in}= \int_{B}\int_{A} \exp\big\{-\frac{1}{2\nu_\lambda} \sum_{j \neq k}\langle\lambda_j,\lambda_k \rangle^2\big\} d\mathcal{G}_k(\lambda_k) d\mathcal{N}_{(-k)}(\{\lambda_j\}_{j\neq k}),\nonumber
\end{align}
where $\mathcal{N}_{(-k)}$ is the $\small{\mbox{NeMO}}$ measure defined using all parent Gaussian processes except the $k$\textsuperscript{th}. We denote $\mathcal{N}_k(A) \propto \int_{A} \exp\big\{-\frac{1}{2\nu_\lambda} \sum_{j \neq k}\langle\lambda_j,\lambda_k \rangle^2\big\} d\mathcal{G}_k(\lambda_k)$, which is the conditional probability measure \citep[Section 5.8]{resnick2019} induced by $\lambda_k\mid\{\lambda_j\}_{j\neq k}$, when $\lambda_1,\ldots,\lambda_K\sim \mathcal{N}$.

To show $\lambda_k\mid\{\lambda_j\}_{j\neq k}$ is a Gaussian process, we will show the conditional characteristic functional $\mathbb{E}_{\mathcal{N}_k}\big[\exp\big\{i\langle U,\lambda_k\rangle\big\}\mid\{\lambda_j\}_{j\neq k}\big]$ coincides with the characteristic functional of a zero-mean Gaussian process with covariance function $C^{\nu_\lambda}_k(\cdot,\cdot)$, for any function $U\in\mathcal{H}_k$ that satisfies $\mathbb{E}_{\mathcal{N}_k}\big[\exp\big\{i\langle U,\lambda_k\rangle\big\}\mid\{\lambda_j\}_{j\neq k}\big] < \infty$. We will make use of the lower Riemann sum, which satisfies $\frac{|T|}{m}\sum_{l=1}^m U(t^*_l)\lambda_k(t^*_l)  \leq \langle U,\lambda_k\rangle$ for $t_l = a + l\frac{(b-a)}{m}, \ l = 0,1,\ldots,m$, $t^*_l\in [t_{l-1},t_l],\ l=1,\ldots,m$, $T = [a,b]$. In general, we use $f(\vv{t}^*) = (f(t_1^*),\ldots,f(t_m^*))^\top$ to denote a vector of function evaluations on the grid $\vv{t}^* = (t_1^*,\ldots,t_m^*)^\top$.

First, note
\begin{align}
    \mathbb{E}_{\mathcal{G}_k}\Big[\exp\Big\{-\frac{1}{2\nu_\lambda}\sum_{j\neq k}\langle\lambda_j,\lambda_k\rangle^2\Big\}\mid\{\lambda_j\}_{j\neq k}\Big] \nonumber & \\ & \hspace{-2in} = \mathbb{E}_{\mathcal{G}_k}\Big[\underset{m\rightarrow\infty}{\lim}\exp\Big\{-\frac{|T|^2}{2\nu_\lambda m^2}\lambda_j(\vv{t}^*)^\top \Lambda_{-k}(\vv{t}^*)\Lambda_{-k}(\vv{t}^*)^\top\lambda_j(\vv{t}^*)\Big\}\mid\{\lambda_j\}_{j\neq k}\Big] \label{eq:prop4_eq1}\\
    & \hspace{-2in} = \underset{m\rightarrow\infty}{\lim}\int_{\mathbb{R}^m } \exp\Big\{-\frac{|T|^2}{2\nu_\lambda m^2}\lambda_j(\vv{t}^*)^\top  \Lambda_{-k}(\vv{t}^*)\Lambda_{-k}(\vv{t}^*)^\top\lambda_j(\vv{t}^*)\Big\}\nonumber\\
    & \hspace{-2 in} \times \big(2\pi \det\big[C_k(\vv{t}^*,\vv{t}^*)\big]\big)^{-\frac{1}{2}}\exp\Big\{-\frac{1}{2}\lambda_j(\vv{t}^*)^\top C_k(\vv{t}^*,\vv{t}^*)^{-1}\lambda_j(\vv{t}^*)\Big\}d\lambda_k(\vv{t}^*) \label{eq:prop4_eq2}\\
    & \hspace{-2 in} =  \det\big[C_k(\vv{t}^*,\vv{t}^*)\big]^{-\frac{1}{2}} \det\Big[C_k(\vv{t}^*,\vv{t}^*)^{-1} + \frac{|T|^2}{\nu_\lambda m^2}  \Lambda_{-k}(\vv{t}^*)\Lambda_{-k}(\vv{t}^*)^\top\Big]^{-\frac{1}{2}}\nonumber
\end{align}
where the interchange of expectation and the limit from Equation \eqref{eq:prop4_eq1} to \eqref{eq:prop4_eq2} follows from the bounded convergence theorem, since 

$$\exp\big\{-\frac{|T|^2}{2\nu_\lambda m^2}\lambda_j(\vv{t}^*)^\top \Lambda_{-k}(\vv{t}^*)\Lambda_{-k}(\vv{t}^*)^\top\lambda_j(\vv{t}^*)\big\}\leq 1.$$ 

Similarly, note
\begin{align}
    \mathbb{E}_{\mathcal{G}_k}\Big[\exp\big\{i\langle U,\lambda_k\rangle\big\}\exp\Big\{-\frac{1}{2\nu_\lambda}\sum_{j\neq k}\langle\lambda_j,\lambda_k\rangle^2\Big\}\mid\{\lambda_j\}_{j\neq k}\Big] \nonumber & \\ & \hspace{-3.2in} = \underset{m\rightarrow\infty}{\lim}\int_{\mathbb{R}^m }  \exp\Big\{i\frac{|T|}{m}U(\vv{t}^*)^\top\lambda_k(\vv{t}^*)\Big\}\big(2\pi \det\big[C_k(\vv{t}^*,\vv{t}^*)\big]\big)^{-\frac{1}{2}} \nonumber \\ &
    \hspace{-3.25in}\times \exp\Big\{-\frac{1}{2}\lambda_j(\vv{t}^*)^\top \big (C_k(\vv{t}^*,\vv{t}^*)^{-1} + \frac{|T|^2}{\nu_\lambda m^2}  \Lambda_{-k}(\vv{t}^*)\Lambda_{-k}(\vv{t}^*)^\top\big)\lambda_j(\vv{t}^*)\Big\}d\lambda_k(\vv{t}^*) \label{eq:prop4_eq3}, 
\end{align}
where the interchange of expectation and limit in Equation \eqref{eq:prop4_eq3} follows from the bounded convergence theorem, since $$\big|\exp\big\{i\frac{|T|}{m}U(\vv{t}^*)^\top\lambda_k(\vv{t}^*)-\frac{|T|^2}{2\nu_\lambda m^2}\lambda_j(\vv{t}^*)^\top \Lambda_{-k}(\vv{t}^*)\Lambda_{-k}(\vv{t}^*)^\top\lambda_j(\vv{t}^*)\big\}\big|\leq 1.$$

Finally, consider
\begin{align}
    \mathbb{E}_{\mathcal{N}_k}\big[\exp\big\{i\langle U,\lambda_k\rangle\big\}\mid\{\lambda_j\}_{j\neq k}\big] & = \frac{\mathbb{E}_{\mathcal{G}_k}\big[\exp\big\{i\langle U,\lambda_k\rangle\big\}\exp\big\{-\frac{1}{2\nu_\lambda}\sum_{j\neq k}\langle\lambda_j,\lambda_k\rangle^2\big\}\mid\{\lambda_j\}_{j\neq k}\big]}{\mathbb{E}_{\mathcal{G}_k}\big[\exp\big\{-\frac{1}{2\nu_\lambda}\sum_{j\neq k}\langle\lambda_j,\lambda_k\rangle^2\big\}\mid\{\lambda_j\}_{j\neq k}\big]}\nonumber\\
    & \hspace{-1.75 in} = \underset{m\rightarrow\infty}{\lim}\int_{\mathbb{R}^m} \exp\big\{i\frac{|T|}{m}U(\vv{t}^*)^\top\lambda_k(\vv{t}^*)\big\} (2\pi)^{-\frac{1}{2}} \\
    &\times\det\big[C_k(\vv{t}^*,\vv{t}^*)^{-1} + \frac{|T|^2}{\nu_\lambda m^2}  \Lambda_{-k}(\vv{t}^*)\Lambda_{-k}(\vv{t}^*)^\top\big]^{\frac{1}{2}} \nonumber \\
    & \hspace{-1.75 in} \times \exp\Big\{-\frac{1}{2}\lambda_j(\vv{t}^*)^\top \big(C_k(\vv{t}^*,\vv{t}^*)^{-1} + \frac{|T|^2}{\nu_\lambda m^2}  \Lambda_{-k}(\vv{t}^*)\Lambda_{-k}(\vv{t}^*)^\top\big)\lambda_j(\vv{t}^*)\Big\}d\lambda_k(\vv{t}^*)\nonumber \\
    & \hspace{-1.75 in} = \underset{m\rightarrow\infty}{\lim}\exp\Big\{-\frac{1}{2}\frac{|T|^2}{m^2}U(\vv{t}^*)^\top\big(C_k(\vv{t}^*,\vv{t}^*)^{-1} + \frac{|T|^2}{\nu_\lambda m^2}  \Lambda_{-k}(\vv{t}^*)\Lambda_{-k}(\vv{t}^*)^\top\big)^{-1}U(\vv{t}^*)\Big\} \label{eq:prop4_eq4} \\
    & \hspace{-1.75 in} = \underset{m\rightarrow\infty}{\lim}\exp\bigg\{-\frac{1}{2}\frac{|T|^2}{m^2}U(\vv{t}^*)^\top\Big(C_k(\vv{t}^*,\vv{t}^*) - \frac{|T|}{m}C_k(\vv{t}^*,\vv{t}^*)\Lambda_{-k}(\vv{t}^*) \nonumber\\
    & \hspace{-1.75 in} \times \big(\nu_\lambda I_{K-1} + \frac{|T|^2}{ m^2}  \Lambda_{-k}(\vv{t}^*)^\top C_k(\vv{t}^*,\vv{t}^*)\Lambda_{-k}(\vv{t}^*)\big)^{-1}\frac{|T|}{m}\Lambda_{-k}(\vv{t}^*)^\top C_k(\vv{t}^*,\vv{t}^*)\Big)U(\vv{t}^*)\bigg\} \label{eq:prop4_eq5} \\
    & \hspace{-1.75 in} = \exp\big\{-\frac{1}{2}\int_{\mathcal{T}}\int_{\mathcal{T}} C_k^{\nu_\lambda}(s,s')U(s)U(s')dsds'\big\}\label{eq:prop4_eq6},
\end{align}
where the matrix inversion from Equations \eqref{eq:prop4_eq4} to \eqref{eq:prop4_eq5}  is computed via the Woodbury's formula \citep{harville1998}. In Equation \eqref{eq:prop4_eq6}, the characteristic functional of $\lambda_k\mid\{\lambda_j\}_{j\neq k}$ coincides with the characteristic functional of a zero-mean Gaussian process with covariance function $C_k^{\nu_\lambda}(\cdot,\cdot)$, as desired.
\end{proof}

\begin{proposition}
    $\lambda_k\mid\{\lambda_j\}_{j\neq k}$ converges in distribution to $\lambda_k^\perp\mid\{\lambda_j\}_{j\neq k}$ as $\nu_\lambda\rightarrow 0$.
\end{proposition}

\begin{proof}
Under the assumption that $H$ is a positive definite matrix, it follows that $\underset{\nu_\lambda\rightarrow0}{\lim}(\nu_\lambda I_{K-1} + H)^{-1} = H^{-1}$. Consider, 
\begin{align}
    \underset{\nu_\lambda\rightarrow0}{\lim}\mathbb{E}_{\mathcal{N}_k}\big[\exp\big\{i\langle U,\lambda_k\rangle\big\}\mid\{\lambda_j\}_{j\neq k}\big] & = \underset{\nu_\lambda\rightarrow0}{\lim}\exp\big\{-\frac{1}{2}\int_{\mathcal{T}}\int_{\mathcal{T}} C_k^{\nu_\lambda}(s,s')U(s)U(s')dsds'\big\}\nonumber\\ 
    & \hspace{-1in} = \exp\big\{-\frac{1}{2}\int_{\mathcal{T}}\int_{\mathcal{T}} C_k(s,s')U(s)U(s')dsds'\big\} \nonumber \\
    & \hspace{-1in} \times \underset{\nu_\lambda\rightarrow0}{\lim}\exp\big\{\frac{1}{2}(\int_\mathcal{T}U(s){h}(s)ds)^\top(\nu_\lambda I_{K-1} + H)^{-1} (\int_\mathcal{T}U(s'){h}(s')ds')\big\} \nonumber\\
    & \hspace{-1in} = \exp\big\{-\frac{1}{2}\int_{\mathcal{T}}\int_{\mathcal{T}} C_k(s,s')U(s)U(s')dsds'\big\} \nonumber \\
    & \hspace{-1in} \times \exp\big\{\frac{1}{2}(\int_\mathcal{T}U(s){h}(s)ds)^\top H^{-1} (\int_\mathcal{T}U(s'){h}(s')ds')\big\} \nonumber\\
    &  \hspace{-1in} = \exp\big\{-\frac{1}{2}\int_{\mathcal{T}}\int_{\mathcal{T}} C_{\Lambda_{(-k)}}^\perp(s,s')U(s)U(s')dsds'\big\}.\nonumber
\end{align}
\end{proof}

\begin{proposition}
Assuming $\{\lambda_1,\ldots,\lambda_K\} \sim \mathcal{N}$ with $\mathcal{N}$, 
    $\langle \lambda_j,\lambda_k\rangle\mid\{\lambda_j\}_{j\neq k}$ converges in probability to $0$ for any $j\neq k$ as $\nu_\lambda\rightarrow 0$.
\end{proposition}

\begin{proof}
First, consider
\begin{align}
    \mathbb{E}_{\mathcal{N}_k}\big[\int_{\mathcal{T}}\int_{\mathcal{T}}|\lambda_k(s)\lambda_k(s')\lambda_j(s)\lambda_j(s')|dsds'\mid\{\lambda_j\}_{j\neq k}\big] & = \nonumber \\ & \hspace{-1in}  \mathbb{E}_{\mathcal{N}_k}[\|\lambda_j\lambda_k\|_1^2|\{\lambda_j\}_{j\neq k}] \label{eq:prop6_eq1} \\
    & \hspace{-1in} \leq \|\lambda_j\|_2^2\mathbb{E}_{\mathcal{N}_k}\big[\|\lambda_k\|_2^2|\{\lambda_j\}_{j\neq k}\big] \nonumber\\
    &  \hspace{-1in}  = \|\lambda_j\|_2^2\int_\mathbb{T} \mathbb{E}_{\mathcal{N}_k}\big[\lambda_k(s)^2|\{\lambda_j\}_{j\neq k}\big]ds \nonumber\\
    & \hspace{-1in} = \|\lambda_j\|_2^2\int_\mathbb{T} C^{\nu_\lambda}_k(s,s)ds \nonumber\\
    & \hspace{-1in} \leq \|\lambda_j\|_2^2\int_\mathbb{T} C_k(s,s)ds \label{eq:prop6_eq2} 
\end{align}
Assuming that the covariance functions of the parent Gaussian Processes are defined so that sample paths are almost surely continuous, Equation \eqref{eq:prop6_eq2} is almost surely finite \citep{marcus1972}. This implies that the expectation and integration operations in the left hand side of Equation \eqref{eq:prop6_eq1} can be interchanged, by Fubini's theorem.

Define $e_j$ to be a binary vector of length $K-1$ that satisfies $\lambda_j(t) = e_j^\top\Lambda_{(-k)}(t)$. For any $\epsilon>0$, by Markov's inequality, 
\begin{align}
    \underset{\nu_\lambda\rightarrow 0}{\lim}P_{\mathcal{N}_k}\big(|\langle \lambda_j,\lambda_k \rangle|>\epsilon|\{\lambda_j\}_{j\neq k}\big) & \leq \underset{\nu_\lambda\rightarrow 0}{\lim}\frac{1}{\epsilon^2}\mathbb{E}_{\mathcal{N}_k}\big[|\langle \lambda_j,\lambda_k \rangle|^2|\{\lambda_j\}_{j\neq k}\big] \nonumber\\
    & \hspace{-1 in}= \underset{\nu_\lambda\rightarrow 0}{\lim}\frac{1}{\epsilon^2}\mathbb{E}_{\mathcal{N}_k}\big[\int_{\mathcal{T}}\int_{\mathcal{T}}\lambda_k(s)\lambda_k(s')\lambda_j(s)\lambda_j(s')dsds'|\{\lambda_j\}_{j\neq k}\big] \nonumber\\
    & \hspace{-1 in}= \underset{\nu_\lambda\rightarrow 0}{\lim}\frac{1}{\epsilon^2}\int_{\mathcal{T}}\int_{\mathcal{T}}\mathbb{E}_{\mathcal{N}_k}[\lambda_k(s)\lambda_k(s')\lambda_j(s)\lambda_j(s')|\{\lambda_j\}_{j\neq k}]dsds' \nonumber\\
    & \hspace{-1 in}= \underset{\nu_\lambda\rightarrow 0}{\lim}\frac{1}{\epsilon^2}\int_{\mathcal{T}}\int_{\mathcal{T}}C_k^{\nu_\lambda}(s,s')\lambda_j(s)\lambda_j(s')dsds' \nonumber\\
    & \hspace{-1 in}= \underset{\nu_\lambda\rightarrow 0}{\lim}\frac{1}{\epsilon^2}\big(e_j^\top H e_j - e_j^\top H(\nu_\lambda I_{K-1} + H)^{-1} H e_j\big) = 0.  \nonumber
\end{align}

\end{proof}

\begin{proposition}
$\eta_{\cdot k}$ converges in distribution to $(I_n - \frac{1}{n}1_n1_n^\top)N(0,I_n)$ as $\nu_\eta \rightarrow 0$.
\end{proposition}

\begin{proof}
First notice, $(I_n - \frac{1}{n}{1}_n{1}_n^\top)N(0,I_n) = N(0,(I_n - \frac{1}{n}{1}_n{1}_n^\top))$ in distribution, since $(I_n - \frac{1}{n}{1}_n{1}_n^\top)$ is a symmetric idempotent matrix. 
Next, for any $\epsilon > 0$, choose $\nu_\eta < \frac{\epsilon n^2}{\|{1}_n{1}_n^\top\|_2}$, and consider 
\begin{eqnarray}
\|\text{var}(\eta_{\cdot k}) - (I_n - \frac{1}{n} 1_n1_n^\top)\|_2 & = & \frac{\nu_\eta}{n\nu_\eta + n^2}\|{1}_n{1}_n^\top\|_2 \nonumber \\
 & < & \frac{\nu_\eta}{n^2}\|{1}_n{1}_n^\top\|_2 <  \epsilon. \nonumber
\end{eqnarray}
Hence, $\text{var}(\eta_{\cdot k})\rightarrow  (I_n - \frac{1}{n}J)$ as $\nu_\eta \rightarrow 0$.
Finally, for $s\in\mathbb{R}^n$, consider the characteristic function of $\eta_{\cdot k}$,
\begin{eqnarray}
   \lim_{\nu_\eta\rightarrow0}\mathbb{E}\big[\exp\big(i s^\top\eta_{\cdot k}\big)\big] &=& \lim_{\nu_\eta\rightarrow0}\exp{\big(-\frac{1}{2}s^\top  \text{var}(\eta_{\cdot k})s\big)}\nonumber\\
    &=& \exp{\big(-\frac{1}{2}s^\top  (I_n - \frac{1}{n}{1}_n{1}_n^\top)s\big)},\nonumber
\end{eqnarray}
as desired.

\end{proof}

\section{Additional Prior Illustration Figures}\label{sec:prior_illustration}

\begin{figure}[b!]
\begin{center}
 \begin{tabular}{c|cc}
\includegraphics[width = 1 in]{prior_figures/simu_lambda_minus_k_high_res.png} & \includegraphics[width = 1 in]{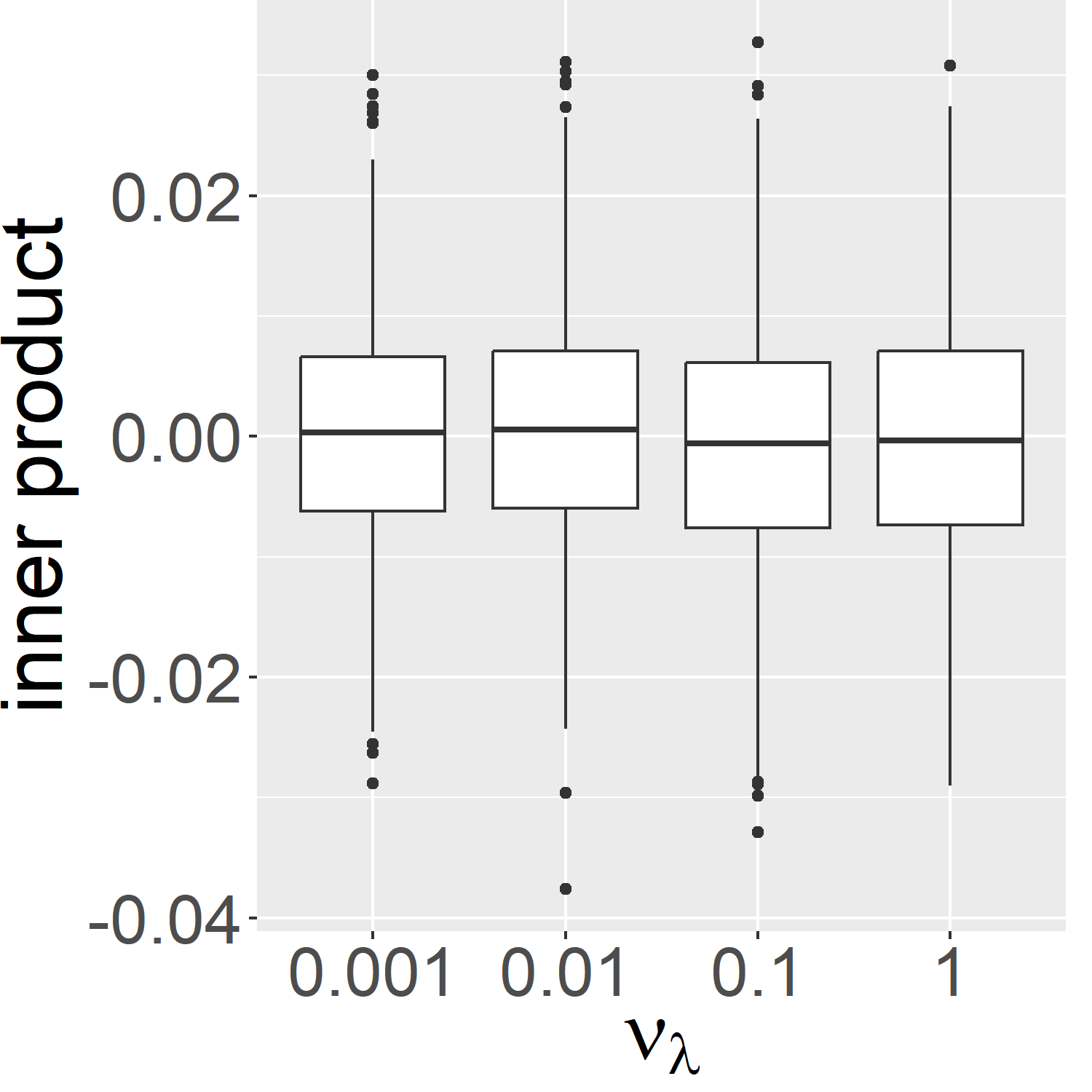} & \includegraphics[width = 1 in]{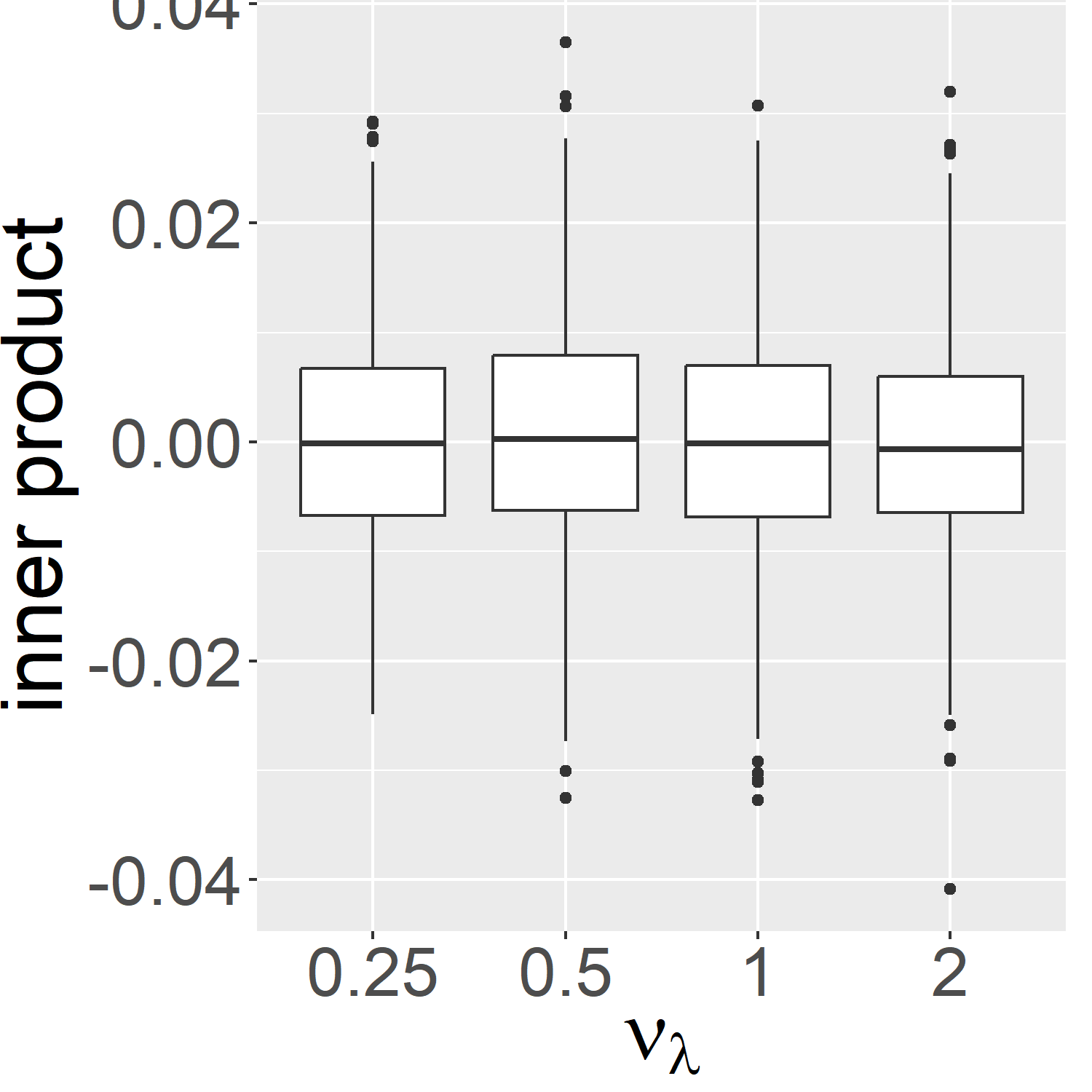}  \\
(a) & (b)  & (c)\\
\end{tabular}
 \begin{tabular}{cccc}
\includegraphics[width = 1 in]{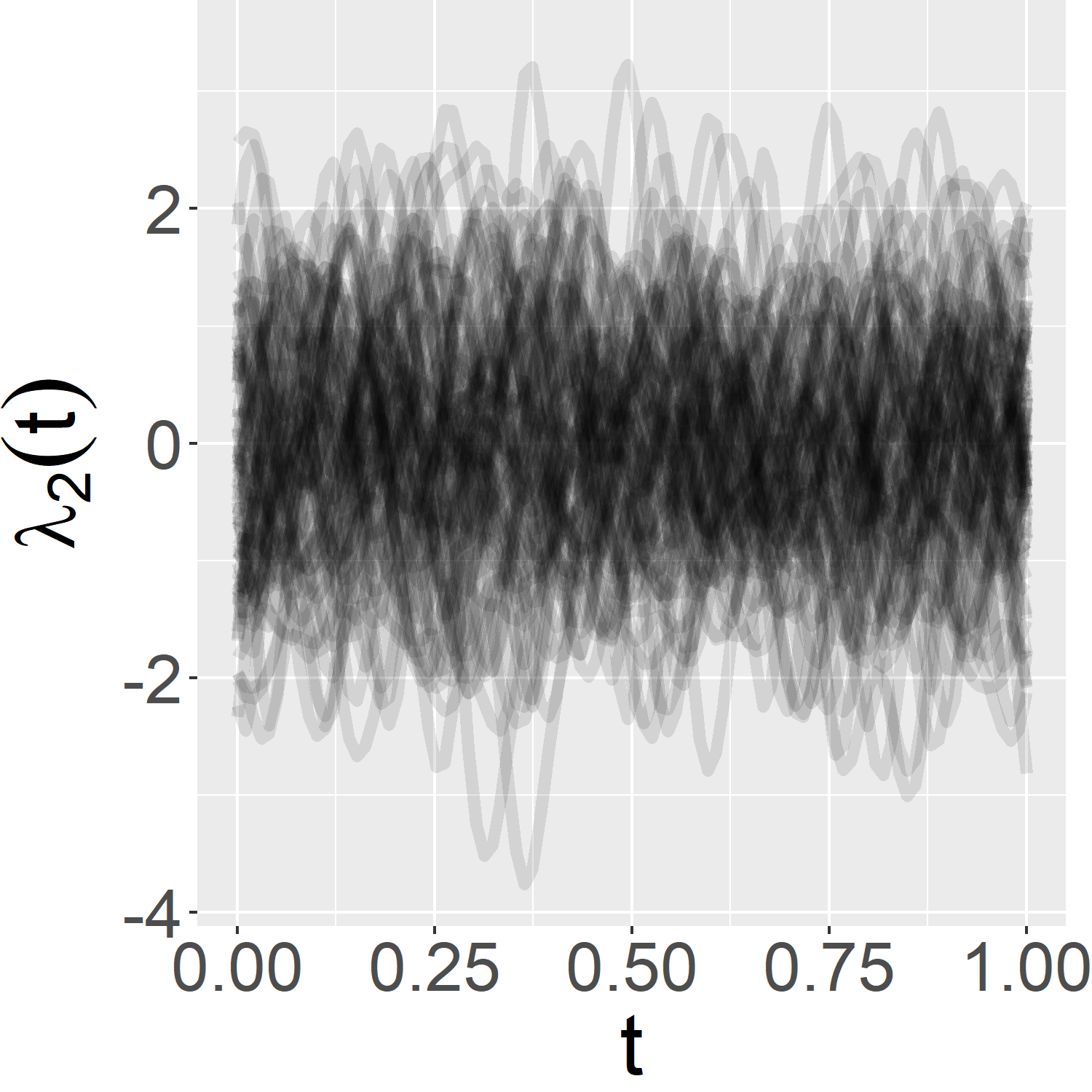} & \includegraphics[width = 1 in]{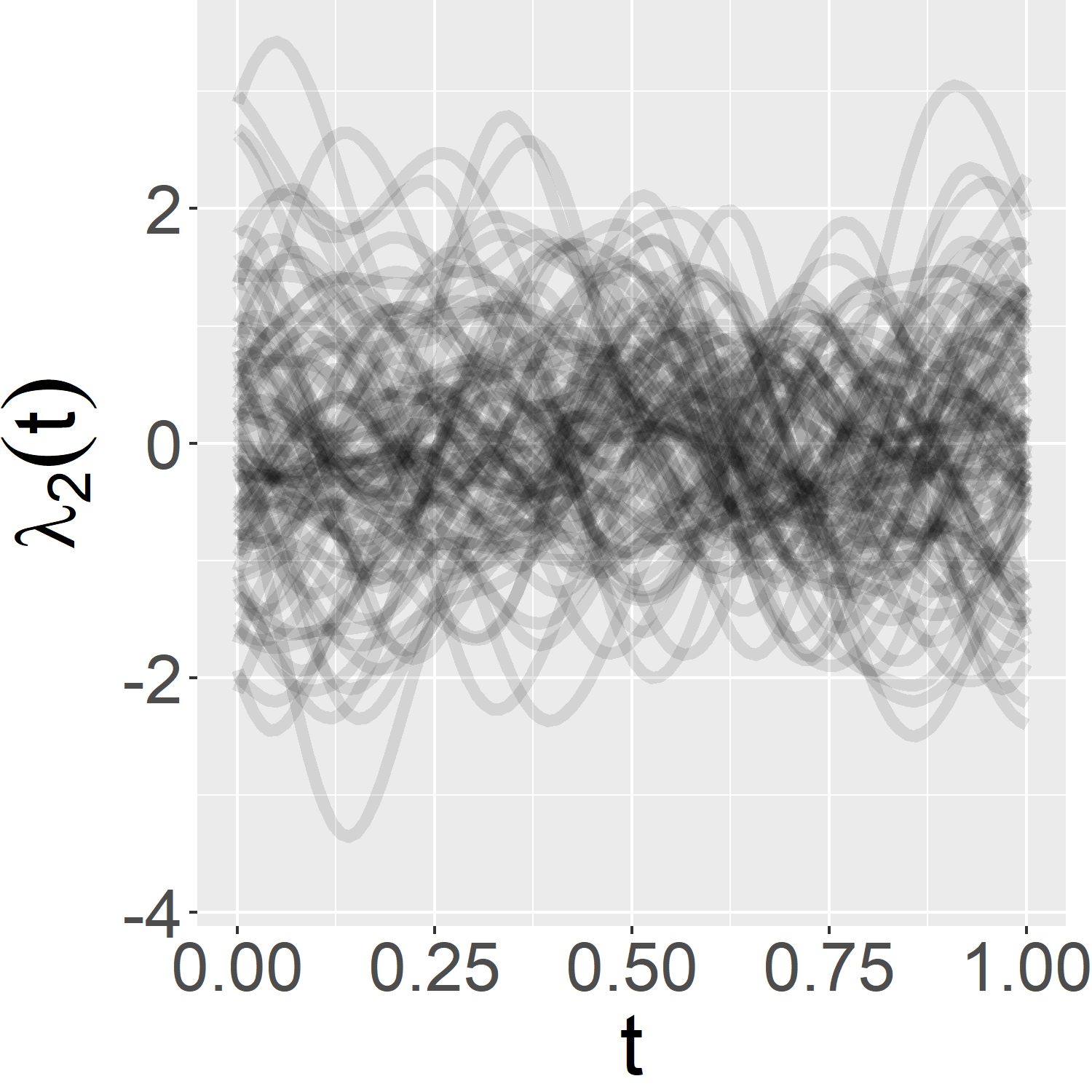} &
\includegraphics[width = 1 in]{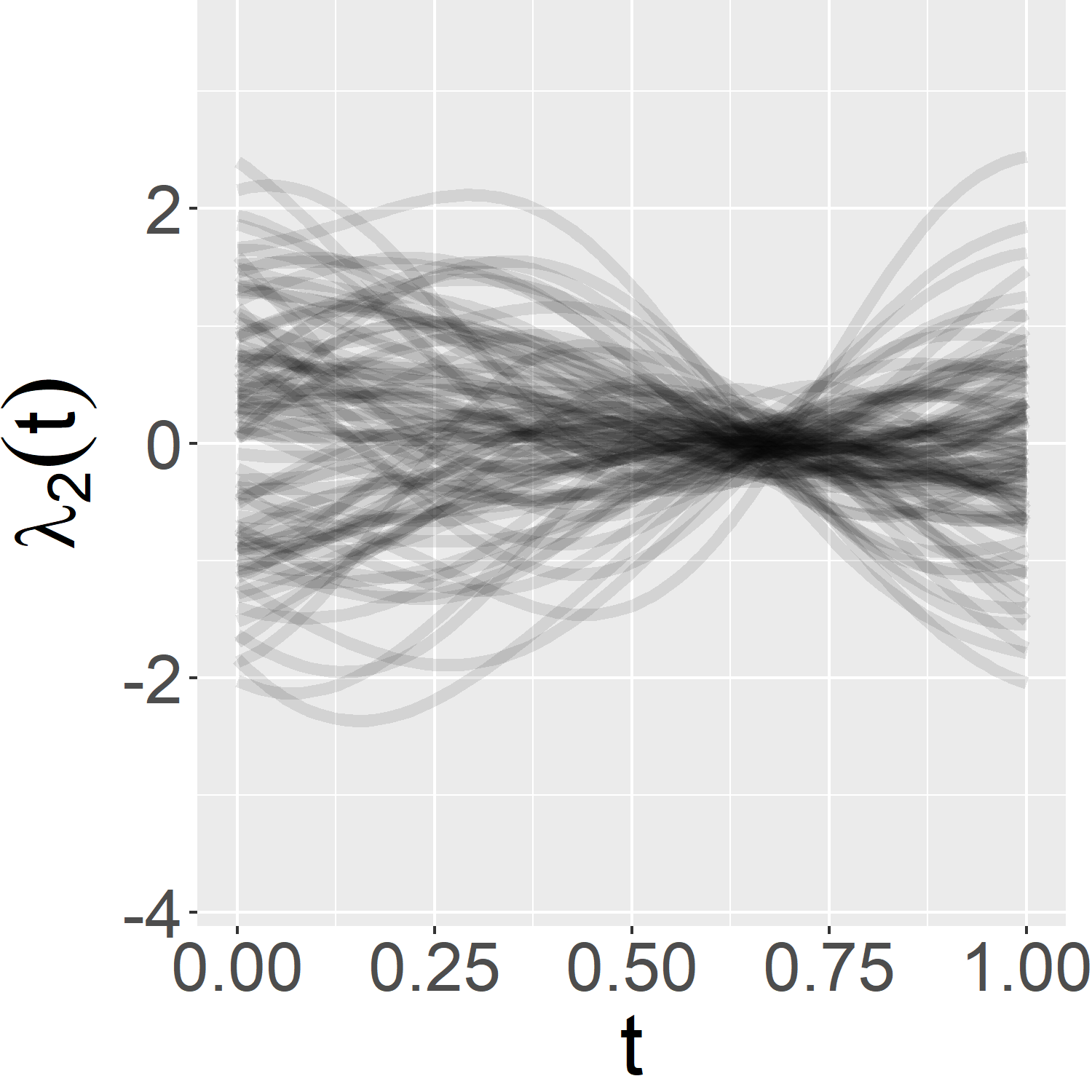} & \includegraphics[width = 1 in]{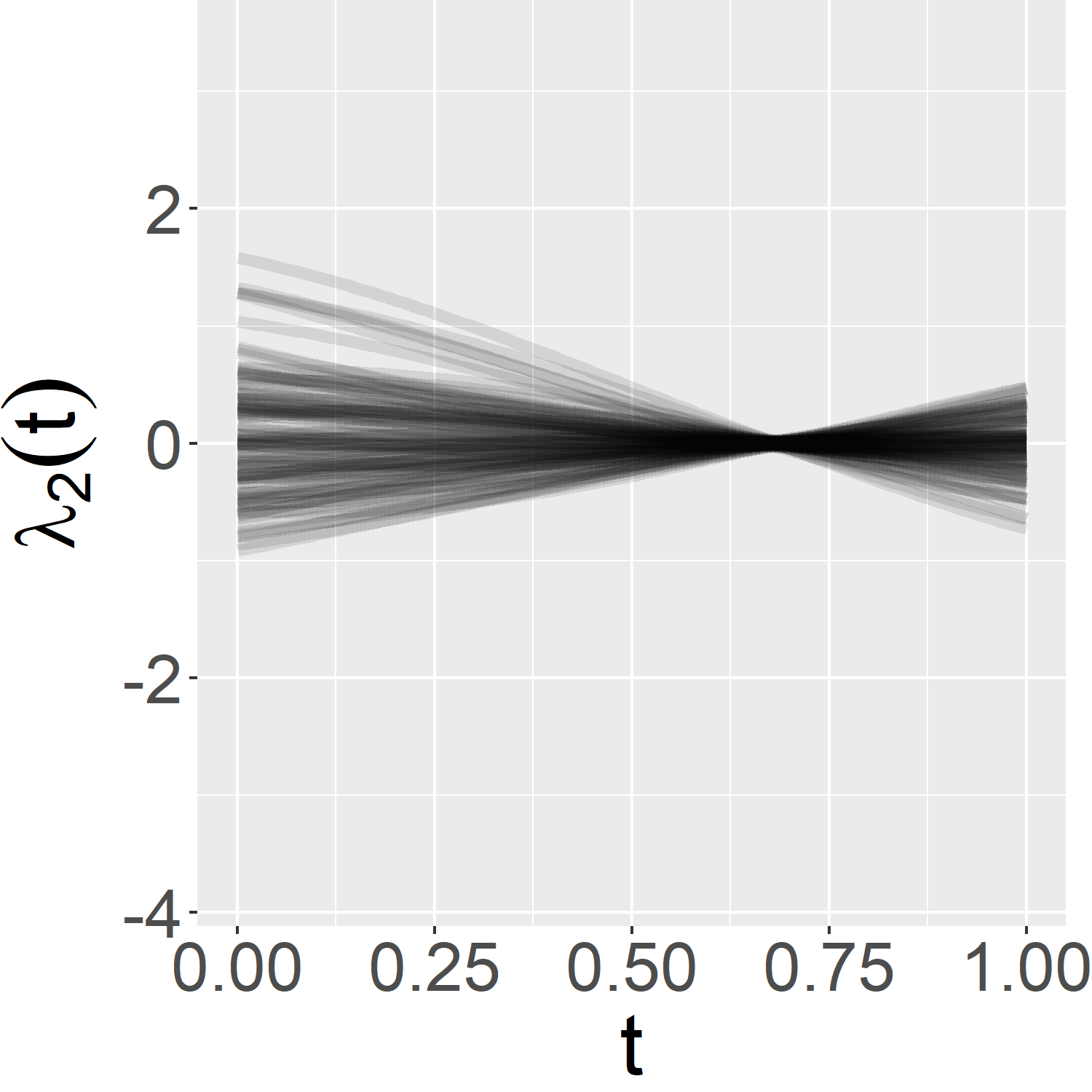}\\
(d) & (e) & (f) & (g) \\
\includegraphics[width = 1 in]{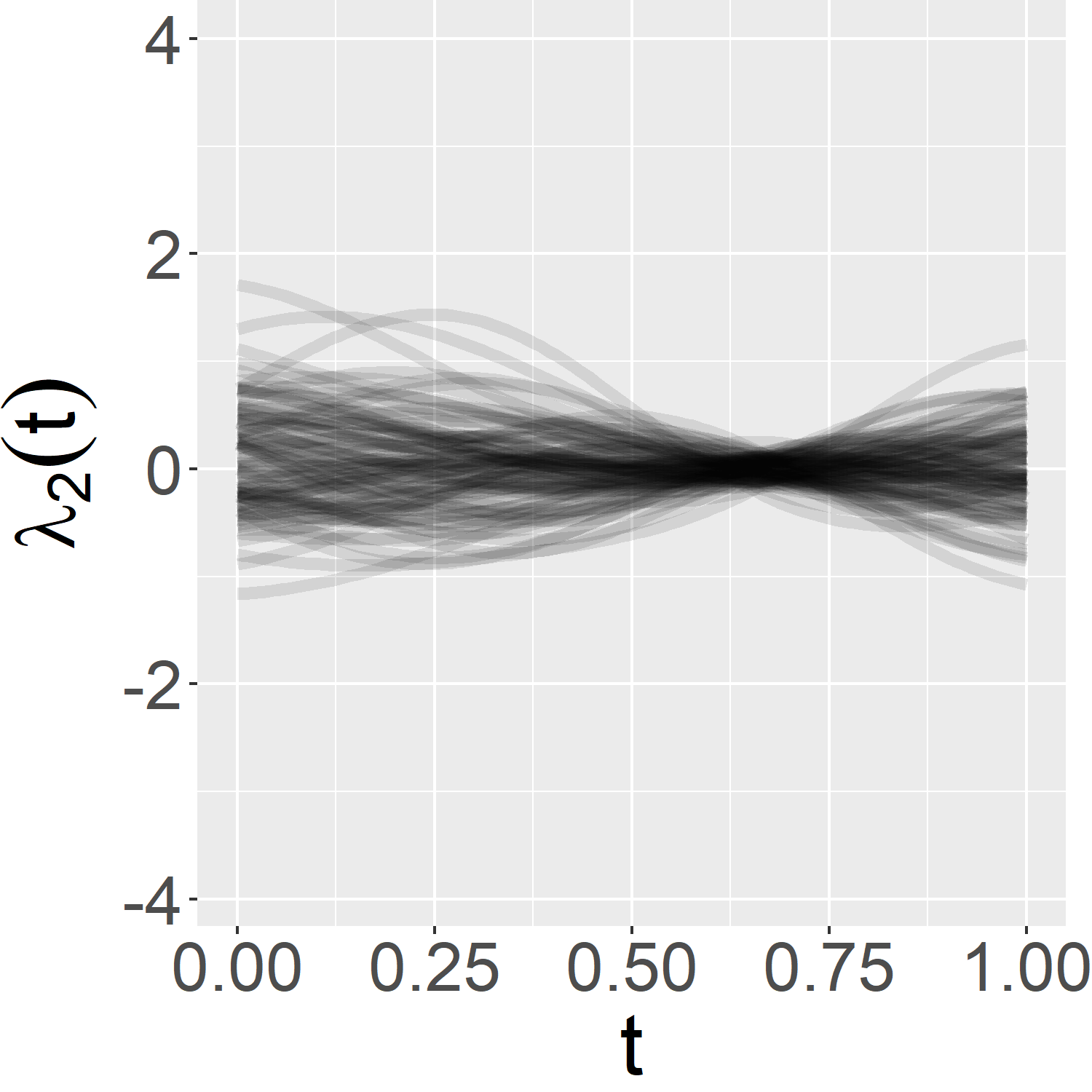} & \includegraphics[width = 1 in]{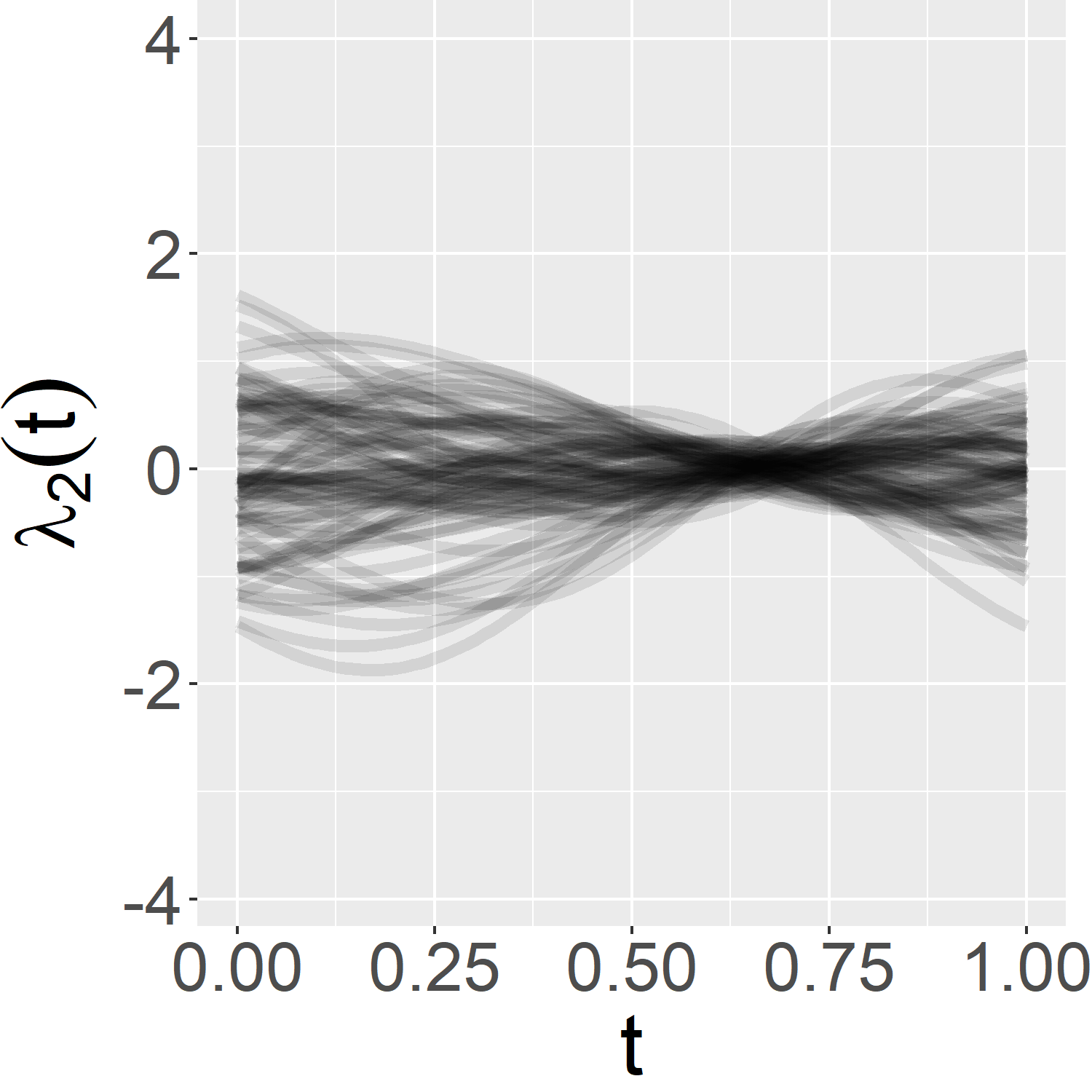} &
\includegraphics[width = 1 in]{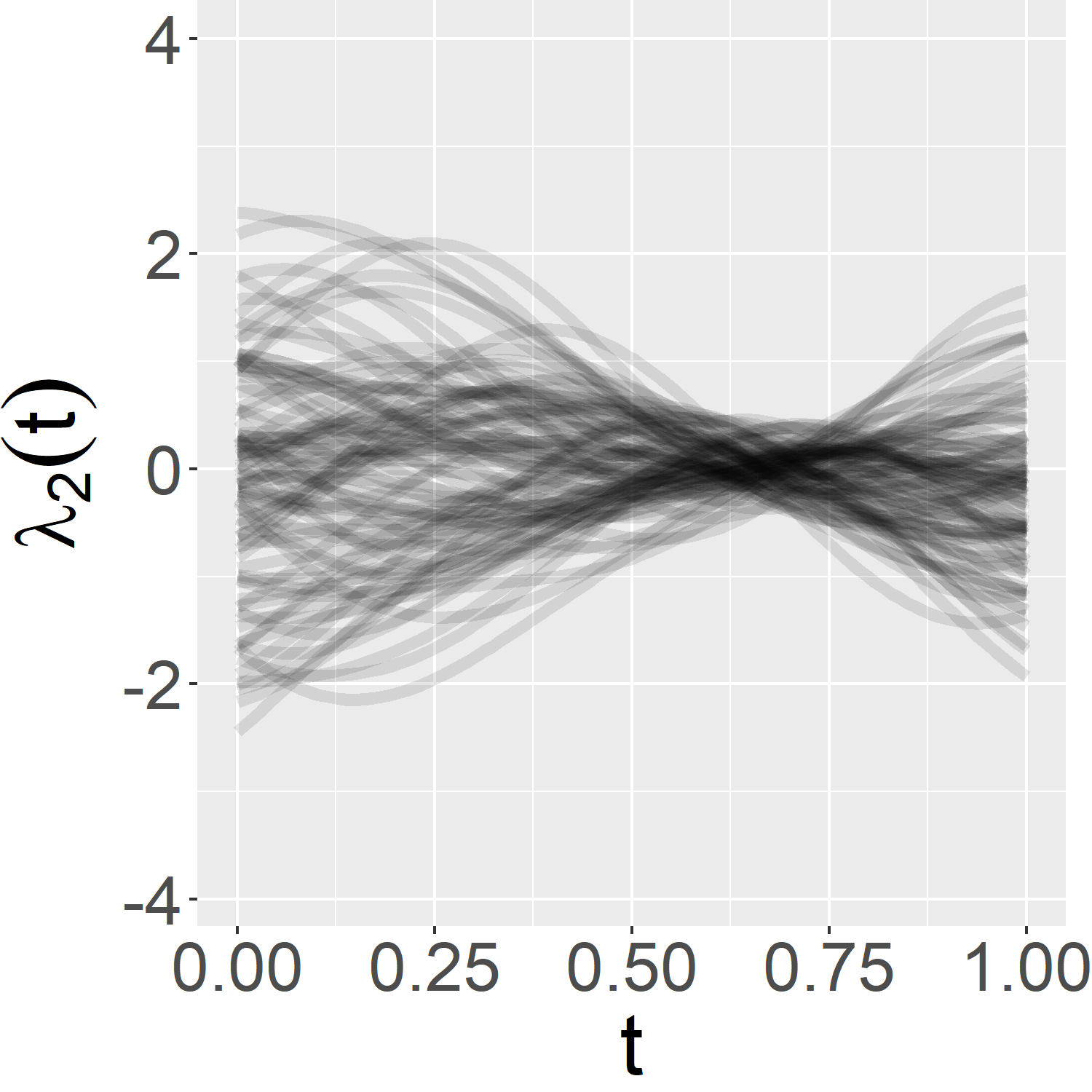} & \includegraphics[width = 1 in]{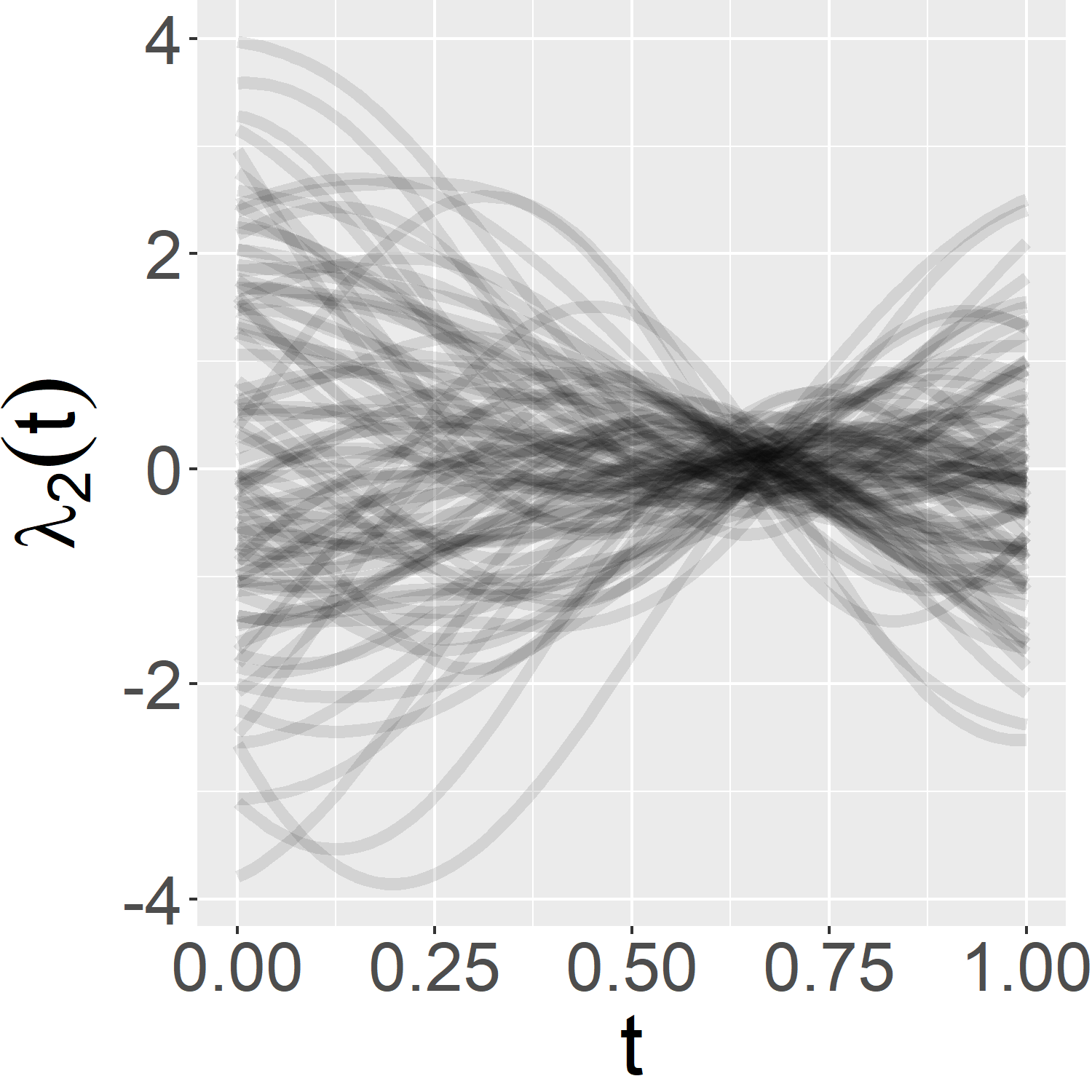}\\
(h) & (i) & (j) & (k) \\
\end{tabular}

\caption{(a) The value of $\lambda_1$ on which $\lambda_2$ is conditioned. Boxplots of 1000 realizations of $\langle\lambda_1,\lambda_2\rangle$, with $\lambda_2\mid\lambda_1 \sim \mbox{GP}(0,C_2^{\nu_\lambda}(\cdot,\cdot))$ (b) $\tau^2_2 = 1,\ \nu_\lambda = 0.0001$ are fixed as $l_2 = 0.001,0.01,0.1,1$ varies, (c) $l_2= 0.01,\ \nu_\lambda= 0.0001$ are fixed as $\tau^2_2 = 0.25,0.5,1,2$ varies. Realizations of $\lambda_2\mid\lambda_1 \sim \mbox{GP}(0,C_2^{\nu_\lambda}(\cdot,\cdot))$ with (d)-(g)  $\tau_2^2 = 1, \nu_\lambda = 0.0001$ fixed, while $l_2^2 = 0.001$, $l_2^2 = 0.01$, $l_2^2 = 0.1$, $l_2^2 = 1$ and with (h)-(k)  $l_2^2 = .1, \nu_\lambda = 0.0001$ fixed, while $\tau_2^2 = 0.25$, $\tau_2^2 = 0.5$, $\tau_2^2 = 1$, $\tau_2^2 = 2$.}
\label{fig:simu_inner_prod}
    \end{center}
\end{figure}

In this section, we illustrate the role of the length-scale and scale parameters of the conditional prior from Section 2.3 in the main paper. Panel (a) of Figure \ref{fig:simu_inner_prod} again displays the value of $\lambda_1$ on which $\lambda_2\mid\lambda_1$ is conditioned.

Figure \ref{fig:simu_inner_prod} shows realizations from the prior with $\tau_2^2 = 1, \nu_\lambda = 0.0001$ fixed, while $l_2^2 = 0.001,\ 0.01,\ 0.1,\ 1$  varies in panels (d)-(g). Notice, as the length-scale parameter increases, the realizations become more smooth. When the realizations are relatively smooth, in panels (c)\&(d), the variance shrinks around $t \approx .65$. This feature is present because of the prior's need to offset the extrema of $\lambda_1$ to enforce orthogonality between the two functions. When the realizations are relatively rough, the shrinking of the variance of $\lambda_2$ is less evident, since rougher functions tend to be close to orthogonal than smoother ones. Panel (b) of Figure \ref{fig:simu_inner_prod} shows a boxplot of the inner product between $\lambda_1$ and 1000 realizations of $\lambda_2$ for the different values of the length-scale parameter.  Varying $l^2_2$ does not appear to have an effect on the inner product between $\lambda_1$ and the realizations of $\lambda_2\mid\lambda_1$. 

Figure \ref{fig:simu_inner_prod} shows realizations from the prior with $l_2^2 = .1, \nu_\lambda = 0.0001$ fixed, and $\tau_2^2 = 0.25,0.5,1,\ 2$ varying in panels (h)-(k). Varying this parameter changes the spread of the realizations, more drastically outside of the region where $t \approx .65$. Panel (c) of Figure \ref{fig:simu_inner_prod} shows a boxplot of the inner product between $\lambda_1$ and 1000 realizations of $\lambda_2\mid\lambda_1$ for the different values of the scale parameter. Again, varying $\tau^2_2$ does not appear to have an effect on the inner product between the $\lambda_1$ and the realizations of $\lambda_2\mid\lambda_1$. 

\section{Markov chain Monte Carlo Implementation}\label{sec:MCMC}

We concisely restate the likelihood and prior distributions in hierarchical model notation. Our hierarchical functional factor analysis model is specified as follows,
\begin{eqnarray*}
 & y_i\mid\mu(\vv{t}),\lambda_1(\vv{t}),\ldots,\lambda_K(\vv{t}),\eta_i \sim \ N_{m_i}(O_i\mu(\vv{t}) + O_i(\lambda_1(\vv{t}),\ldots,\lambda_K(\vv{t}))^\top\eta_i,\sigma^2 I_{m_i}), \\
 & \mu\mid l_\mu,\tau^2_\mu \sim \mbox{GP}(0,C_\mu(\cdot,\cdot)),\ (l^2_\mu)^{-1} \sim \text{gamma}(\alpha_\mu,\beta_\mu), \ \tau^2_\mu \sim \text{half-normal}(\gamma_\mu) \\
  & \lambda_k\mid l_k,\tau^2_k,\Lambda_{(-k)}\sim \mbox{GP}(0,C^{\nu_\lambda}_k(\cdot,\cdot)),\ \eta_{\cdot k}\mid\psi_k \sim N(0,\psi_k(I_n + \frac{1}{\nu_\eta}1_n1_n^\top)^{-1}),\ \\
  & (l^2_k)^{-1} \sim \text{gamma}(\alpha_\lambda,\beta_\lambda), \ \tau^2_k \sim \text{half-normal}(\gamma_\lambda),\ (\psi_k)^{-1}\sim \text{gamma}(\alpha_\eta,\beta_\eta), \\
  & (\sigma^2)^{-1}\sim \text{gamma}(\alpha_\sigma,\beta_\sigma),\ i = 1\ldots, n,\ k = 1,\ldots, K.
\end{eqnarray*}
Our hierarchical generalized functional factor analysis extension for binary functional data is 
\begin{eqnarray*}
    &    z_i(\vv{t}_i) =
    \begin{cases}
    1,\quad z_i^*(\vv{t}_i)>0,\\
    0,\quad \text{otherwise},
    \end{cases} 
    \\
    & z^*_i(\vv{t}_i)\mid\mu(\vv{t}),\lambda_1(\vv{t}),\ldots,\lambda_K(\vv{t}),\eta_i \text{ind}{\sim} \ N_{m_i}(O_i\mu(\vv{t}) + O_i(\lambda_1(\vv{t}),\ldots,\lambda_K(\vv{t}))^\top\eta_i,I_{m_i}),\\
    & \mu\mid l_\mu,\tau^2_\mu \sim \mbox{GP}(0,C_\mu(\cdot,\cdot)),\ (l^2_\mu)^{-1} \sim \text{gamma}(\alpha_\mu,\beta_\mu), \ \tau^2_\mu \sim \text{half-normal}(\gamma_\mu) \\
    & \lambda_k\mid l_k,\tau^2_k,\Lambda_{(-k)}\sim \mbox{GP}(0,C^{\nu_\lambda}_k(\cdot,\cdot)),\ \eta_{\cdot k}\mid\psi_k \sim N(0,\psi_k(I_n + \frac{1}{\nu_\eta}1_n1_n^\top)^{-1}),\\
    & (l^2_k)^{-1} \sim \text{gamma}(\alpha_\lambda,\beta_\lambda), \ \tau^2_k \sim \text{half-normal}(\gamma_\lambda),\\
    & (\psi_k)^{-1}\sim \text{gamma}(\alpha_\eta,\beta_\eta),\ i = 1,\ldots,n,\ k = 1,\ldots, K.
\end{eqnarray*}
The $\psi_k$ auxiliary parameters introduced in the priors for $\eta_k$ govern prior variance following the parameter expansion approach of \cite{ghosh2009}, which introduced this technique in the context of multivariate latent factor models. Inference of the $\psi_k$ parameters themselves is not of interest. Saved posterior samples of $\psi_k$ are used to scale saved posterior samples of $\lambda_k$ and $\eta_{\cdot k}$. This approach has been shown to diminish posterior dependence and lead to better Markov chain Monte Carlo mixing.

Throughout the algorithms presented in this section, we use the superscript $cur$ and $can$ to denote current and candidate states of parameters in the Markov chain, and use superscripted square brackets to denote saved parameter values. The notation $\cdot\mid-$ is used to denote a parameter conditioned on all other parameters and the data.

\begin{algorithm}[!h]
\caption{Metropolis-within-Gibbs for Functional Factor Analysis}
\label{alg:fpca}
\begin{algorithmic}[1]
\State Randomly initialize current states of all parameters using their prior distribution
\For{ iter = 1 : N}
\State Draw $\mu^{cur}(\vv{t})\mid-\sim \pi(\mu(\vv{t})\mid-)$ from Equation \eqref{eq:muCond} 
\State Propose $l^{can}_\mu \sim N(l^{cur}_\mu,\Sigma_{l_\mu}^{prop})$
\State Compute $\alpha_{l_\mu} := \frac{\pi(\mu^{cur}(\vv{t})\mid l^{can}_\mu,\tau^{2,cur}_\mu)\pi(l^{can}_\mu)}{\pi(\mu^{cur}(\vv{t})\mid l^{cur}_\mu,\tau^{2,cur}_\mu)\pi(\tau^{2,cur}_\mu)}$
\State If $\text{unif}(0,1) < \alpha_{l_\mu}$, set $l^{cur}_\mu = l^{cur}_\mu$
\State Propose $\tau^{2,can}_\mu \sim N(\tau^{2,cur}_\mu,\Sigma_{\tau^2_\mu}^{prop})$
\State Compute $\alpha_{\tau^2_\mu} := \frac{\pi(\mu^{cur}(\vv{t})\mid l^{cur}_\mu,\tau^{2,can}_\mu)\pi(\tau^{2,can}_\mu)}{\pi(\mu^{cur}(\vv{t})\mid l^{cur}_\mu,\tau^{2,cur}_\mu)\pi(\tau^{2,cur}_\mu)}$
\State If $\text{unif}(0,1) < \alpha_{\tau^2_\mu}$, set $\tau^{2,cur}_\mu  = \tau^{2,can}_\mu $
\For{k in 1 : K}
\State Draw $\lambda_k^{cur}(\vv{t})\mid-\sim \pi(\lambda_k(\vv{t})\mid-)$ from Equation \eqref{eq:lambdaCond} 
\State Propose $l^{can}_k \sim N(l^{cur}_k,\Sigma_{l_k}^{prop})$
\State Compute $\alpha_{l_k} := \frac{\pi(\lambda_k^{cur}(\vv{t})\mid l^{can}_k,\tau^{2,cur}_k)\pi(l^{can}_k)}{\pi(\lambda_k^{cur}(\vv{t})\mid l^{cur}_k,\tau^{2,cur}_k)\pi(l^{cur}_k)}$
\State If $\text{unif}(0,1) < \alpha_{l_k}$, set $l^{cur}_k = l^{can}_k$
\State Propose $\tau^{2,can}_k \sim N(\tau^{2,cur}_k,\Sigma_{\tau^2_k}^{prop})$
\State Compute $\alpha_{\tau^2_k} := \frac{\pi(\lambda_k^{cur}(\vv{t})\mid l^{cur}_k,\tau^{2,can}_k)\pi(\tau^{2,can}_k)}{\pi(\lambda_k^{cur}(\vv{t})\mid l^{cur}_k,\tau^{2,cur}_k)\pi(\tau^{2,cur}_k)}$
\State If $\text{unif}(0,1) < \alpha_{\tau^2_k}$, set $\tau^{2,cur}_k  = \tau^{2,can}_k $
\State Draw $\eta_{\cdot k}^{cur}\mid-\sim \pi(\eta_{\cdot k}\mid-)$ from Equation \eqref{eq:etaCond} 
\State Draw $\psi_k^{cur}\mid-\sim \pi(\psi_k\mid-)$ from Equation \eqref{eq:psiCond} 
\EndFor
\State Draw $\sigma^{2,cur}\mid-\sim \pi(\sigma^{2}\mid-)$ from Equation \eqref{eq:sigmaCond} 
\State Save \begin{eqnarray*}
   &\mu^{[iter]}(\vv{t}),l_\mu^{[iter]},\tau^{2,[iter]}_\mu\{\lambda_k(\vv{t})^{[iter]},l_k^{[iter]},\tau^{2,[iter]}_k,\eta_{\cdot k}^{[iter]},\psi_k^{[iter]}\}_{k = 1}^K,\sigma^{2,[iter]} \\ 
   & := \mu^{cur}(\vv{t}),l_\mu^{cur},\tau^{2,cur}_\mu\{\lambda_k(\vv{t})^{cur},l_k^{cur},\tau^{2,cur}_k,\eta_{\cdot k}^{cur},\psi_k^{cur}\}_{k = 1}^K,\sigma^{2,cur}
\end{eqnarray*} 
\EndFor
\end{algorithmic}
\end{algorithm}

For algorithm to sample parameters from functional factor analysis model presented in Algorithm \ref{alg:fpca}, the full conditional distribution of $\mu(\vv{t})$ is 
\begin{eqnarray}\label{eq:muCond}
   \mu(\vv{t})\mid- & \sim & N(\mathbb{E}[\mu(\vv{t})\mid-],\mathbb{V}[\mu(\vv{t})\mid-]), \\
   \mathbb{V}[\mu(\vv{t})\mid-] & = & \Big(C_\mu^{-1}(\vv{t},\vv{t}) + \frac{1}{\sigma^2}\sum_{i = 1}^nO_i^\top O_i\Big)^{-1}\nonumber \\
   \mathbb{E}[\mu(\vv{t})\mid-] & = & \frac{1}{\sigma^2}\mathbb{V}[\mu(\vv{t})\mid-]\sum_{i = 1}^n\Big(O_i^\top y_i(\vv{t}_i) - O_i^\top O_i\sum_{k = 1}^k\eta_{i,k}\lambda_k(\vv{t})\Big).\nonumber 
\end{eqnarray}
The full conditional distribution of $\lambda_k(\vv{t}),\ k = 1,\ldots,K$ is 
\begin{eqnarray}\label{eq:lambdaCond}
   \lambda_k(\vv{t})\mid- & \sim & N(\mathbb{E}[\lambda_k(\vv{t})\mid-],\mathbb{V}[\lambda_k(\vv{t})\mid-]), \\
   \mathbb{V}[\lambda_k(\vv{t})\mid-] & = & \big(C_k^{-1}(\vv{t},\vv{t}) + \frac{1}{\nu_\lambda}W\Lambda_{(-k)}(\vv{t})\Lambda_{(-k)}(\vv{t})^\top W + \frac{1}{\sigma^2}\sum_{i = 1}^n\eta_{i,k}^2O_i^\top O_i\big)^{-1}\nonumber \\
   \mathbb{E}[\lambda_k(\vv{t})\mid-] & = & \nonumber \frac{1}{\sigma^2}\mathbb{V}[\lambda_k(\vv{t})\mid-]\sum_{i = 1}^n\big(\eta_{i,k} O_i^\top y_i(\vv{t}_i) - O_i^\top O_i\mu(\vv{t})- O_i^\top O_i\sum_{j \neq k}\eta_{i,j}\lambda_j(\vv{t})\big)  .
\end{eqnarray}
The full conditional distribution of $\eta_{\cdot k},\ k = 1,\ldots,K$ is 
\begin{eqnarray}\label{eq:etaCond}
   \eta_{\cdot k}\mid- & \sim & N(\mathbb{E}[\eta_{\cdot k}\mid-],\mathbb{V}[\eta_{\cdot k}\mid-]), \\
   \mathbb{V}[\eta_{\cdot k}\mid-]  & = & \Big(\text{diag}\big(\lambda_k(\vv{t})^\top O_i^\top O_i \lambda_k(\vv{t})\big)_{i = 1}^n + \frac{1}{\psi_k}(I_n + \frac{1}{\nu_\eta}1_n1_n^\top)\Big)^{-1}\nonumber \\
    \mathbb{E}[\eta_{\cdot k}\mid-]& = & \frac{1}{\sigma^2}\mathbb{V}[\eta_{\cdot k}\mid-]\Big[\big(y_i(\vv{t}_i) - O_i\mu(\vv{t})- O_i\sum_{j \neq k}\eta_{i,j}\lambda_j(\vv{t})\big)^\top O_i\lambda_k(\vv{t})\Big]_{i = 1}^n \nonumber,  
\end{eqnarray}
where $\text{diag}(\cdot)_{i = 1}^n$ denotes a diagonal matrix with diagonal elements given by the argument, and $[\cdot]_{i = 1}^n$ is a vector with elements given by the argument. \\
The full conditional distribution of $\psi_k,\ k = 1,\ldots,K$ is 
\begin{eqnarray}\label{eq:psiCond}
   (\psi_k)^{-1}\mid- & \sim & \text{gamma}\big(\alpha_\eta + \frac{n}{2},\beta_\eta \frac{1}{2}\eta_{\cdot k}^\top(I_n + \frac{1}{\nu_\eta}1_n1_n^\top)^{-1}\eta_{\cdot k}\big)
\end{eqnarray}
The full conditional distribution of $\sigma^2$ is 
\begin{eqnarray}\label{eq:sigmaCond}
   (\sigma^2)^{-1}\mid- \sim \text{gamma}\Big(\alpha_\sigma + \frac{1}{2}\sum_{i = 1}^n m_i, & \nonumber \\ & \hspace{-1.1in} \beta_\sigma + \frac{1}{2}\sum_{i = 1}^n\sum_{j = 1}^{m_i}\big(y_i(t_{i,j}) - \mu(t_{i,j}) + \sum_{k = 1}^K\eta_{i,k}\lambda_k(t_{i,j})\big)^2\Big).
\end{eqnarray}

A Metropolis-within-Gibbs algorithm for our generalized functional factor analysis approach for binary functional data is presented in Algorithm \ref{alg:gfpca}. The latent variable formulation of the probit model leads to a simple step for sampling the latent processes $z^*_i,\ i = 1,\ldots,n$ as described in \cite{albert1993}. Conditioning on the latent processes, the sampling of the model parameters is similar to the functional factor analysis setting. 

\begin{algorithm}[!h]
\caption{Metropolis-within-Gibbs for Generalized Functional Factor Analysis}
\label{alg:gfpca}
\begin{algorithmic}[1]
\State Randomly initialize current states of all parameters using their prior distribution
\For{ iter = 1 : N}
\For{ i = 1 : n}
\For{ j = 1 : $m_i$}
\State \hspace{-.75 in} Draw $z^{*,cur}_i(t_{i,j})\mid-\sim$  $$\begin{cases}
    \text{truncated-normal}_{(0,\infty)}(\mu^{cur}(t_{i,j}) + \sum_{k = 1}^K\eta^{cur}_{i,k}\lambda^{cur}_k(t_{i,j}),1),\ z_i(t_{i,j})=1,\\
    \text{truncated-normal}_{(-\infty,0)}(\mu^{cur}(t_{i,j}) + \sum_{k = 1}^K\eta^{cur}_{i,k}\lambda^{cur}_k(t_{i,j}),1),\ z_i(t_{i,j})=0.
    \end{cases}$$ 
\EndFor
\EndFor
\State $\mu^{cur}(\vv{t}),l_\mu^{cur},\tau^{2,cur}_\mu\{\lambda_k(\vv{t})^{cur},l_k^{cur},\tau^{2,cur}_k,\eta_{\cdot k}^{cur},\psi_k^{cur}\}_{k = 1}^K$ can be sampled using steps 3 - 19 of Algorithm \ref{alg:fpca}, replacing $y_i := z^*_i,\ i=1,\ldots,n \text{ and } \sigma^2 := 1$
\State Save \begin{eqnarray*}
   &\mu^{[iter]}(\vv{t}),l_\mu^{[iter]},\tau^{2,[iter]}_\mu\{\lambda_k(\vv{t})^{[iter]},l_k^{[iter]},\tau^{2,[iter]}_k,\eta_{\cdot k}^{[iter]},\psi_k^{[iter]}\}_{k = 1}^K \\ 
   & := \mu^{cur}(\vv{t}),l_\mu^{cur},\tau^{2,cur}_\mu\{\lambda_k(\vv{t})^{cur},l_k^{cur},\tau^{2,cur}_k,\eta_{\cdot k}^{cur},\psi_k^{cur}\}_{k = 1}^K
\end{eqnarray*} 
\EndFor
\end{algorithmic}
\end{algorithm}

\section{Simulation Experiments}\label{sec:additionalSim}

In this section, we study the performance of our approach under inferential settings that are motivated by the goals of the analysis of the Cebu Longitudinal Health and Nutrition Survey in the main paper. 

\subsection{Functional Factor Analysis for Gaussian Observations}\label{sec:ffa_sim}

We generate $n = 100$ functions on a grid, $\vv{t}$, with $m = 30$ equally spaced points, following 
\begin{equation}\nonumber
    f_i(\vv{t}) = \mu(\vv{t}) + \{\lambda_1(\vv{t}),\ldots,\lambda_K(\vv{t})\}^\top\eta_i,\ i=1,\ldots,n.
\end{equation}
With $K = 2$, we generate $\mu,\ \lambda_1,\ \lambda_2,\ \eta_{\cdot 1},\ \eta_{\cdot 2}$ according to the priors presented in Section 3.2 of the main paper with $l^2_\mu = l^2_{\lambda_1} = l^2_{\lambda_2} = .4$ and $\tau^2_\mu = \tau^2_{\lambda_1} = \tau^2_{\lambda_2} = 1$. Simulated observations are generated according to $y_i(\vv{t}_i) = f_i(\vv{t}_i) + \epsilon_i(\vv{t}_i)$, with $\epsilon_i(\vv{t}_i)\sim N_{m_i}(0,\sigma^2I_{m_i})$ and $\sigma^2 = 1$. The subject specific grids $\vv{t}_i$ are generated by randomly omitting $25\%,\ 50\%,\ \text{and }75\%$ of the points in the common grid, $\vv{t}$. Estimation performance is measured using mean integrated squared error for how well different methods are able to capture $\mu$, $\lambda_1$, $\lambda_2$, and all $f_i,\ i = 1,\ldots,n$. 

As competing methods, we use \texttt{R} package \texttt{fdapace} to estimate model parameters according to the principal analysis through conditional expectation methodology \citep{yao2005}. The approach of \cite{yao2005} uses kernel smoothing to estimate a mean and covariance function from sparse functional data, where individual function estimates rely on these population quantities. Additionally, we use the \texttt{R} package \texttt{fpca} to implement the methodology of \cite{peng2009}. They use a pre-specified basis to represent factor loadings, and inference is carried out through the basis coefficients. The work of \cite{crainiceanu2010} provides an empirical Bayes approach, where the authors mean-center the observations, and estimate a covariance surface using penalized thin-plate splines, from which estimates of factor loadings are computed. While fixing the factor loadings to their estimated values, the authors formulate priors for latent factors and error variance, and perform Markov chain Monte Carlo using \texttt{WinBUGS}. 

To compare these methods, we replicate data generation $100$ times. Figure \ref{fig:sims}\footnote{Abbreviations in figure legends: CG2010 - \cite{crainiceanu2010}; PP2009 - \cite{peng2009}; WZSG2019 - \cite{wrobel2019}; YMW2005 - \cite{yao2005}; ZLLZ2021 - \cite{zhong2021}.} shows mean integrated squared error results for different model parameters and levels of sparsity in panel (a). Boxplots show differences between mean integrated squared error of model parameters between a competing method and our $\small{\mbox{NeMO}}$ approach within a simulation replicate.  For all methods, estimation performance deteriorates as the sparsity level increases. Consistently, the estimation error from $\small{\mbox{NeMO}}$ is lower than the others. 

We are also interested in performance in inferring the number of latent factors, $K = 2$. We use the approach outlined in Section 3.5 of the main paper for the $\small{\mbox{NeMO}}$ model, while using default choices for competing methods. In \cite{yao2005} the number of latent factors is set to describe $95\%$ of variability in the observations. In \cite{peng2009}, the number of factors is selected to minimize a cross validation score. The implementation of \cite{crainiceanu2010} assumes the number of factors is given; we provide the method the correct number. Across sparsity levels and replicates, the correct number of factors is chosen $99.7\%$ of the time for $\small{\mbox{NeMO}}$, $2.3\%$ for \cite{yao2005} and $66\%$ for \cite{peng2009}.

\begin{figure}[t!]
\begin{center}
 \begin{tabular}{c}
 Continuous Data \\
\includegraphics[width = 5 in]{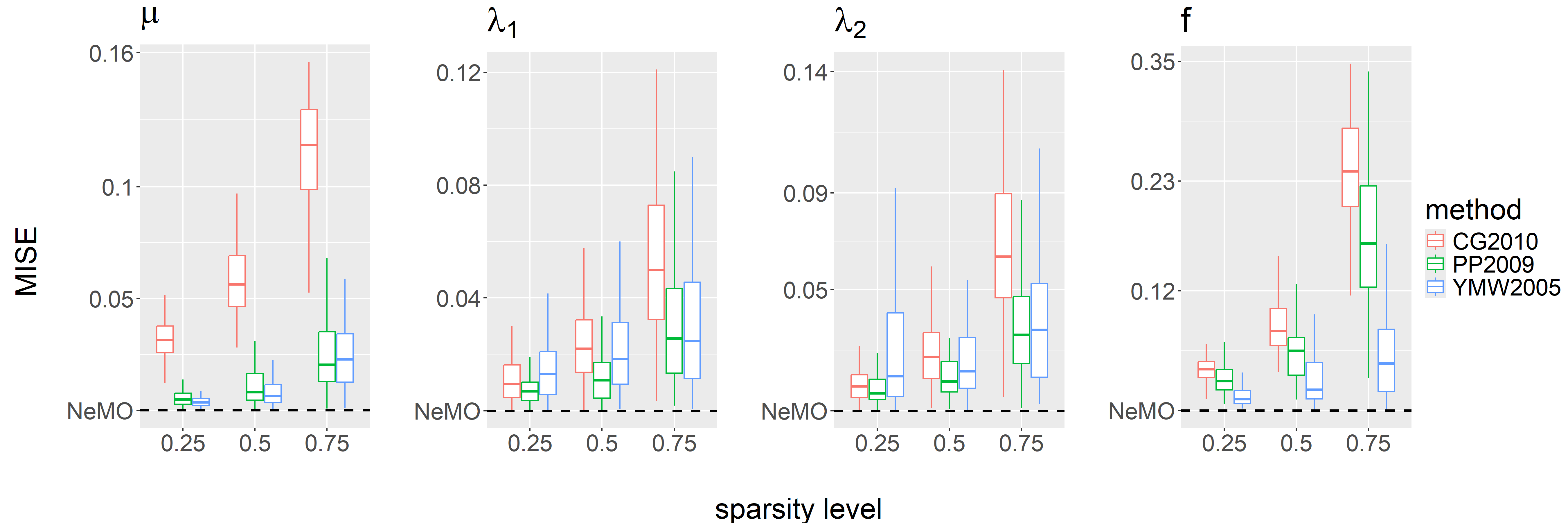} \\
(a) 
\end{tabular}

 \begin{tabular}{c|c}
 Binary Data & Continuous Data \\
 & (Coverage)\\
 \includegraphics[width = 1.75 in]{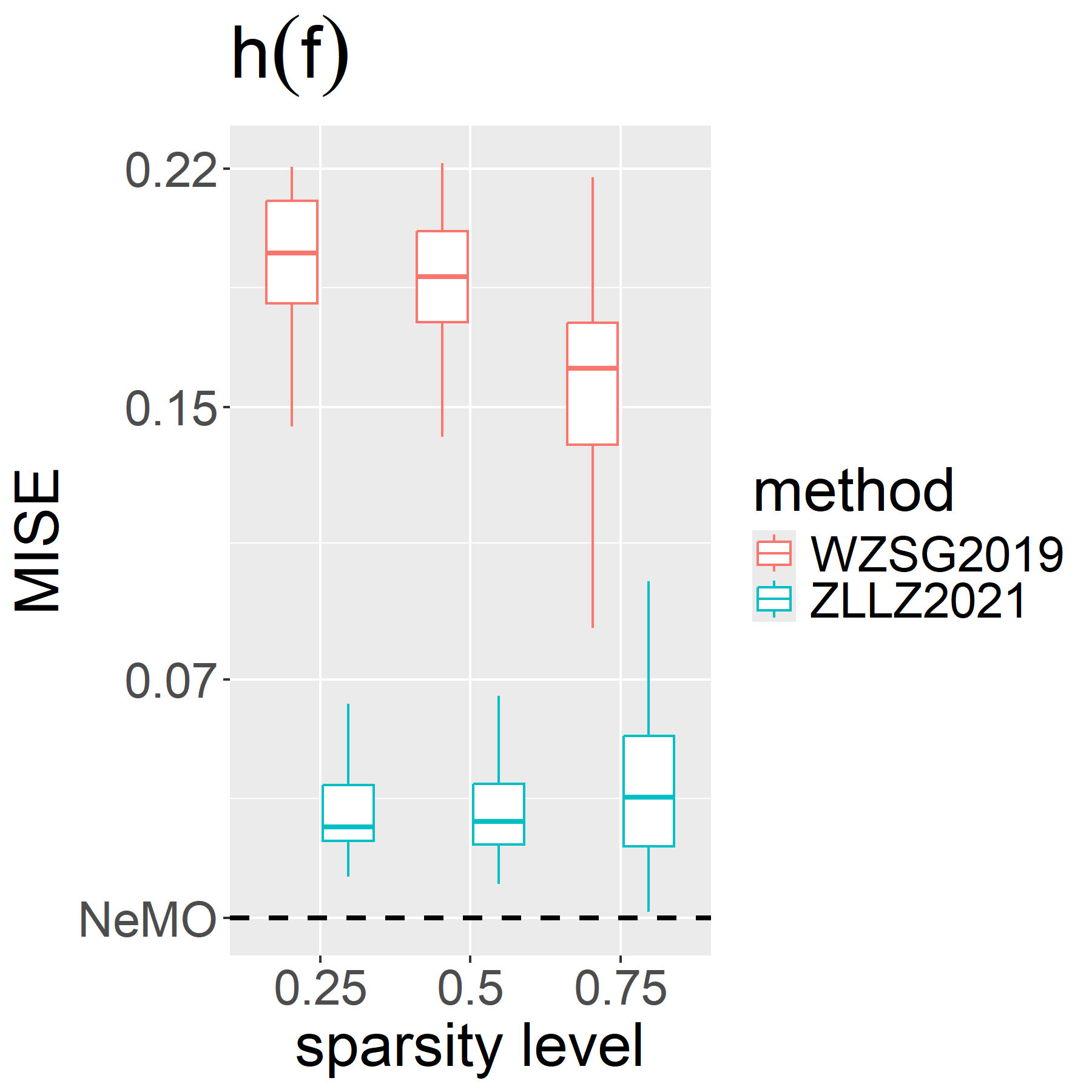} &
\includegraphics[width = 1.55 in]{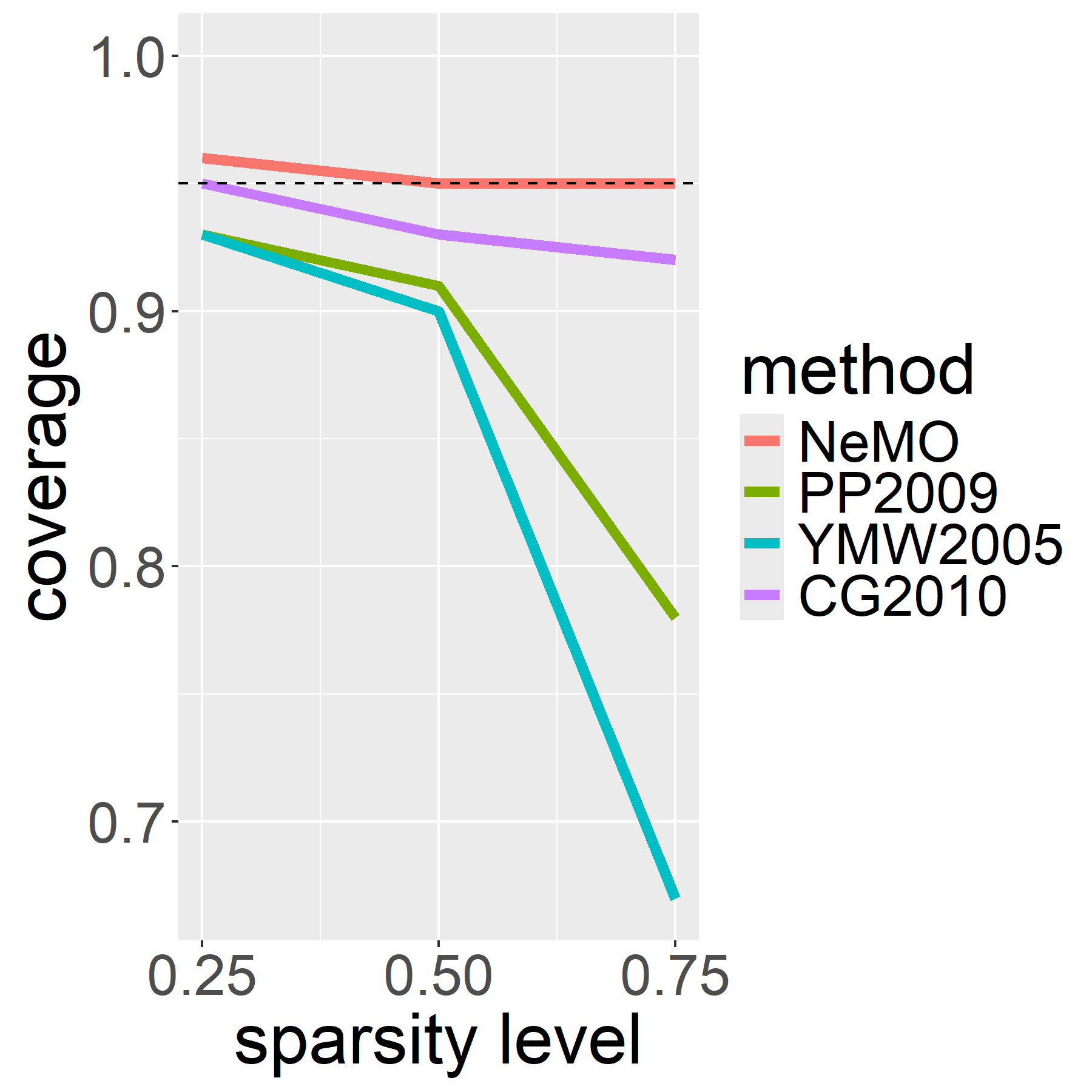}\\ (b) & (c) 
\end{tabular}
\caption{(a) \& (b) Mean integrated squared error for estimated parameters at different sparsity levels in the simulations of Sections \ref{sec:ffa_sim} and \ref{sec:bffa_sim}, respectively. (c) Coverage of credible/confidence intervals for different methods at different sparsity levels in simulations of Section \ref{sec:sim_regression}.}\label{fig:sims}
    \end{center}
\end{figure}


\subsection{Estimation for Non-Gaussian Observations}\label{sec:bffa_sim}

In the generalized functional factor analysis setting, we generate $n = 100$ binary functions on a grid $\vv{t}$ with $m = 30$ equally spaced points according to  \begin{equation}\nonumber
    \mbox{pr}\{y_i(\vv{t}) = 1\} = h
    \{f_i(\vv{t})\} = \Phi\big[\mu(\vv{t}) + \{\lambda_1(\vv{t}),\ldots,\lambda_K(\vv{t})\}^\top\eta_i\big],\ i=1,\ldots,n, 
\end{equation}
where $h(\cdot) = \Phi(\cdot)$ is the standard normal cumulative distribution function. With $K = 2$, we generate $\mu,\ \lambda_1,\ \lambda_2,\ \eta_{\cdot 1}\ \eta_{\cdot 2}$ according to the priors presented in Section 3.2 of the main paper with $l^2_\mu = l^2_{\lambda_1} = l^2_{\lambda_2} = .4$ and $\tau^2_\mu = \tau^2_{\lambda_1} = \tau^2_{\lambda_2} = 10$. To study the effect of sparsity, we randomly omit $25\%,\ 50\%,\ \text{and }75\%$ of the observations. Estimation performance of $h(f_i),\ i = 1,\ldots,n$ is measured through mean integrated squared error.

From our model, we use the posterior mean of $h[\mu(\vv{t}) + \{\lambda_1(\vv{t}),\ldots,\lambda_K(\vv{t})\}^\top\eta_i]$ to produce an estimate of $h(f_i)$ based on $1,000$ Markov chain Monte Carlo iterations. For competitors, we use the \texttt{bfpca} function from the \texttt{registr} package in \texttt{R} to estimate $h(f_i)$ using the methodology described in \cite{wrobel2019}. We also estimate $h(f_i)$ using the methods of \cite{zhong2021} from the \texttt{SLFPCA} \texttt{R} package. 
Panel (b) of Figure \ref{fig:sims} shows the mean integrated squared error results. Each of the boxplots show the difference between the mean integrated squared error of model parameters between a competing method and our $\small{\mbox{NeMO}}$ approach within a simulation replicate. Mean integrated squared error from $\small{\mbox{NeMO}}$ is consistently below that of the others. 

The implementations of \cite{wrobel2019} and \cite{zhong2021} do not enable selection of the number of latent factors, so they are fixed to the true number. Using the approach outlined in Section 3.5 of the main paper for our $\small{\mbox{NeMO}}$ model, we select the correct number of latent factors $100\%$ of the time when $25\%$ or $50\%$ of the observations are missing, and $89\%$ of the time when $75\%$ of the data are missing. 


\subsection{Coverage of Regression Coefficients}\label{sec:sim_regression}

We study our model's ability to serve as the foundation for inferential tasks in a latent factor regression setting \citep{montagna2012}, especially as it relates to the analysis of the Cebu Longitudinal Health and Nutrition Survey in the main paper. Data are generated as in Section \ref{sec:ffa_sim} with $K = 1$ and $\eta_{i,1} = \Theta x_{i,1} + \xi_{i,1},$ where $x_{i,1}$ is an observed covariate, $\Theta$ is a regression coefficient and $\xi_{i,1}\sim N(0,1)$ are independent. In each of 100 replicates, $\Theta$ and $x_{i,1}$ are independently generated as standard normal random variables. We study frequentist coverage of $95\%$ credible intervals for $\Theta$ based on $10,000$ Markov chain Monte Carlo samples. For comparison, we estimate $\eta_{i,1}, i = 1,\ldots,n$, using the methods of \cite{yao2005} and \cite{peng2009} from Section \ref{sec:ffa_sim} and regress the estimated factors onto the covariate. For the method of \cite{crainiceanu2010}, the latent factor regression is a hierarchical layer in the model. Figure \ref{fig:sims} panel (c) show interval coverage when $\sigma^2 = 1$. These results illustrate the importance of uncertainty propagation between estimation and inference. The coverage for \cite{peng2009} and \cite{yao2005} is very low due to conditioning on first stage estimates. The low coverage of \cite{crainiceanu2010} is due to the empirical Bayes setup, which fails to account for uncertainty in estimating factor loadings. In this setting $\small{\mbox{NeMO}}$ propagates uncertainty between estimation and inference improving performance.


\subsection{Simulation Example under Model Misspecification}

\begin{figure}[t!]
\begin{center}
 \begin{tabular}{cc}
\includegraphics[width = 1.25 in]{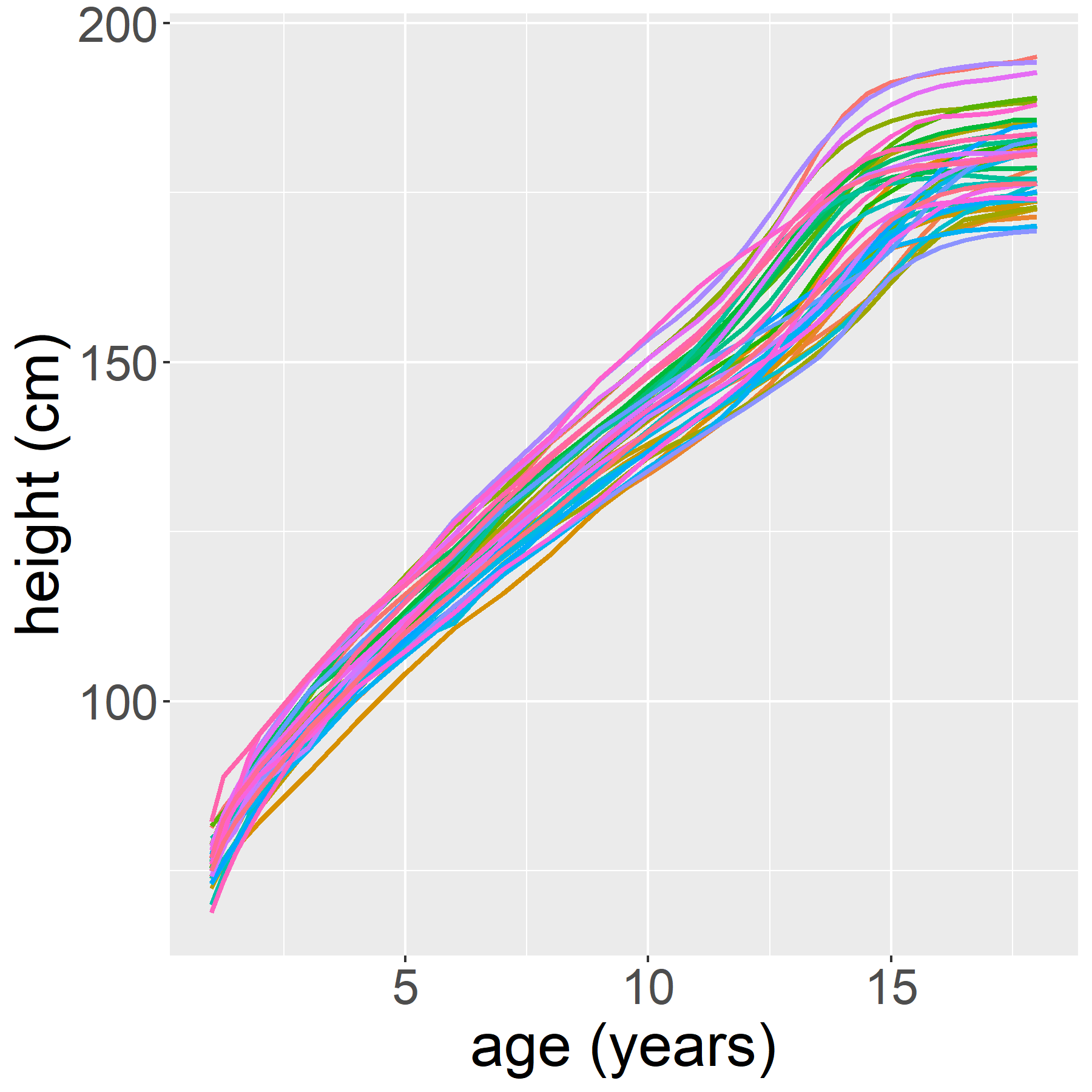} & \includegraphics[width = 4 in]{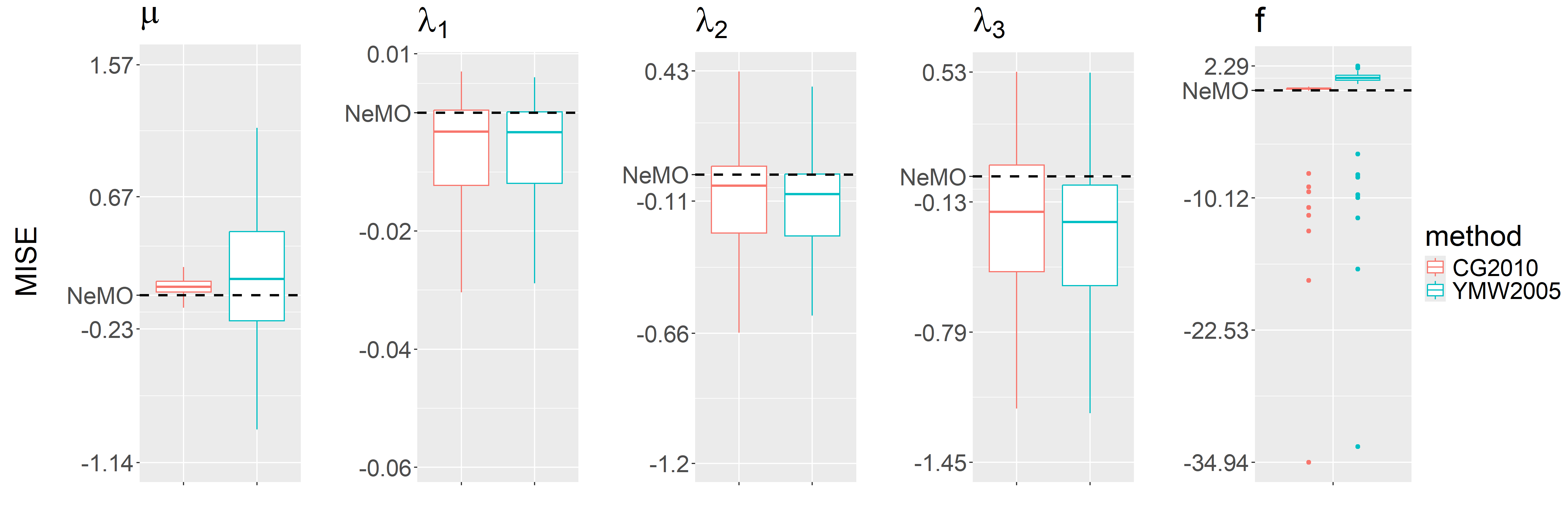} \\
(a) & (b) 
\end{tabular}
\caption{(a) Observed growth curves for $n = 39$ male subjects in the Berkeley growth study. (b) Mean integrated squared error results of estimated model parameters.}\label{fig:growthSim}
    \end{center}
\end{figure}

In this simulated example we generate data that replicate the variability in observations from the Berkeley growth study, described in \citep{ramsay2005}. Panel (a) of Figure \ref{fig:growthSim} displays the height in centimeters of the 39 male children in the study. Using the \texttt{fdapace} \texttt{R} package we inferred the components of an functional factor analysis model that implements the methodology of \cite{yao2005}, under the following observation model
\begin{equation}
    y_i(\vv{t}) = \mu(\vv{t}) + (\lambda_1(\vv{t}),\ldots,\lambda_K(\vv{t}))^\top\eta_i + \epsilon_i(\vv{t}),\ i=1,\ldots,n.
\end{equation}
The ages $\vv{t}$ correspond to quarterly measurements in the first year of age, annual measurements from age two to age eight, then biannual measurements from age 8 to 18. The \texttt{fdapace} package selected $K = 3$, so that the estimated fraction of variance explained by the loadings is 95\%. In this simulation, we sample independent Gaussian latent factors and additive error with variance estimated from \texttt{fdapace}, and plug estimates of the mean and factor loadings into the above observation model to generate functional datasets with variability similar to that of Figure \ref{fig:growthSim}(a). Using the simulated observations, we measure estimation performance of various methods using mean integrated squared error.

For the $\small{\mbox{NeMO}}$ model, we use the posterior mean of $\mu(\vv{t}) + (\lambda_1(\vv{t}),\ldots,\lambda_K(\vv{t}))^\top\eta_i$ to produce an estimate of $f_i$ based on $1,000$ Markov chain Monte Carlo samples. In terms of other functional factor analysis methods, we compare to \cite{yao2005}, \cite{peng2009}, and \cite{crainiceanu2010}. We replicate the entire data generating mechanism $100$ times for comparison. Panel (b) of Figure \ref{fig:growthSim}\footnote{Abbreviations in figure legend: CG2010 - \cite{crainiceanu2010}; YMW2005 - \cite{yao2005}.} shows mean integrated squared error results. The results from the \cite{peng2009} method is excluded from the figure because the mean integrated squared error was orders of magnitudes larger than the other  methods. The mean integrated squared error computed using $\small{\mbox{NeMO}}$, \cite{yao2005}, and \cite{crainiceanu2010} estimates are all comparable. Occasionally, the $\small{\mbox{NeMO}}$ model selects $K<3$, which results in the mean integrated squared error to be larger than the competitors, although this only happened in 8 of the 100 replicates. Even in this realistic simulation setting under model misspecification, our functional factor analysis model using $\small{\mbox{NeMO}}$ processes performs well compared with competing methods.

\section{Additional Information for the Analysis of the Cebu Longitudinal Health and Nutrition Survey Dataset}\label{sec:cebuAdd}

\subsection{Hierarchichal Model Formulation}  

In Section 5 of the main paper, we motivated a flexible model for weight of young children while accounting for scalar and longitudinal covariates. The observation model is restated as follows, 
\begin{eqnarray} 
y_i(\vv{t}_i) & = & O_i\mu(\vv{t}) + O_i(\lambda_1(\vv{t}),\ldots,\lambda_K(\vv{t}))^\top\eta_i + O_i\sum_{j = 1}^4\int_{T}\beta_j(s,\vv{t})z_{i,j}(s)ds + \epsilon_i(\vv{t}_i) \label{eq:obs_model1}\\
 \eta_i & = & \Theta x_i + {\xi}_i \label{eq:obs_model2}
\end{eqnarray}
Similar to the hierarchical functional factor analysis model presented in section \ref{sec:MCMC}, we specify the following model components: 
\begin{eqnarray*}
 & \mu\mid l_\mu,\tau^2_\mu \sim \mbox{GP}(0,C_\mu(\cdot,\cdot)),\ (l^2_\mu)^{-1} \sim \text{gamma}(\alpha_\mu,\beta_\mu), \ \tau^2_\mu \sim \text{half-normal}(\gamma_\mu) \\
  & \lambda_k\mid l_k,\tau^2_k,\Lambda_{(-k)}\sim \mbox{GP}(0,C^{\nu_\lambda}_k(\cdot,\cdot)),\ k = 1,\ldots,K \\ 
  & (l^2_k)^{-1} \sim \text{gamma}(\alpha_\lambda,\beta_\lambda), \ \tau^2_k \sim \text{half-normal}(\gamma_\lambda)\ k = 1,\ldots, K \\
  &  \epsilon_i(\vv{t}_i)\mid\sigma^2\sim N(0,I_{m_i}),\ i = 1\ldots,n,\ (\sigma^2)^{-1} \sim \text{gamma}(\alpha_\sigma,\beta_\sigma).
\end{eqnarray*}
We incorporate scalar covariates, $x_i$, through the latent factors, $\eta_i = B x_i + {\xi}_i$. For these model parameters, we assign 
\begin{eqnarray*}
 & \Theta_{k,q}\mid\psi_k\sim N(0,\psi_k),\ q=1,\ldots,4,\ {\xi}_k\mid\psi_k \sim N(0,\psi_k(I_n + \frac{1}{\nu_\eta}1_n1_n^\top)^{-1}),\ k = 1,\ldots, K \\
  & (\psi_k)^{-1}\sim \text{gamma}(\alpha_\xi,\beta_\xi), k = 1,\ldots, K 
\end{eqnarray*}

As noted in the main paper, the longitudinal covariates are sparsely recorded. Based on initial EDA of the data, the illness indicator covariates, $z_{i,2},z_{i,3}, z_{i,4},\ i = 1,\ldots,n$, do not appear to have much structured variability across subjects. For model based imputation of these covariates, we model the probability of illness at time $s$ as $p(z_{i,j}(s) = 1) = \Phi(\mu^{z_j}(s)),\ i = 1,\ldots,n, j = 2,3,4$. Alternatively, we do suspect that there is structured variability in the breastfeeding indicator $z_{i,1}$, given typical breastfeeding patterns and biological constraints. We model $P(z_{i,1}(t) = 1) = \Phi(\mu^{z_1}(s) + (\lambda^z_1(s),\ldots,\lambda^z_{K_z}(s))^\top{\eta}^z_i)$ as in the generalized functional factor analysis setting of our model. For these components, we assign 
\begin{eqnarray*}
    & \mu^{z_j}(\vv{s})\mid l_{\mu^{z_j}},\tau^2_{\mu^{z_j}} \sim \mbox{GP}(0,C_{\mu^{z_j}}(\cdot,\cdot)), j = 1\ldots,4 \\ & (l^2_{\mu^{z_j}})^{-1}\ \sim \text{gamma}(\alpha_{\mu^z},\beta_{\mu^z}), \ \tau^2_{\mu^{z_j}} \sim \text{half-normal}(\gamma_{\mu^z}),\ j = 1,\ldots,4 \\
    & \lambda^z_k\mid l^z_k,\tau^{2,z}_k,\Lambda^z_{(-k)}\sim \mbox{GP}(0,C^{\nu^z_\lambda}_{k^z}(\cdot,\cdot)),\ \eta^z_{\cdot k}\mid\psi^z_k \sim N(0,\psi^z_k(I_n + \frac{1}{\nu_\eta}1_n1_n^\top)^{-1}),\\
    & (l^{2,z}_k)^{-1} \sim \text{gamma}(\alpha_{\lambda^z},\beta_{\lambda^z}), \ \tau^{2,z}_k \sim \text{half-normal}(\gamma_{\lambda^z}),\ (\psi^z_k)^{-1}\sim \text{gamma}(\alpha_{\eta^z},\beta_{\eta^z}),\\ & k = 1,\ldots, K^z.
\end{eqnarray*}
For grid points where the covariates are not observed, $s^{miss}_{i,z_j},\ j = 1,\ldots,4$, values are imputed for $z_{i,j}(s^{miss}_{i,z_j})$ based on the above models. These values are sampled throughout the Markov chain Monte Carlo algorithm to account for the uncertainty due to the missing data. 

Finally, we can specify the model components for the functional linear model component \citep{ramsay1991}. Each $\beta_{(s,t)}$ is represented via a basis expansion. For the historical linear model component for the breastfeeding indicator, we use the tent basis described in \cite{malfait2003}. We denote these basis elements with $V_{p,1}(s,t),\ p = 1,\ldots,P$, so that $\int_{T}\beta_1(s,t)z_{i,1}(s)ds = \sum_{p = 1}^{P_1}\rho_{p,1}\int_{T}V_{p,1}(s,t)z_{i,1}(s)ds$, where the integrals are approximated numerically based on the imputed breastfeeding indicator covariate. For the illness covariates, the integral form equates to $\int_{T}\beta_j(s,t)z_{i,j}(s)ds = \beta_j(t)z_{i,j}(t),\ j = 2,3,4$ in the concurrent linear model \citep{hastie1993}. We represent each with a b-spline basis, $\beta_j(t) = \sum_{p = 1}^{P_j}\rho_{p,j}V_{p,j}(t)$. Inference for these model components is carried out through inference on the basis coefficients for which we specify 
\begin{eqnarray*}
    & \rho_{p,j} \sim N(0,\tau^2_{\rho}),\ p = 1,\ldots,P_j,\ j = 1,\ldots,4.
\end{eqnarray*}

\subsection{Additional Figures for Posterior Inference}\label{sec:add_post_figs}

\begin{figure}[t]
\begin{center}
 \begin{tabular}{c}
\includegraphics[width = 3.5 in]{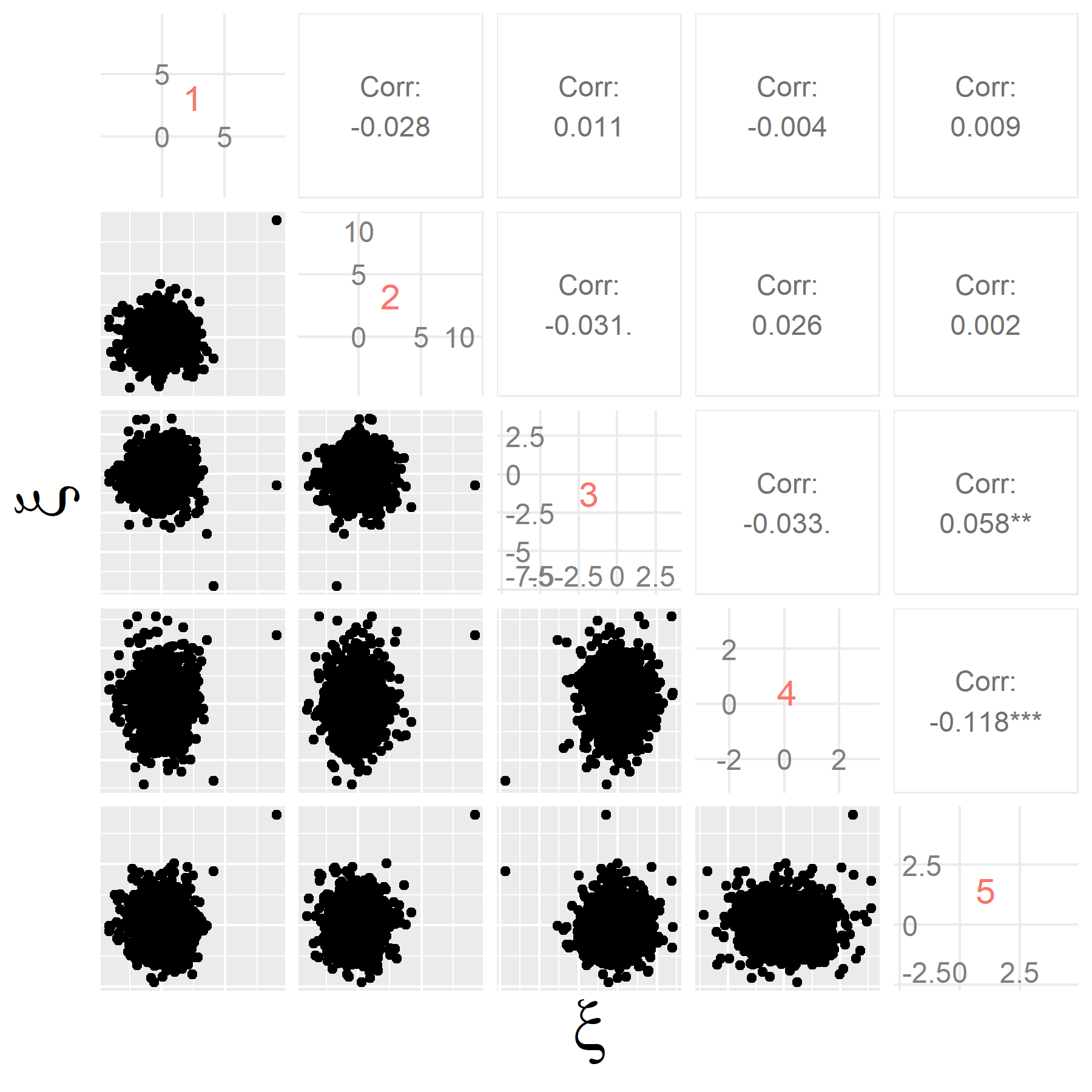} \\
(a) \\
\includegraphics[width = 3.5 in]{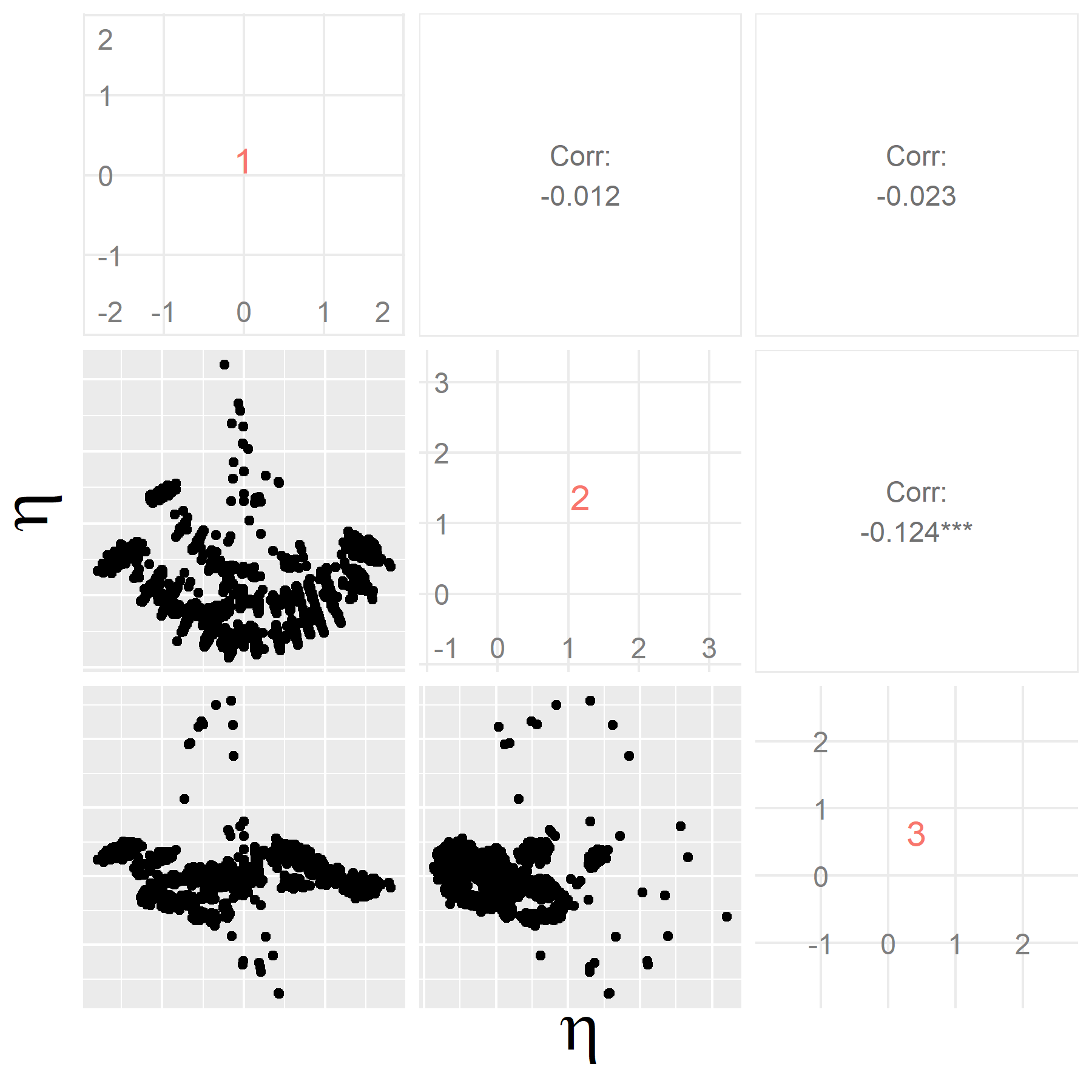} \\
(b) 
\end{tabular}
\caption{Scatter plot matrices of posterior means of latent factors for (a) $\xi$ (weight) and (b) $\eta^z$ (breastfeeding). }\label{fig:cebu_scores}
    \end{center}
\end{figure}

\begin{figure}[t]
\begin{center}
 \begin{tabular}{cc}
\includegraphics[width = 2 in]{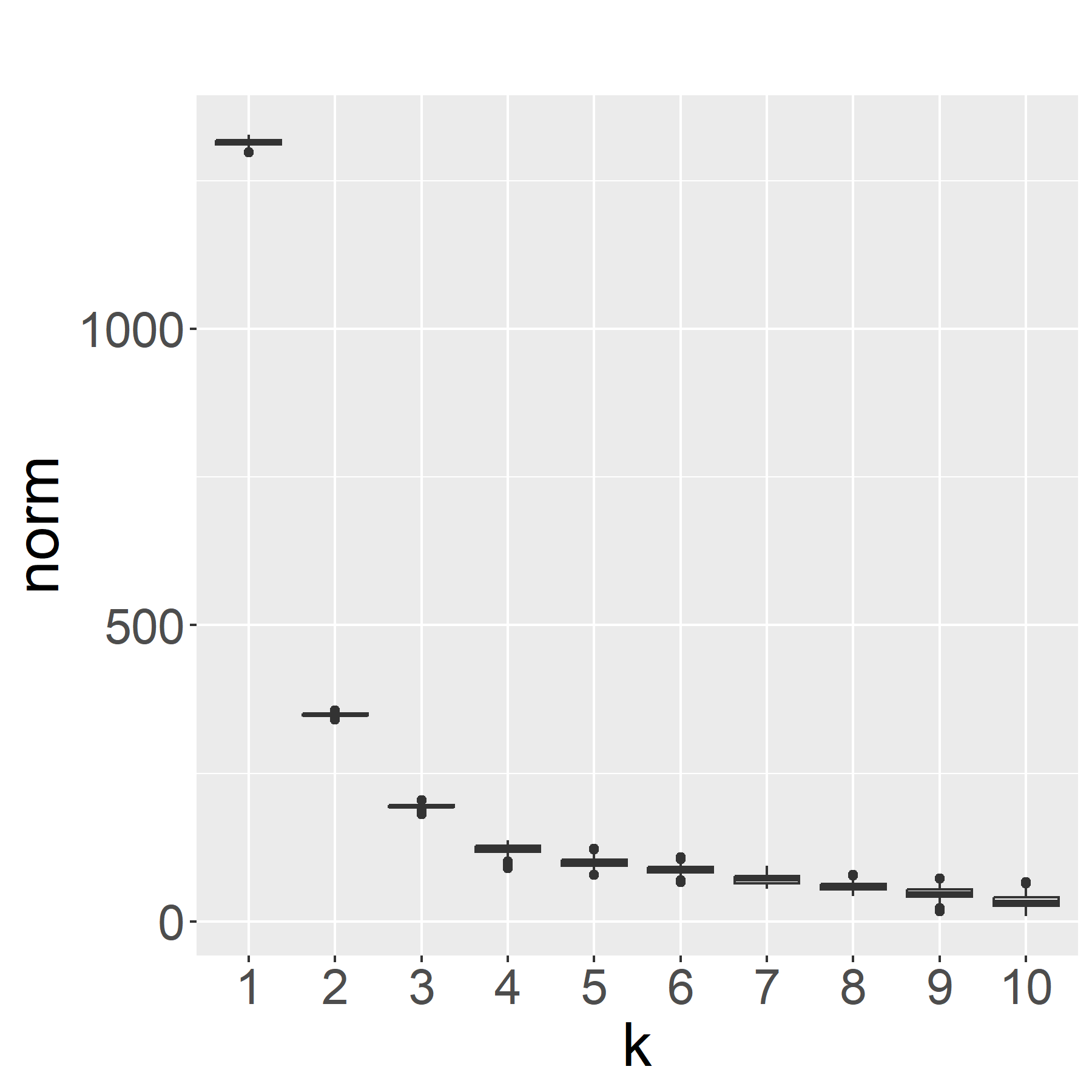} & \includegraphics[width =  2 in]{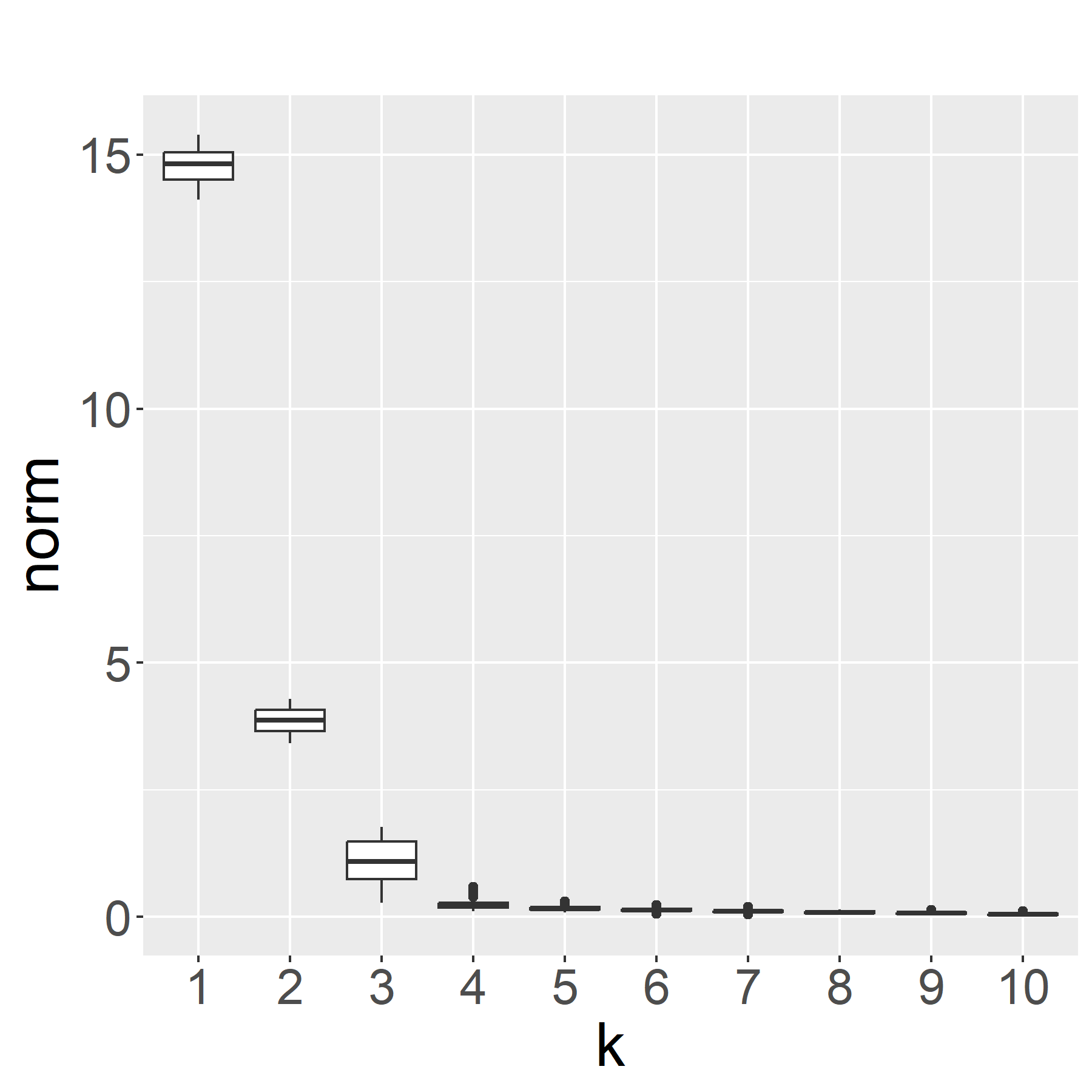} \\
(a) & (b) 
\end{tabular}
\caption{Scree plots that display boxplots of posterior samples of norms of factor loadings for the (a) weight and (b) breastfeeding trajectories of children in the Cebu Longitudinal Health and Nutrition Survey.}\label{fig:cebu_FPC_scree}
    \end{center}
\end{figure}

\begin{figure}[t]
\begin{center}
 \begin{tabular}{cccc}
\includegraphics[width = 1.25 in]{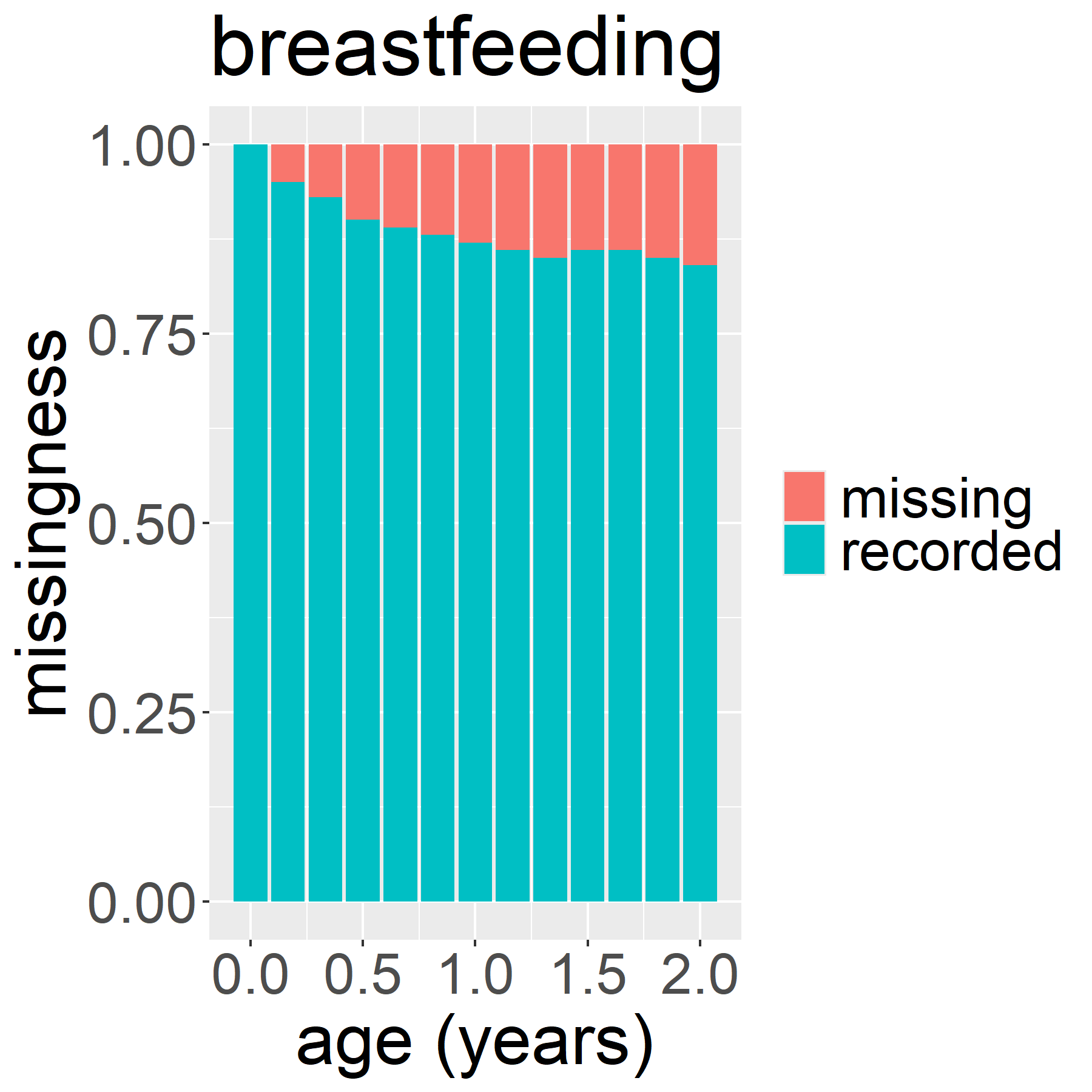} & \includegraphics[width = 1.25 in]{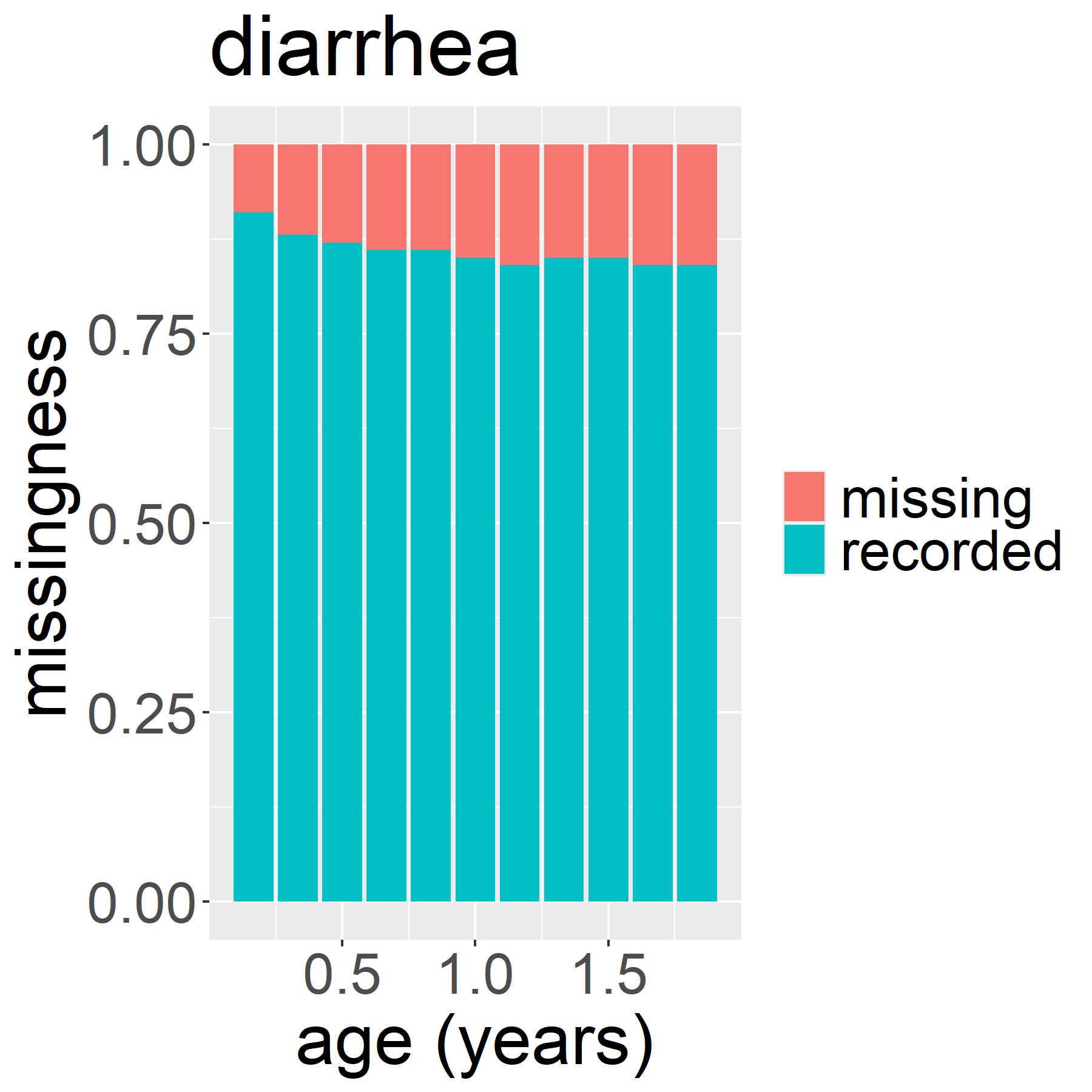} & \includegraphics[width = 1.25 in]{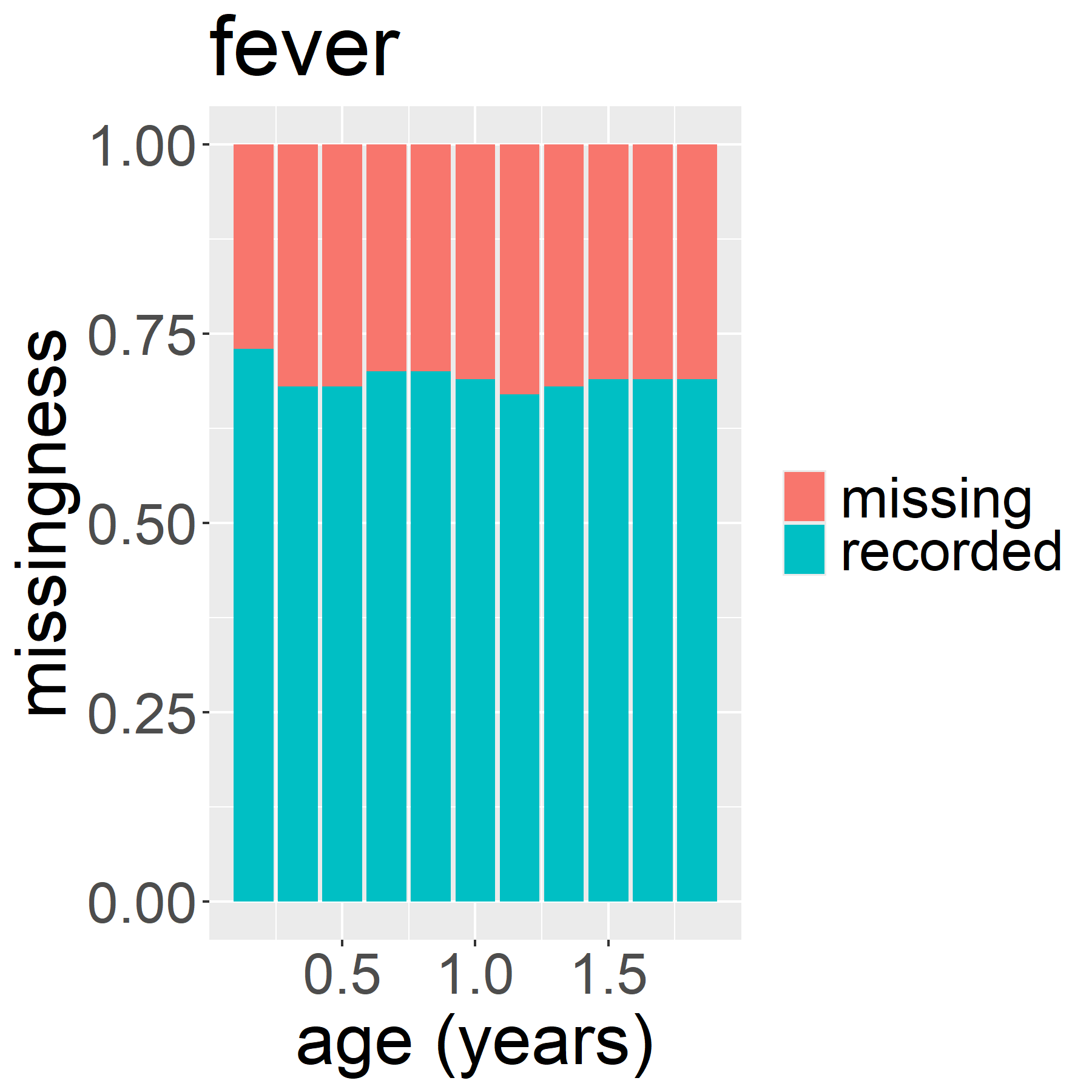} & \includegraphics[width = 1.25 in]{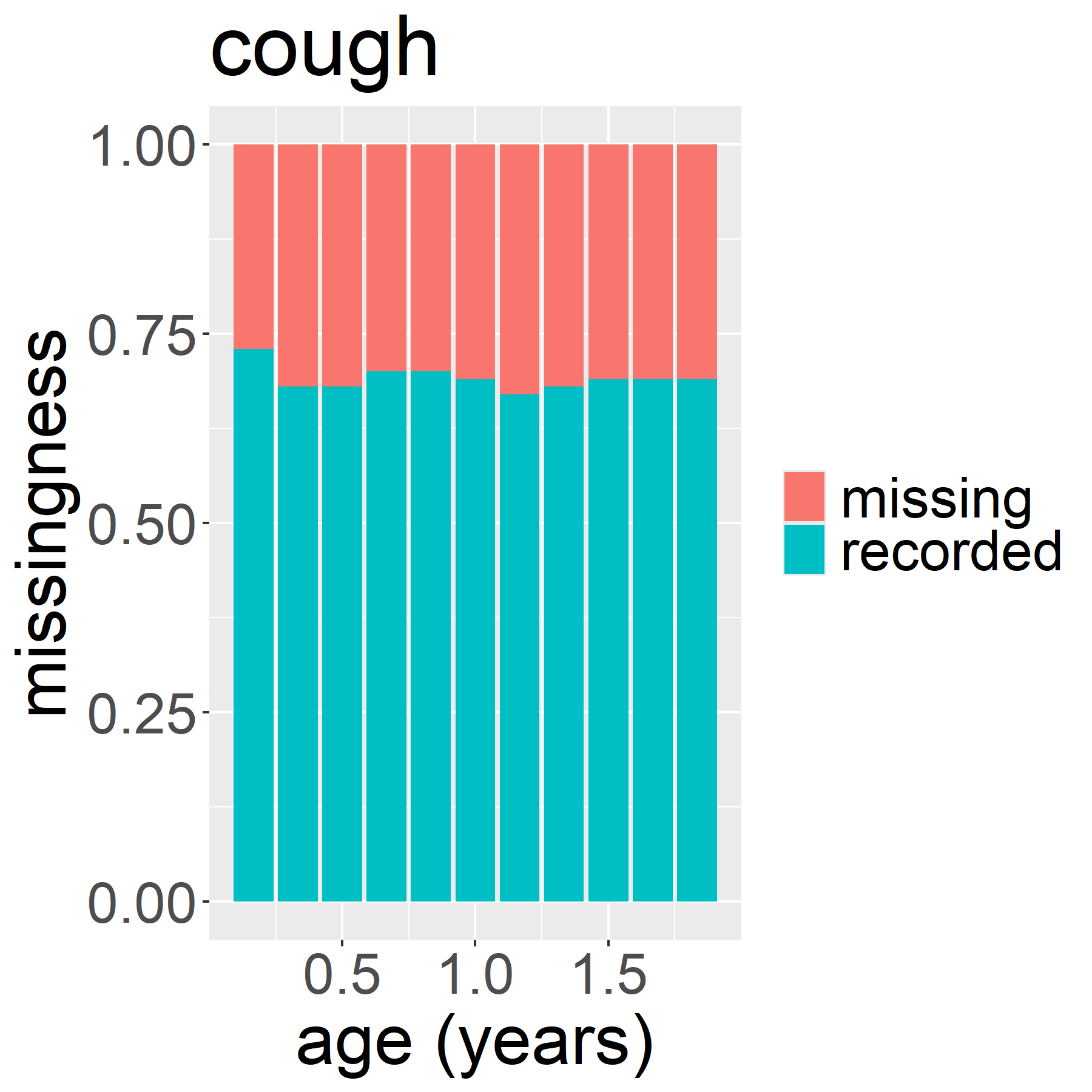}  \\
(a) & (b) & (c) & (d)
\end{tabular}
\caption{(a)-(d) Proportion of missing data by age for each of the binary longitudinal covariates (breastfeeding, diarrhea, fever, cough) in the Cebu Longitudinal Health and Nutrition Survey.}\label{fig:cebu_cov_missing}
    \end{center}
\end{figure}

\begin{figure}[t]
\begin{center}
 \begin{tabular}{ccc}
\includegraphics[width = 1.75 in]{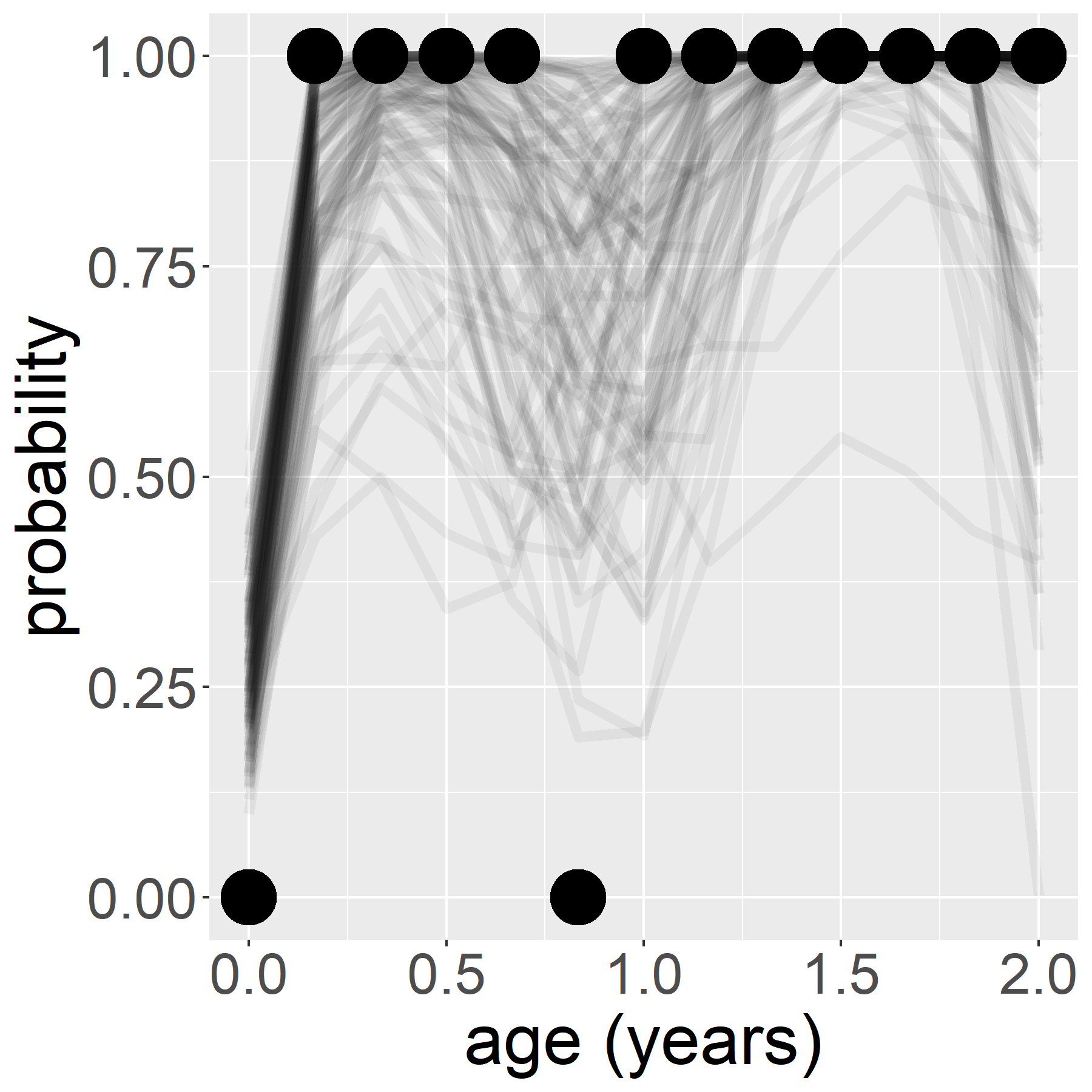} & \includegraphics[width = 1.75 in]{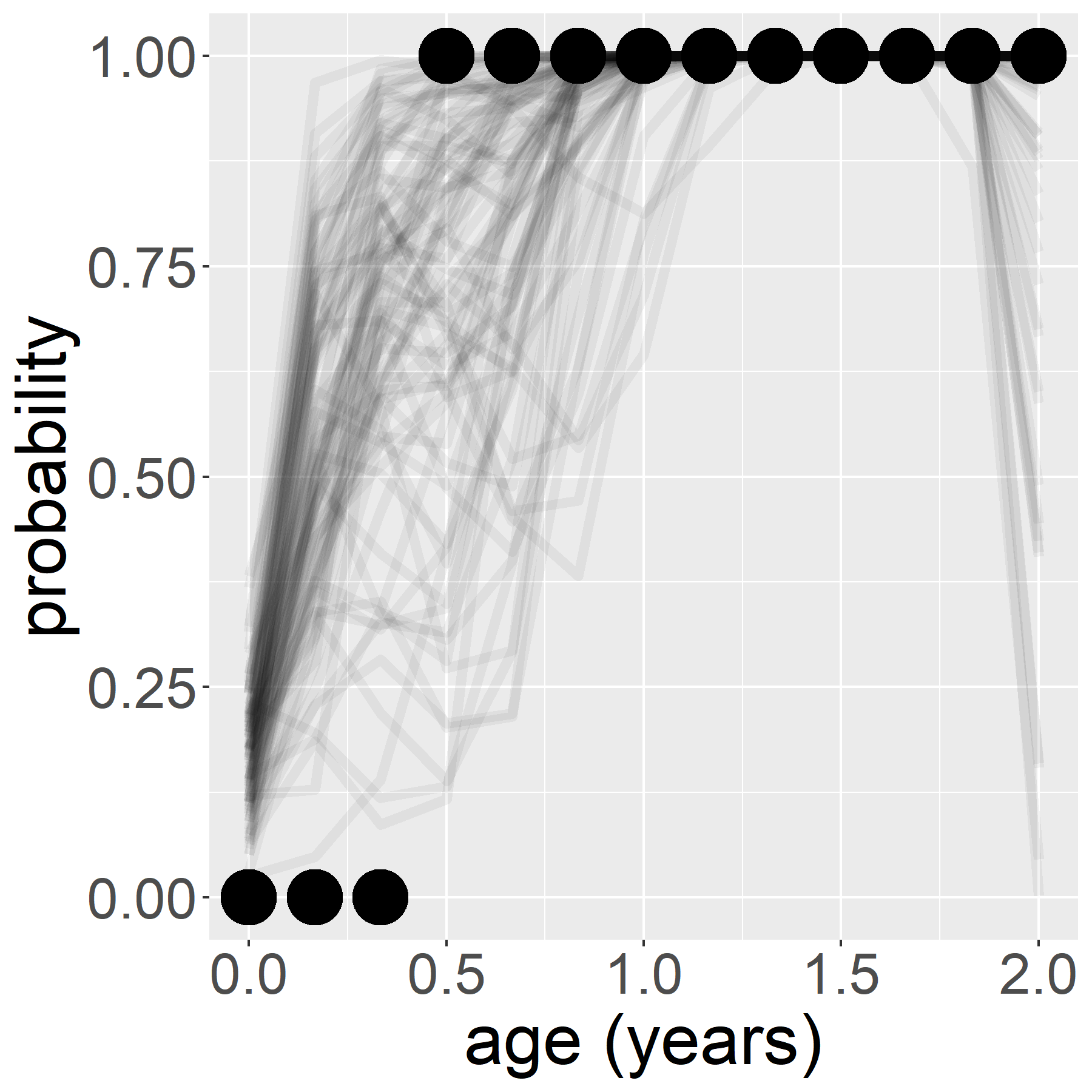} & \includegraphics[width = 1.75 in]{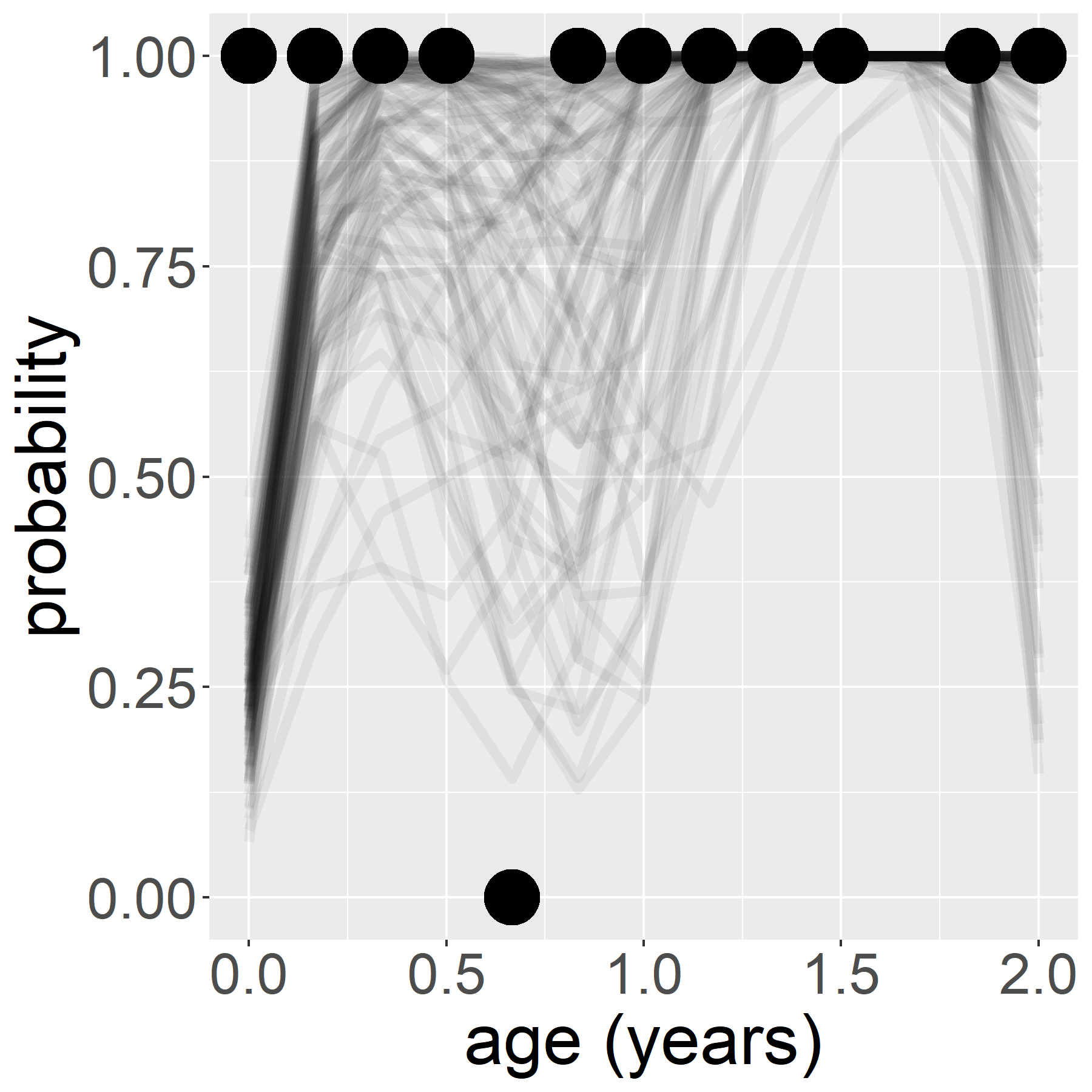} \\
(a) & (b) & (c) \\

\includegraphics[width = 1.75 in]{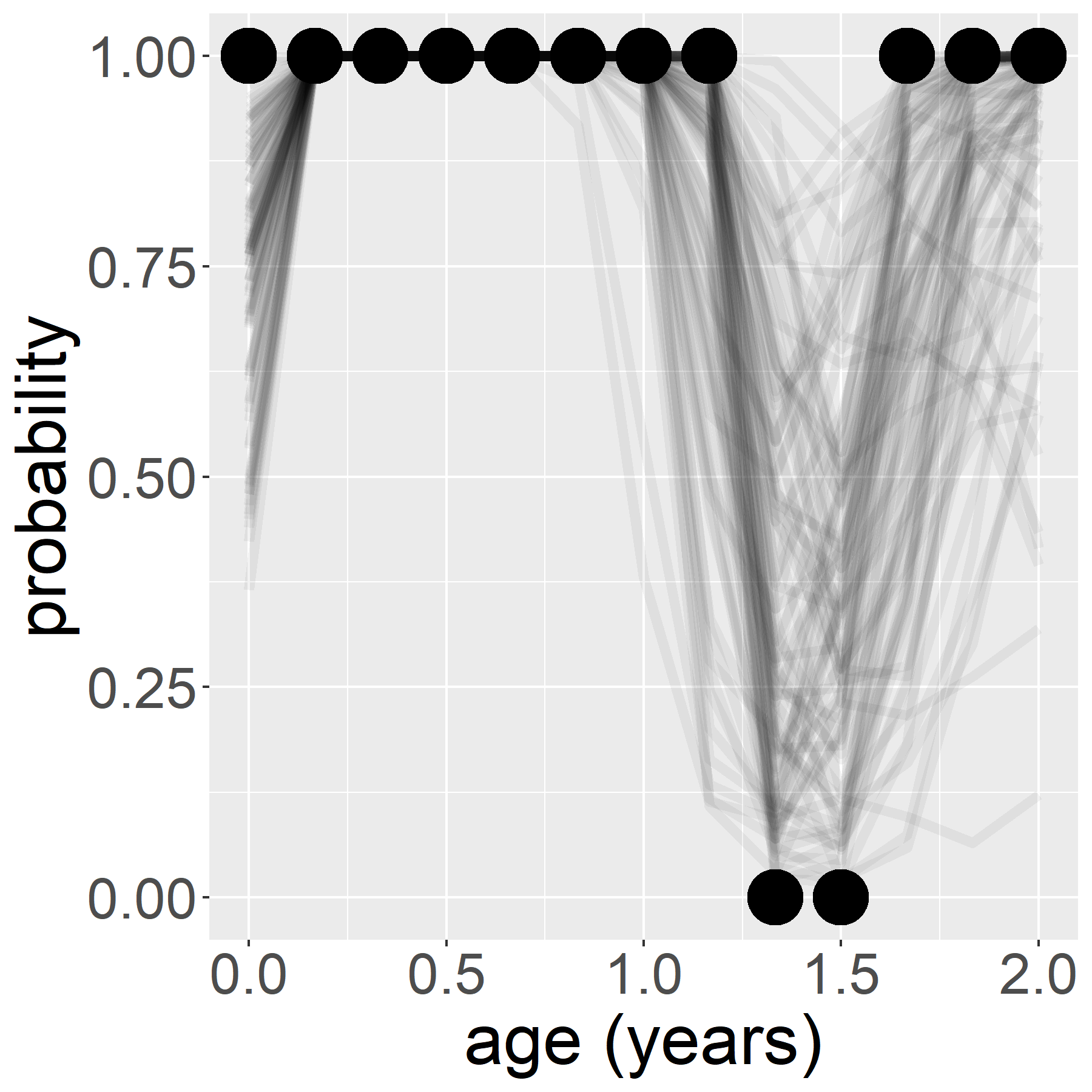} & \includegraphics[width = 1.75 in]{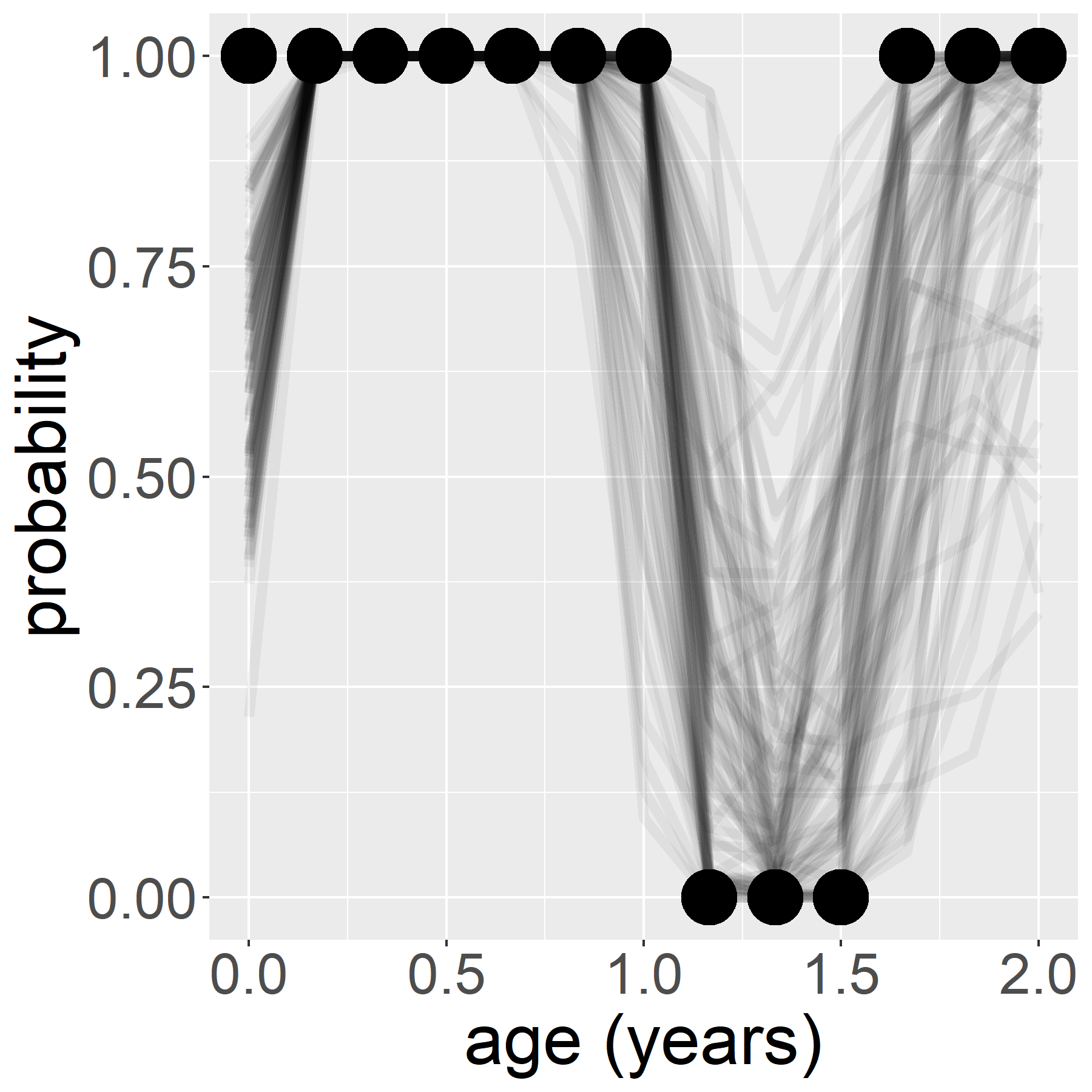} & \includegraphics[width = 1.75 in]{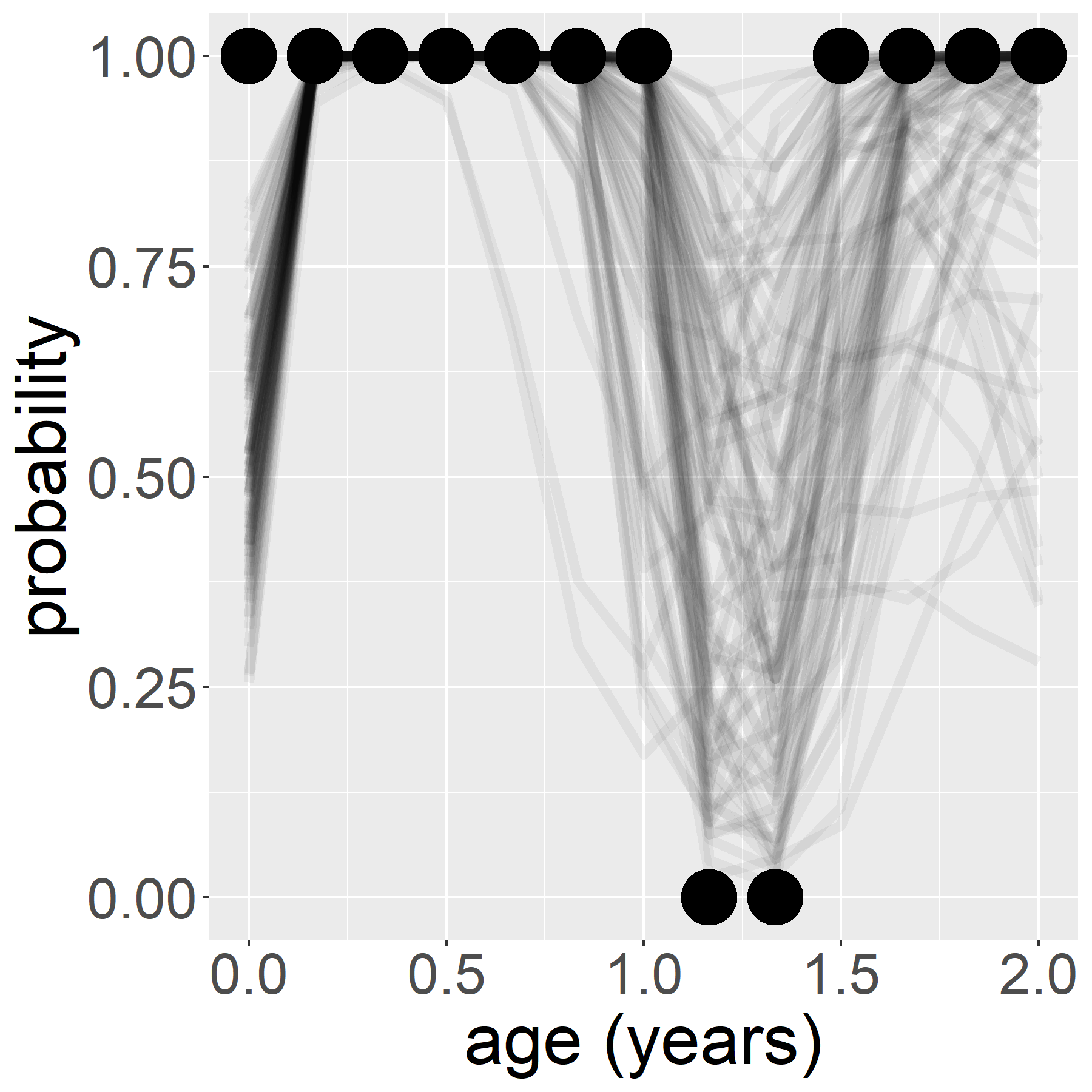} \\
(d) & (e) & (f) \\
\end{tabular}
\caption{(a)-(f) Posterior samples of fitted atypical breastfeeding trajectories, $\Phi(\mu^{z_1}(\protect\vv{t}) + (\lambda^z_1(\protect\vv{t}),\lambda^z_2(\protect\vv{t}),\lambda^z_3(\protect\vv{t}))^\top\eta^z_i)$ overlaid on $z_{i,1}(\protect\vv{t}_i)$ for $i = 212,\ 2782,\ 2809,\ ,1516,\ 1844,\ 2701$. }\label{fig:cebu_bf_outliers}
    \end{center}
\end{figure}

In the main paper, the inferential results for the analysis of the Cebu Longitudinal Health and Nutrition Survey are reported with $K = 5$ and $K^z = 3$. These numbers of factors were selected based on the approach outlined in Section 3.5 of the main paper. Figure \ref{fig:cebu_FPC_scree} displays scree plots of norms of inferred loadings for modeling weight and breastfeeding trajectories based on a larger number of latent factors. In both of these cases, the first few loadings are significantly larger in magnitude compared to trailing loadings. 

Figure \ref{fig:cebu_scores} shows latent factors for weight and breastfeeding status. In panel (a), the factors of weight appear much more normal, relative to the latent factors associated with breastfeeding status shown in panel (b). We investigate atypical observations for breastfeeding in Figure \ref{fig:cebu_bf_outliers}. Outliers in terms of the second factor are shown in panels (a)-(c). In general, it looks like these outliers correspond to the situation where a child was reported not to have breastfed erratically at an early age. Panels (d)-(f) show outliers in terms of the third factor. These outliers correspond to the situation where a child was reported not to have breastfed sometime after age 1, but continued to breastfeed after age 1.5. In addition to these outliers, the non-normal structure of the factors is replicated in exploratory data analyses presented in Section \ref{sec:freq_model}.

\subsection{Data Analysis using  \citep{crainiceanu2010}}\label{sec:freq_model}

\begin{figure}[t]
\begin{center}
 \begin{tabular}{cc}
\includegraphics[width = 1.25 in]{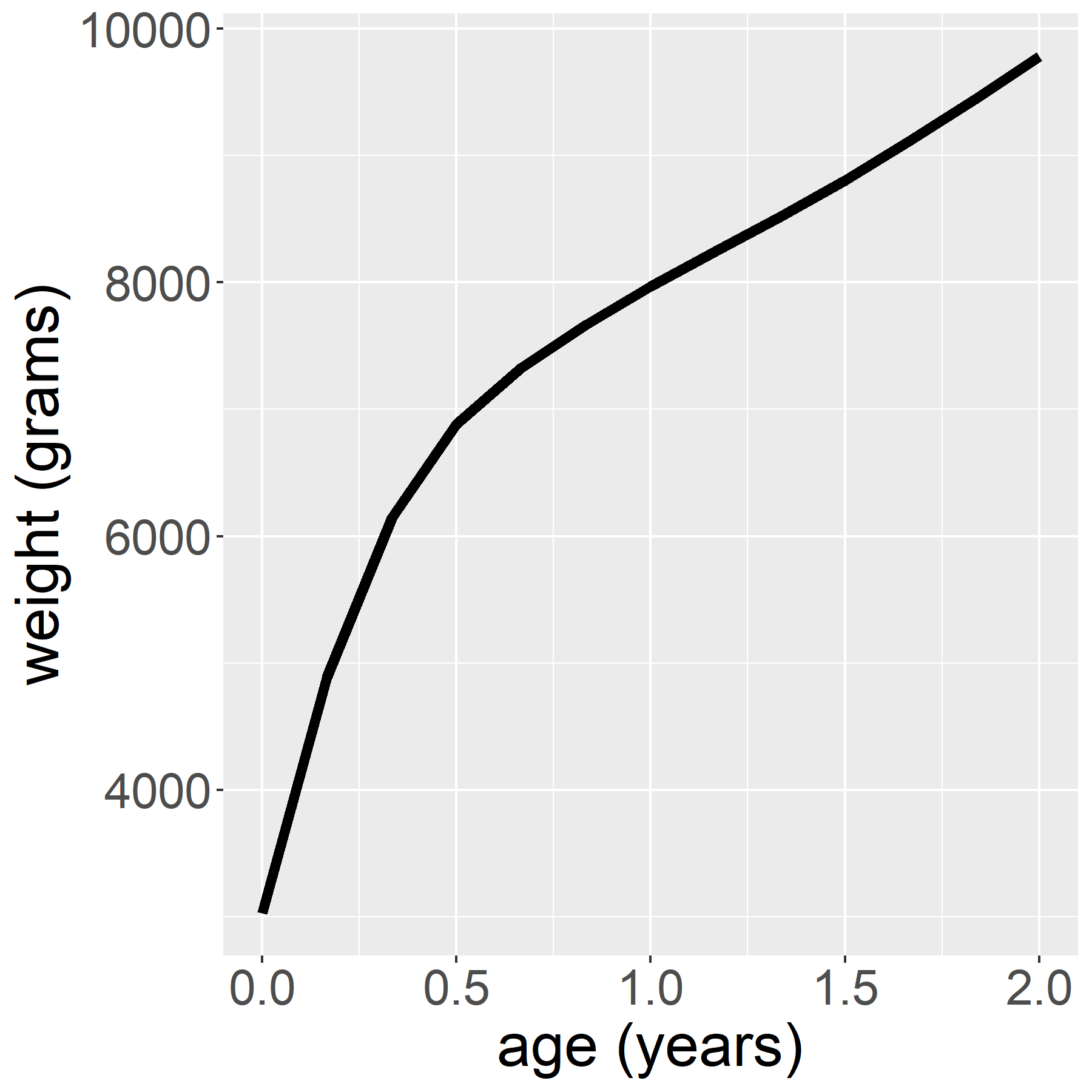}&\includegraphics[width = 1.25 in]{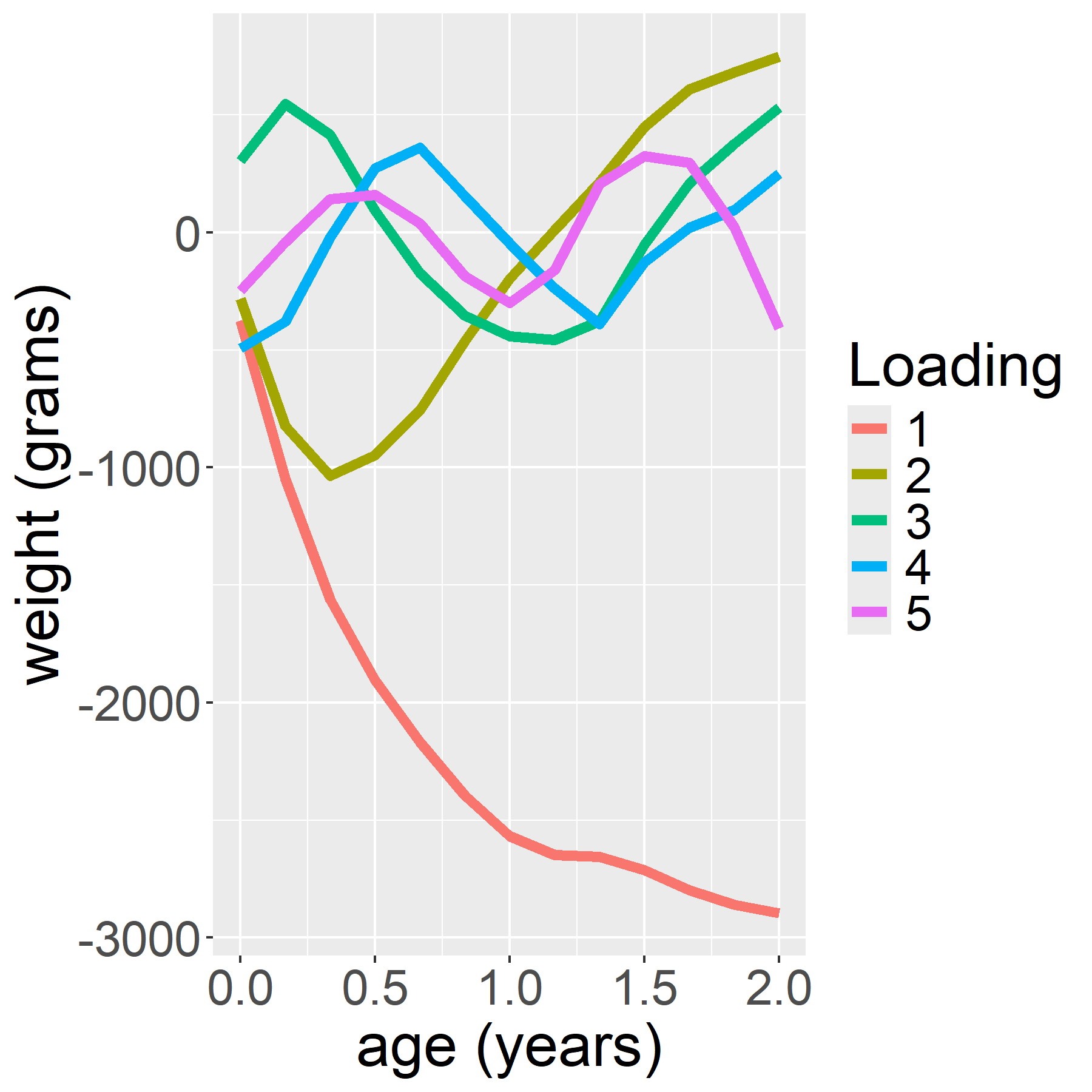} \\
(a) & (b)  \\
\end{tabular}
\begin{tabular}{ccc}
\includegraphics[width = 3 in]{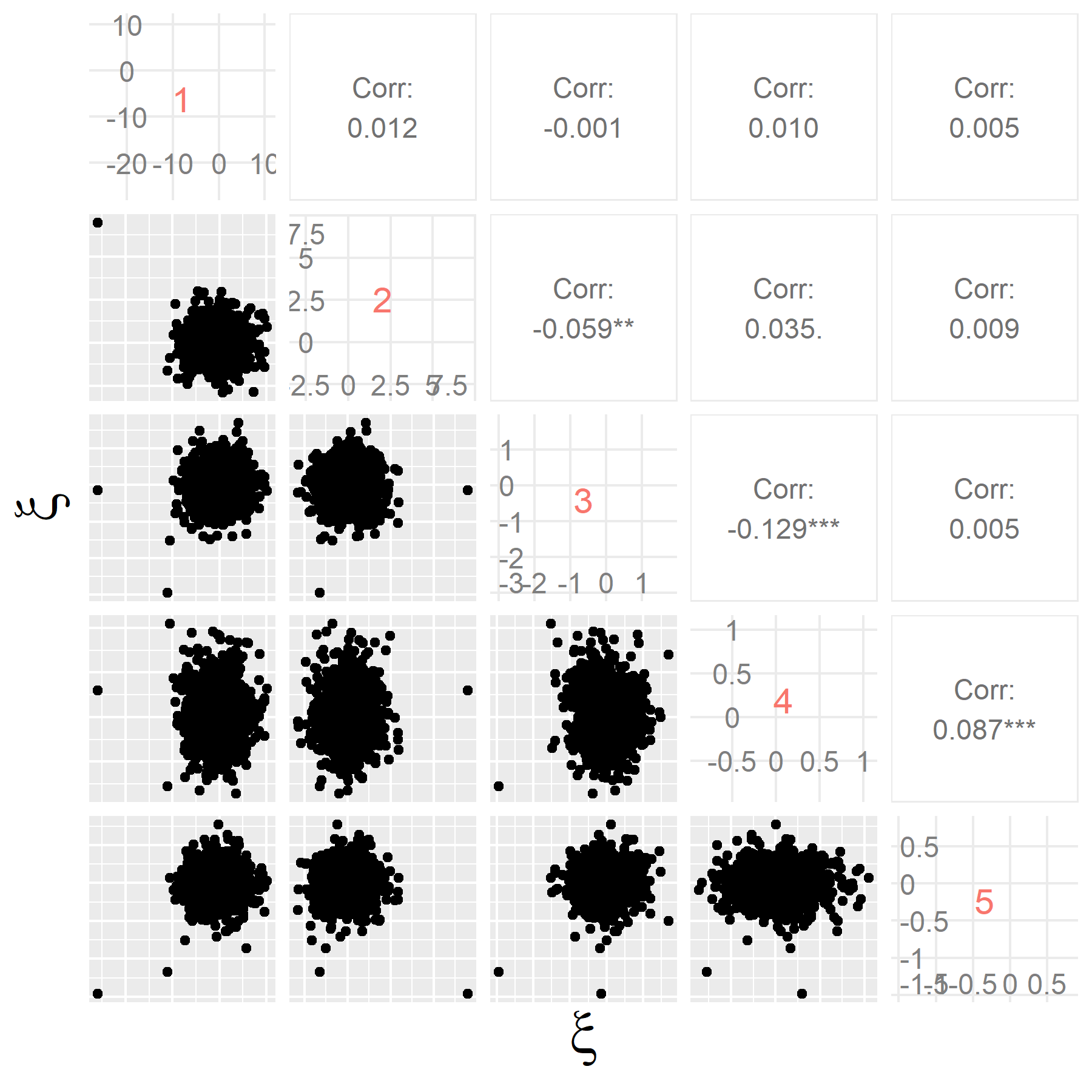}\\
(c) \\
\end{tabular}
 \begin{tabular}{cccc}
\includegraphics[width = 1.25 in]{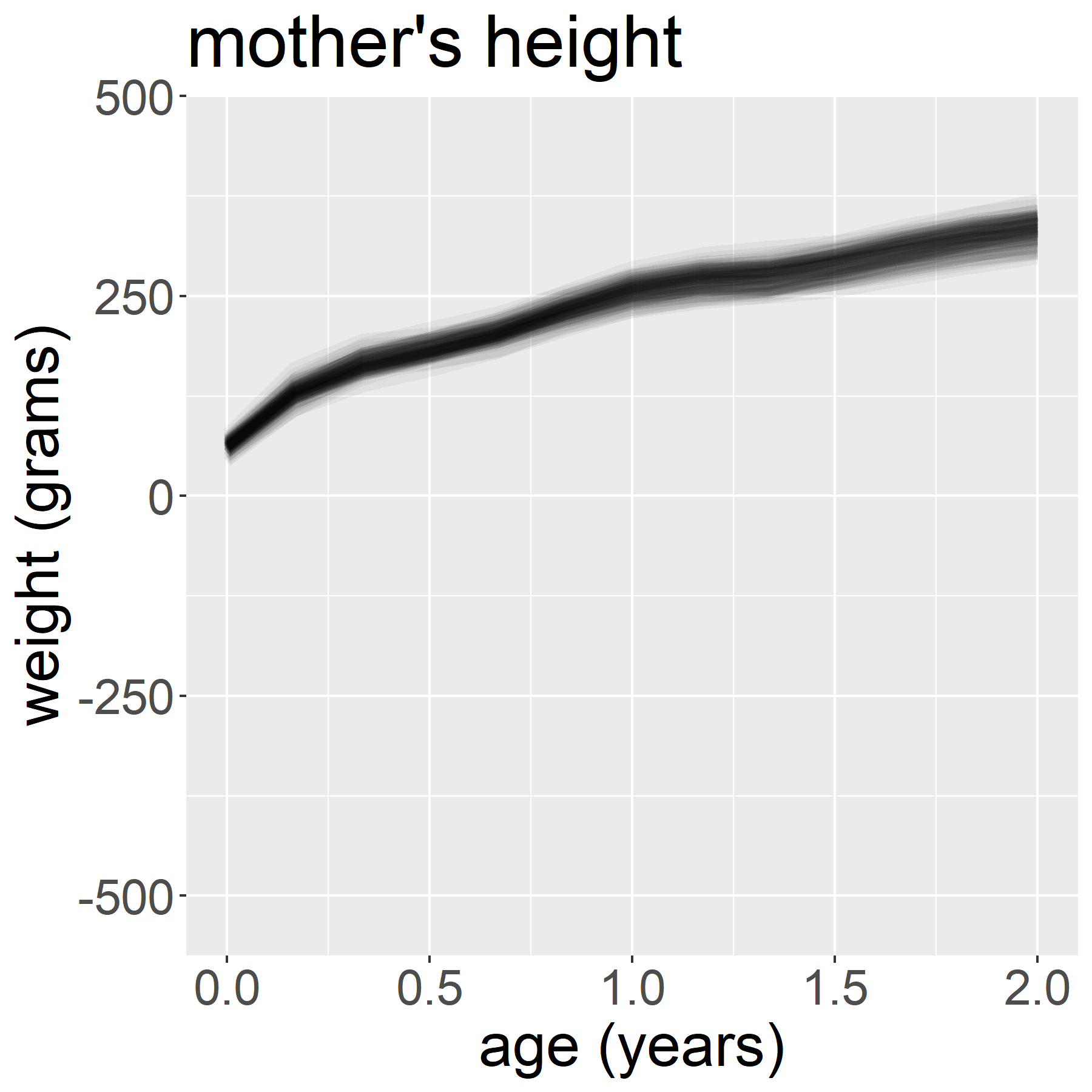}& \includegraphics[width = 1.25 in]{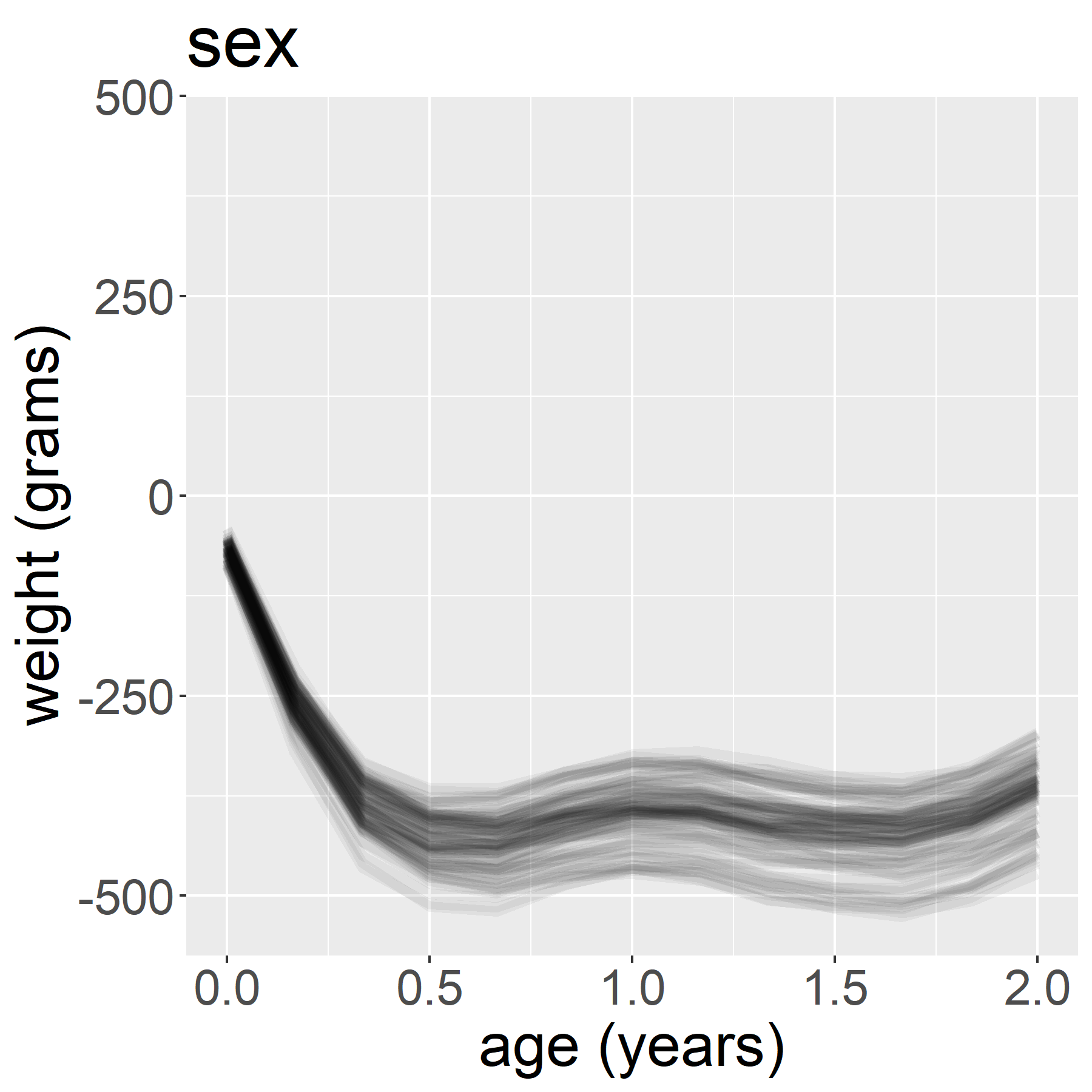}&\includegraphics[width = 1.25 in]{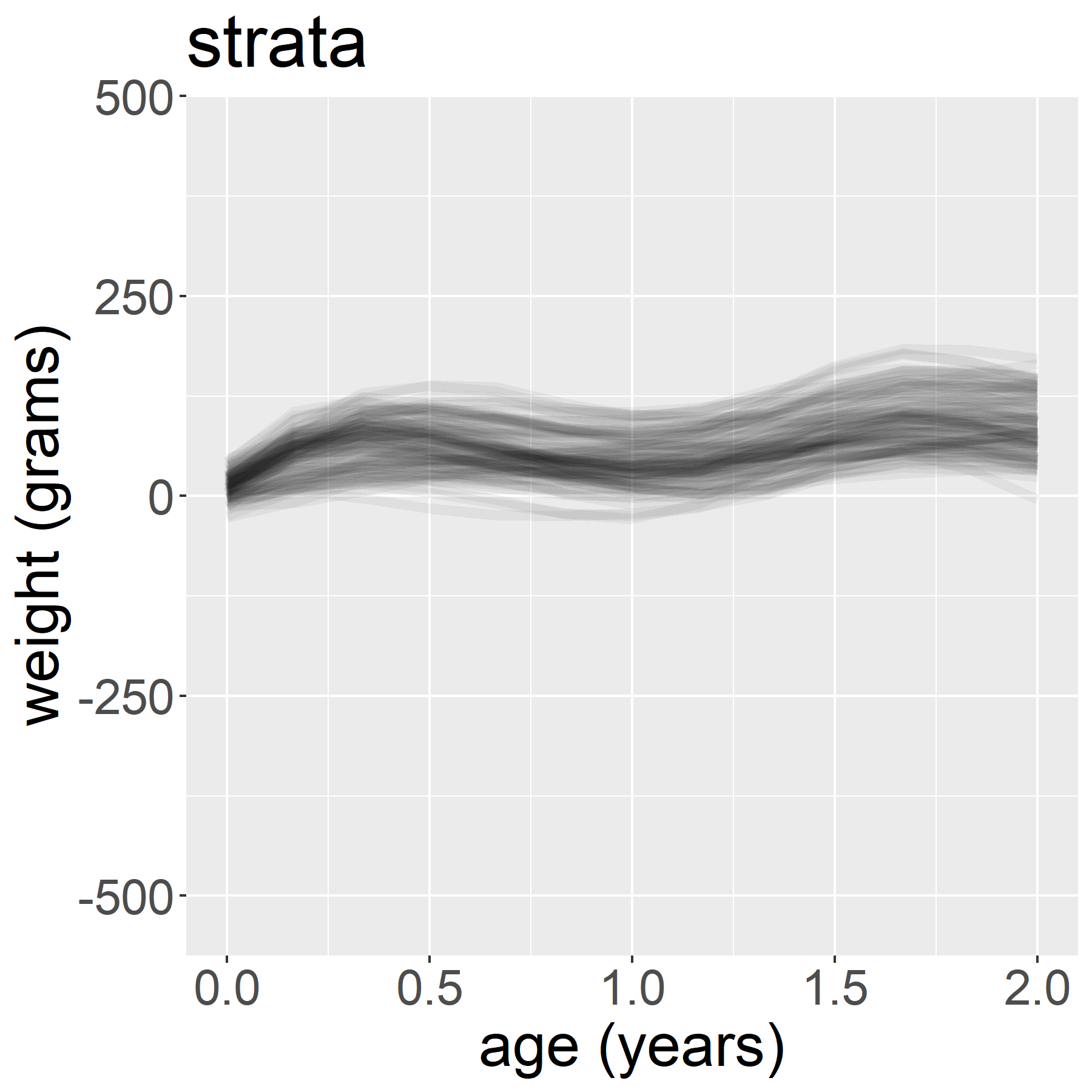}&\includegraphics[width = 1.25 in]{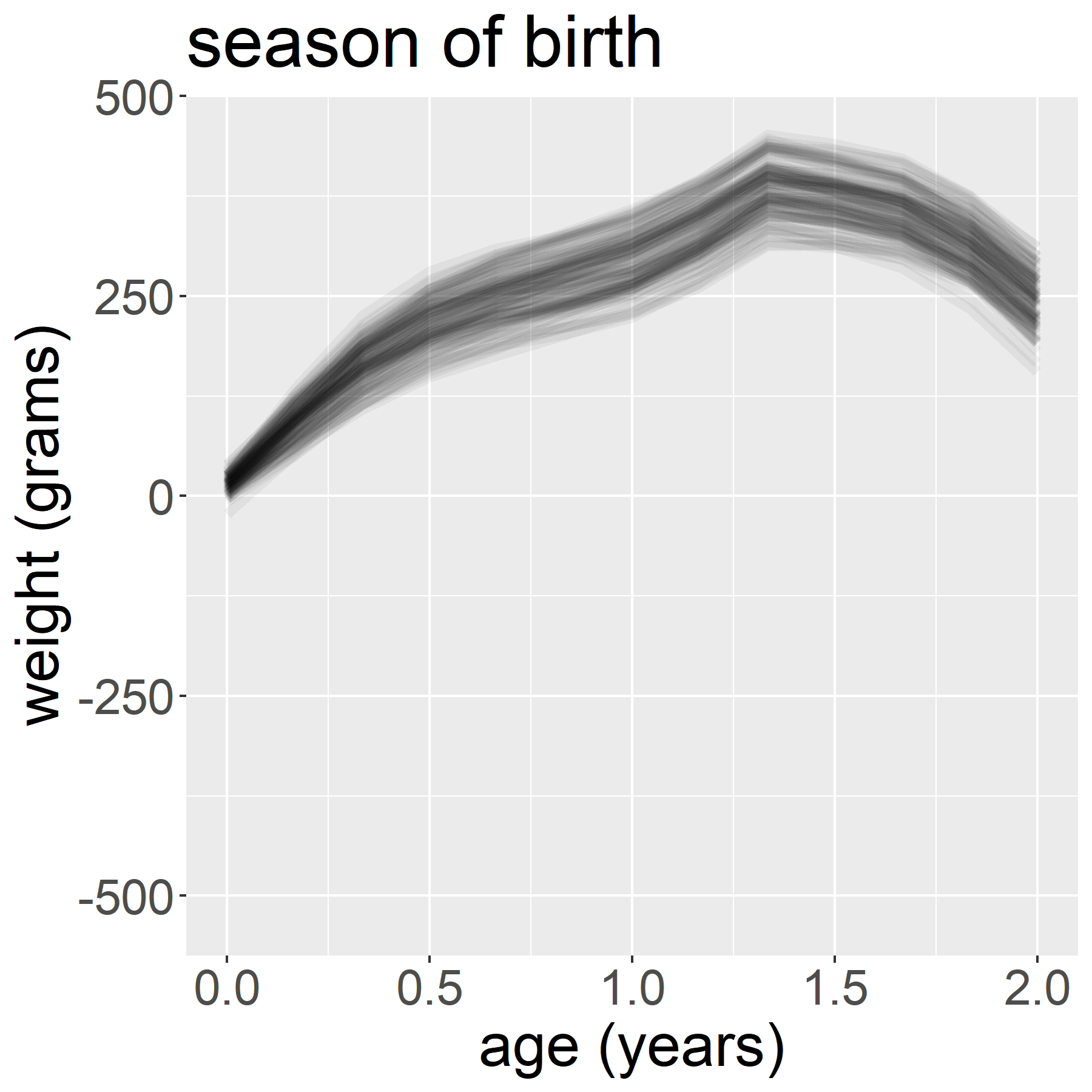} \\
(d) & (e) & (f) & (g) \\
\end{tabular}
\caption{Inferred parameters related to functional factor analysis of weight and scalar covariates using the model of \cite{crainiceanu2010}: (a) $\hat{\mu}(t)$. (b) $\sqrt{\hat\psi_1}\hat{\lambda}_1,\ldots,\sqrt{\hat\psi_K}\hat{\lambda}_K$. (c) Posterior means of $\xi_{i,k}\ i = 1,\ldots,n\ k = 1,\ldots,K$. (d)-(g) Posterior samples of $(\hat{\lambda}_1(t),\ldots,\hat{\lambda}_K(t))\Theta_q,\ q = 1\ldots,4$. }\label{fig:wb_cebu_fpca}
    \end{center}
\end{figure}

\begin{figure}[t]
\begin{center}
 \begin{tabular}{ccc}
\includegraphics[width = 1.25 in]{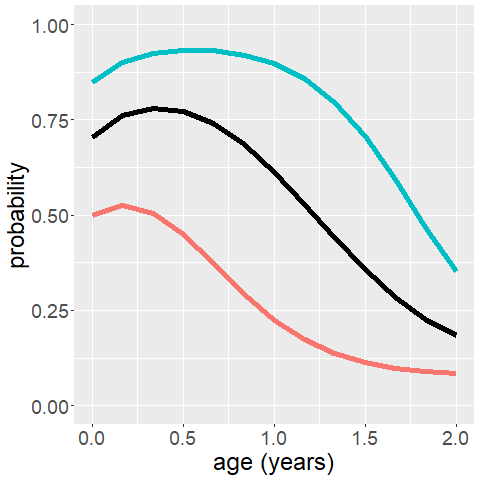}& \includegraphics[width = 1.25 in]{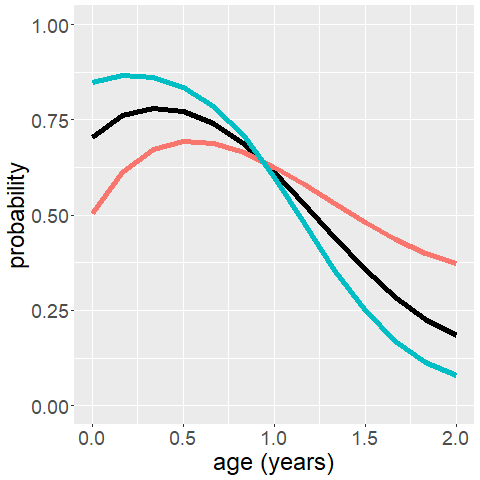} & \includegraphics[width = 1.25 in]{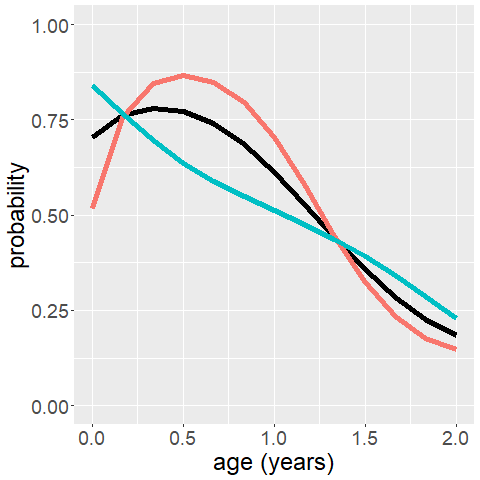} \\
(a) & (b) & (c) \\
\end{tabular}
 \begin{tabular}{c}
 \includegraphics[width = 2.5 in]{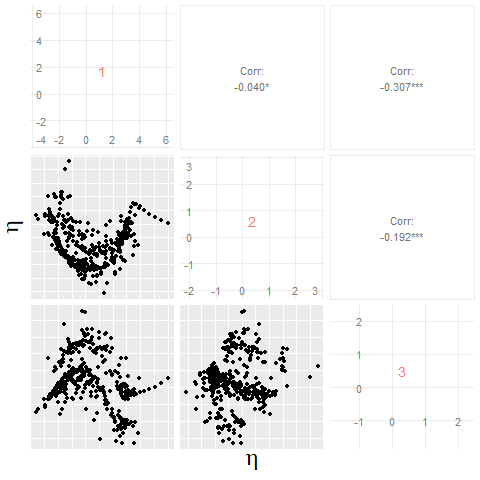}\\
 (d)\\
\end{tabular}
 \begin{tabular}{ccccc}
\includegraphics[width = 1.0 in]{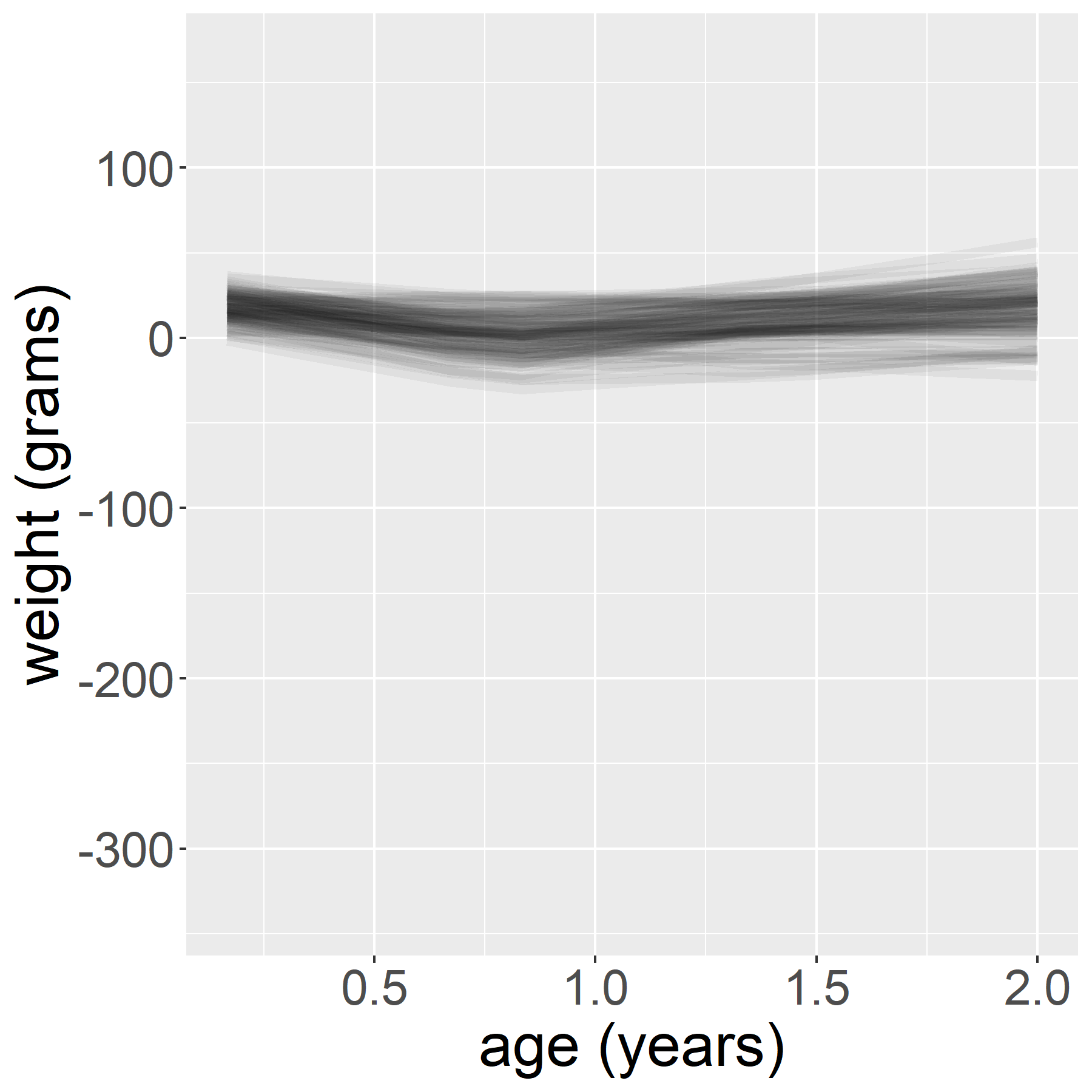}&\includegraphics[width = 1.0 in]{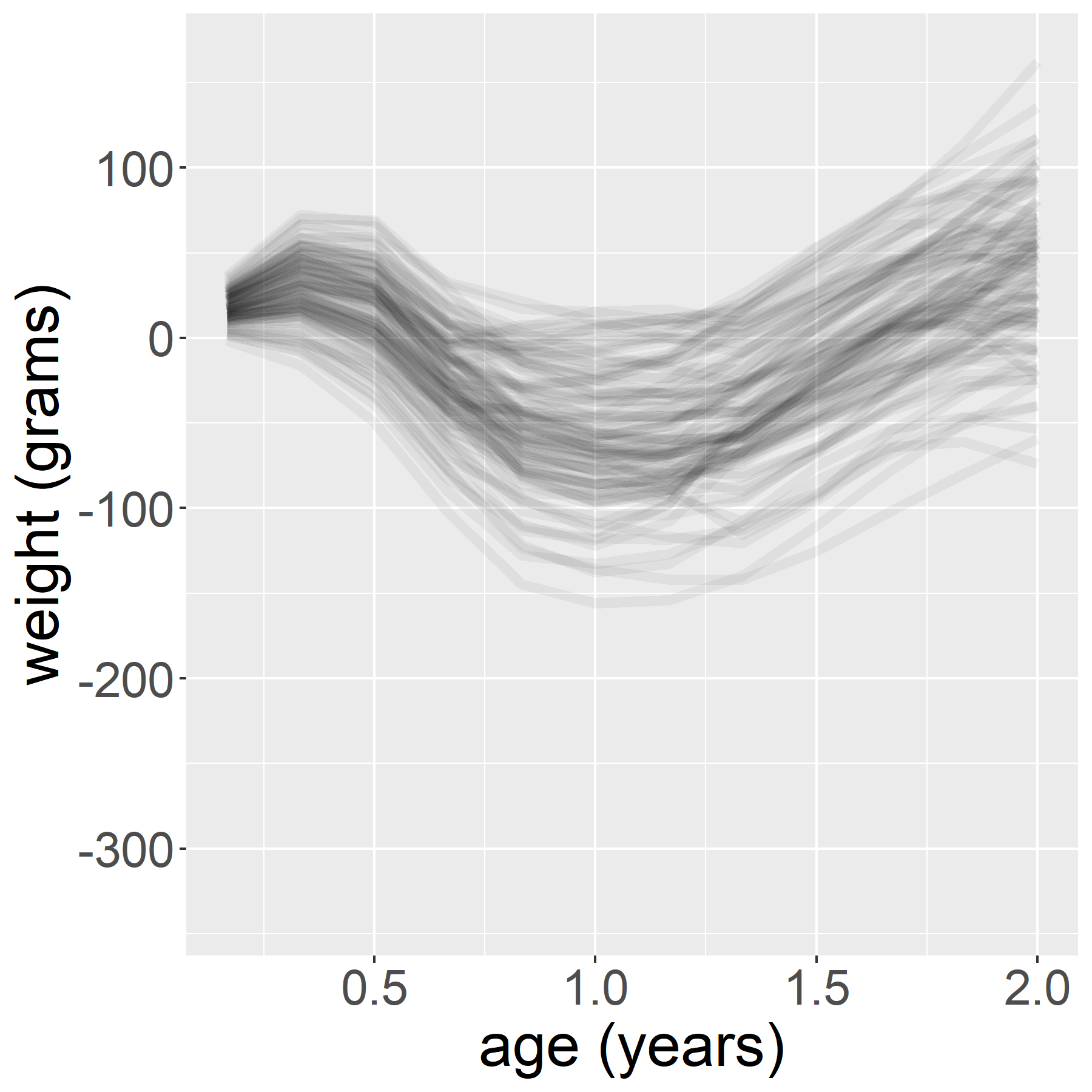}&\includegraphics[width = 1.0 in]{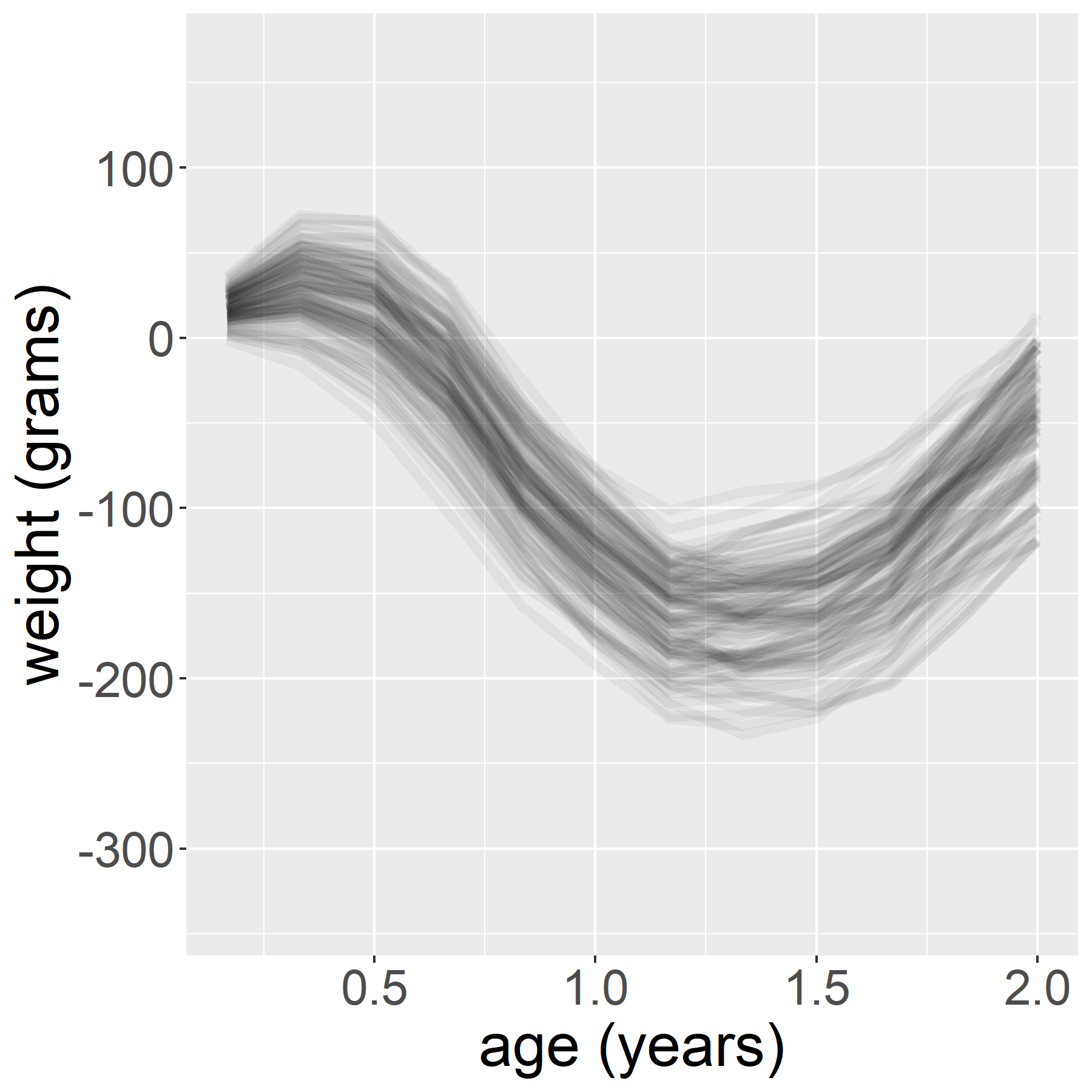}&\includegraphics[width = 1.0 in]{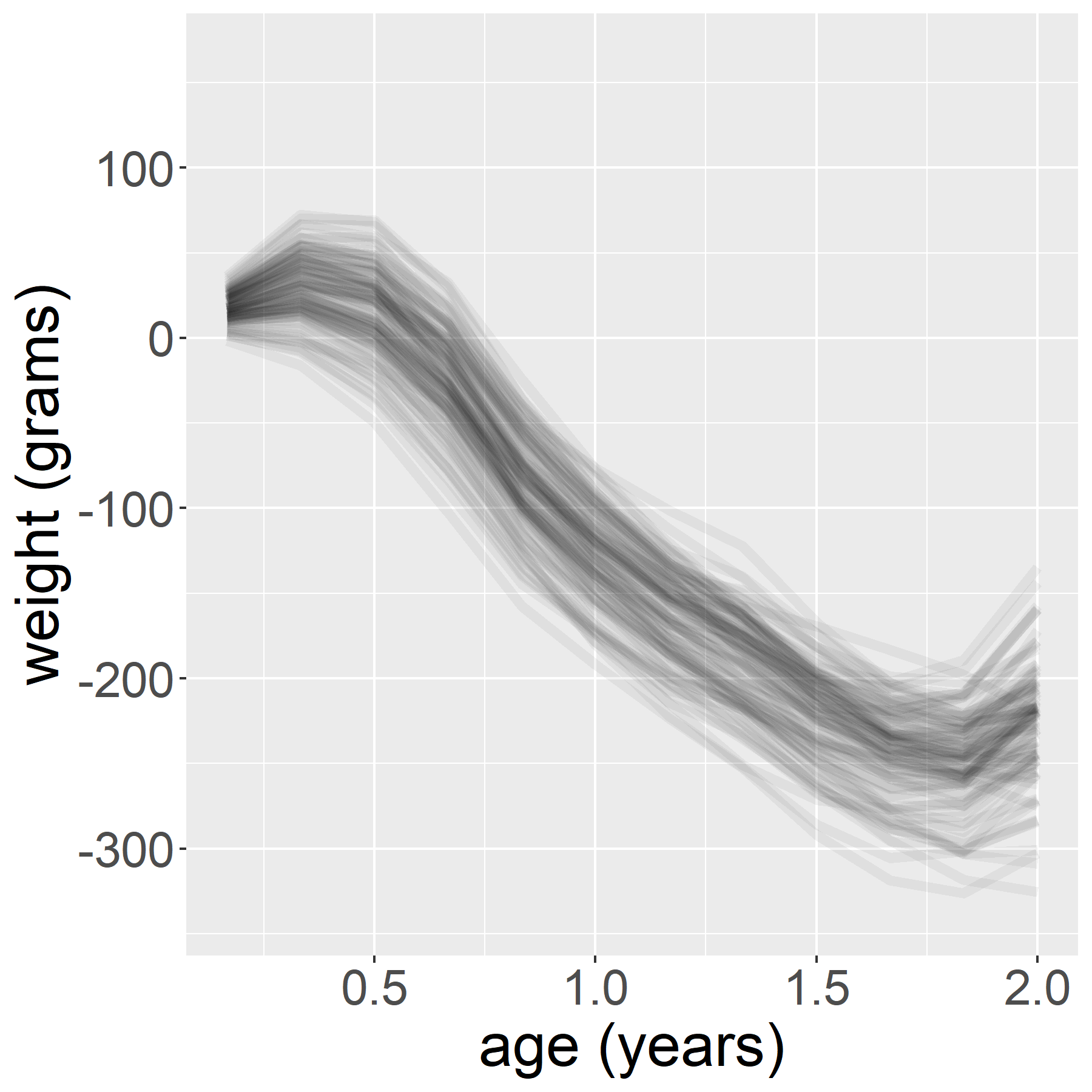}&\includegraphics[width = 1.0 in]{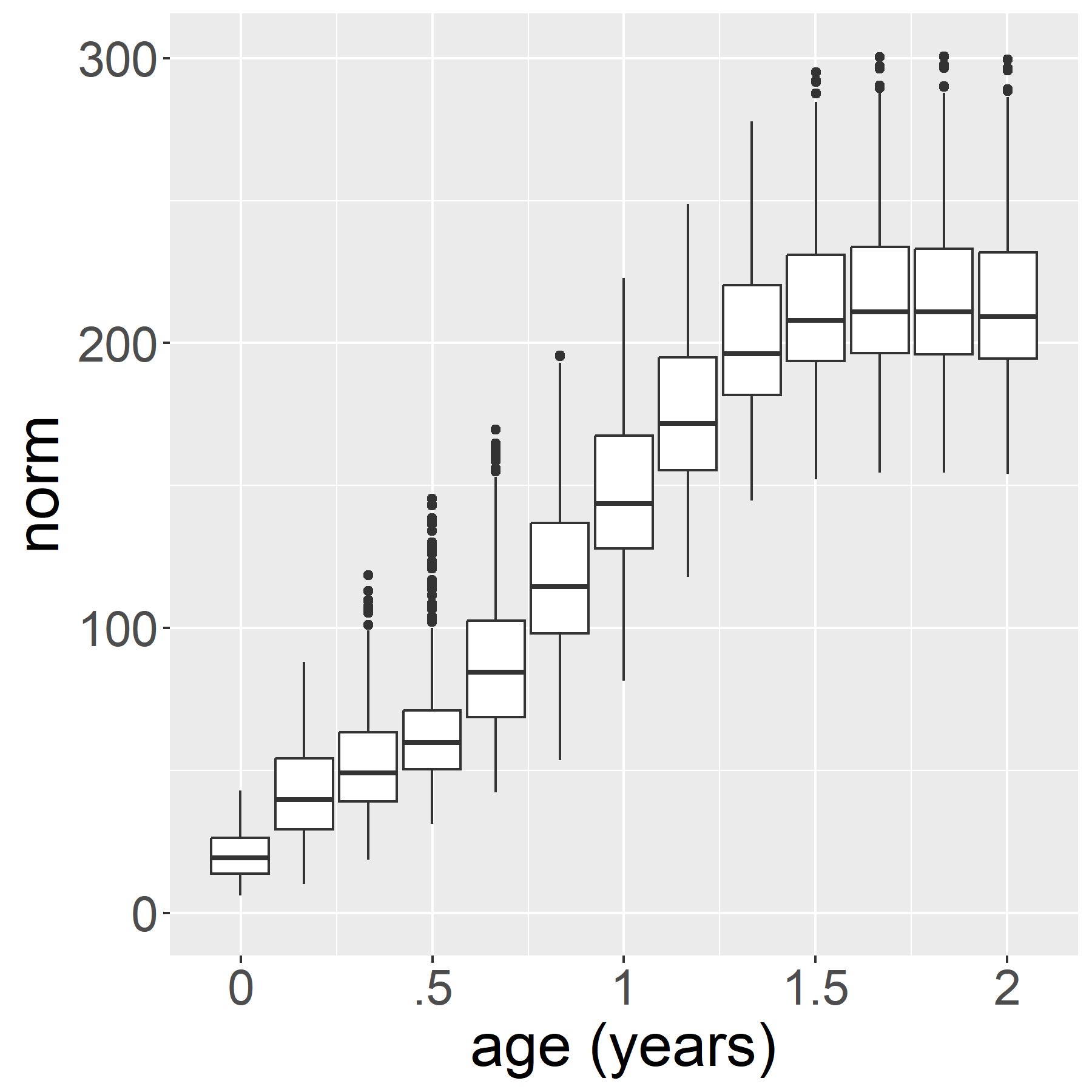} \\
(e) & (f) & (g) & (h) & (i) 
\end{tabular}
\caption{Inference for parameters related to modeling breastfeeding status on weight using the \cite{crainiceanu2010} model: (a)-(c) $\text{logit}^{-1}(\hat{\mu}^{z_1} \pm \text{sd}(\hat{\eta}^z_{1,k},\ldots,\hat{\eta}^z_{n,k})\hat{\lambda}^z_k(t)$ for $k = 1,2,3$. (d) latent factors, $\eta^z_{i,k}\ i =1,\ldots,n,\ k = 1,2,3$. (e)-(h) Posterior samples of $\int_{\mathcal{T}} \beta_1 1_{s\leq s'}ds$ $s' = 0, 0.5,1,2$. (f) Boxplots of posterior samples of $\|\int_{\mathcal{T}} \hat{\beta_1} 1_{s\leq s'}ds\|_2$ for different values of $s'$.} \label{fig:wb_cebu_bf}
    \end{center}
\end{figure}

\begin{figure}[t]
\begin{center}
 \begin{tabular}{cccc}
\includegraphics[width = 1.25 in]{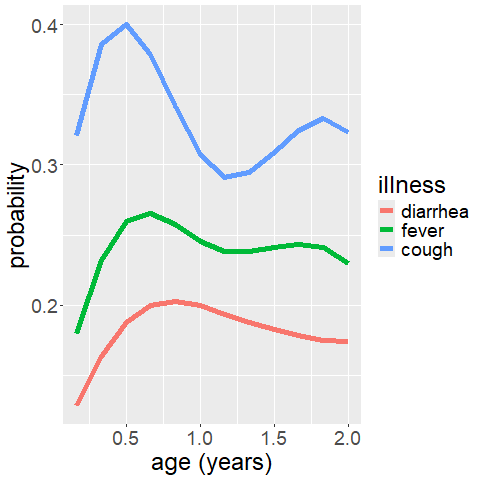}&\includegraphics[width = 1.25 in]{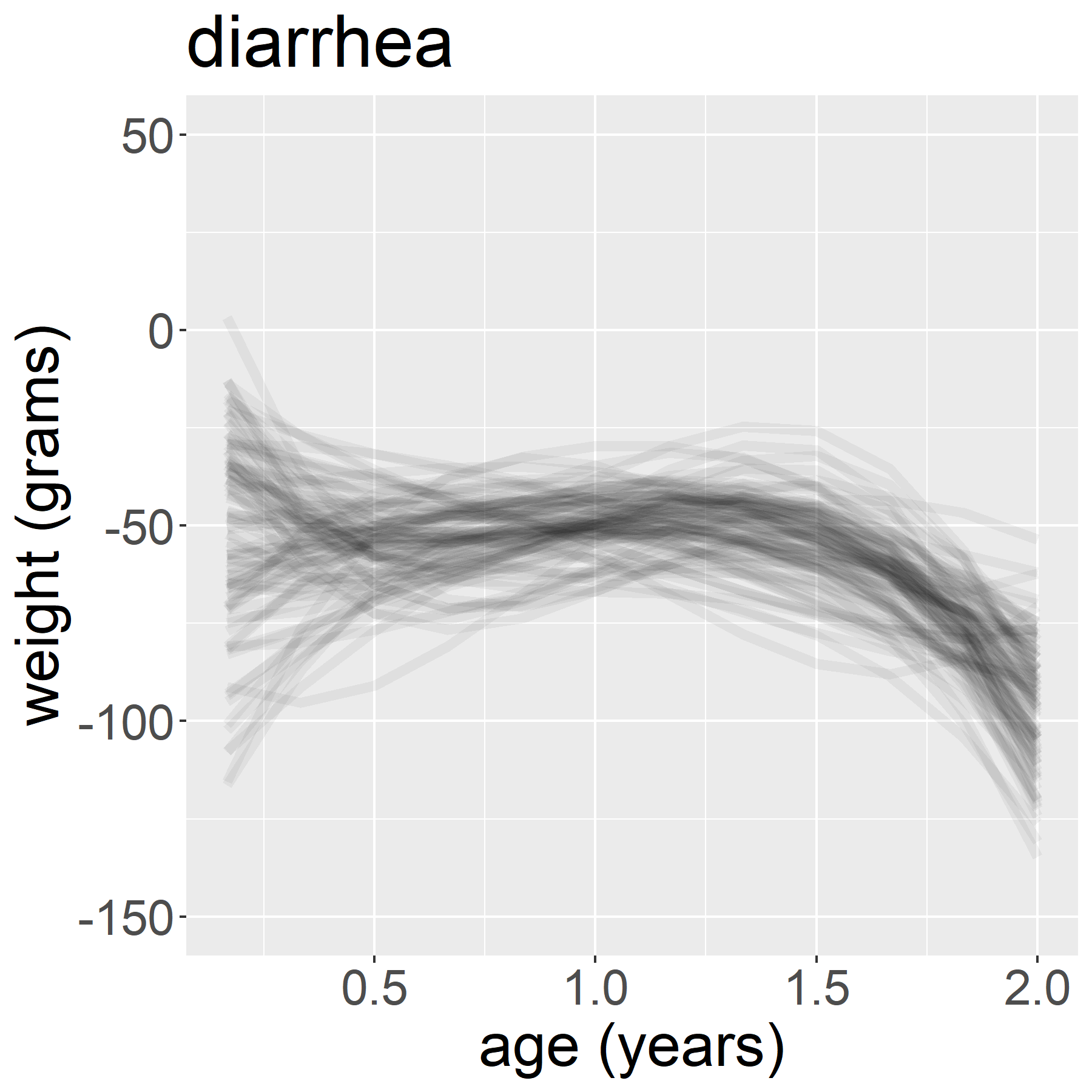} &\includegraphics[width = 1.25 in]{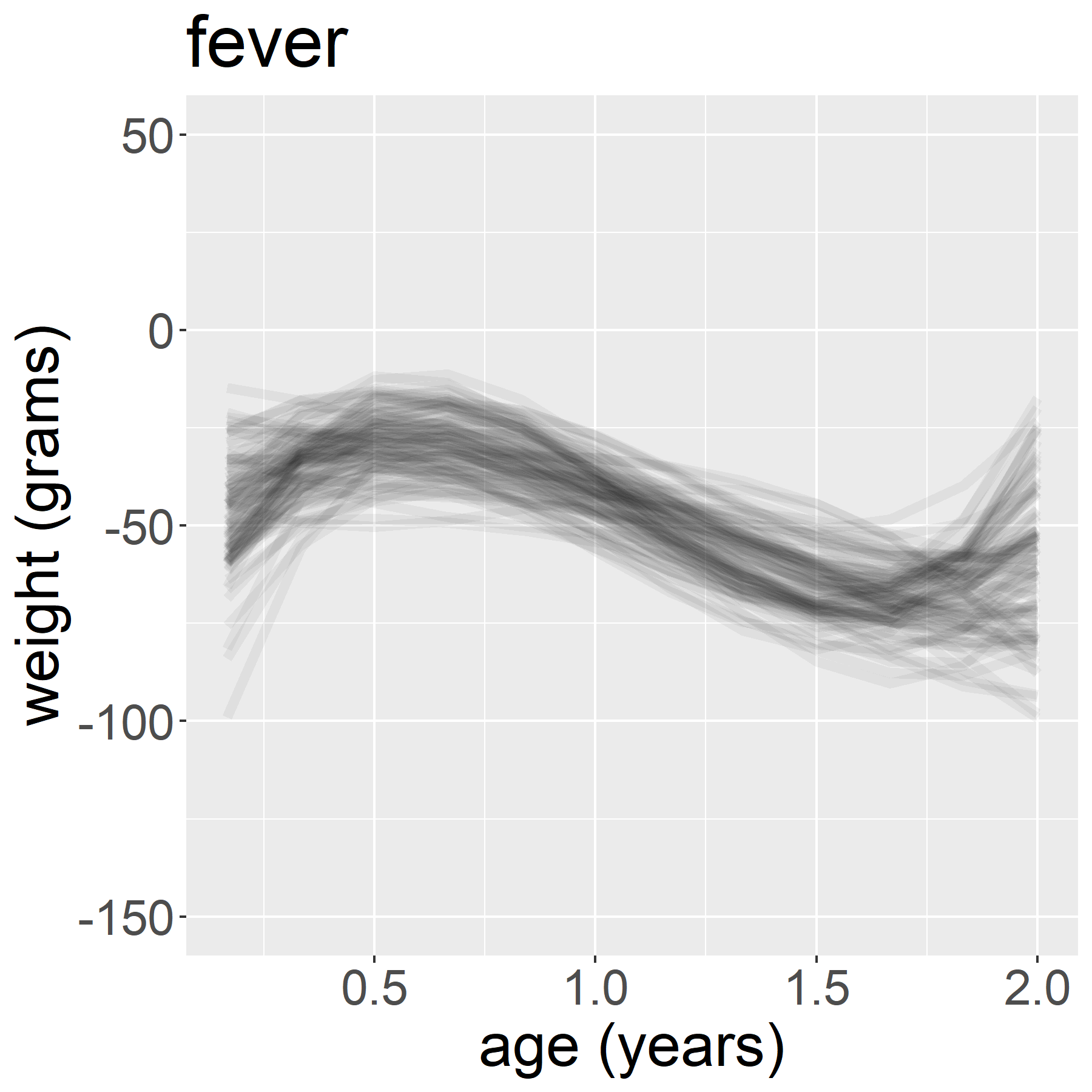} &\includegraphics[width = 1.25 in]{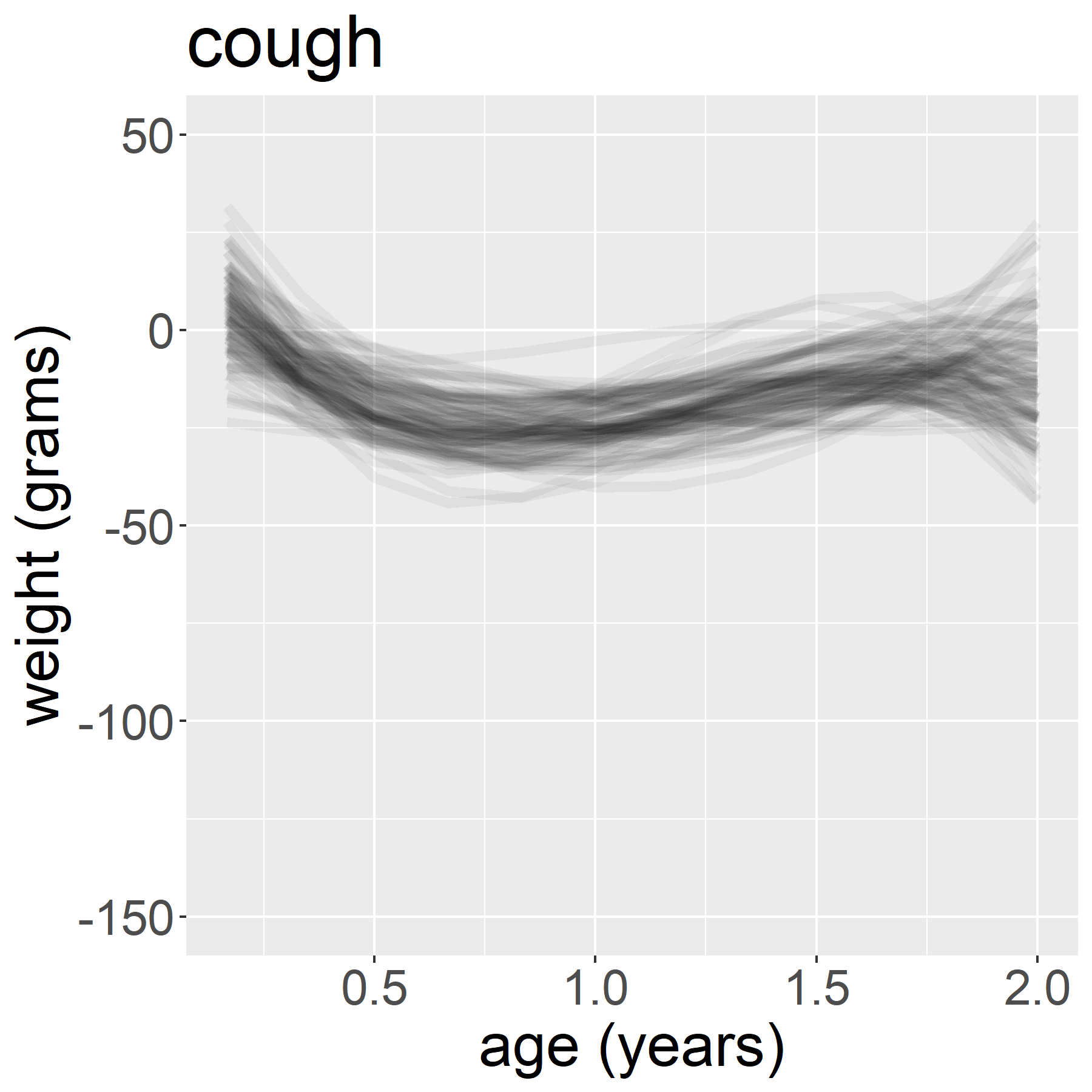} \\
(a) & (b) & (c) & (d) \\
\end{tabular}
\caption{Inferred functions related to modelling the effects of illness on weight using the \cite{crainiceanu2010} model: (a) $\text{logit}^{-1}(\hat\mu^{z_j}(t))$. (b)-(d) Posterior samples of $\beta_j(t)$ for $j = 2,3,4$.}\label{fig:wb_cebu_ill}
    \end{center}
\end{figure}

\begin{figure}[t]
\begin{center}
\begin{tabular}{ccc}
\includegraphics[width = 1.75 in]{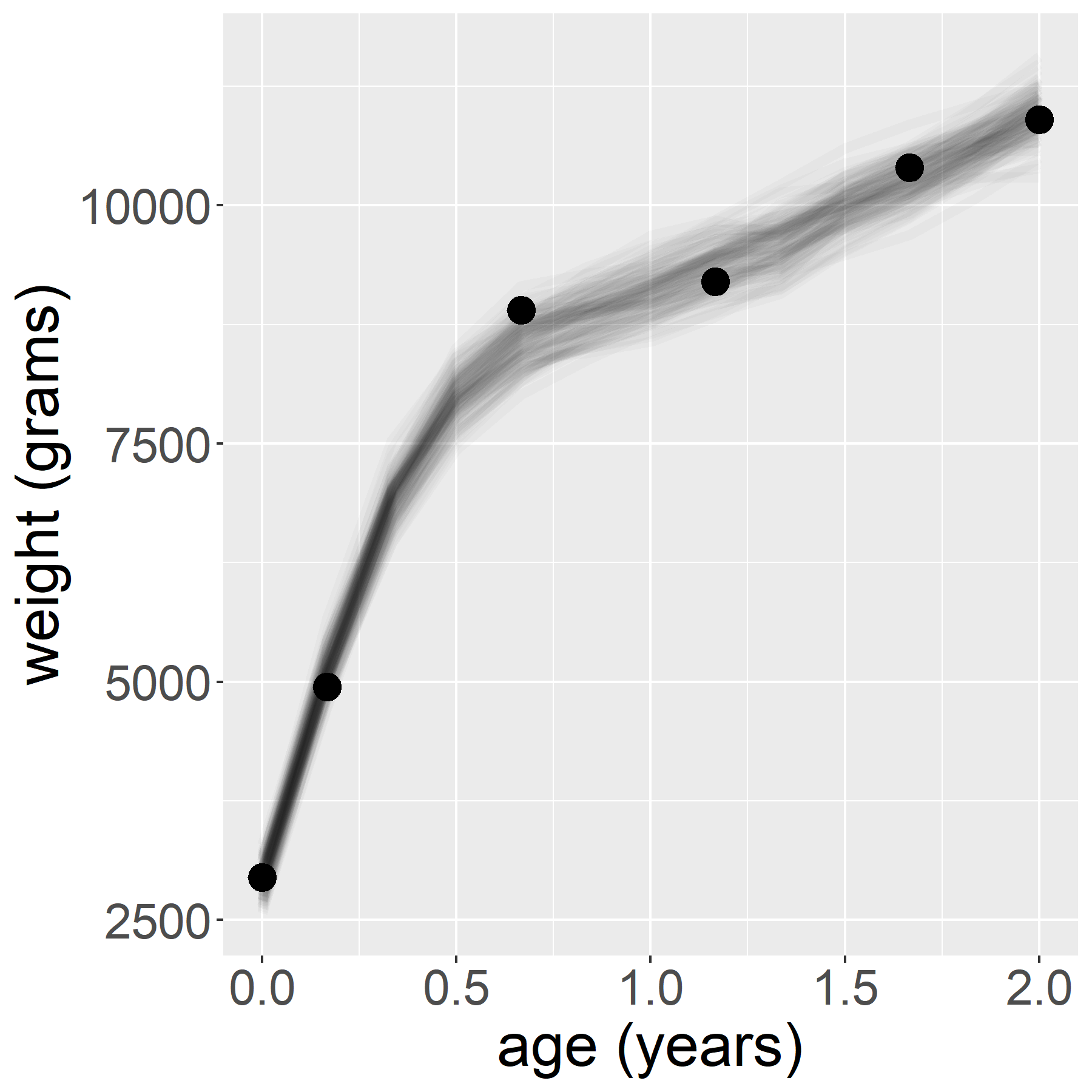}& \includegraphics[width = 1.75 in]{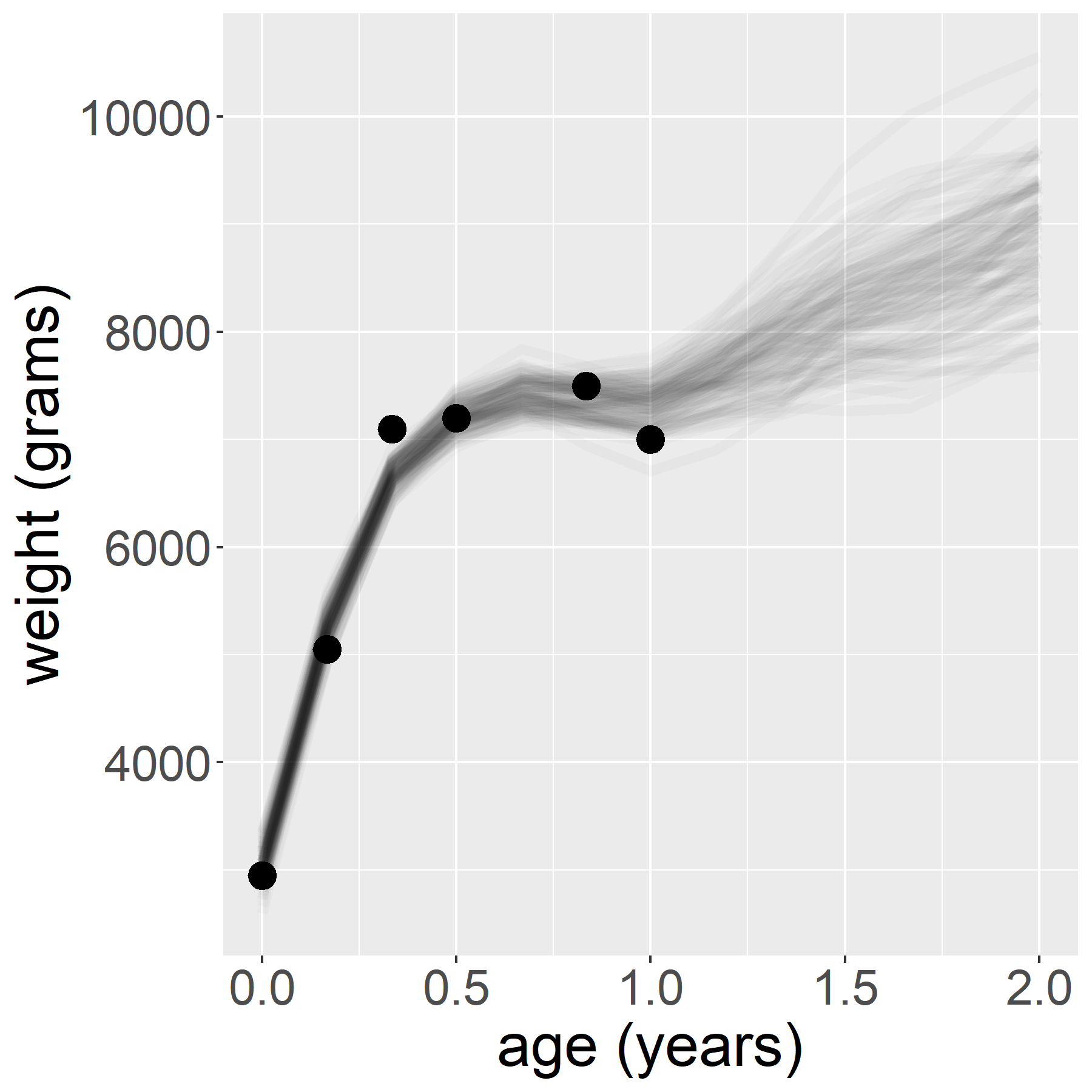}&\includegraphics[width = 1.75 in]{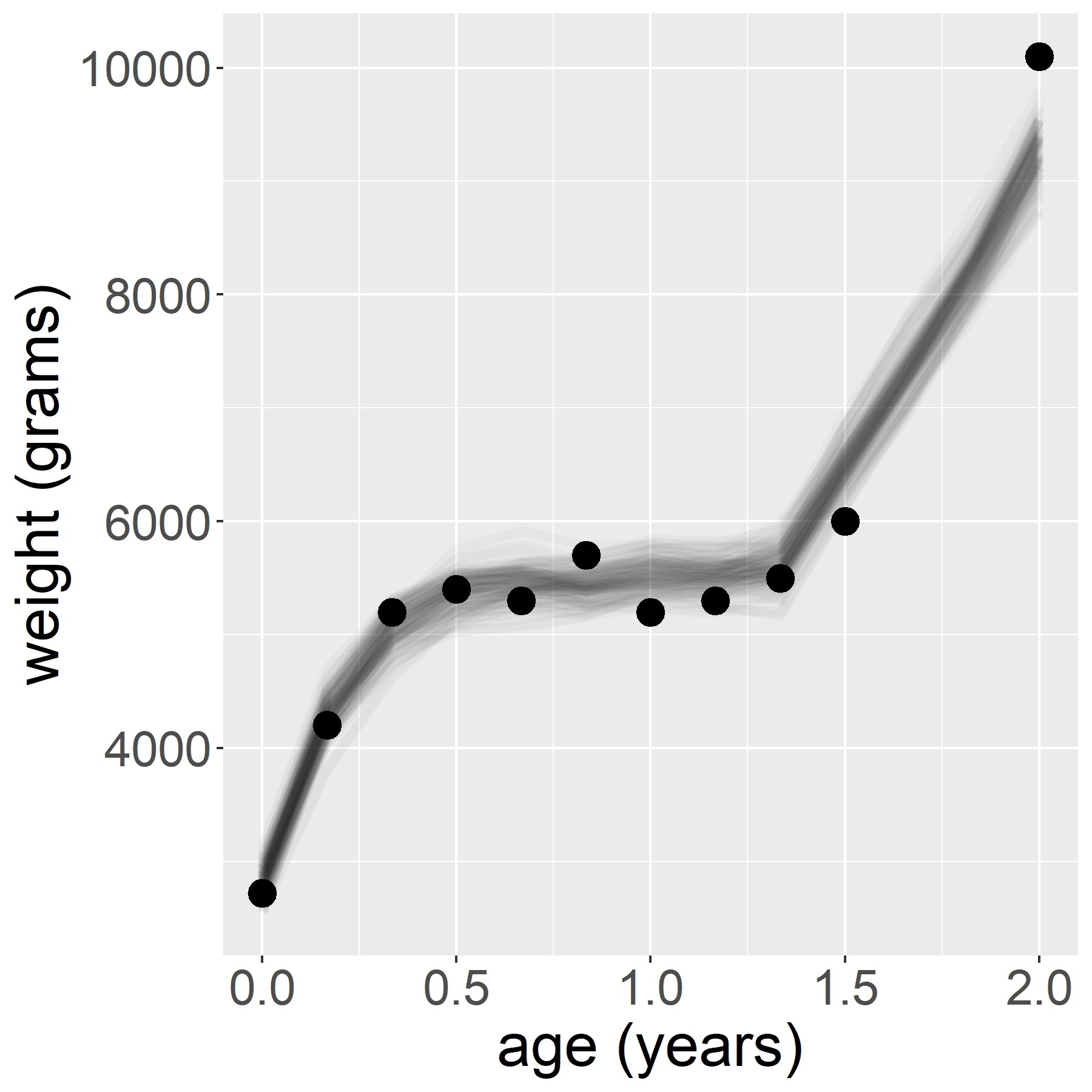} \\
(a) & (b) & (c) \\
\end{tabular}
\caption{Inferred weight trajectories for different subjects using the \cite{crainiceanu2010} model: (a)-(c) Posterior samples of $\hat{\mu}(\protect\vv{t}_i) + (\hat{\lambda}_1(\protect\vv{t}_i),\ldots,\hat{\lambda}_K(\protect\vv{t}_i))^\top\hat{\eta}_i + O_i\sum_{j = 1}^4\int_{T}\beta_j(s,\protect\vv{t})z_{i,j}(s)ds$, overlaid on $y_i(\protect\vv{t}_i)$, $i = 421,\ 1626,\ 2205$.}\label{fig:wb_cebu_fit}
    \end{center}
\end{figure}

In this section, we analyze the Cebu dataset with methods discussed in \cite{crainiceanu2010}, which we compare to the results from the $\small{\mbox{NeMO}}$ model. The observation model is given by Equations \eqref{eq:obs_model1}\&\eqref{eq:obs_model2}. In line with their empirical Bayes setup, $\mu(\vv{t})$ is estimated using the cross sectional mean of the data, and $\lambda_1,\ldots,\lambda_K$ are estimated from the mean-centered data using a penalized spline based approach. Additionally, we complete longitudinal covariates using model based imputation, with the generative model
\begin{align}
    P(z_{i,1}(t) = 1) & =  \text{logit}^{-1}(\mu^{z_1} + (\lambda^z_1(t),\ldots\lambda^z_{K^z}(t))^\top\eta^z_i), \nonumber \\
    P(z_{i,j}(t) = 1) & = \text{logit}^{-1}(\mu^{z_j}) ,\ j= 2,3,4, \nonumber 
\end{align}
where model components are estimated using the \texttt{slfpca} \texttt{R} package. 

The mean of the response and laodings are fixed at point estimates, $\hat\mu(\vv{t}),\hat{\lambda}_1(\vv{t}),\ldots,\hat{\lambda}_K(\vv{t})$ in subsequent modeling. Remaining model parameters are assigned conjugate prior distributions and posterior inference is carried out through Gibbs sampling. As suggested by \cite{crainiceanu2010}, \textit{a priori} $\Theta_{k,q}\sim N(0,\hat\phi_k)$ and $\xi\sim N(0,\hat\phi_k)$, where the prior variances are eigenvalues of the estimated covariance matrix that determine the estimated loadings. For the regression coefficients between the scalar covariates and the latent factors, we assume independent, identically distributed standard normal priors. We use the same bases discussed in the main paper to simplify the integral equations in the functional linear model components that govern the relationships between the longitudinal covariates and weight. Basis coefficients are given normal conjugate priors, consistent with the $\small{\mbox{NeMO}}$ model. 

Figure \ref{fig:wb_cebu_fpca} panels (a) and (b) show the estimated mean process and loadings. These estimated functions are similar in shape and magnitude to the loadings presented in the main paper. Panels (d)-(f) of this figure show the inferred effects of the scalar covariates by displaying posterior samples of $(\hat{\lambda_1}(t),\ldots,\hat{\lambda_K}(t))\hat{\Theta}_q,\ q = 1\ldots,4$. For $q = 1,2,3$ the resulting functions match the results presented in the main paper. The function for $q = 4$ in the main paper is inferred to be closer to zero with a considerable amount of uncertainty. The discrepancy between these two results could be from conditioning on estimated loadings in the empirical Bayes setting. 

Figure \ref{fig:wb_cebu_bf} panels (a)-(c) visualize estimated generalized functional factor loadings used to describe the structure variability in breastfeeding status estimated from the \texttt{slfpca} \texttt{R} package. These loadings seem to capture if a child is breastfed, and when they stopped breastfeeding. Panels (d) of this figure shows estimated latent factors, $\hat{\eta}^z_i,\ i = 1,\ldots,n$. As in Section \ref{sec:add_post_figs}, the latent factors have a non-Gaussian structure. Panel (e)-(h) display $\|\int_{\mathcal{T}} \hat{\beta_1} 1_{s\leq s'}ds\|_2$, for different values of $s'$. The general shape of inferred functions is consistent with those presented in the main paper, however, the functions in panels (e)-(h) are slightly smaller in magnitude. As shown in panel (f), there is a diminishing effect of prolonged breastfeeding. 

Figure \ref{fig:wb_cebu_ill} panel (a) displays the estimated average probability of a child experiencing the illnesses by age. As in the model presented in the main paper, all of these estimated functions have a prominent peak at early ages. The estimated regression functions, $\beta_j(t),\ j = 2,3,4$, shown in panels (b)-(d), match their corresponding interpretations in the main paper. It appears that experiencing diarrhea and fever have a more negative effect on weight than experiencing cough. 

The inferences made in this section are generally consistent with the inferences presented in the main paper. We use the widely applicable information criterion to compare the predictive ability of the model presented in this section and the $\small{\mbox{NeMO}}$ model \citep{vehtari2017}. The widely applicable information criterion values are $WAIC_{NeMO} = 69709.10$ for the $\small{\mbox{NeMO}}$ model and $WAIC_{CG} = 69905.07$ for the model of \cite{crainiceanu2010}. The difference in WAICs is $WAIC_{NeMO} - WAIC_{CG} = -195.97$ with standard error $114.04$. Using this criteria, the $\small{\mbox{NeMO}}$ model is preferred in terms of predictive accuracy. 

\subsection{Model Fit Assessment}

\begin{figure}[t!]
\begin{center}
 \begin{tabular}{c}
\includegraphics[width = 3in]{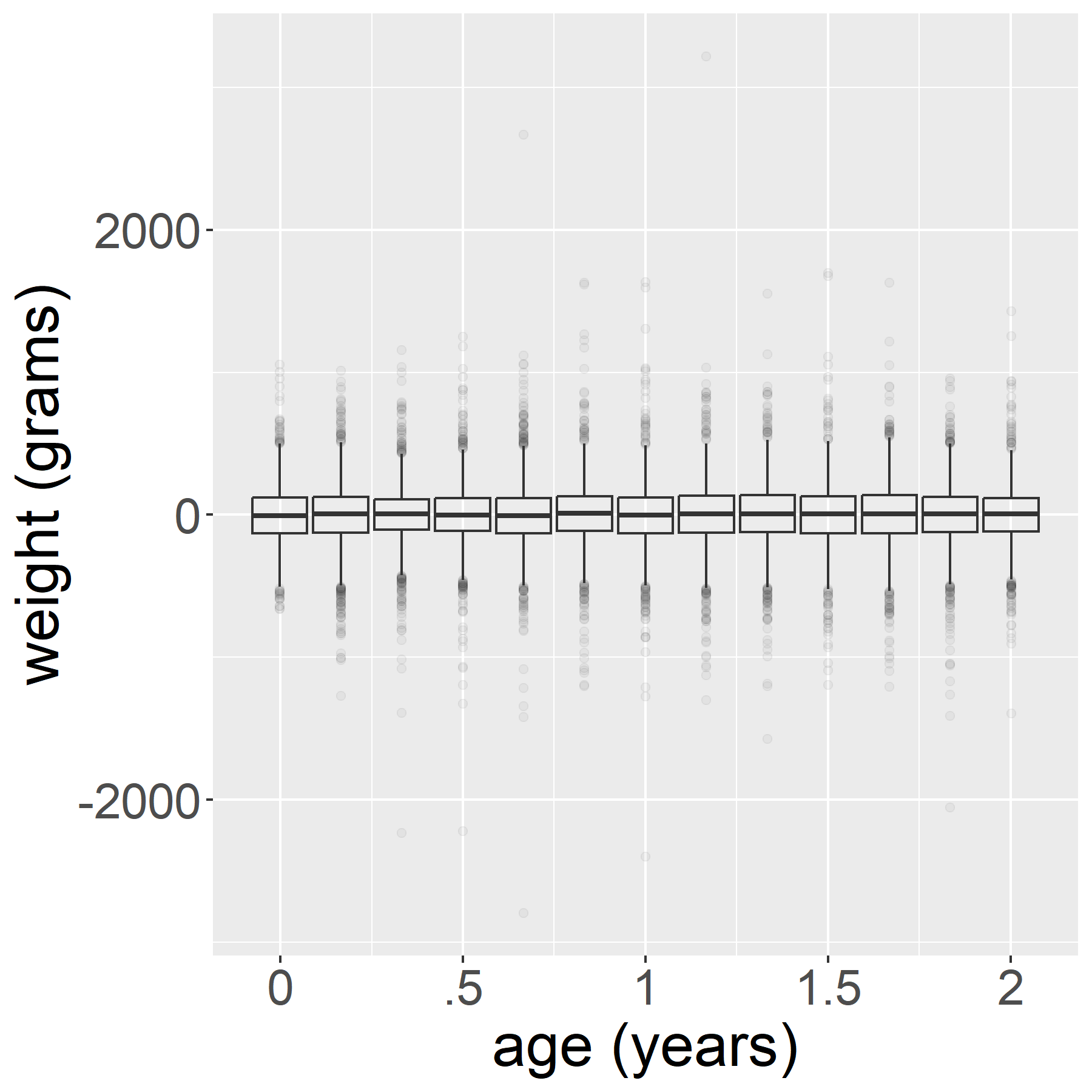}\\
\end{tabular}
\caption{Boxplots of residuals by age.}\label{fig:residuals}
    \end{center}
\end{figure}

\begin{figure}[t!]
\begin{center}
\begin{tabular}{ccc}
\includegraphics[width = 1.5 in]{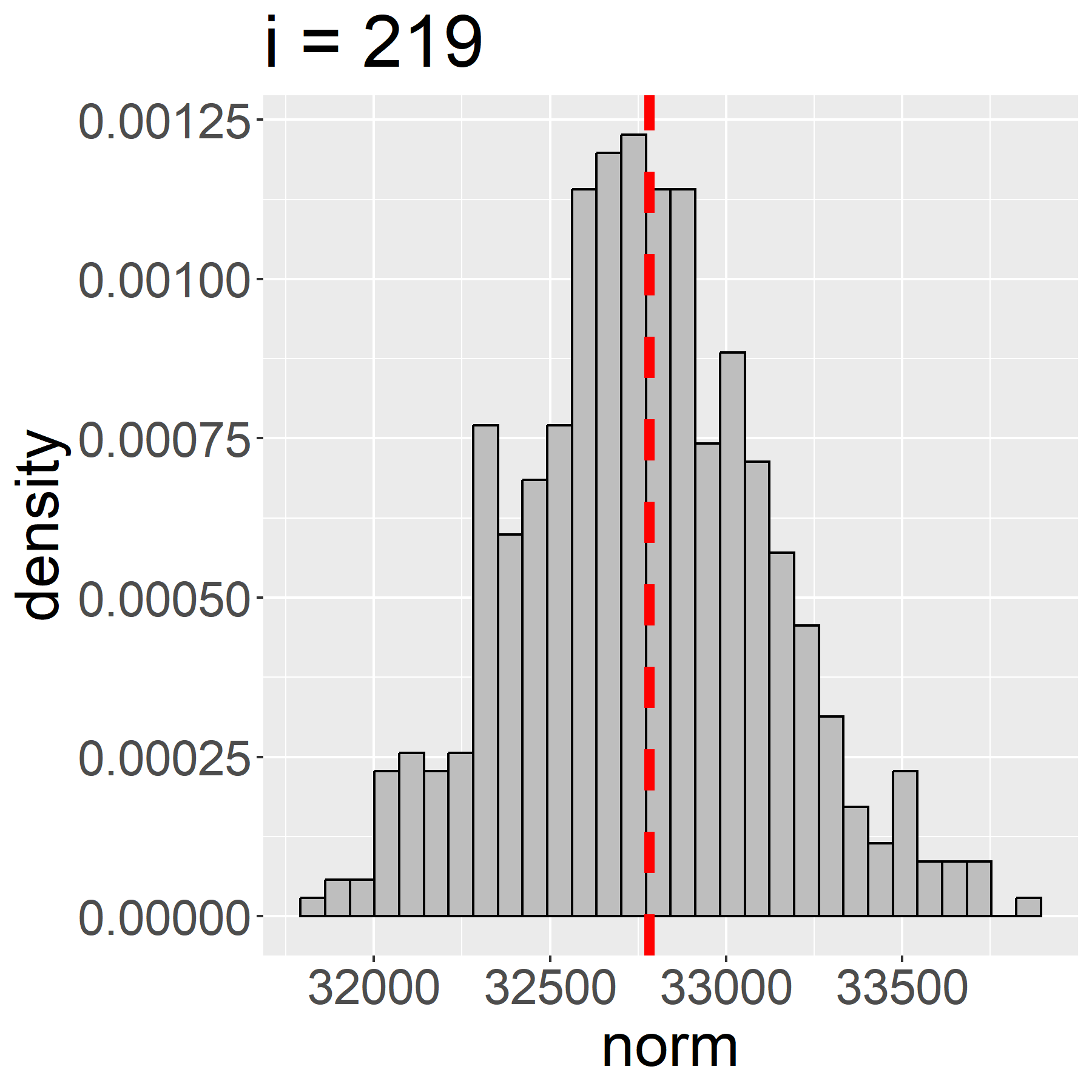}&\includegraphics[width = 1.5 in]{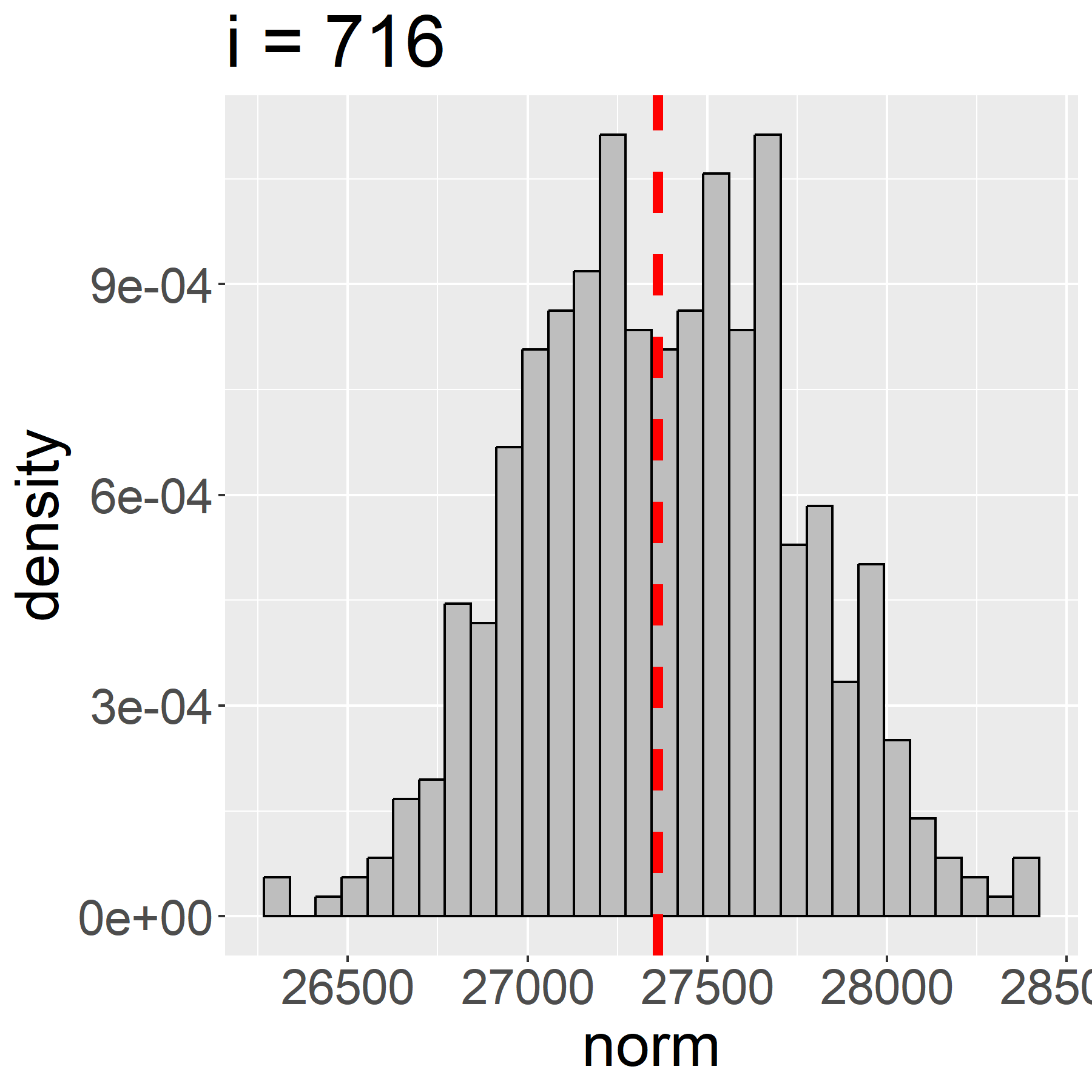}&\includegraphics[width = 1.5 in]{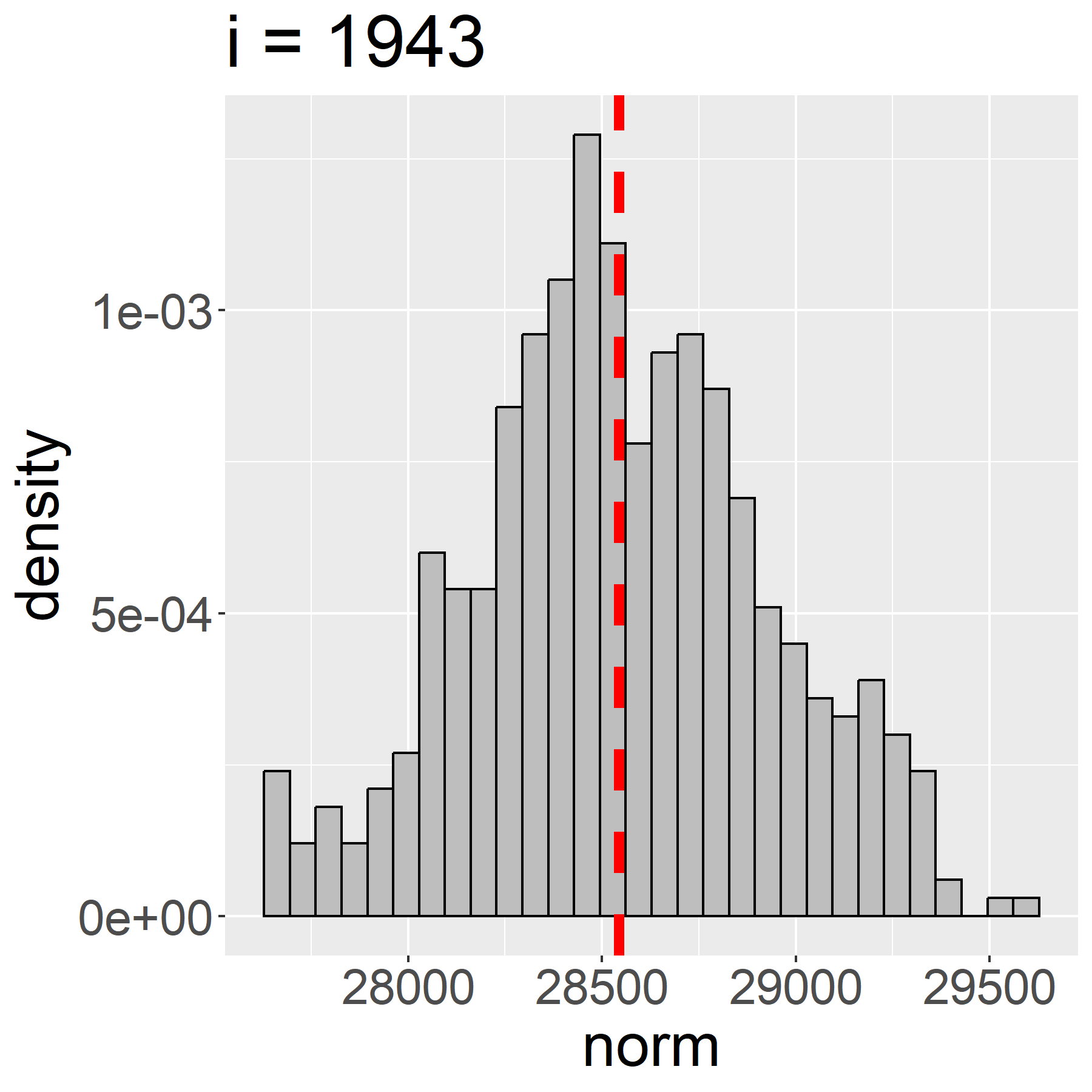} \\
(a) & (b) & (c) \\
\includegraphics[width = 1.5 in]{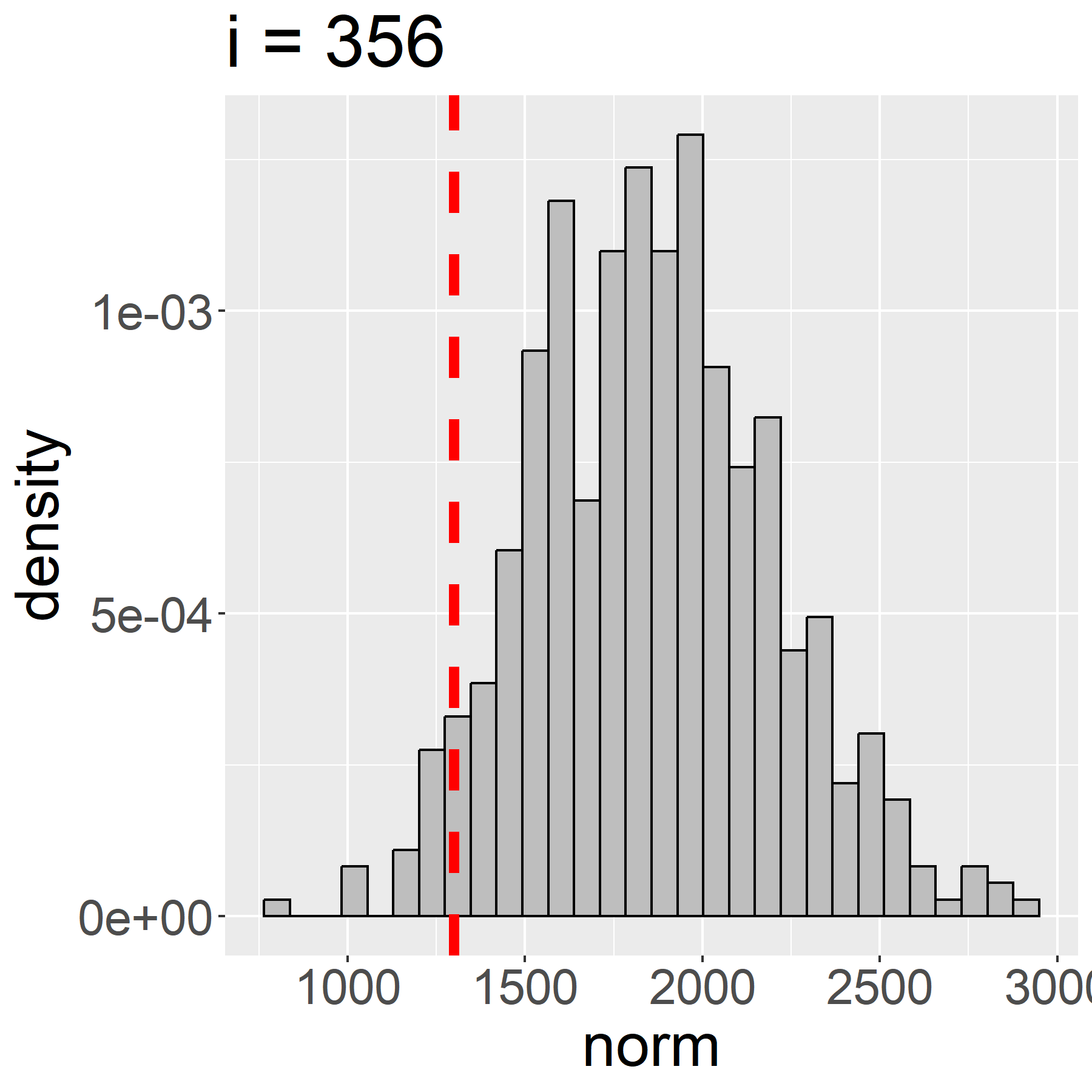} & \includegraphics[width = 1.5 in]{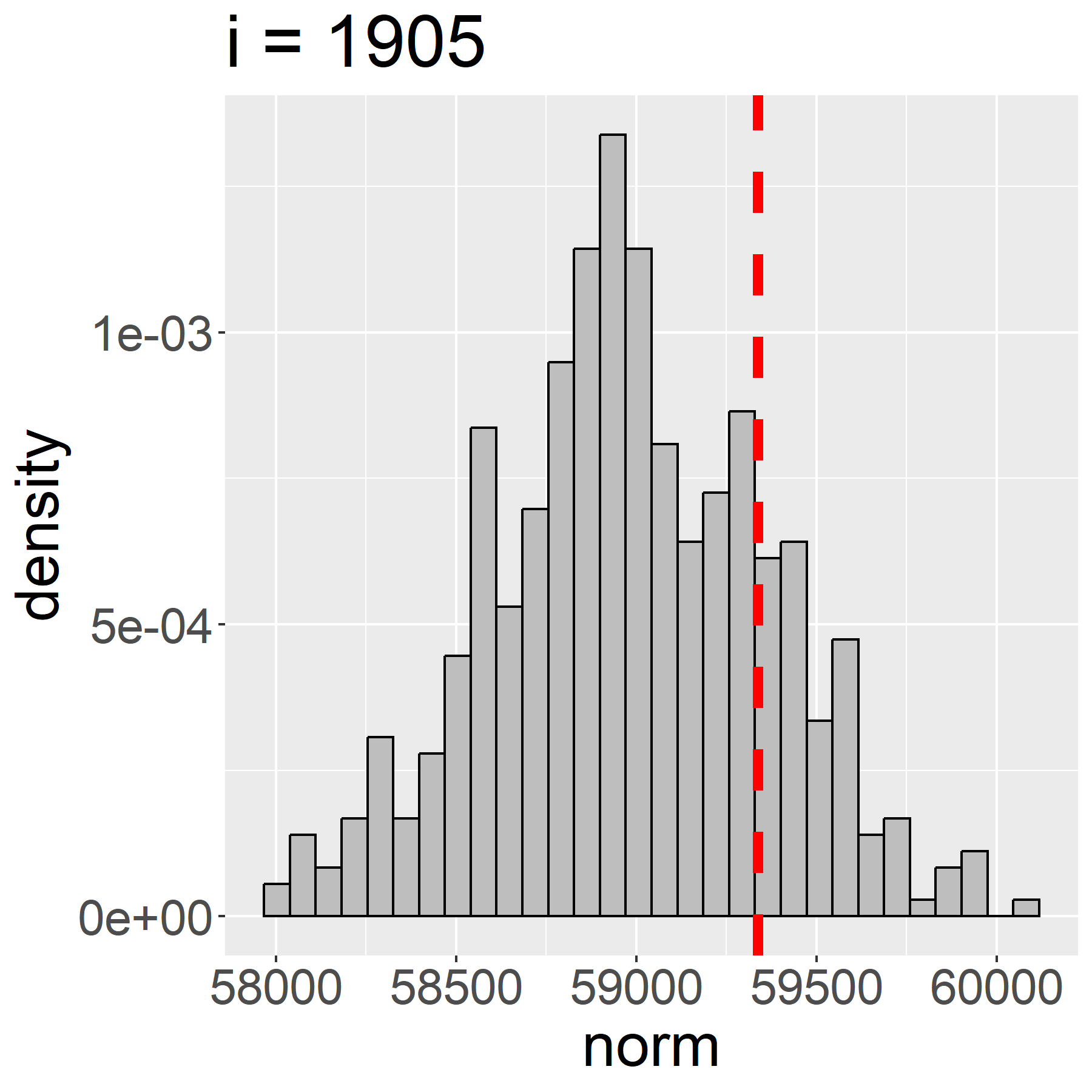} & \includegraphics[width = 1.5 in]{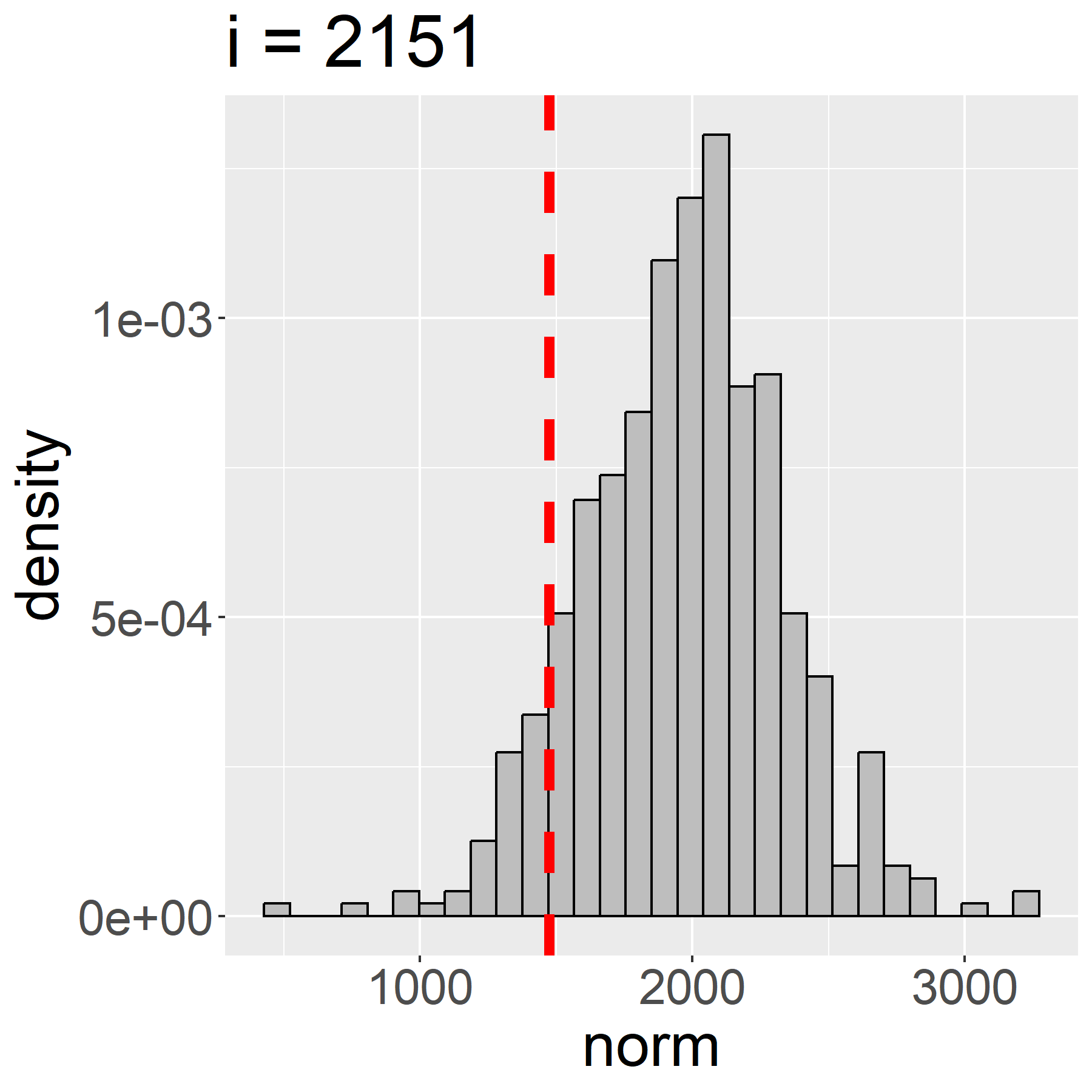} \\
(d) & (e) & (f) 
\end{tabular}
 \begin{tabular}{c}
\includegraphics[width = 2in]{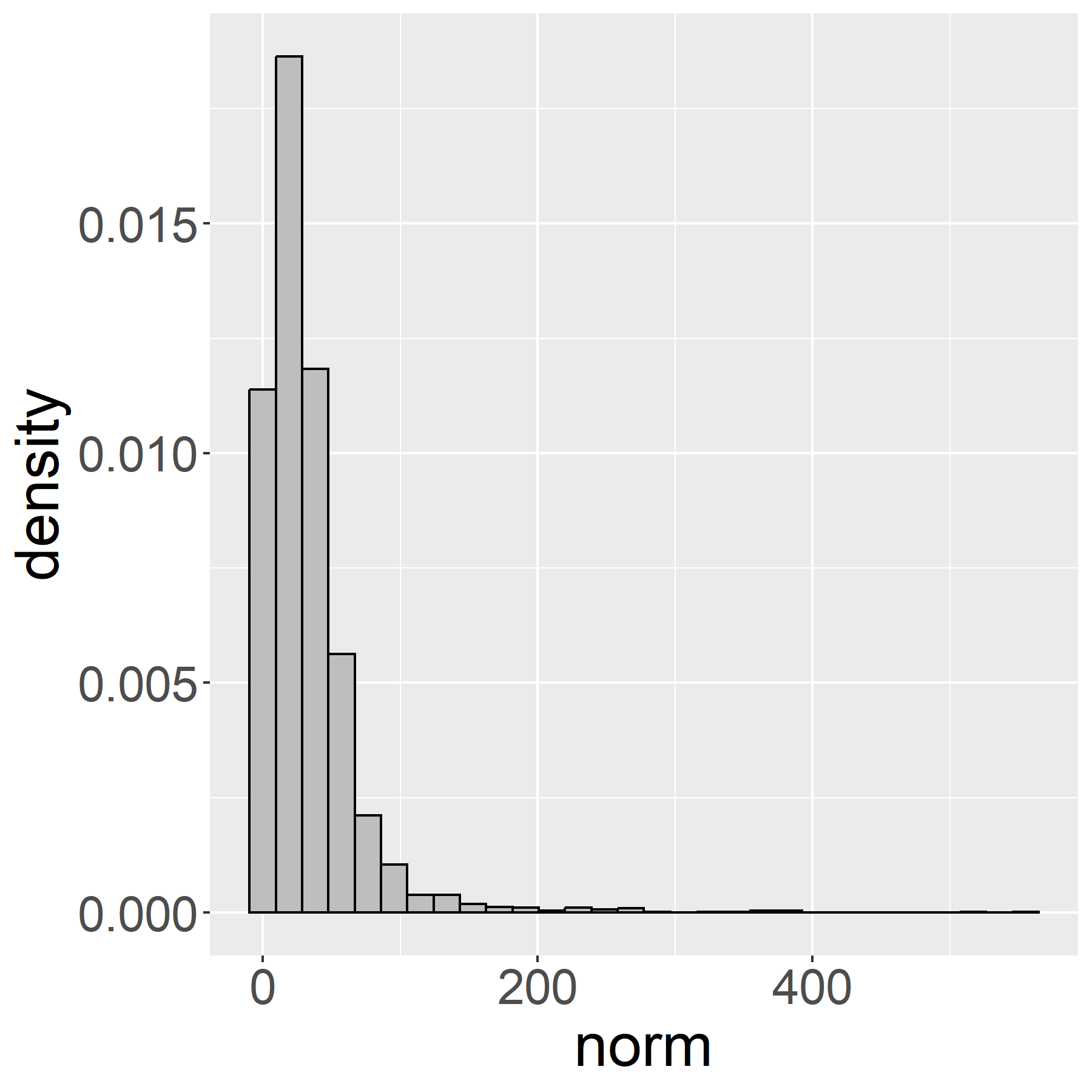}\\
(g)
\end{tabular}
\caption{$T_{mag,y_i}$ based on posterior predictive samples (histogram) and observation (red lines) for random subjects (a)-(c) and subjects whose predicted values were furthest from the observed values (d)-(f). (g) the absolute difference between the mean of $T_{mag,y_i}$ based on posterior predictive samples and the observed value for all subjects. }\label{fig:norm_pred_check}
    \end{center}
\end{figure}

\begin{figure}[t!]
\begin{center}
\begin{tabular}{ccc}
\includegraphics[width = 1.25 in]{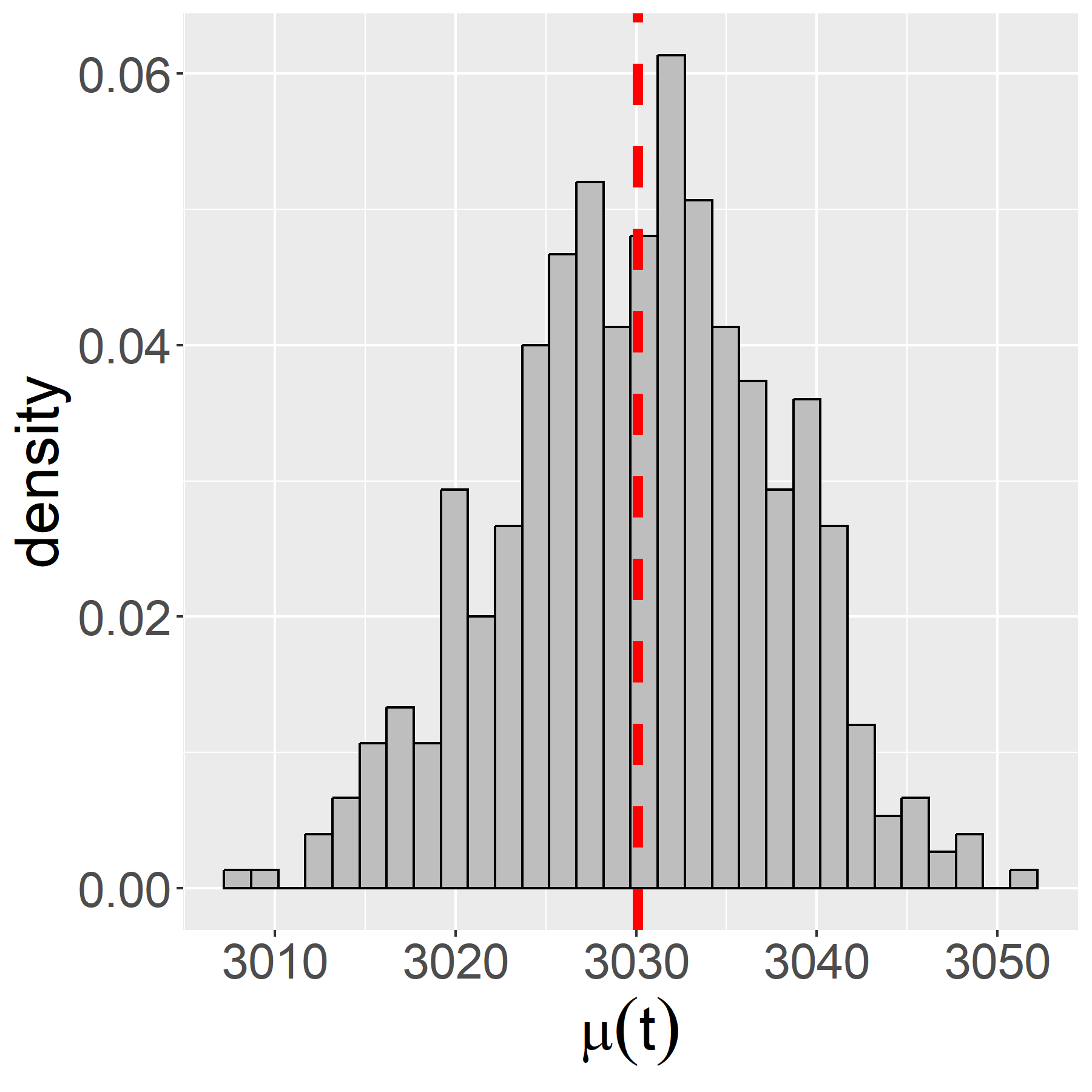}&\includegraphics[width = 1.25 in]{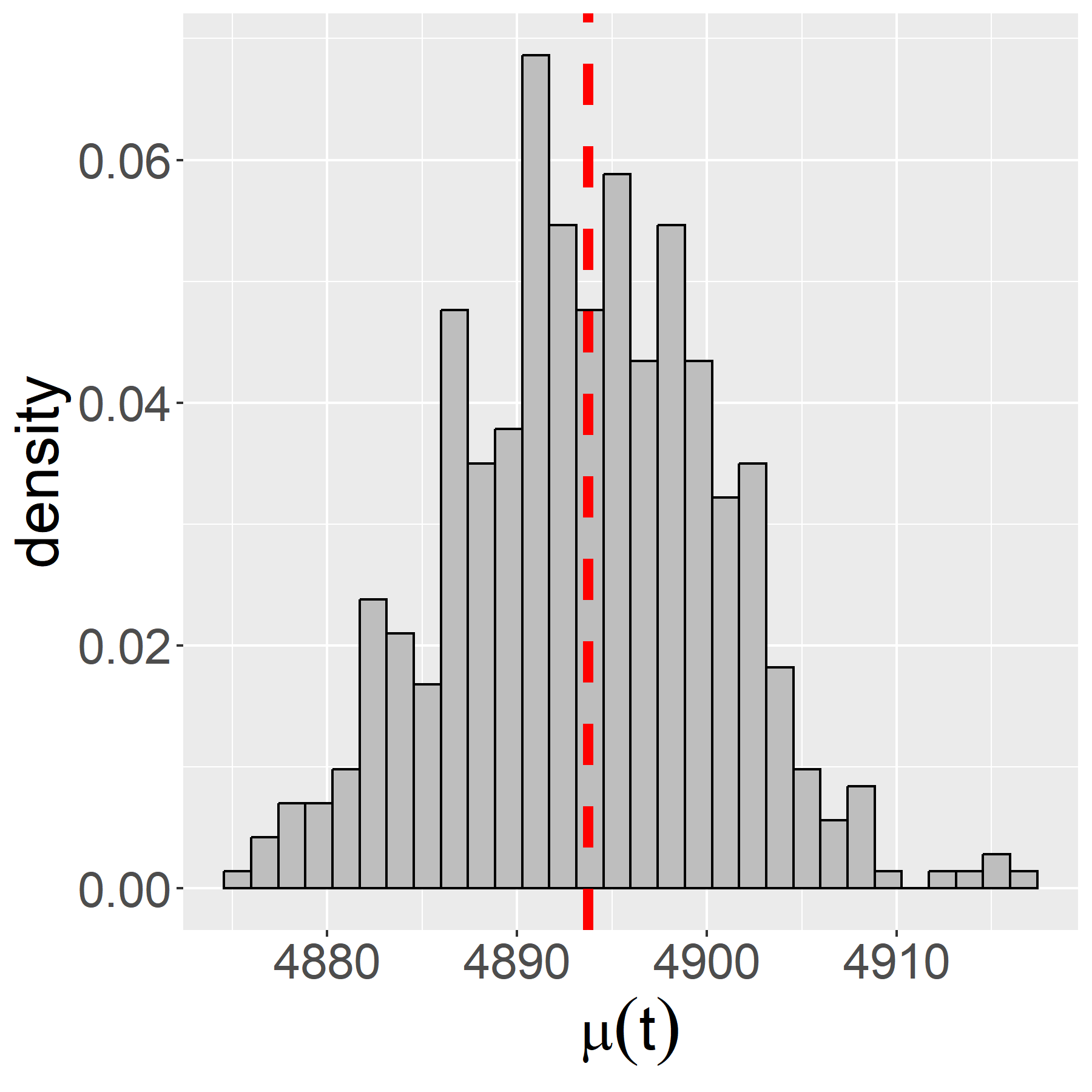}&\includegraphics[width = 1.25 in]{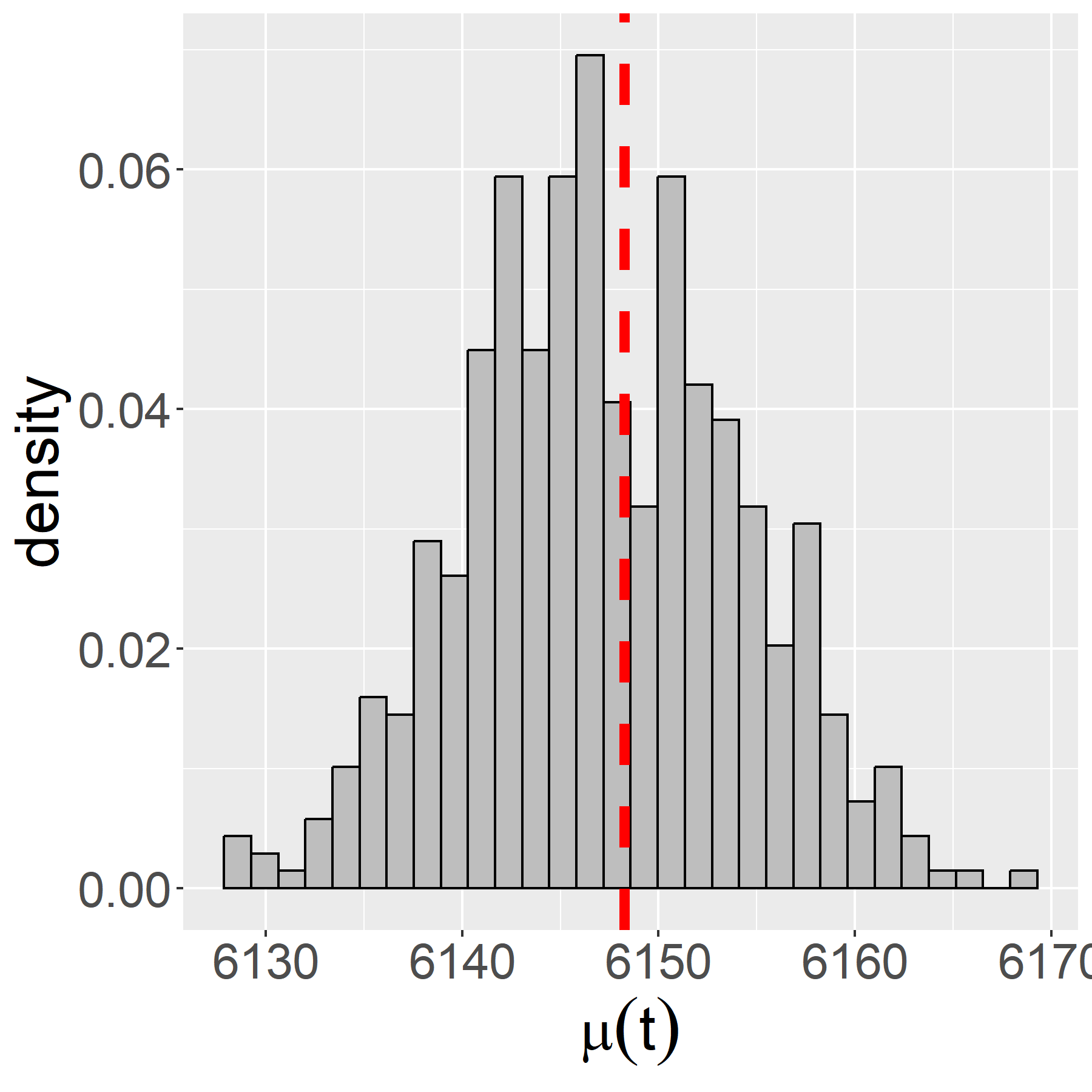} \\
(a) & (b) & (c) \\
\includegraphics[width = 1.25 in]{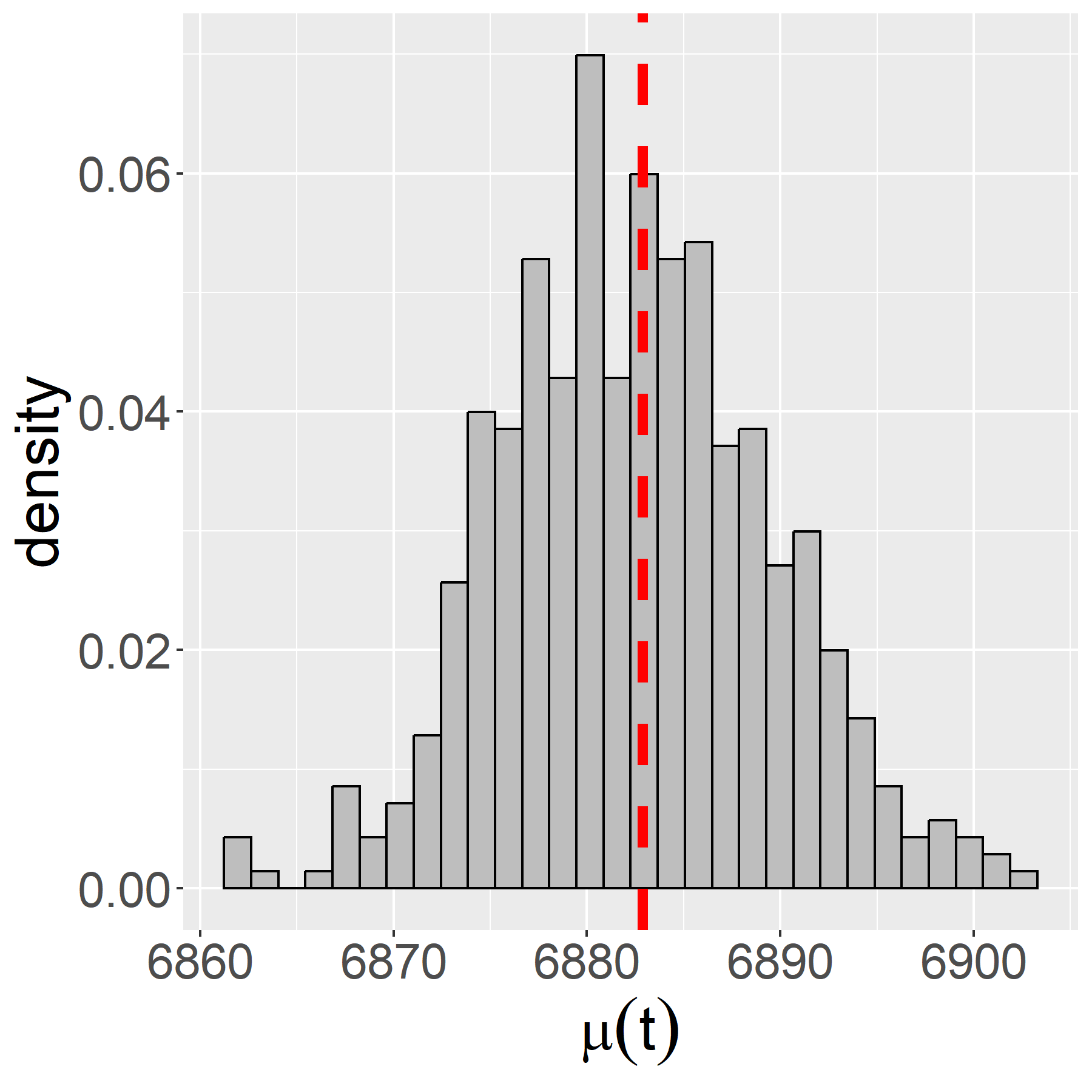}&\includegraphics[width = 1.25 in]{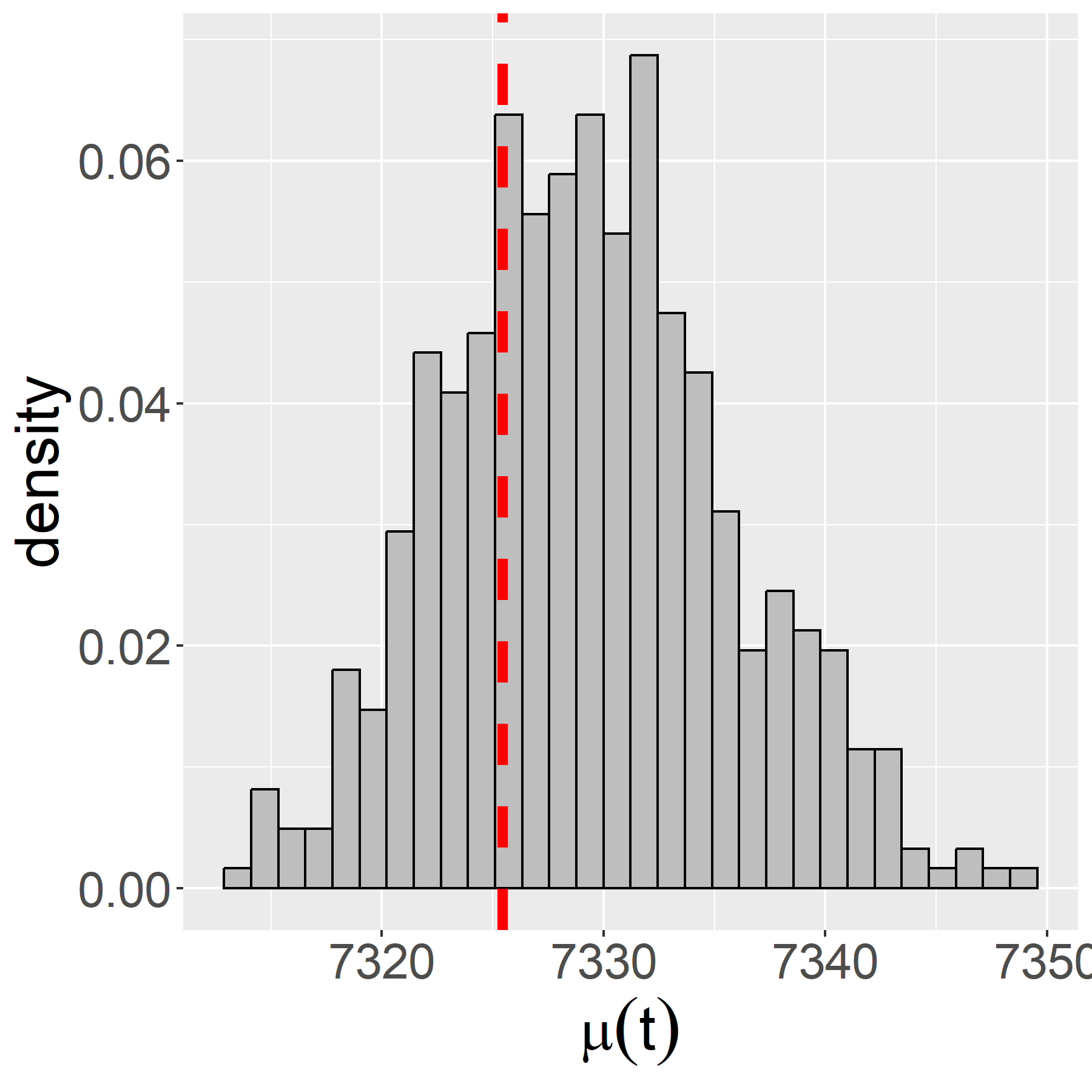}&\includegraphics[width = 1.25 in]{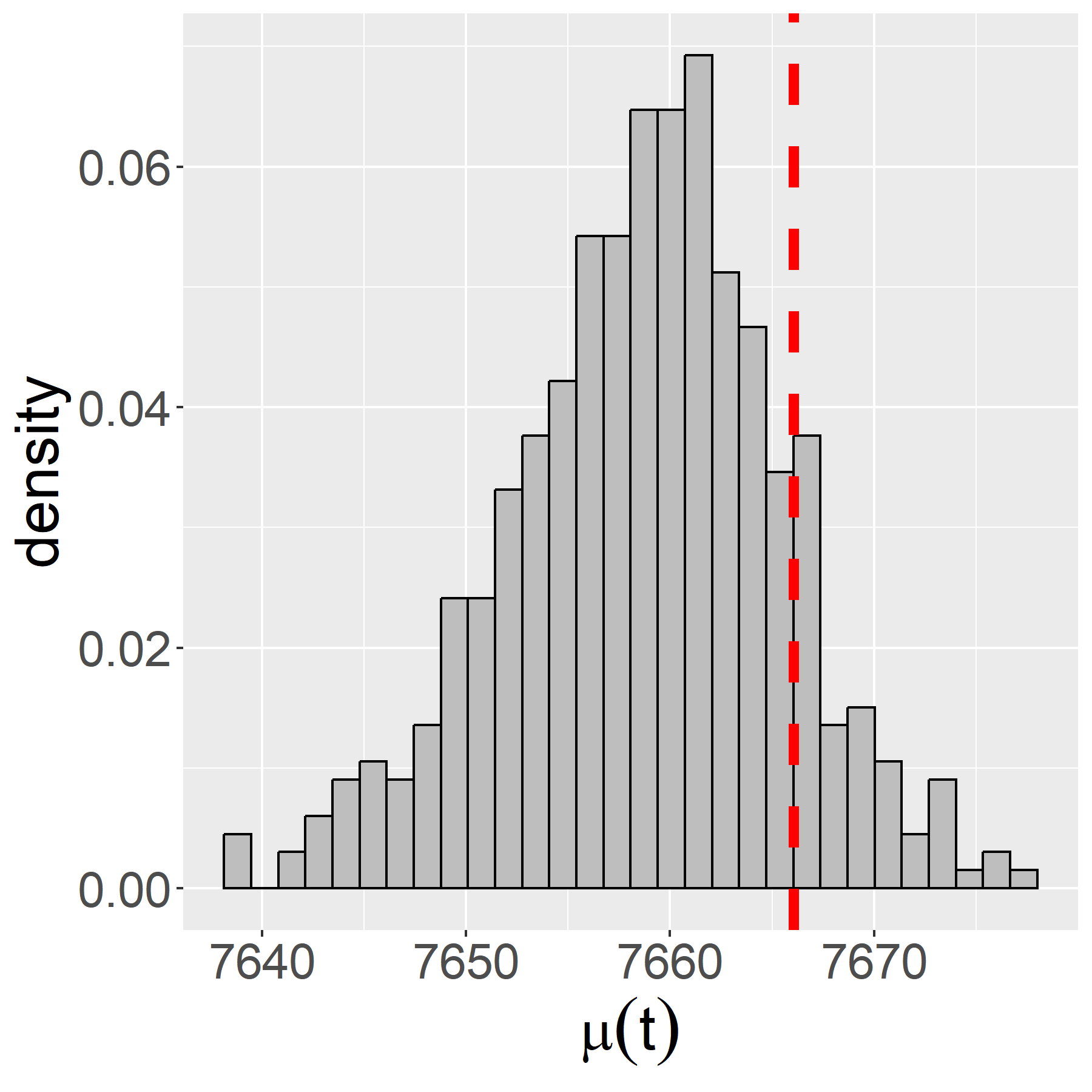} \\
(d) & (e) & (f) \\
\includegraphics[width = 1.25 in]{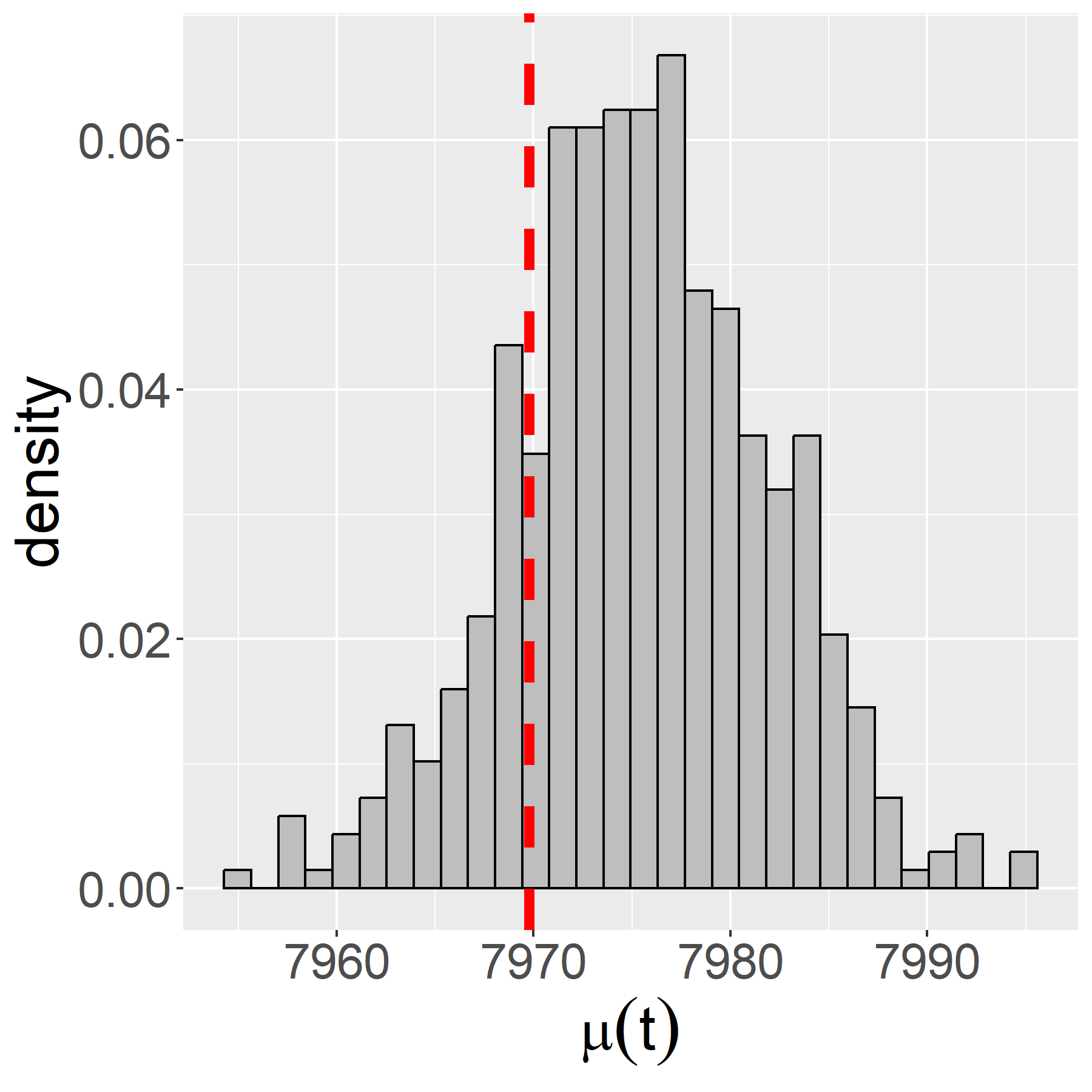}&\includegraphics[width = 1.25 in]{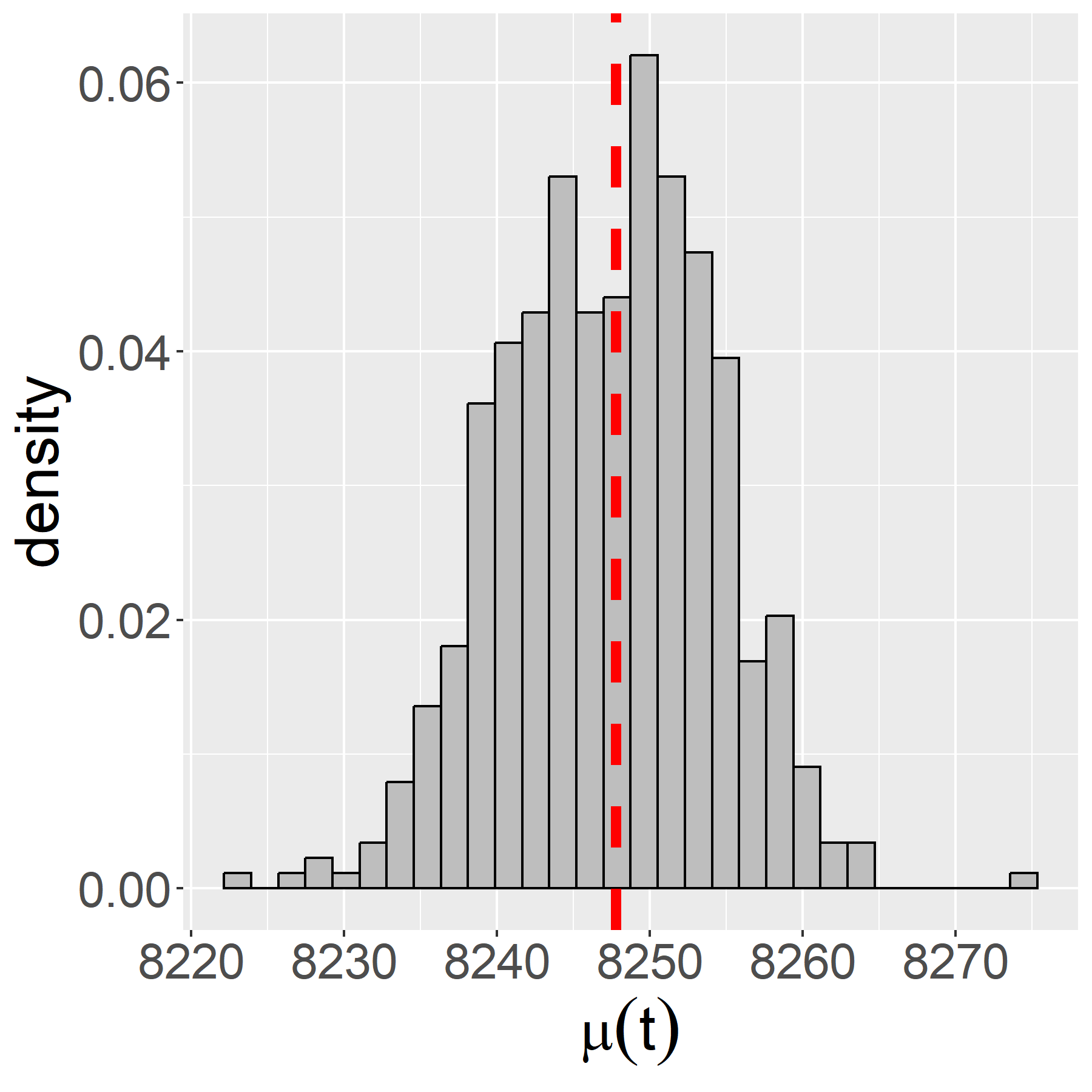}&\includegraphics[width = 1.25 in]{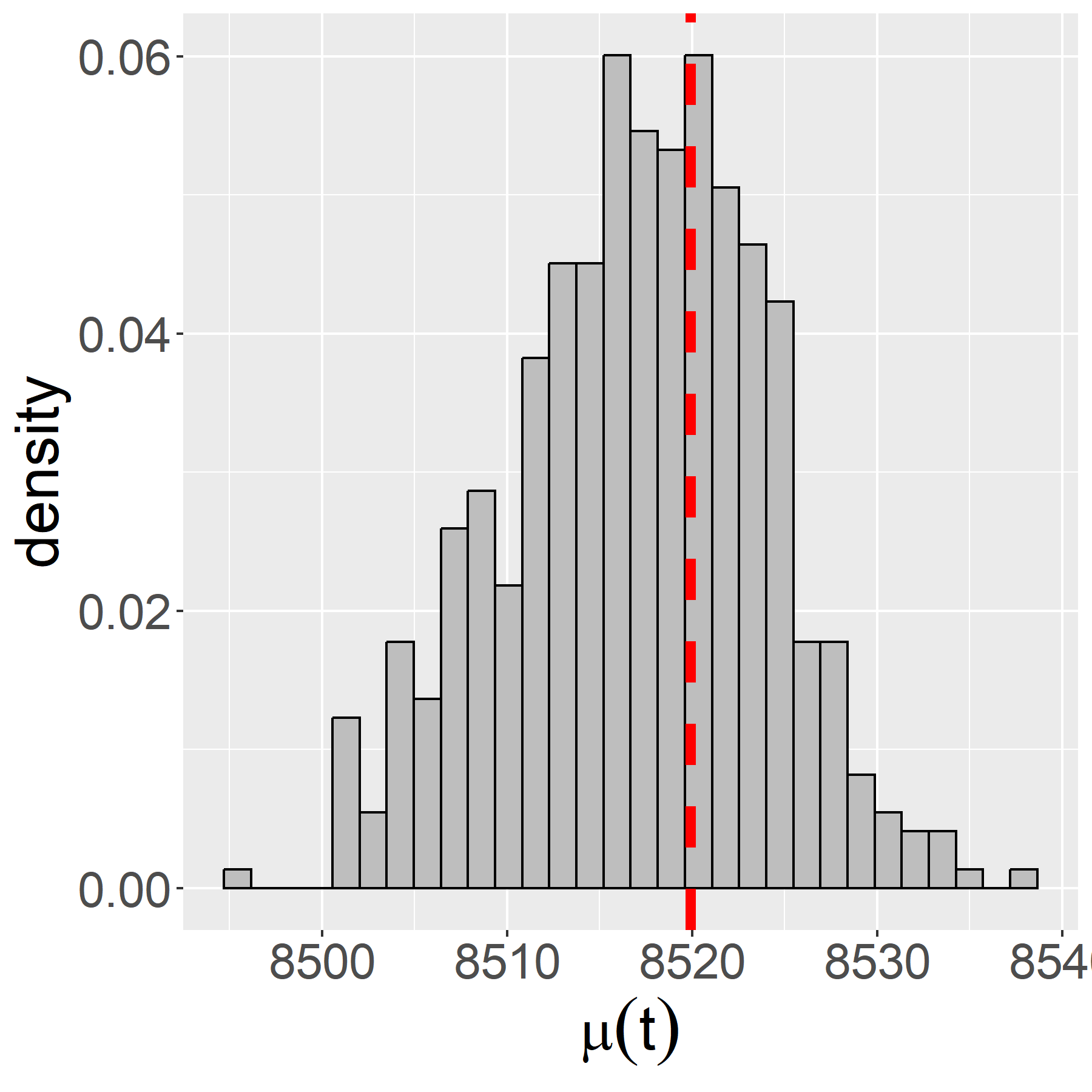} \\
(g) & (h) & (i) \\
\includegraphics[width = 1.25 in]{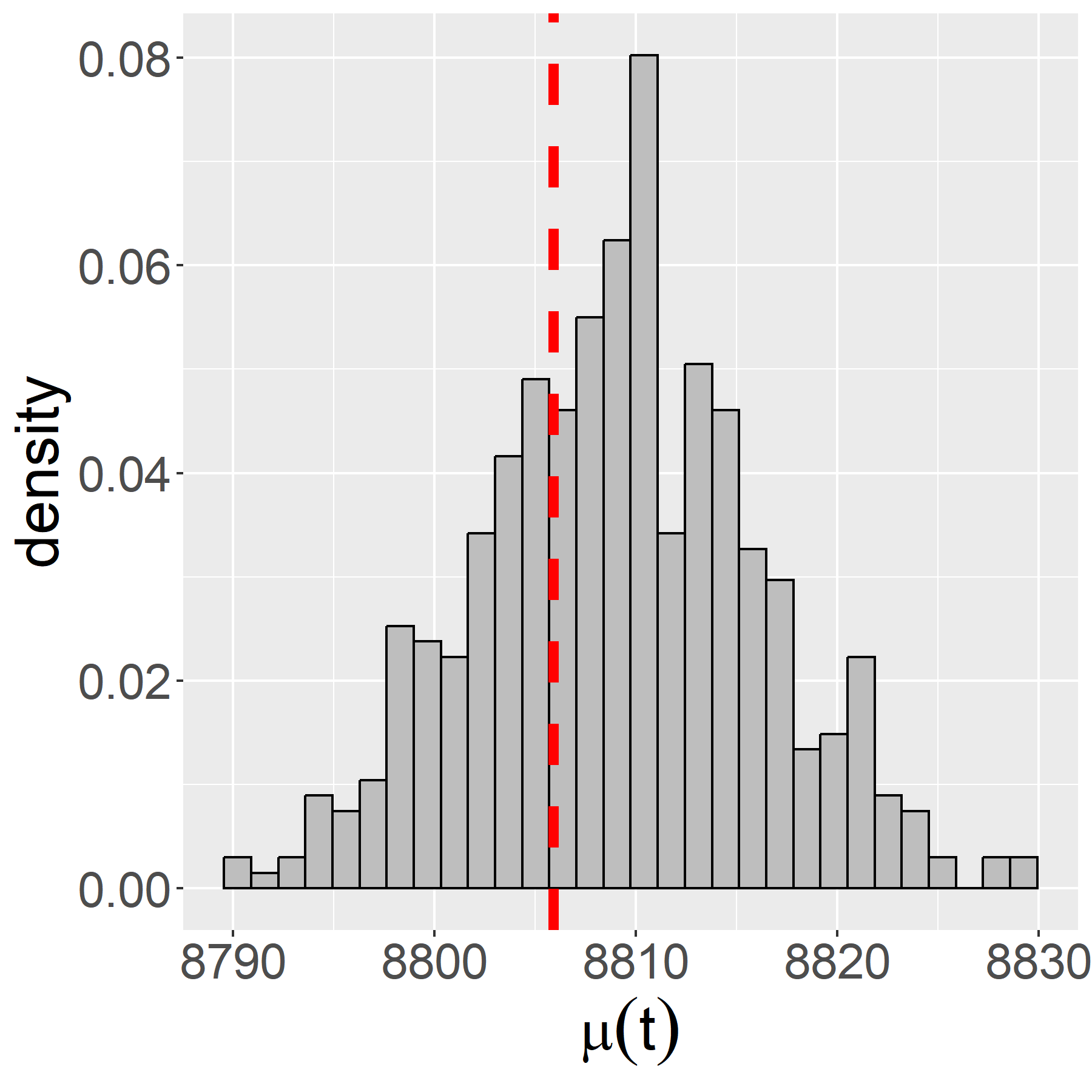}&\includegraphics[width = 1.25 in]{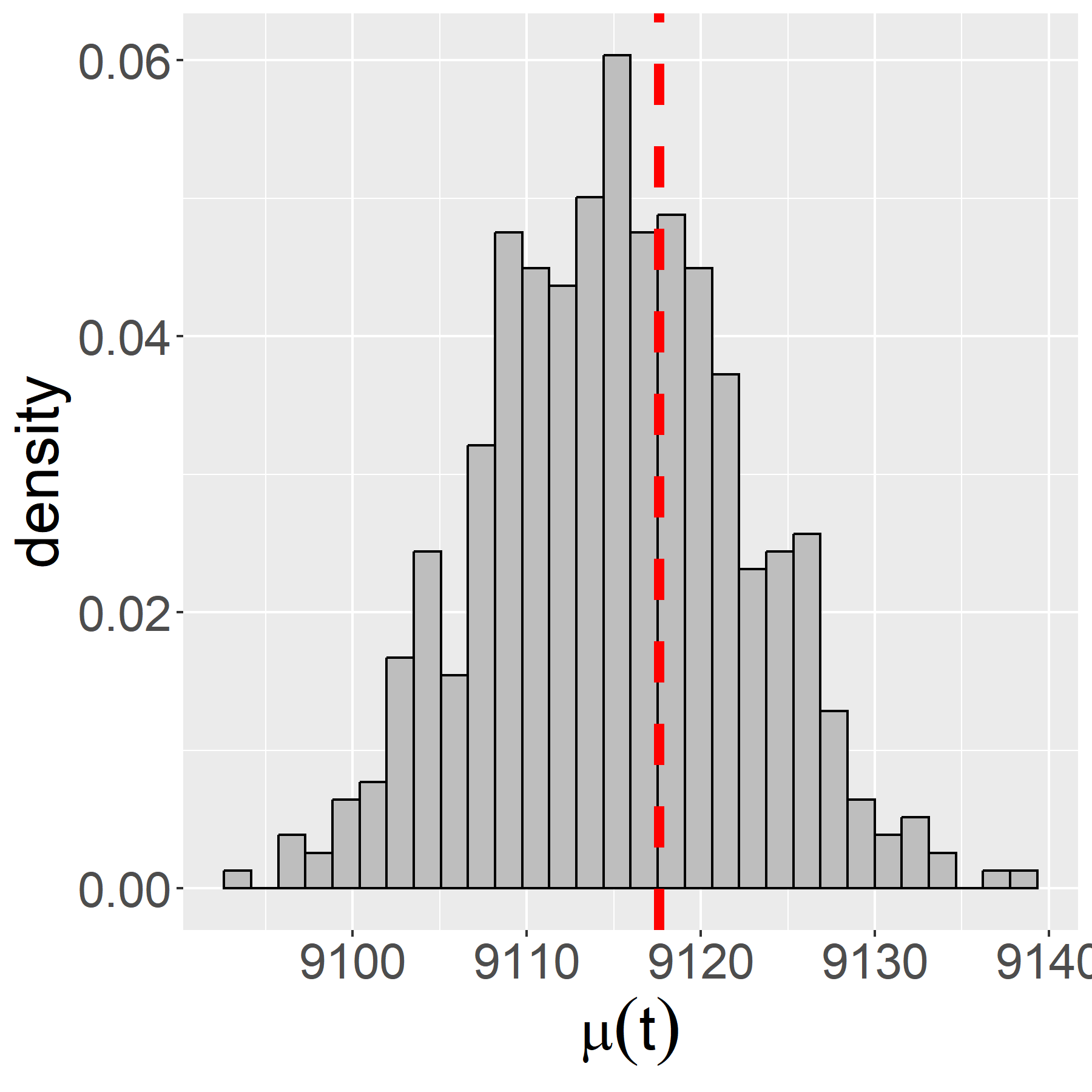}&\includegraphics[width = 1.25 in]{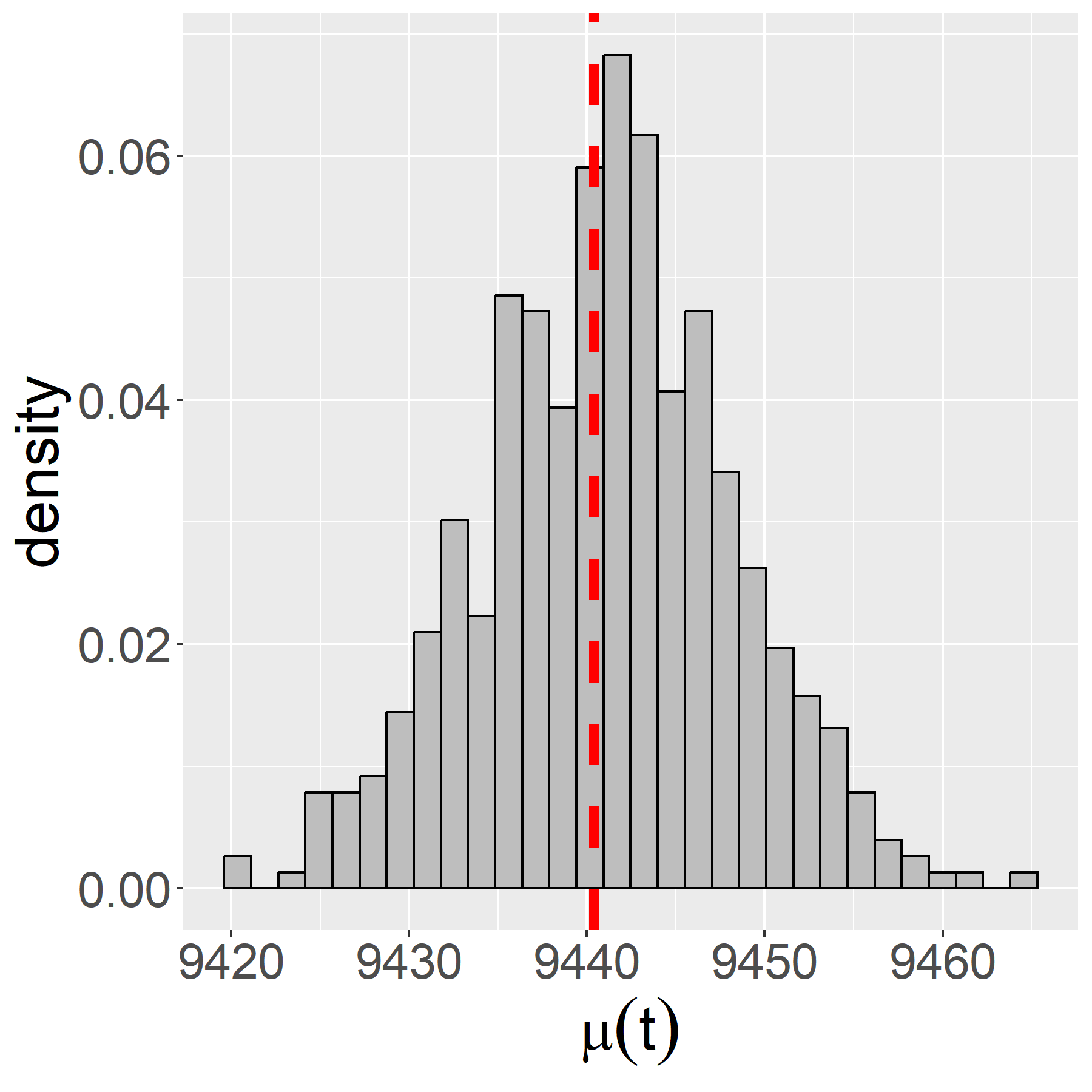} \\
(j) & (k) & (l) \\
\includegraphics[width = 1.25 in]{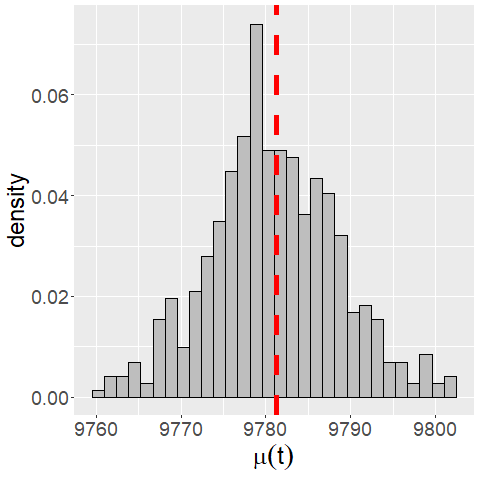}& & \\
(m) & & \\
\end{tabular}
\caption{$T_{mean,(y_1,\ldots,y_n)}$ based on posterior predictive samples (histogram) and observation (red lines) for (a) $t = 0$ years, (b) $t = 2/24$ years, ..., (m) $t = 2$ years. }\label{fig:mean_pred_check}
    \end{center}
\end{figure}

\begin{figure}[t!]
\begin{center}
\begin{tabular}{ccc}
\includegraphics[width = 1.5 in]{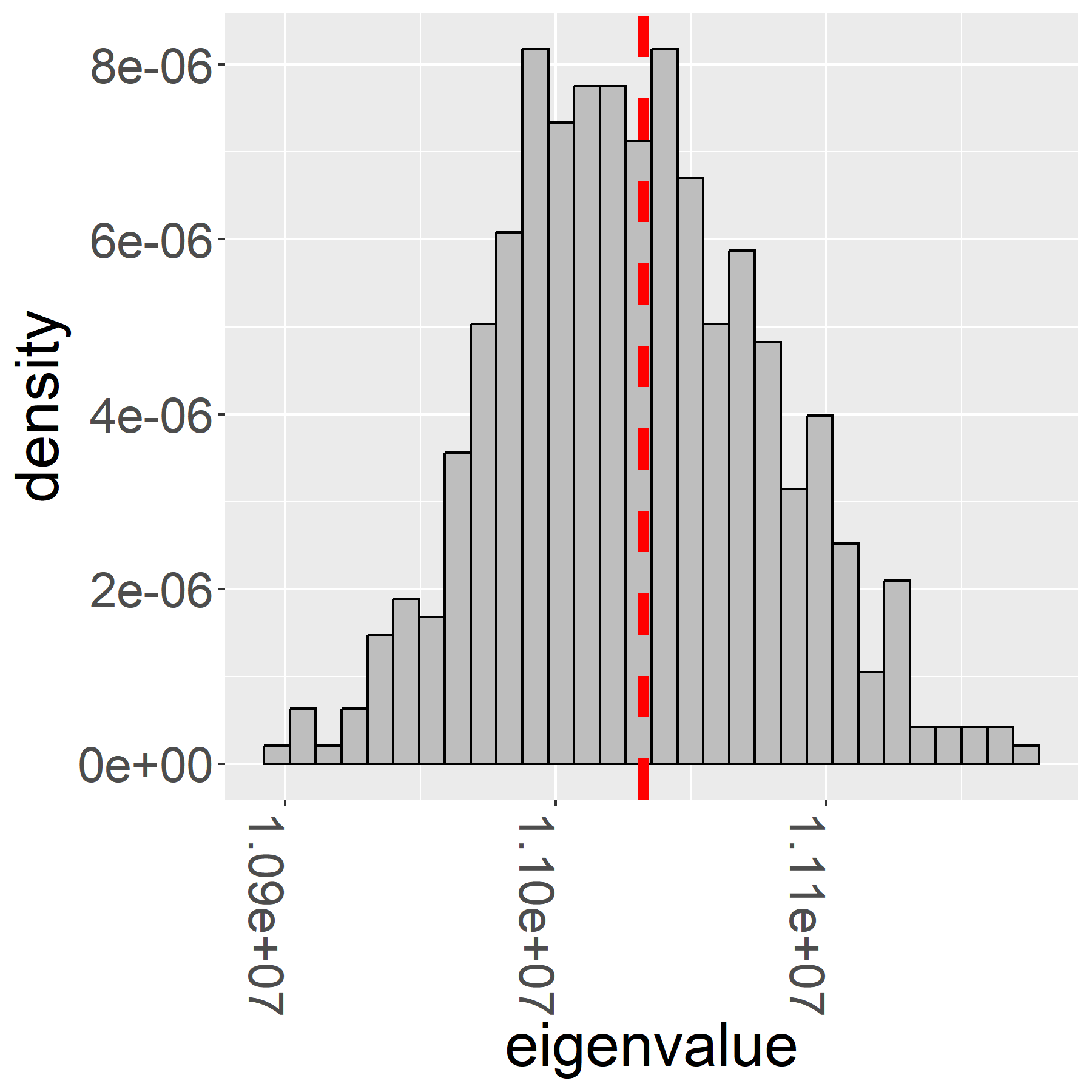}&\includegraphics[width = 1.5 in]{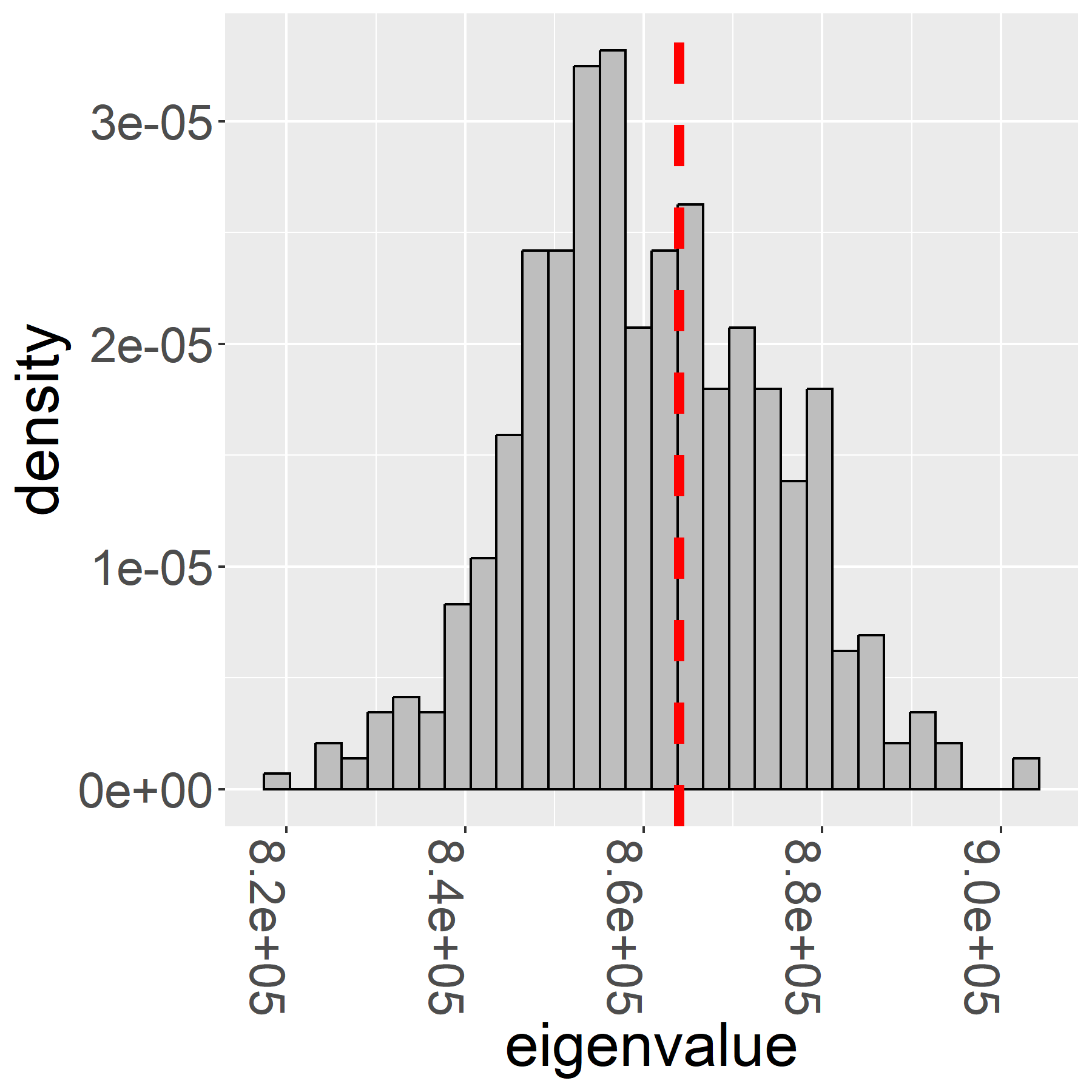}&\includegraphics[width = 1.5 in]{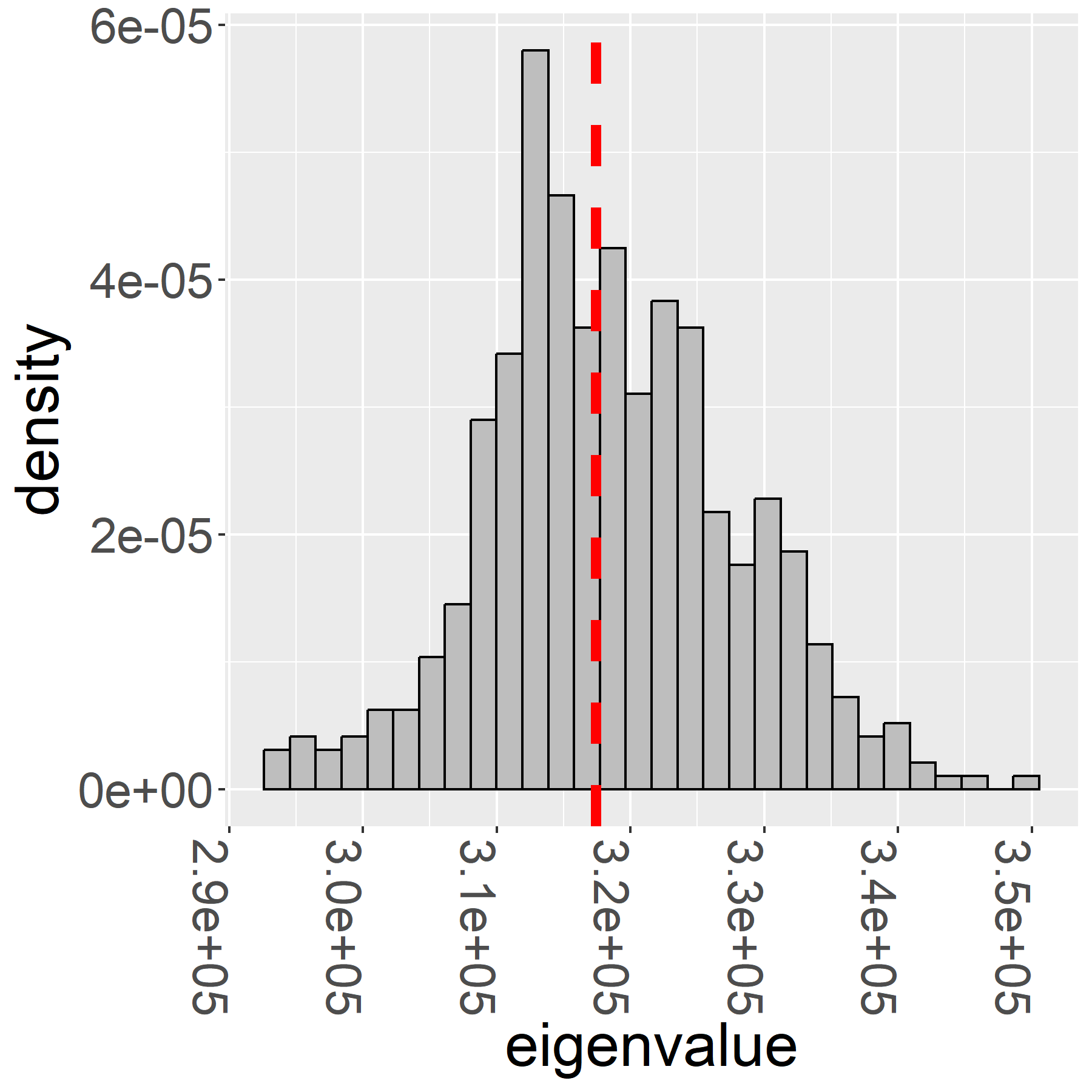} \\
(a) & (b) & (c) 
\end{tabular}
\begin{tabular}{cc}
\includegraphics[width = 1.5 in]{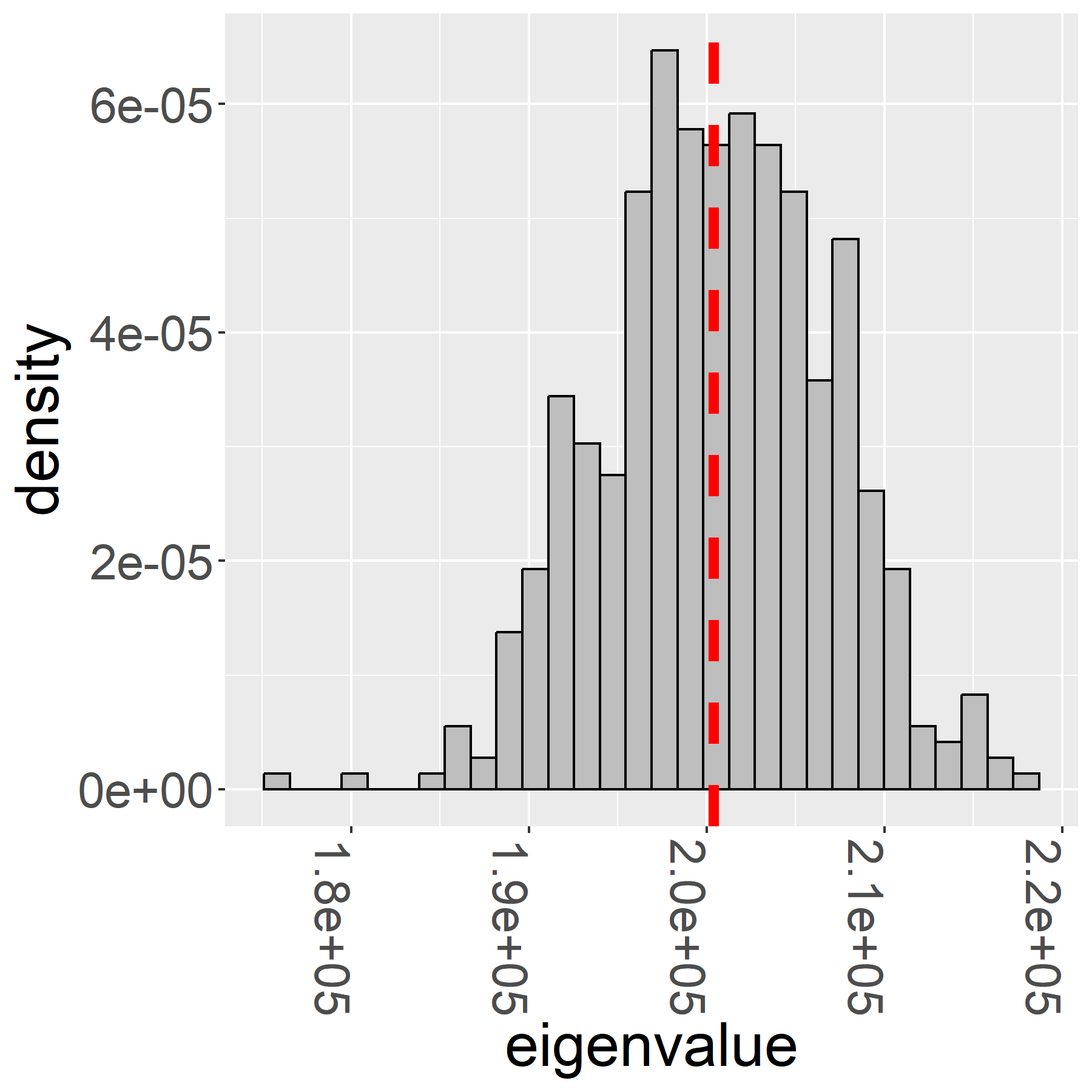}&\includegraphics[width = 1.5 in]{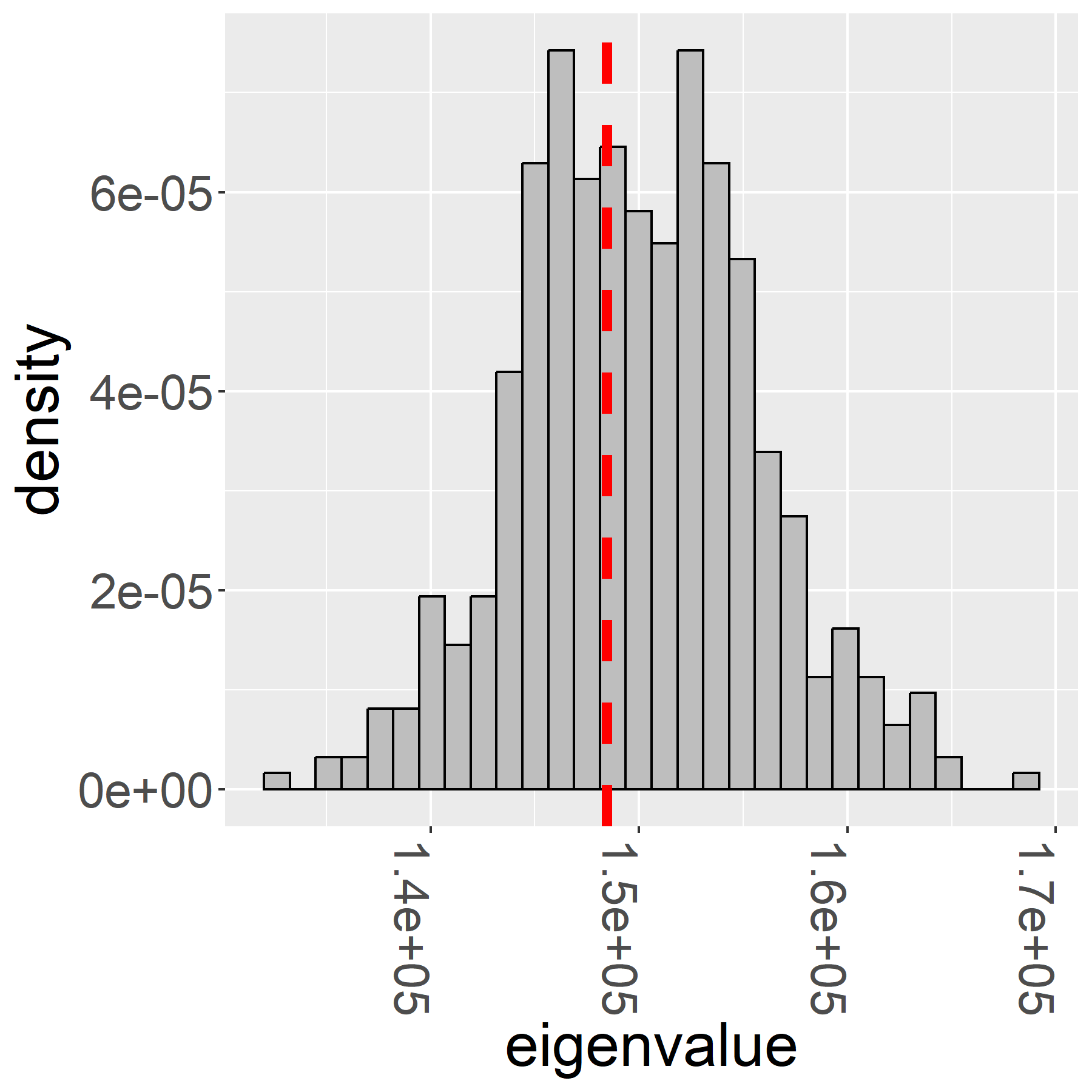} \\
(d) & (e) 
\end{tabular}
\caption{$T_{var,(y_1,\ldots,y_n)}$ based on posterior predictive samples (histogram) and observation (red lines) for (a) $k = 1$, (b) $k = 2$, ..., (e) $k = 5$.}\label{fig:var_pred_check}
    \end{center}
\end{figure}

In this section, we assess the fit of the $\small{\mbox{NeMO}}$ model to the Cebu data presented in the main paper. 

Let $f_i(\vv{t}) = \mu(\vv{t}) + (\lambda_1(\vv{t}),\ldots,\lambda_K(\vv{t}))^\top\eta_i + \sum_{j = 1}^4\int_{T}\beta_j(s,\vv{t})z_{i,j}(s)ds$ denote the function that underlies the sparse and noisy observation $y_i(\vv{t}_i)$, $f^{[iter]}_i(\vv{t})$ denote posterior Markov chain Monte Carlo samples, and let $\hat{f}_i(\vv{t})$ denote the posterior mean computed from Markov chain Monte Carlo samples. Residuals are estimated as $\hat\epsilon_i(\vv{t}_i) = y_i(\vv{t}_i) - \hat{f}_i(\vv{t}_i)$. Figure \ref{fig:residuals} shows a boxplot of estimated residuals for all subjects, $i = 1,\ldots,n$, by age. Visually, it appears that there is no mean trend across ages, and there does not appear to be substantial heteroskedasticity. 

We further assess the fit of the model using posterior predictive checks \citep{gelman2013}. Under the assumed model, posterior predictive samples are generated via $y_{i}^{pred,[iter]}(\vv{t}_i)\sim N(f^{[iter]}_i,\sigma^{[iter],2}I_{m_i})$. Quantities of interest, $T$, based on the posterior predictive samples are compared to the same quantities computed using the actual observations. We consider the following quantities of interest used to study the fit of our model,
\begin{eqnarray}
T_{mag,y_i} & = & \|y_i\|_2,\ i = 1,\ldots,n \nonumber \\
T_{mean,(y_1,\ldots,y_n)}(t) & = & \frac{1}{n}\sum_{i = 1}^n y_i(t),\ t = 0,\frac{2}{24},\ldots,2\  \nonumber \\
T_{var,(y_1,\ldots,y_n)}(k) & = & \text{eig}_k\bigg(\text{cov}(y(s),y(t)) \bigg),\ k = 1\ldots,K, \nonumber
\end{eqnarray}
where $\text{eig}_k$ denotes the $k$\textsuperscript{th} largest eigenvalue of the sample covariance matrix based on $y_1,\ldots,y_n$.
The quantity $T_{mag}$ can be used to assess how well the model captures the magnitude of individual subjects, $T_{mean}$ can be used to assess how well the model captures the mean trend across subjects, and $T_{var}$ can be used to assess how well the modes of variability are captured across subjects. 

Figure \ref{fig:norm_pred_check} panels (a)-(f) show histograms of $T_{mag,y_i}$ based on posterior predictive draws for different subjects, where the red line represents $T_{mag,y_i}$ based on actual observed values. Panels (a)-(c) represent random samples of subjects in the population, while panels (d)-(f) represent subjects having posterior predictive mean of $T_{mag,y_i}$ furthest from the observed value. Even in the worst case scenario, the predictions based on the model cover the observed value. Panel (g) of this figure shows a histogram of the absolute difference between the mean of $T_{mag,y_i}$ based on  posterior predictive samples and the observed value. The absolute difference is mostly concentrated near zero. Figure \ref{fig:mean_pred_check} panels (a)-(f) show histograms of $T_{mean,(y_1,\ldots,y_n)}$ based on posterior predictive draws at different times $t = 0,\ldots,2$, where the red line represents $T_{mean,(y_1,\ldots,y_n)}$ based on actual observed values. These values are computed cross-sectionally by removing missing values. At each time, it appears that the observed value of $T_{mean,(y_1,\ldots,y_n)}$ is covered by the values computed using posterior predictive samples. Figure \ref{fig:var_pred_check} panels (a)-(f) show histograms of $T_{var,(y_1,\ldots,y_n)}$ based on posterior predictive draws for different eigenvalues $k = 1,\ldots,5$, where the red line represents $T_{var,(y_1,\ldots,y_n)}$ based on actual observed values. The sample covariance is computed by using pairwise complete observations. For each eigenvalue the observed value of $T_{var,(y_1,\ldots,y_n)}$ is covered by the values computed using posterior predictive samples. Based on these posterior predictive checks, it is apparent that the model is adequate in describing these individual and population level features of the observations. 

In the main paper, cough, season, and strata are noted as the covariates with the smallest effect on the response. We consider reducing the model presented in the main paper by omitting these covariates. In order to assess fit of reduced models, compared to a full model with all covariates, we use widely applicable information criterion. Results are presented in Table \ref{table:full_v_reduced}. For each case, the full model is preferred. 

\begin{table}
\caption{Widely applicable information criterion (WAIC) for full versus reduced models. In all cases the full model is preferred.\label{table:full_v_reduced}}
\centering
\begin{tabular}{|l||l|l|l|}
\hline
covariate omitted & WAIC     & \begin{tabular}[c]{@{}l@{}}difference\\ (full - reduced)\end{tabular} & \begin{tabular}[c]{@{}l@{}}standard error\\ (difference in WAIC)\end{tabular} \\ \hline \hline
cough            & 69727.66 & -18.40                                                                & 26.41                                                                         \\ \hline
season           & 69719.45 & -10.19                                                                & 33.76                                                                         \\ \hline
strata           & 69718.37 & -9.11                                                                 & 26.41                                                                         \\ \hline
\end{tabular}
\end{table}

\subsection{Markov chain Monte Carlo Trace Plots}

The inferential results of the main paper are based on Markov chain Monte Carlo samples. In this section we plot thinned Markov chain Monte Carlo samples to diagnose convergence. The diagnostic plots for all parameters do not indicate that there is an issue with Markov chain Monte Carlo convergence. Figure \ref{fig:trace_fpca} displays trace plots related to the inferential results of Figures 3 and 4 of the main paper. Figure \ref{fig:trace_fpca_reg} displays trace plots related to the inferential results of Figure 5 of the main paper. Figure \ref{fig:trace_bf} displays trace plots related to the inferential results of Figures 6 and 7 of the main paper. Figure \ref{fig:trace_ill} displays trace plots related to the inferential results of Figure 8 of the main paper

\begin{figure}[t]
\begin{center}
 \begin{tabular}{ccc}
\includegraphics[width = 1.25 in]{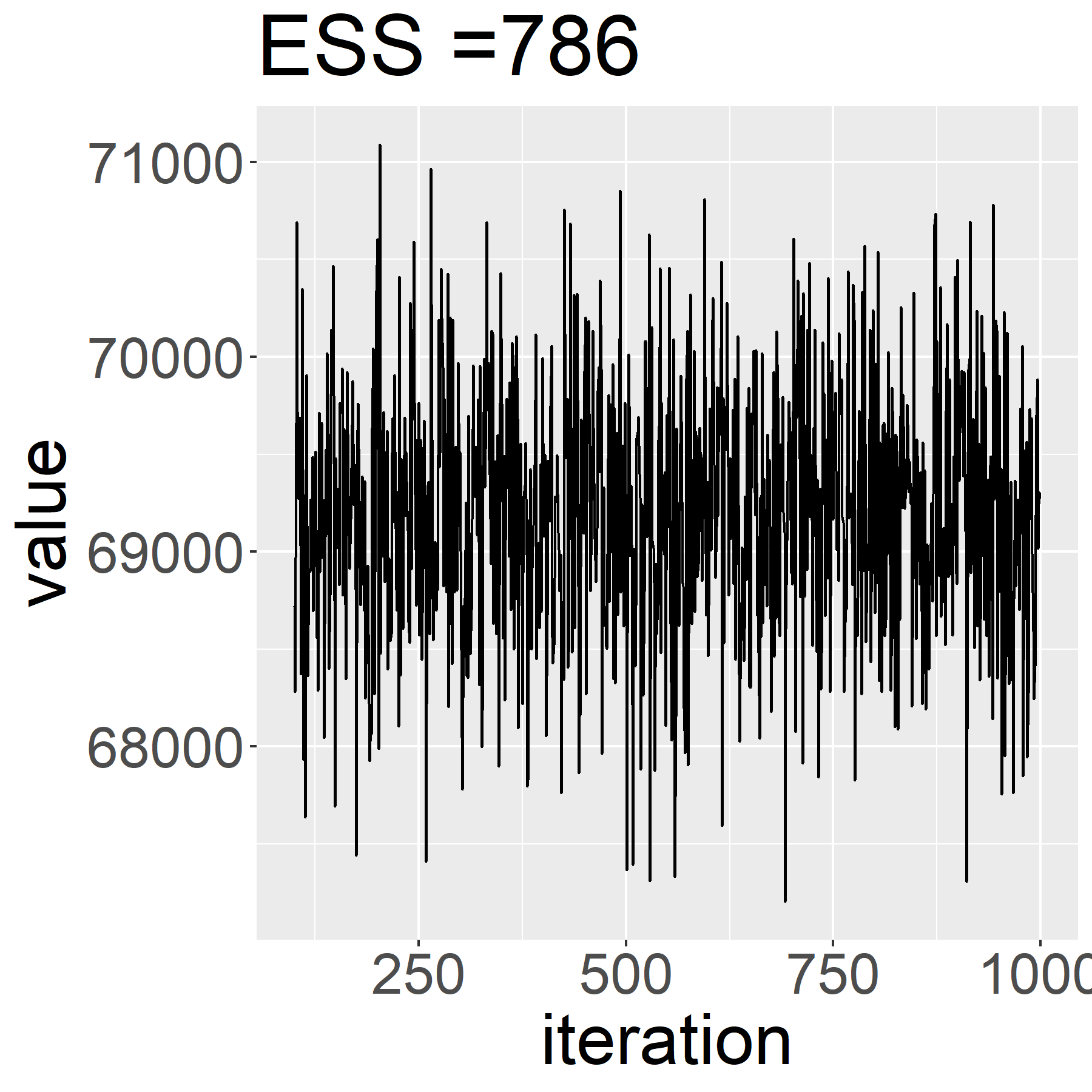} & \includegraphics[width = 1.25 in]{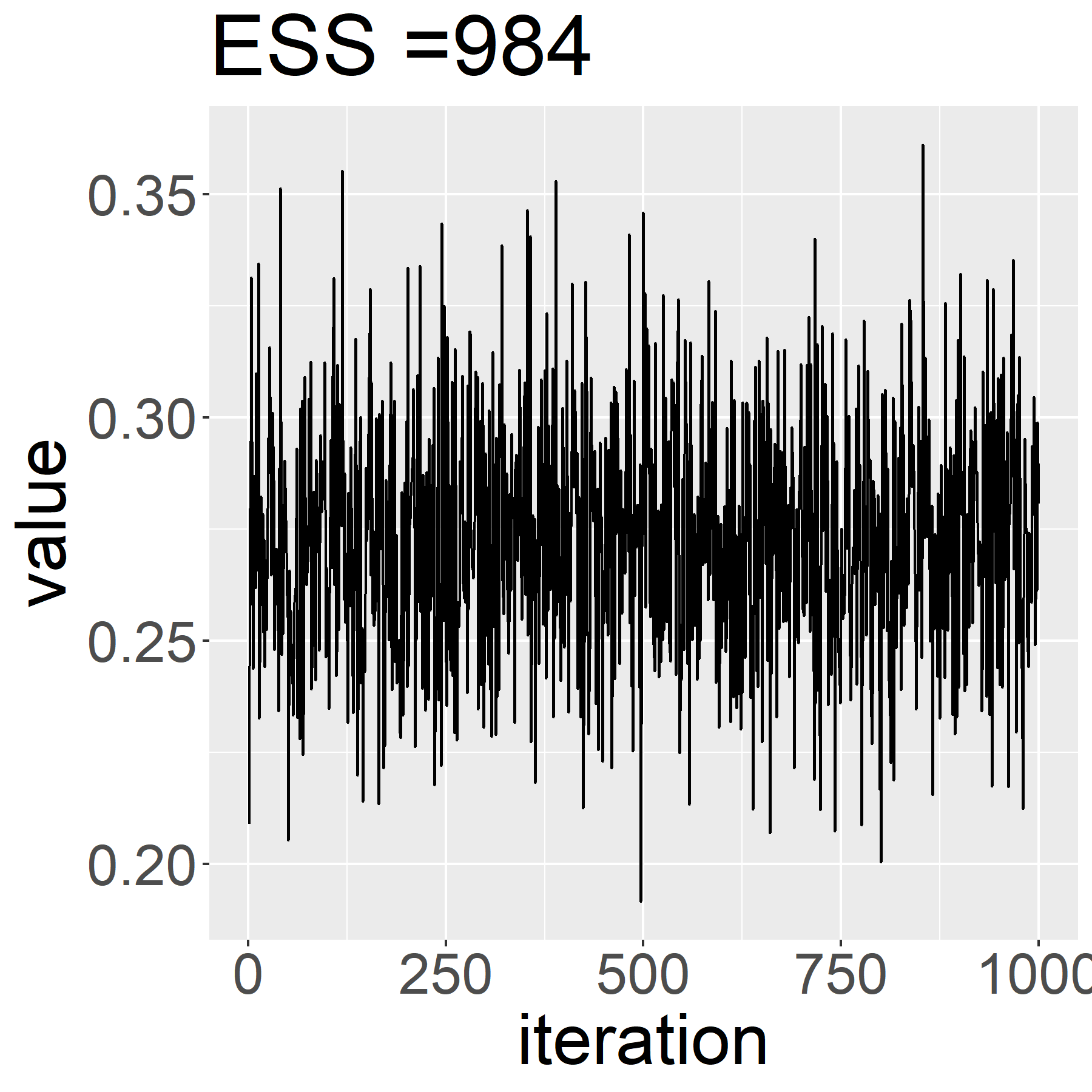} & \includegraphics[width = 1.25 in]{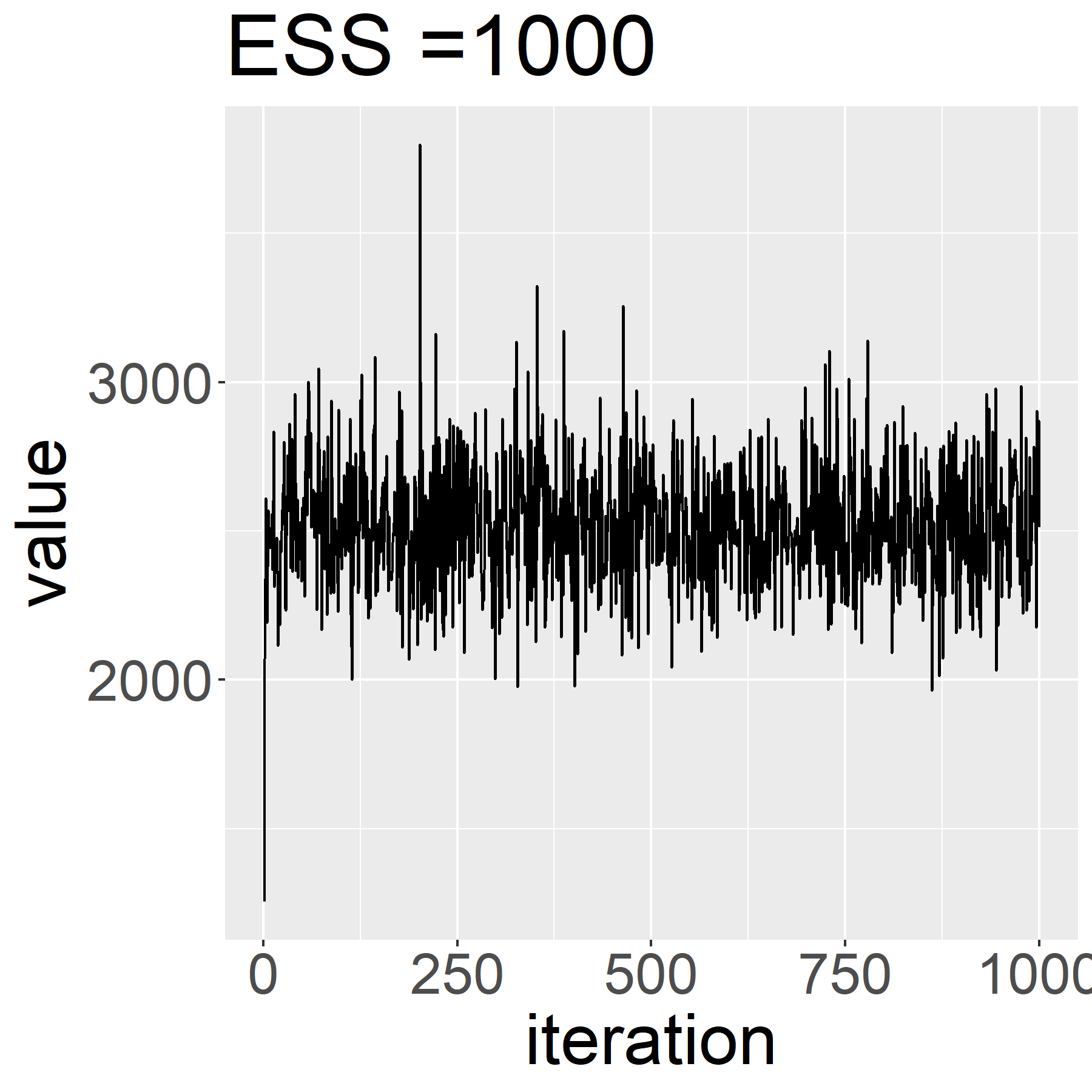}  \\
$\sigma^2$ & $l_\mu^2$ & $\tau^2_\mu$ \\
\end{tabular}
 \begin{tabular}{ccccc}
\includegraphics[width = 1.0 in]{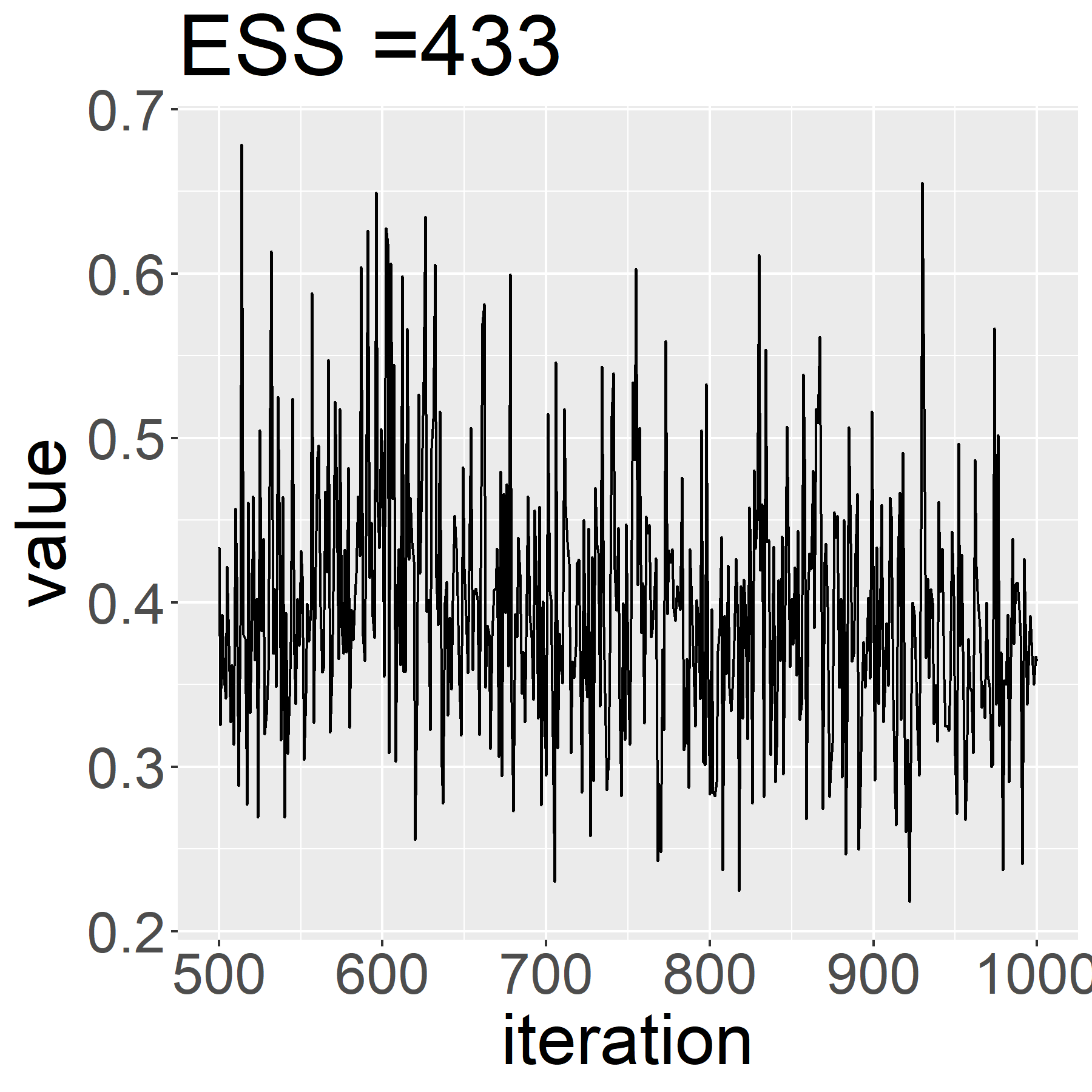} & \includegraphics[width = 1.0 in]{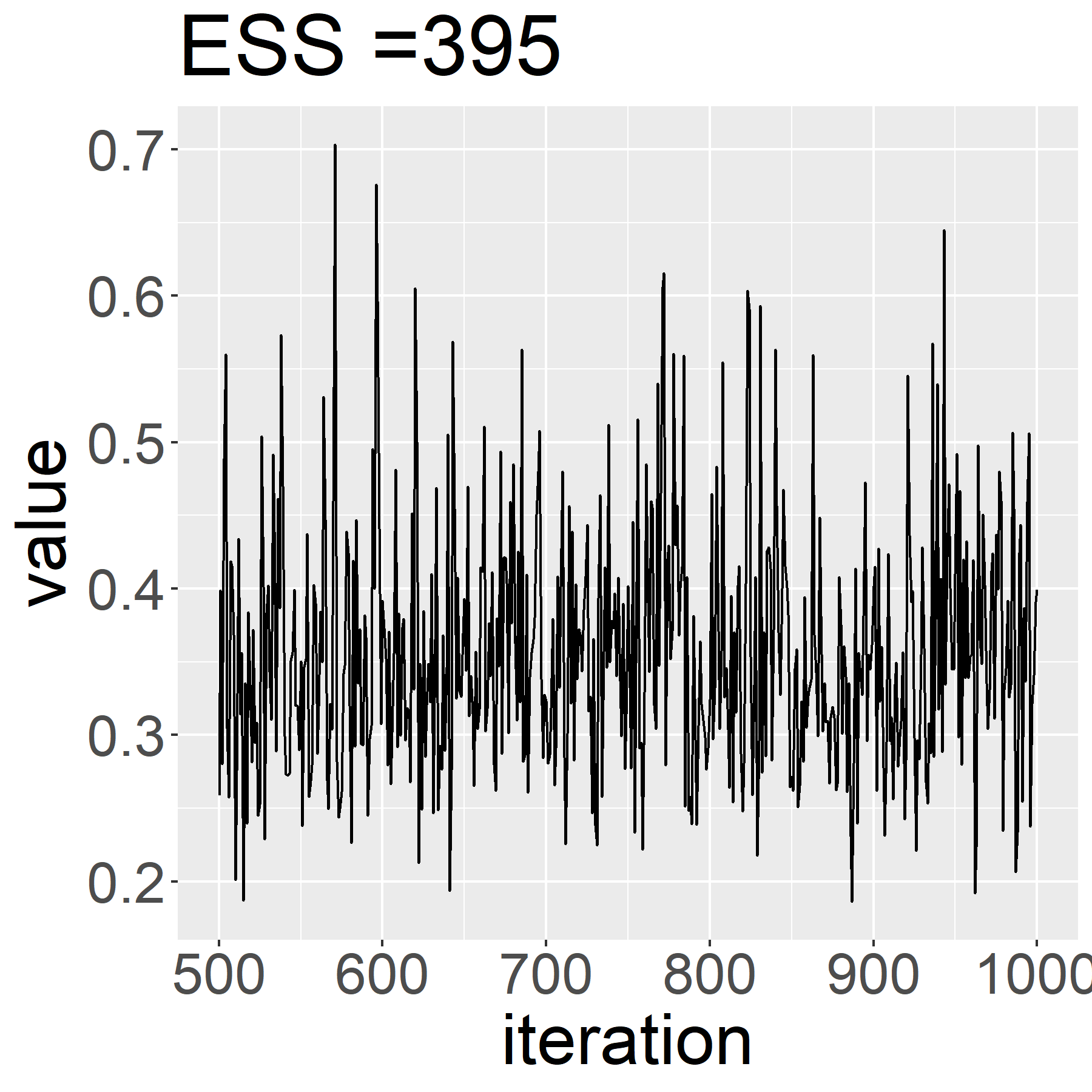} & \includegraphics[width = 1.0 in]{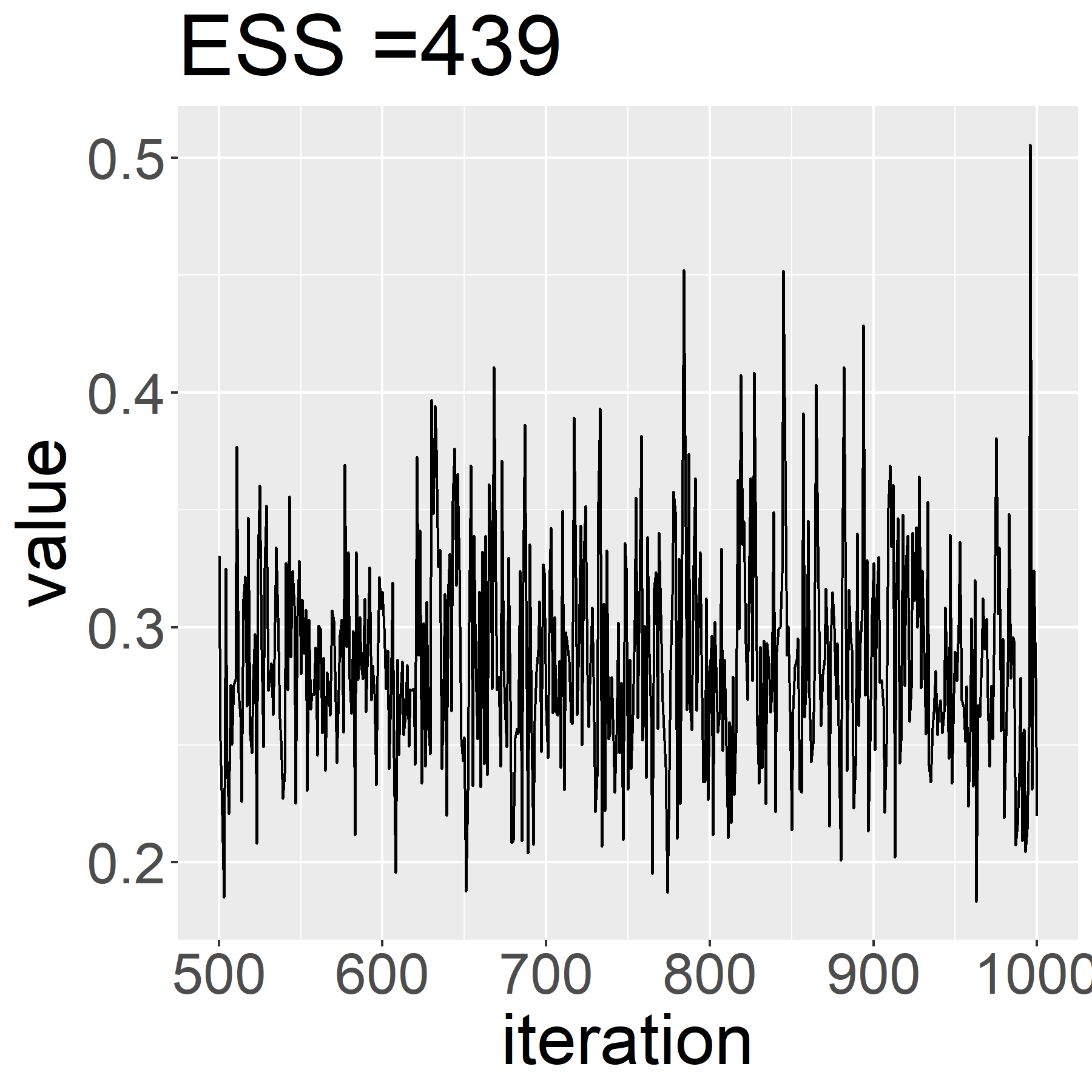}  & \includegraphics[width = 1.0 in]{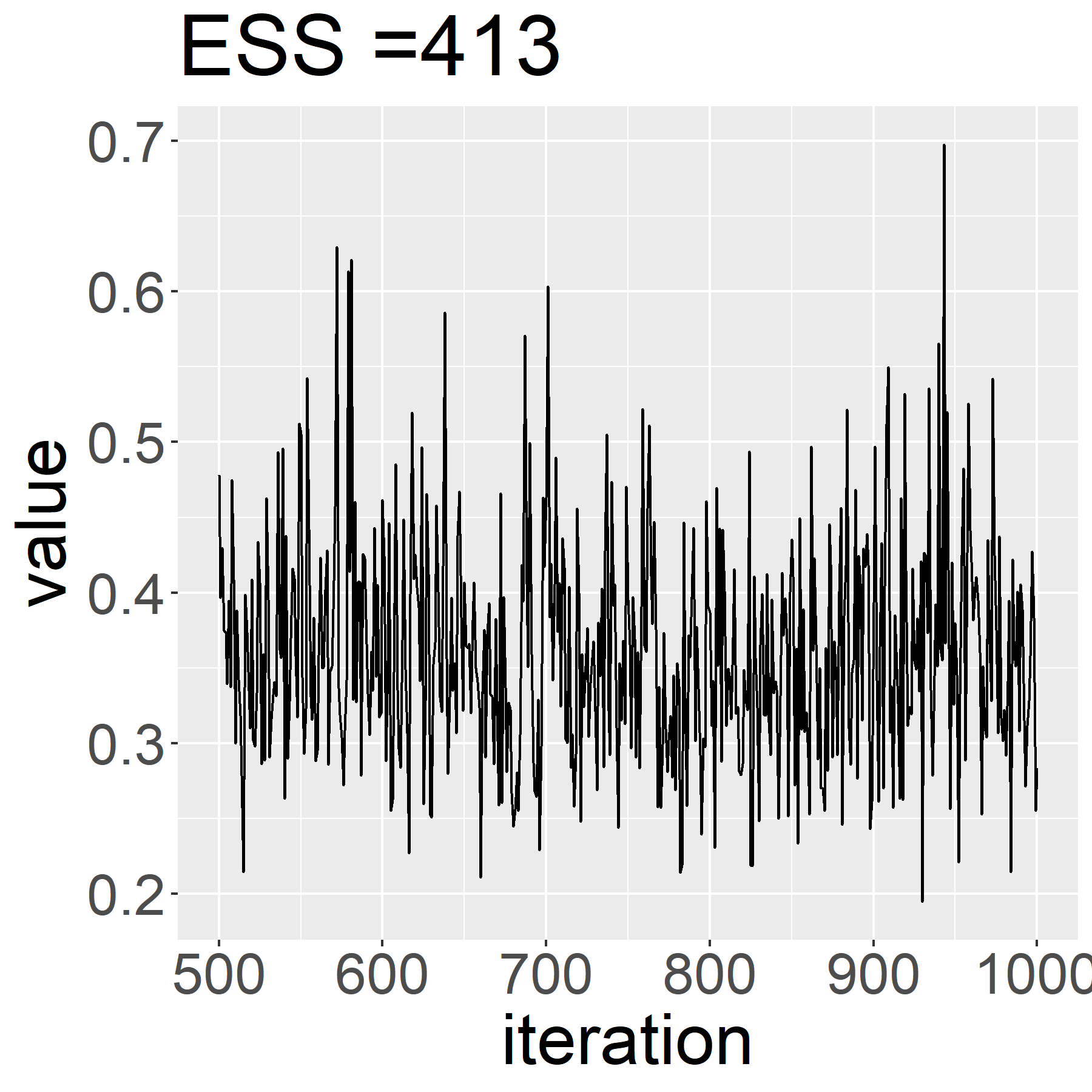}  & \includegraphics[width = 1.0 in]{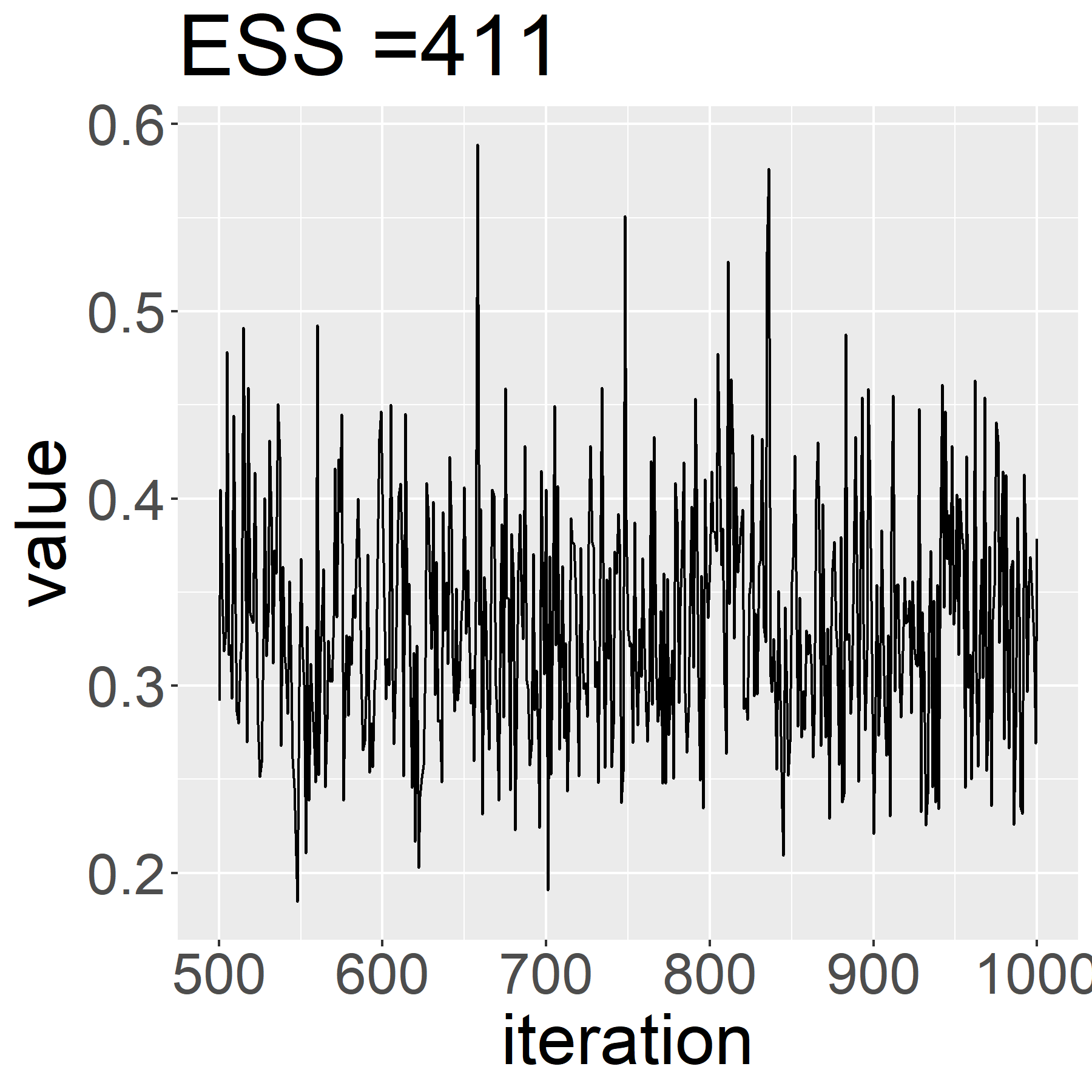} \\
$l_1^2 $ & $l_2^2$& $l_3^2$ & $l_4^2$ &  $l_5^2$ \\
\includegraphics[width = 1.0 in]{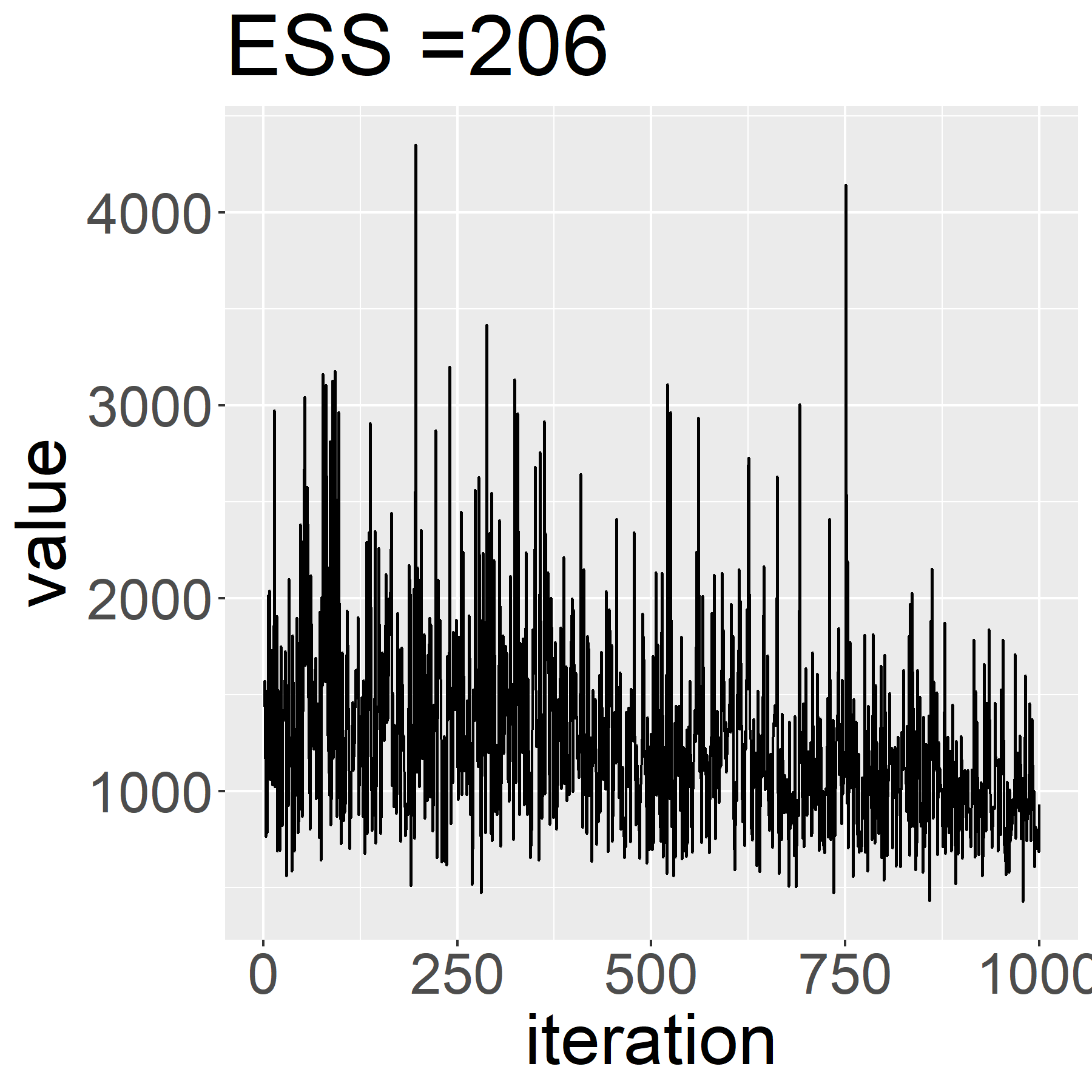} & \includegraphics[width = 1.0 in]{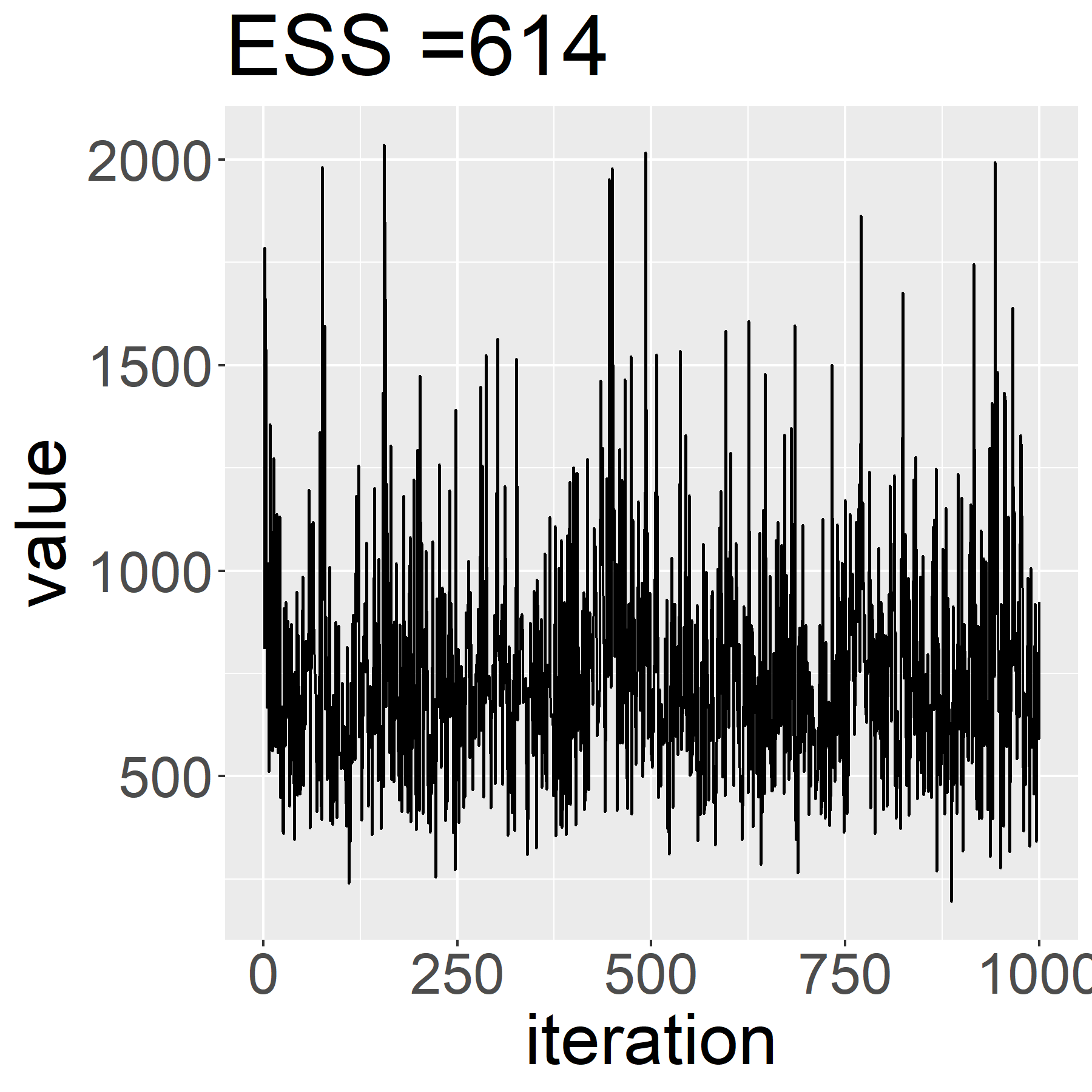} & \includegraphics[width = 1.0 in]{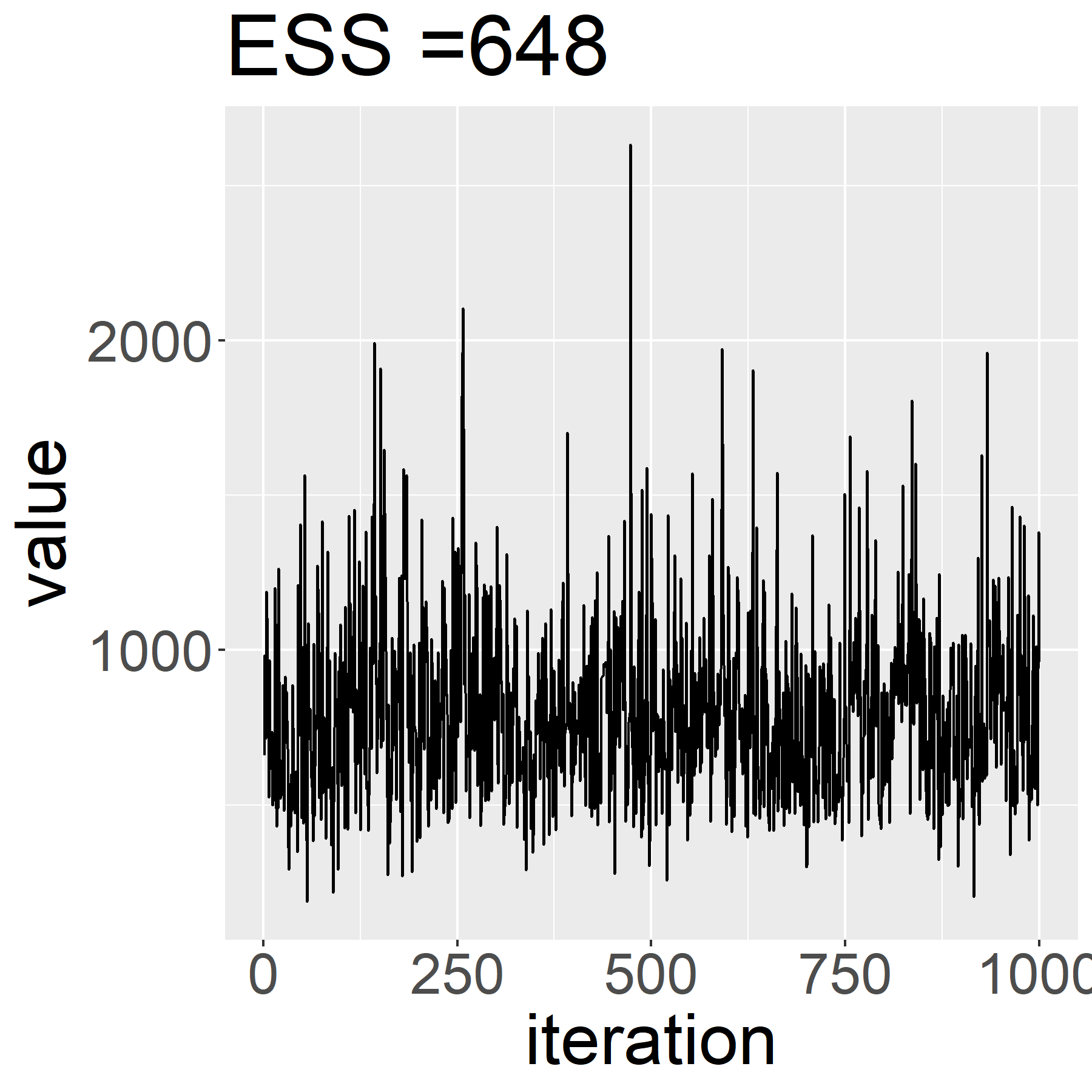}  & \includegraphics[width = 1.0 in]{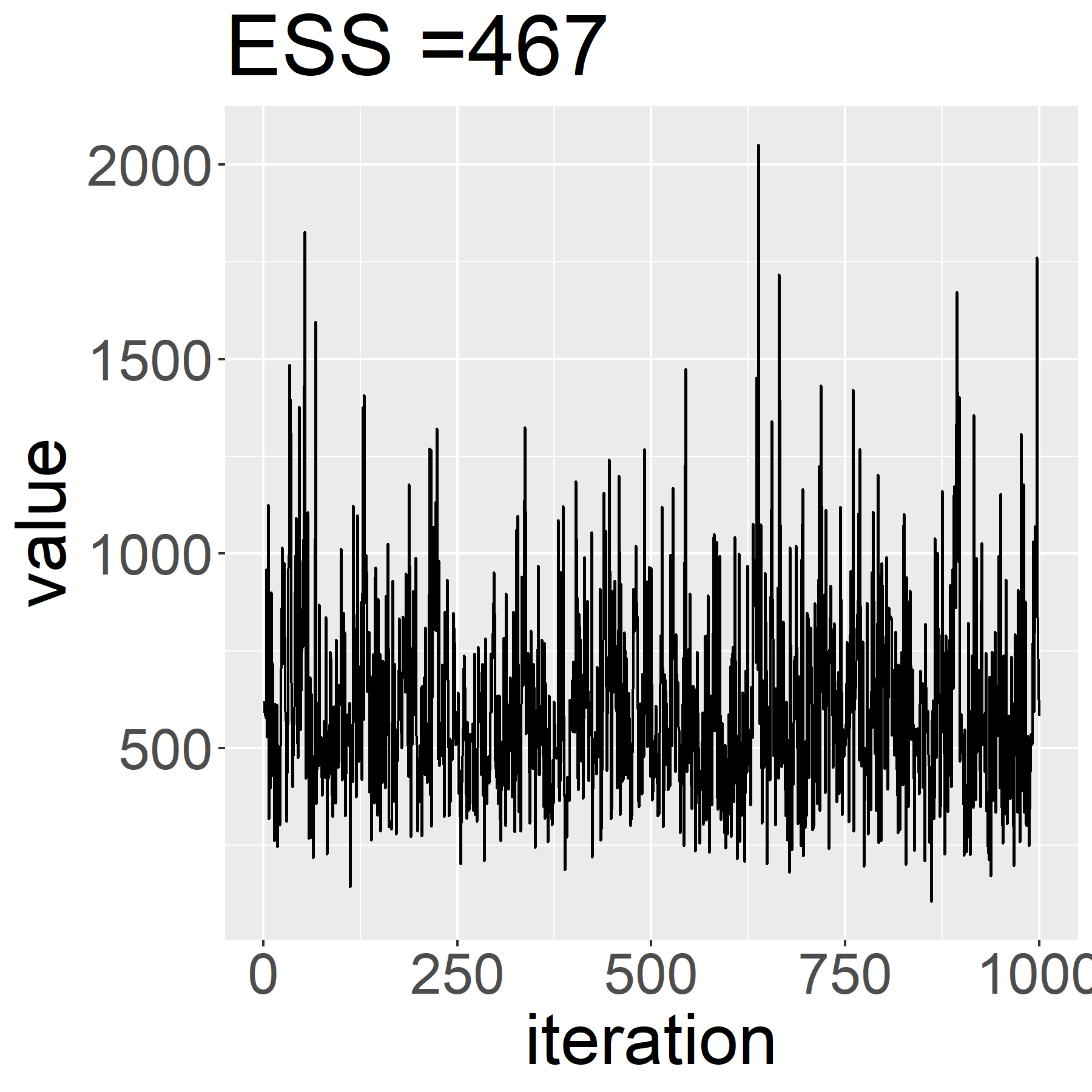}  & \includegraphics[width = 1.0 in]{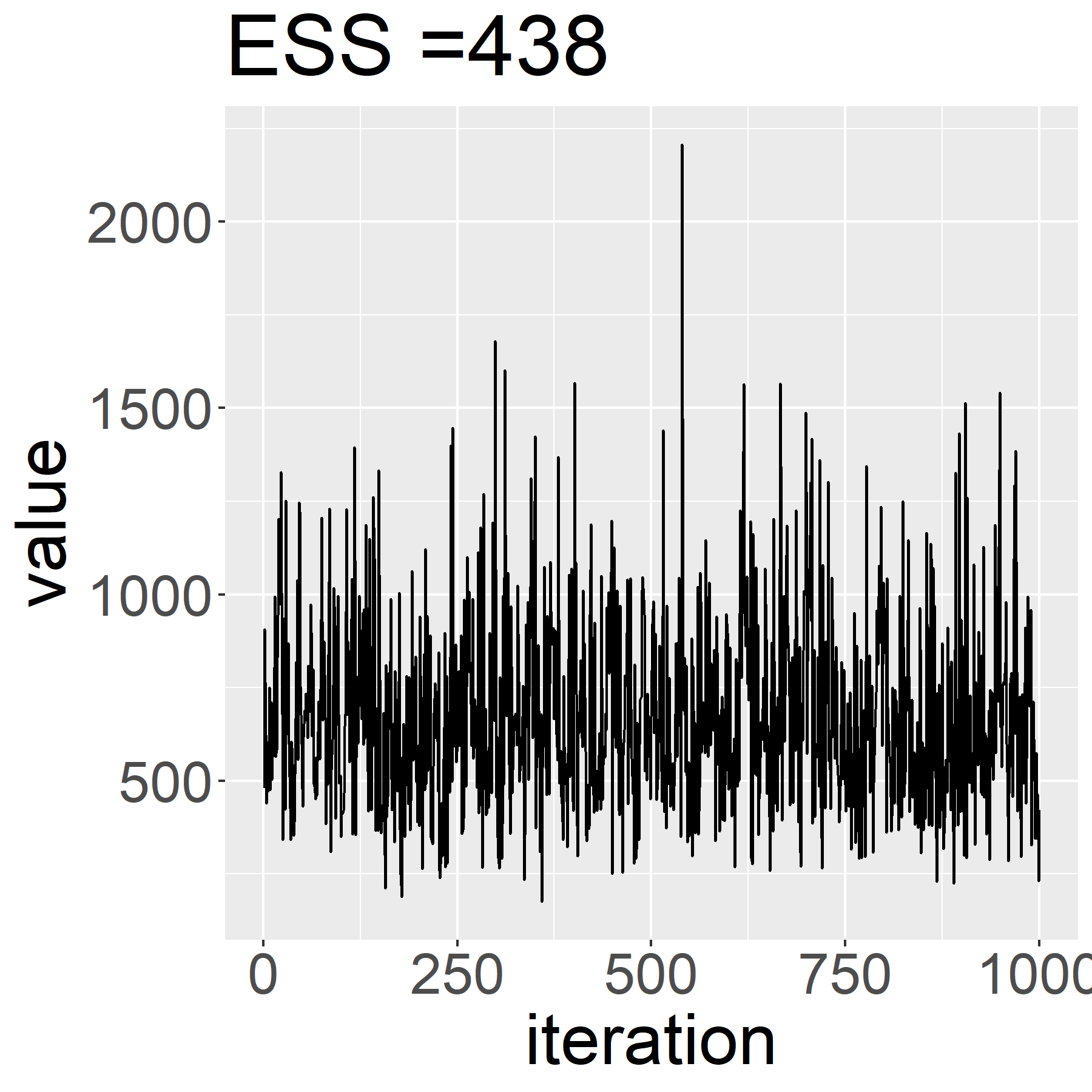} \\
$\tau_1^2 $ & $\tau_2^2$& $\tau_3^2$ & $\tau_4^2$ &  $\tau_5^2$ \\
\includegraphics[width = 1.0 in]{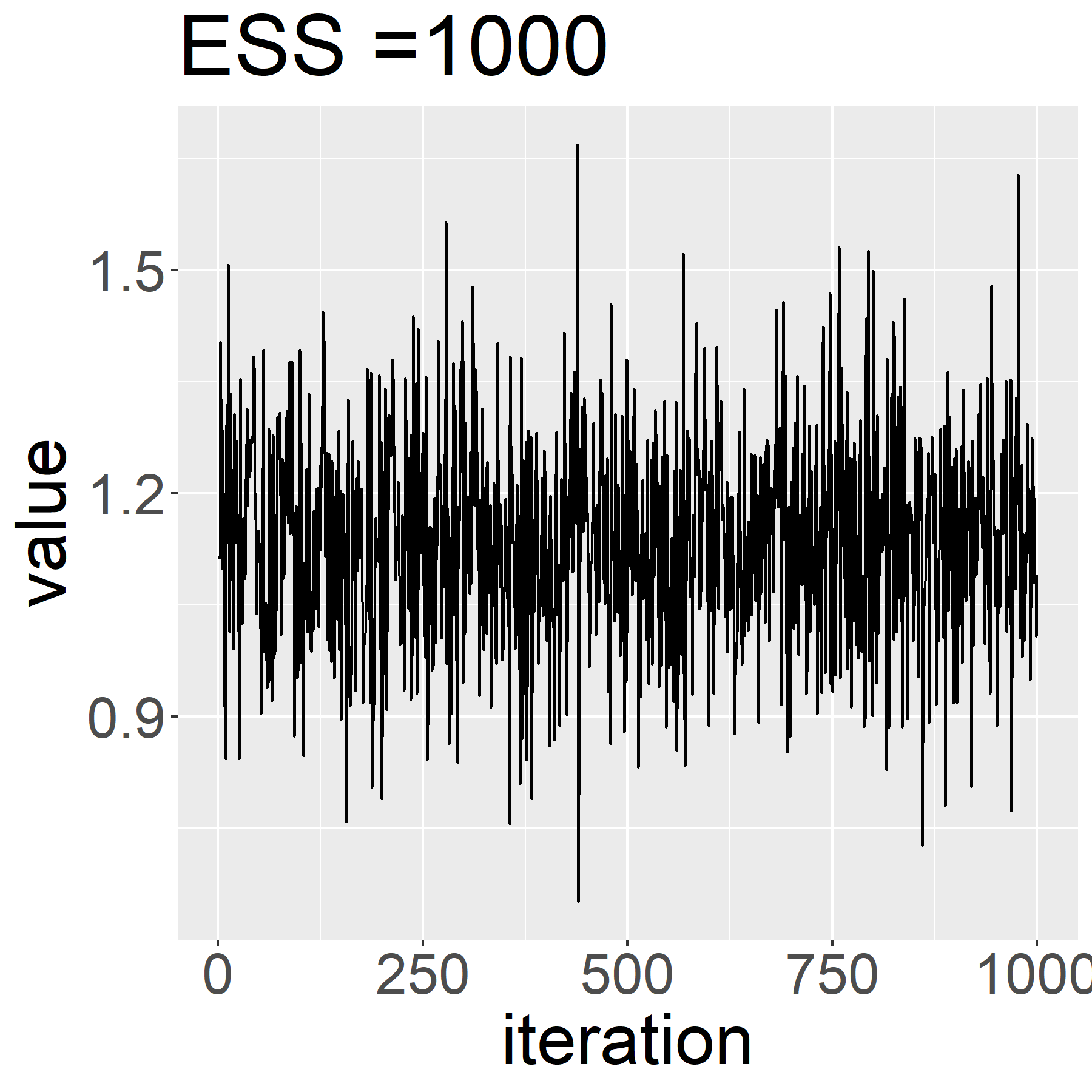} & \includegraphics[width = 1.0 in]{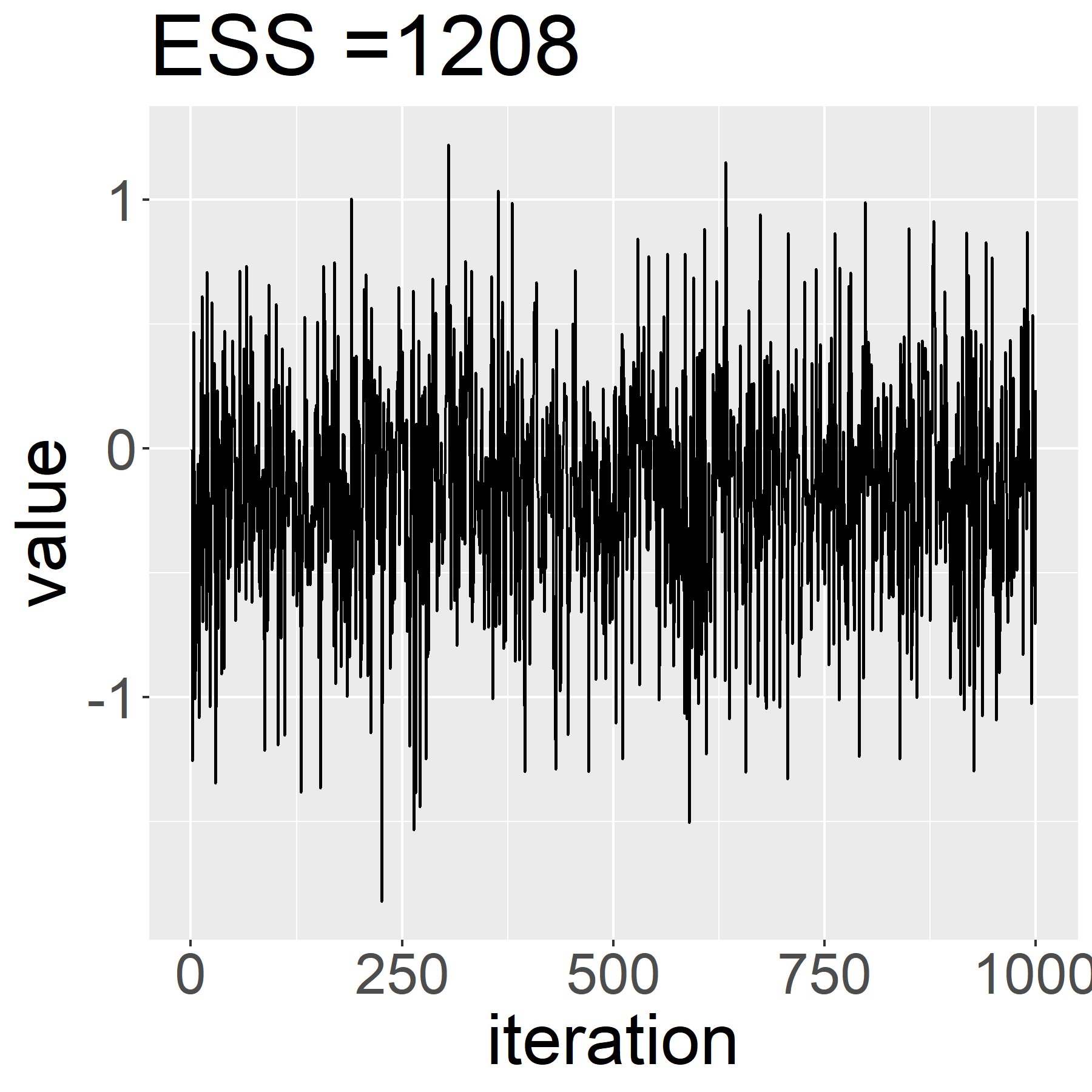} & \includegraphics[width = 1.0 in]{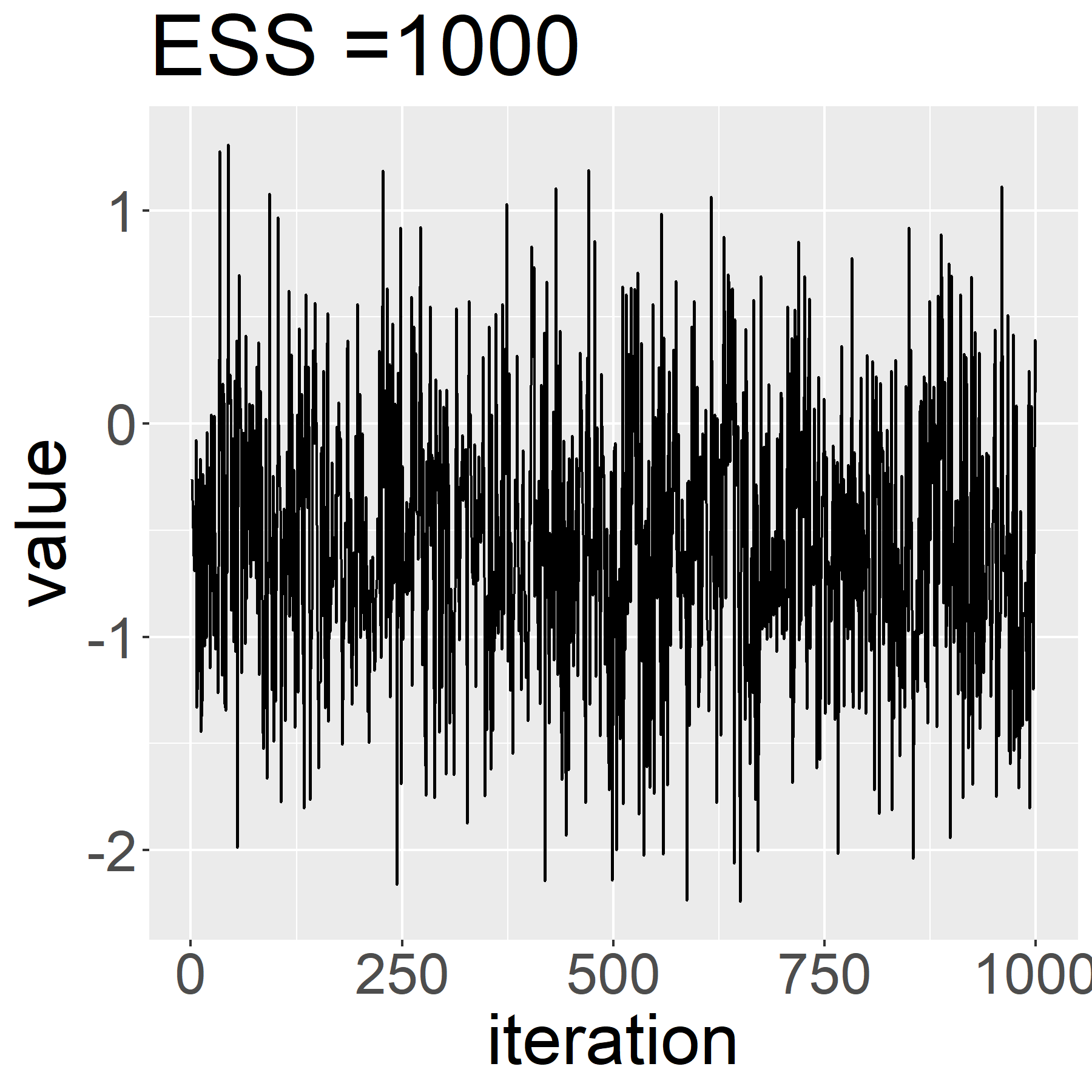}  & \includegraphics[width = 1.0 in]{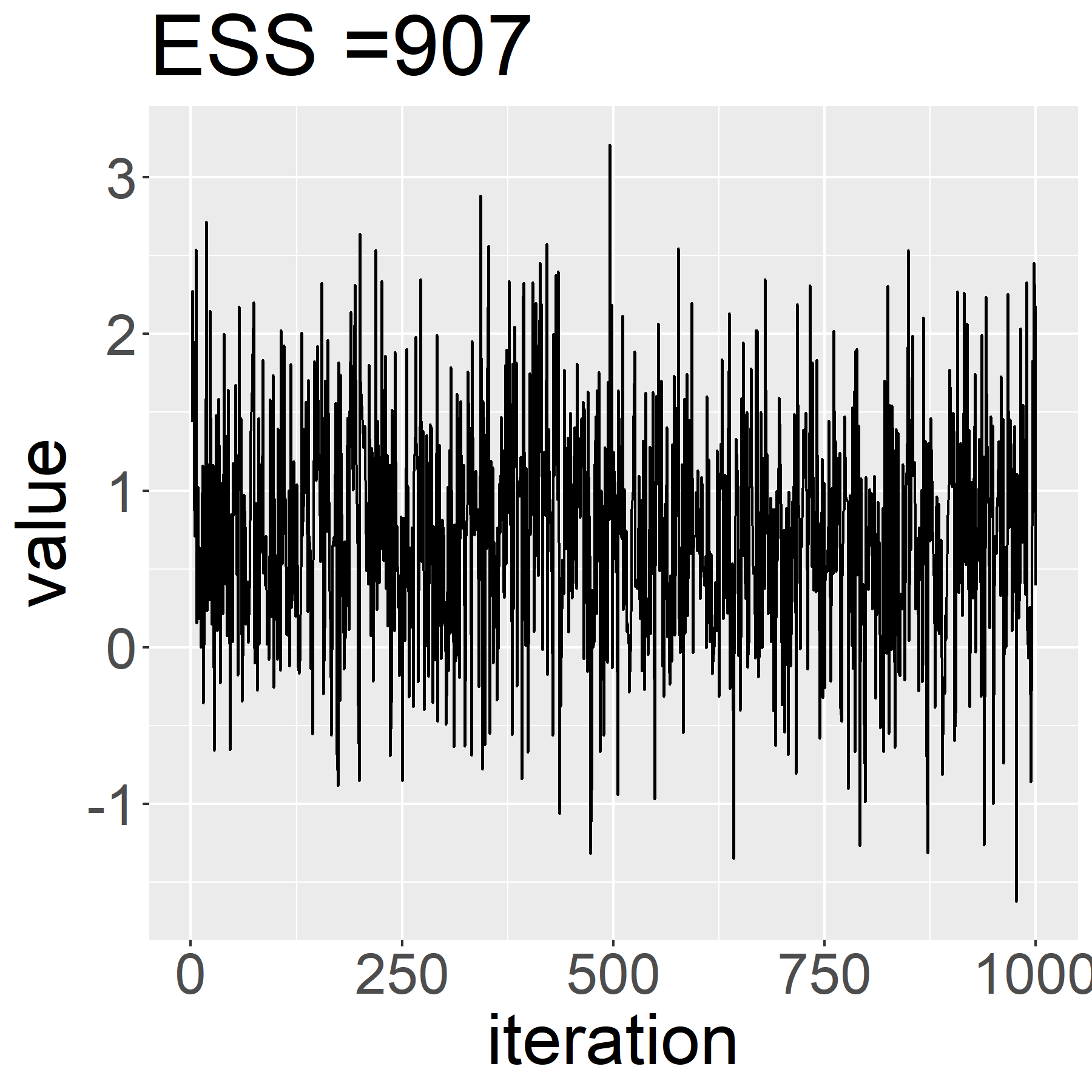}  & \includegraphics[width = 1.0 in]{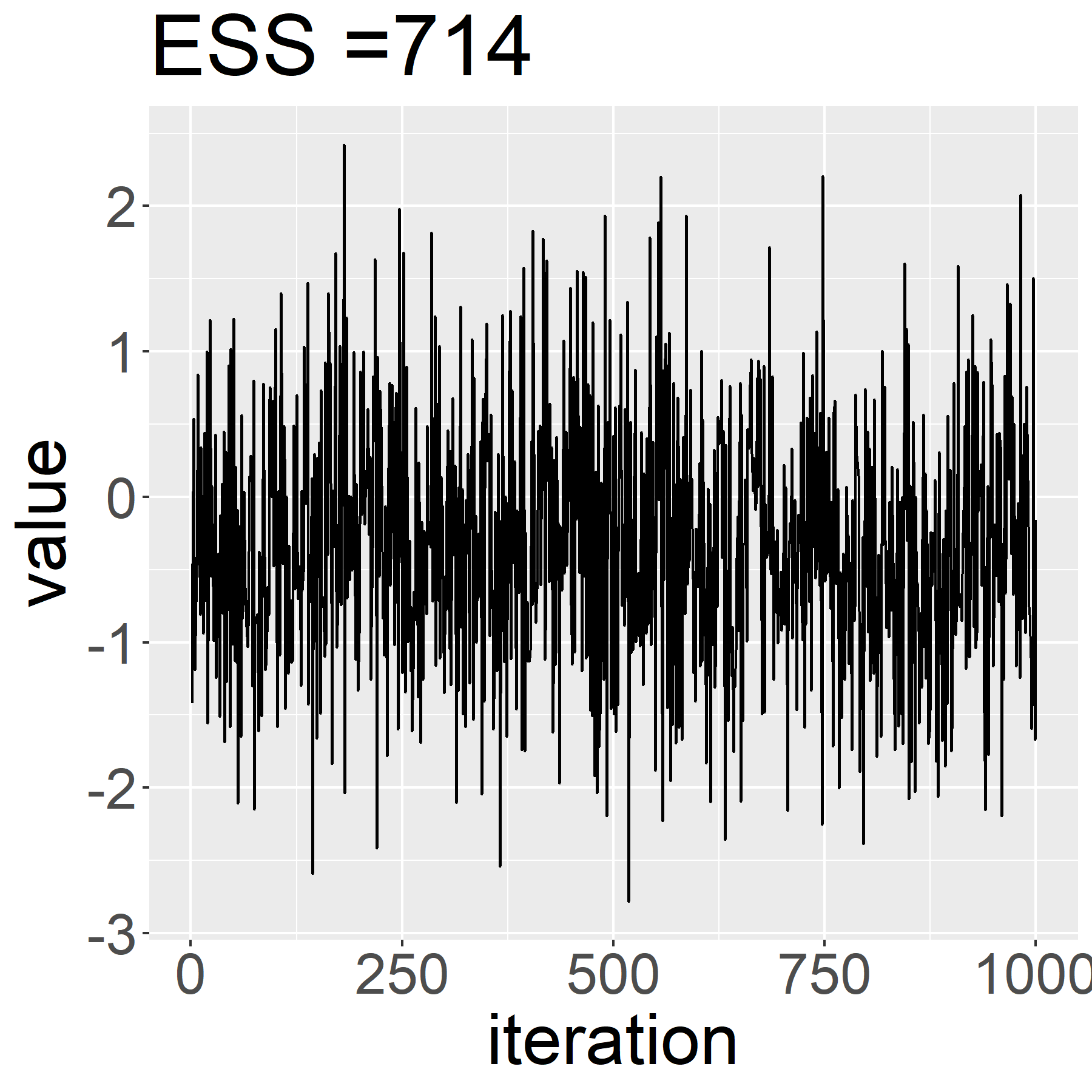} \\
$\xi_{421,1}$ & $\xi_{421,2}$ & $\xi_{421,3}$ & $\xi_{421,4}$ & $\xi_{421,5}$ \\
\end{tabular}
\caption{Trace plots related to the functional factor analysis results presented Figures 3 and 4 in the main paper.}\label{fig:trace_fpca} 
    \end{center}
\end{figure}

\begin{figure}[t]
\begin{center}
 \begin{tabular}{cccc}
\includegraphics[width = 1.0 in]{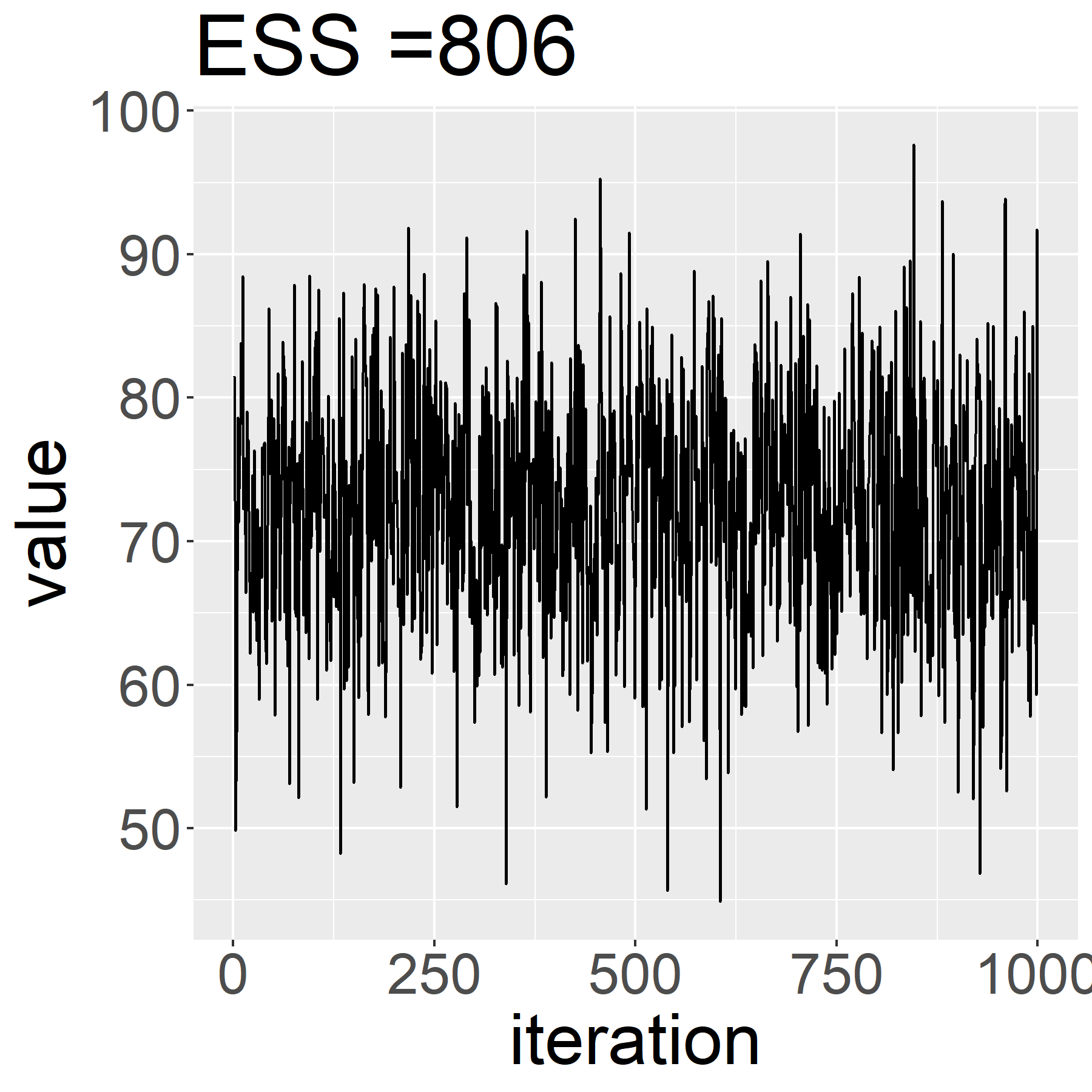} & \includegraphics[width = 1.0 in]{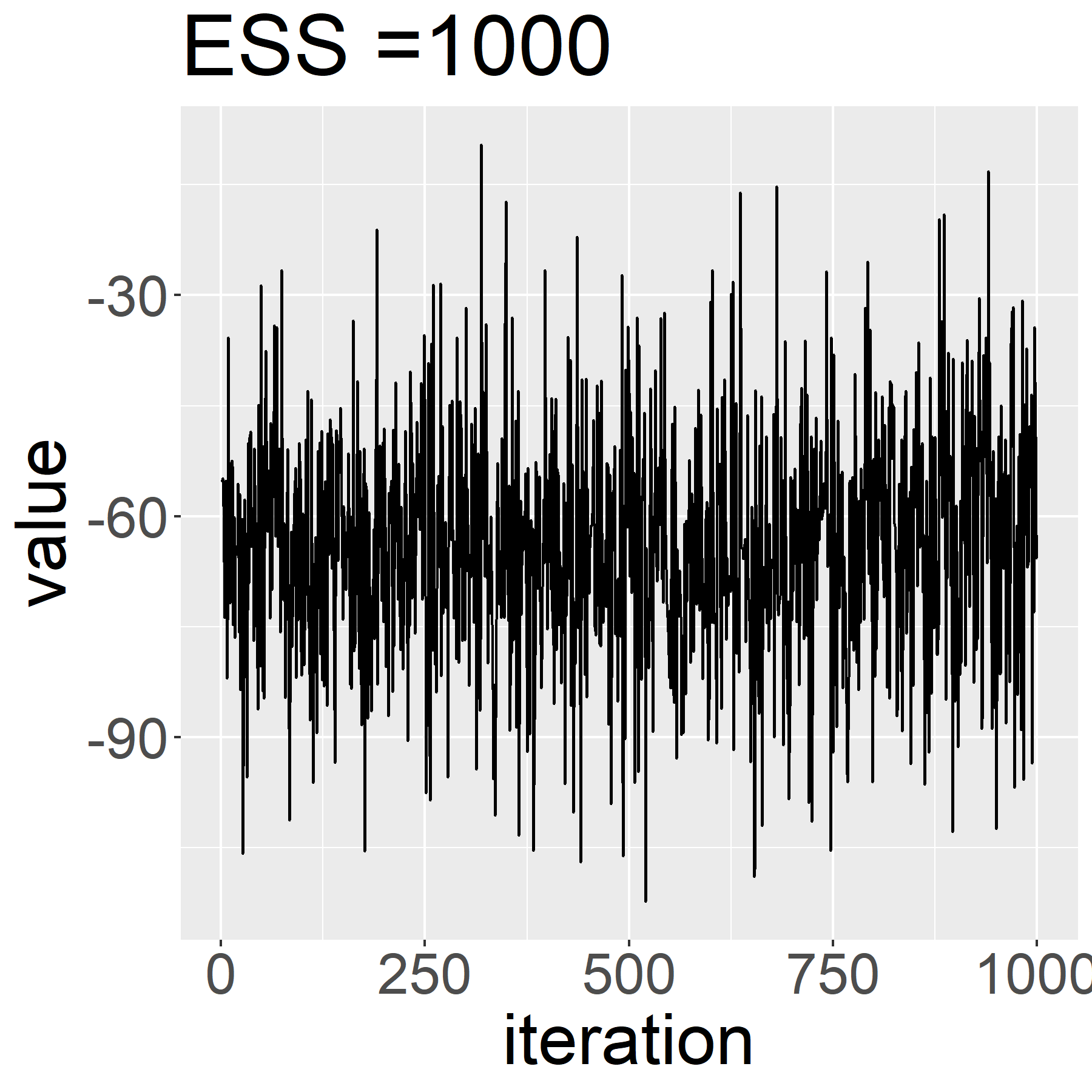} & \includegraphics[width = 1.0 in]{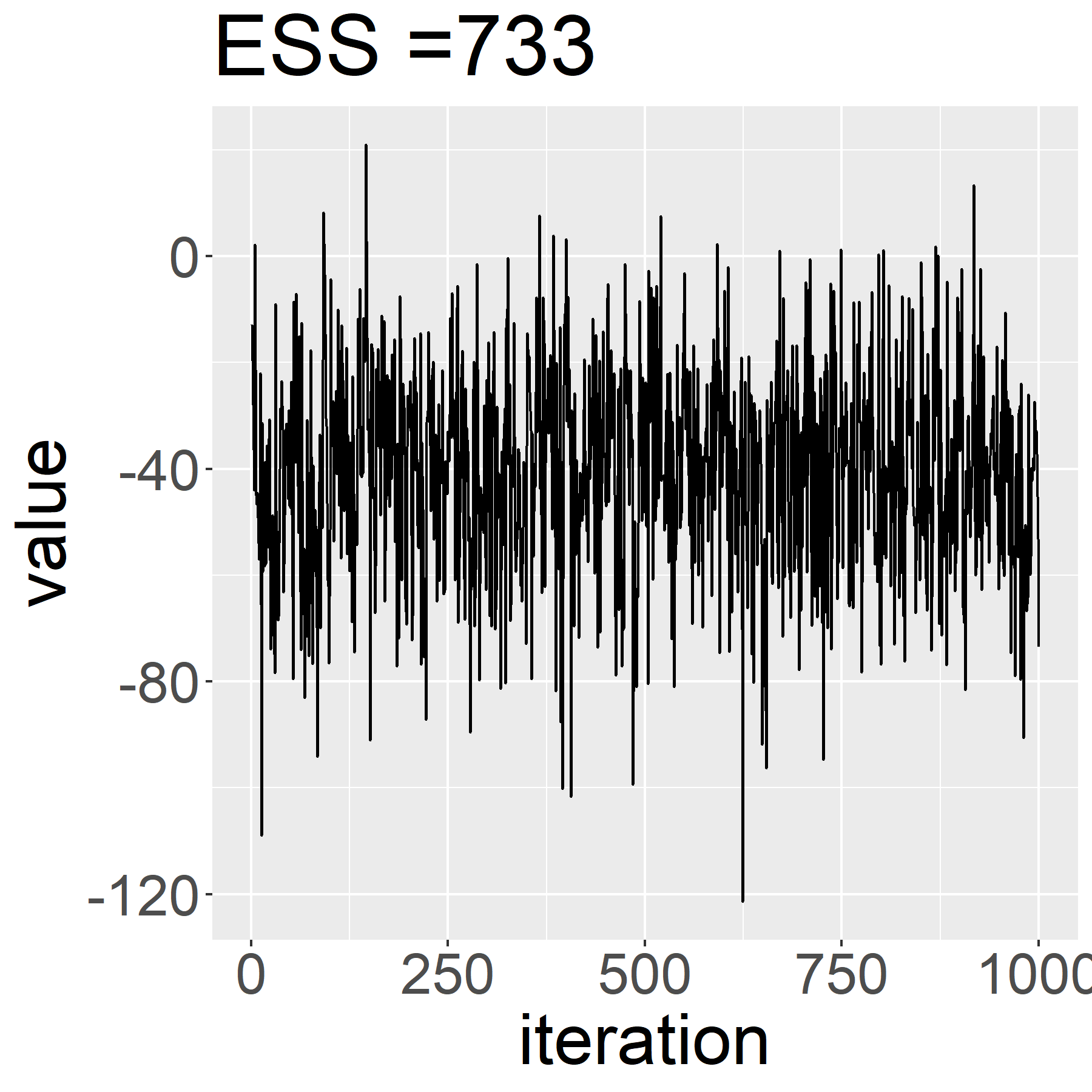}  & \includegraphics[width = 1.0 in]{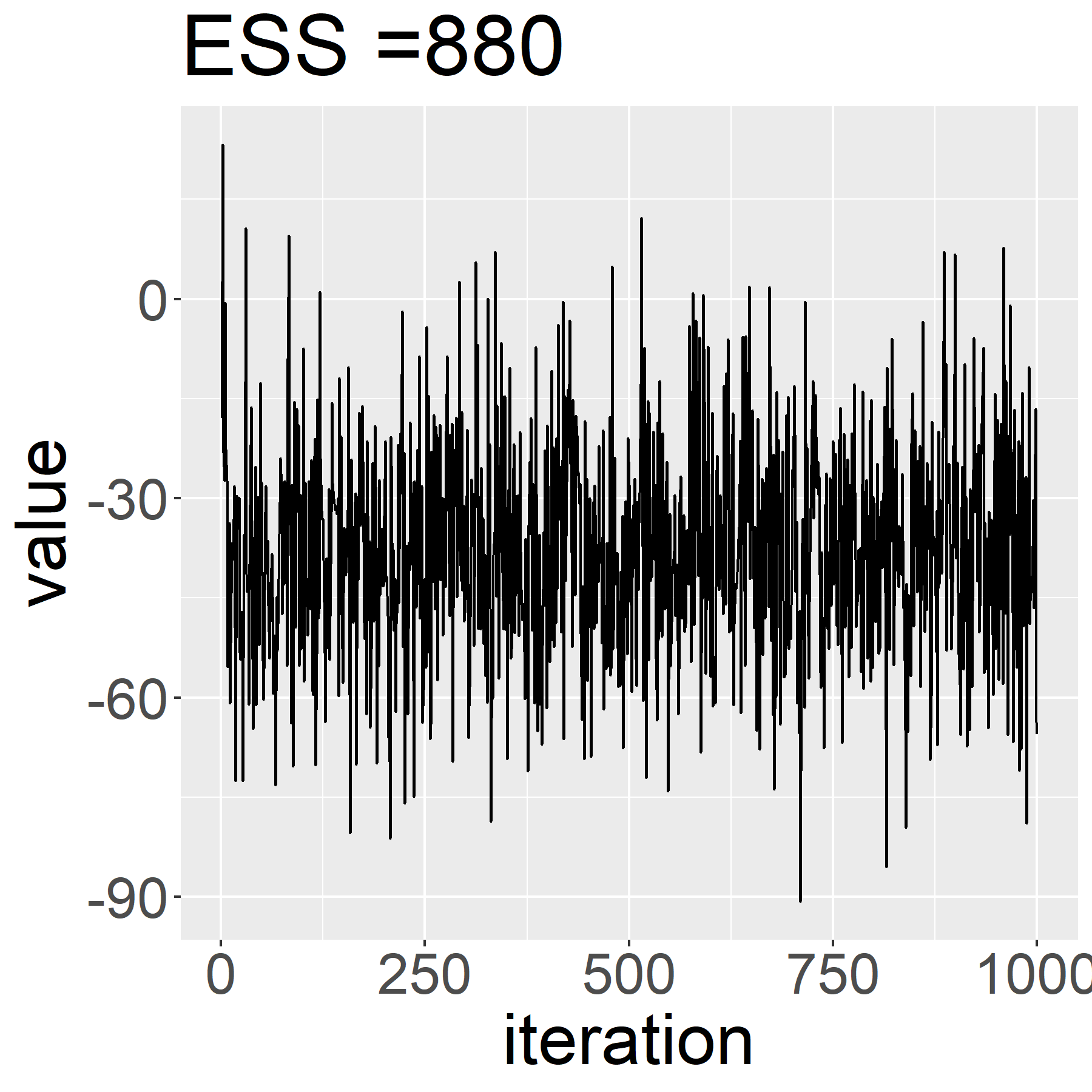} \\
\end{tabular}

$(\lambda_1(t_1),\ldots,\lambda_5(t_1))^\top \Theta_{q}$ for $q = 1,\ldots,4$ \\

 \begin{tabular}{cccc}
\includegraphics[width = 1.0 in]{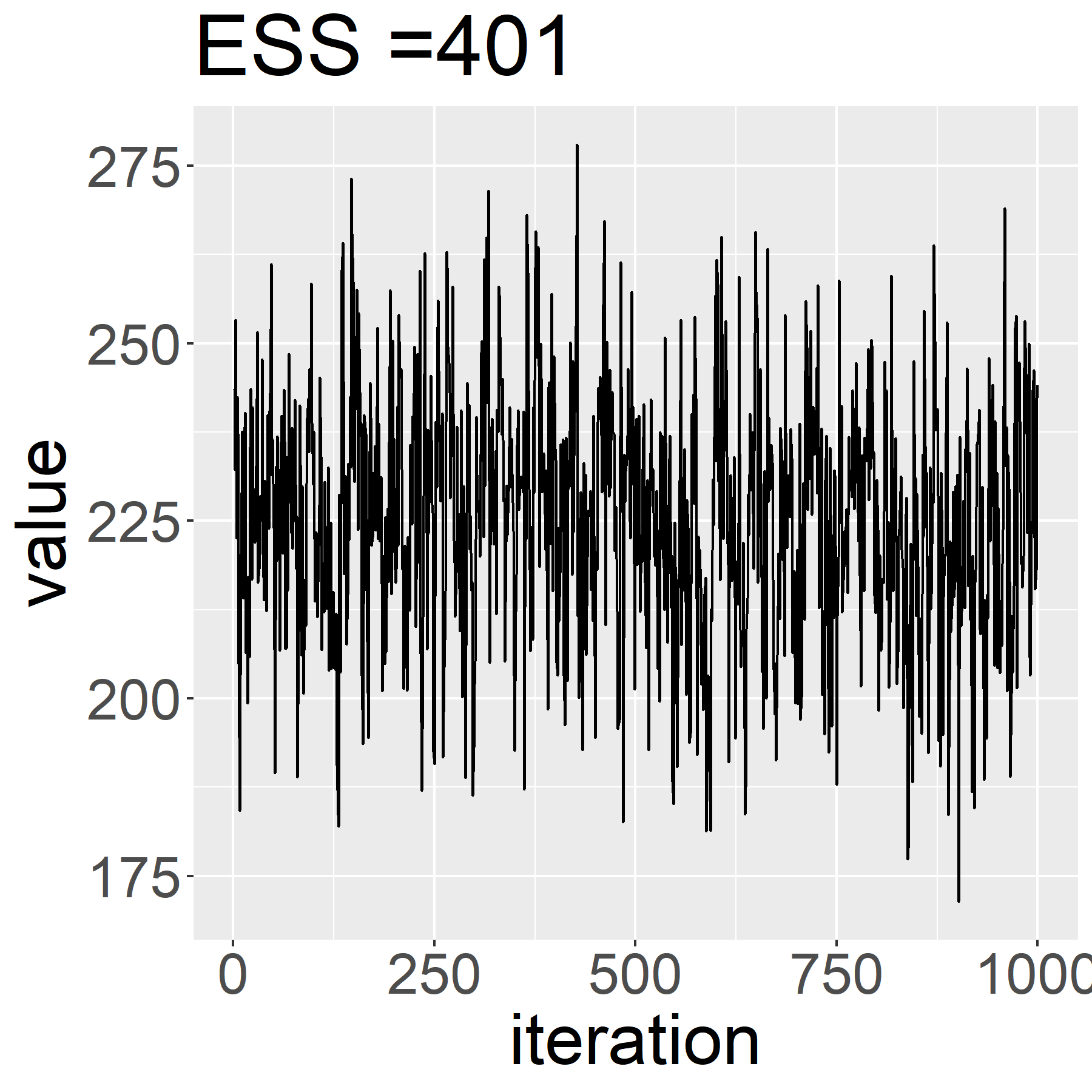} & \includegraphics[width = 1.0 in]{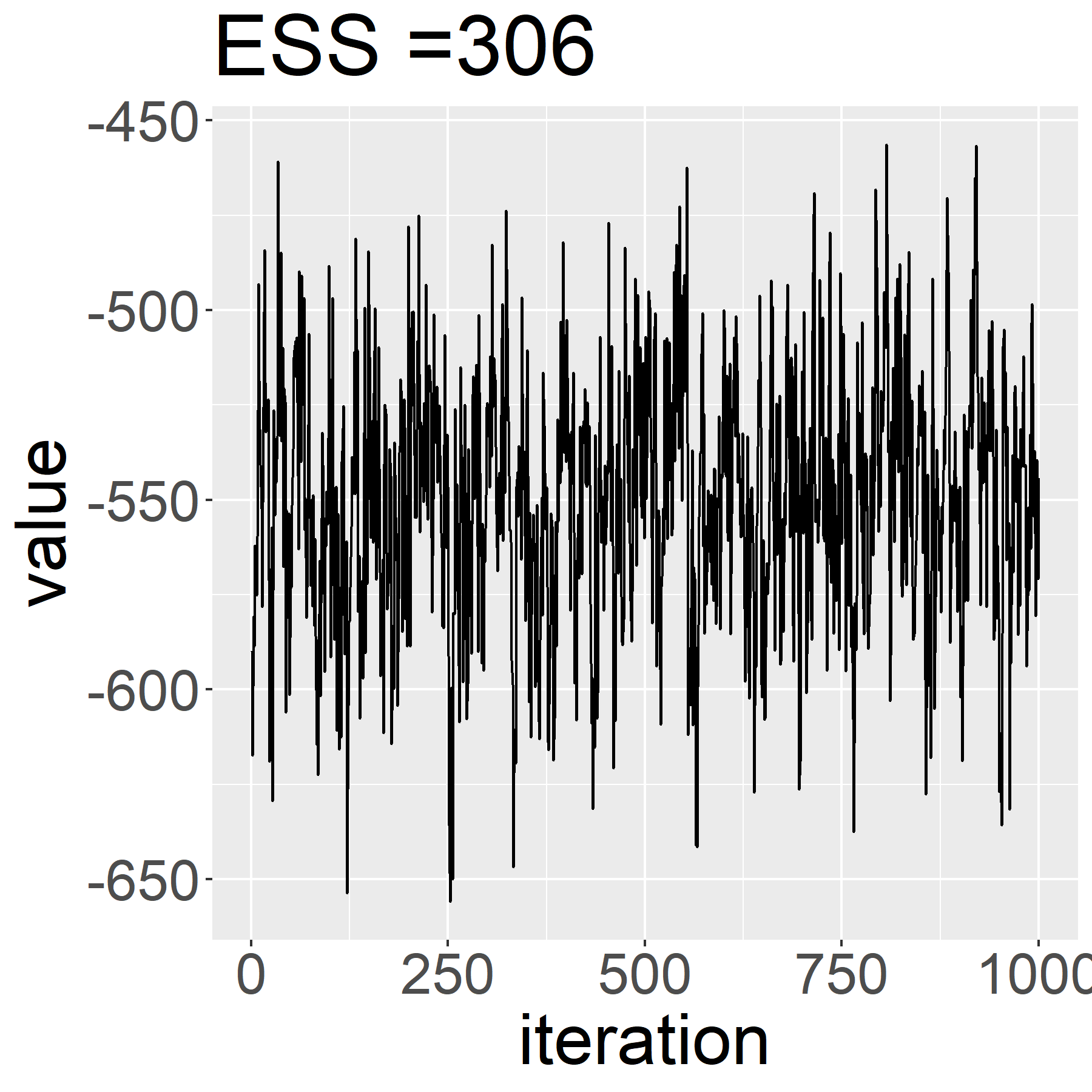} & \includegraphics[width = 1.0 in]{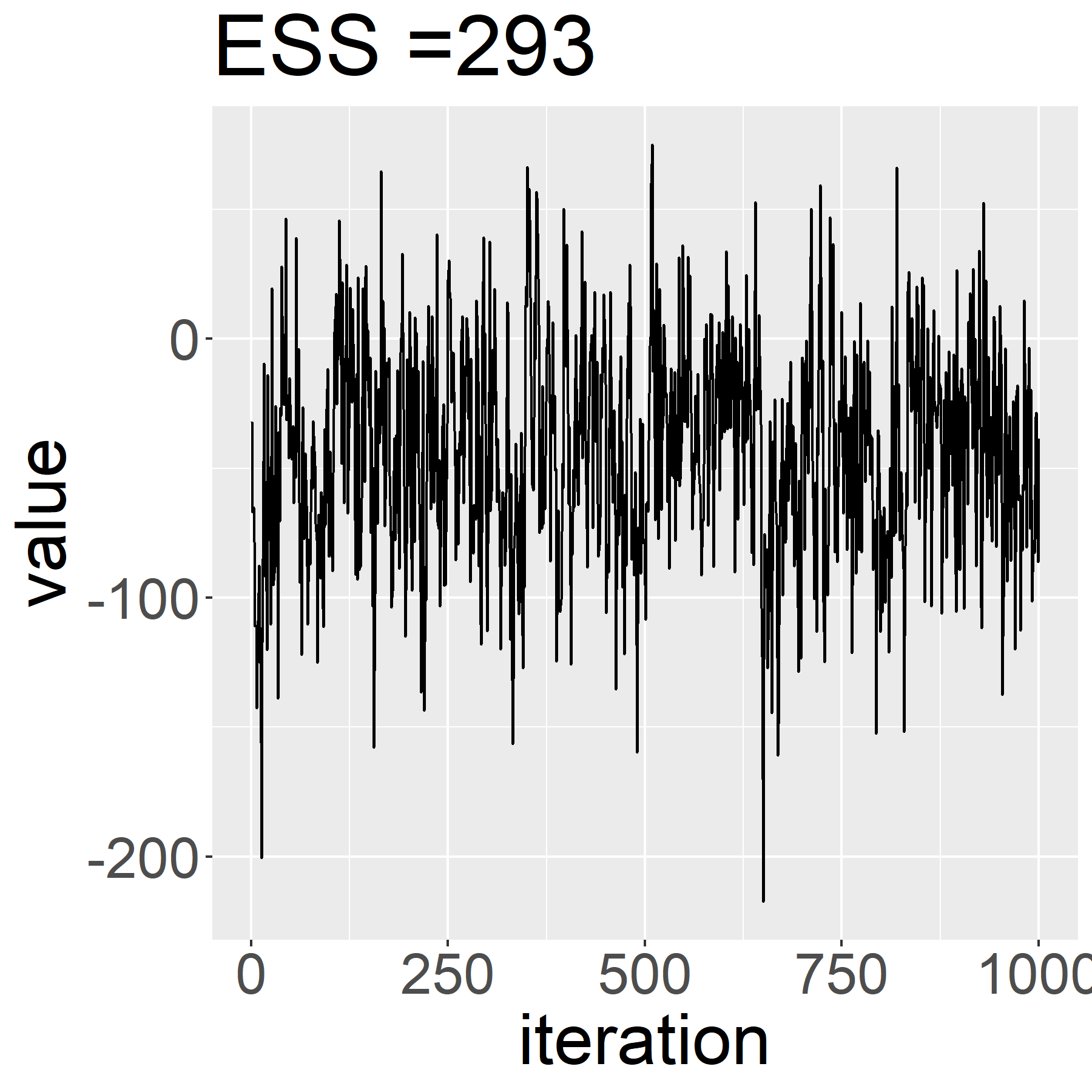}  & \includegraphics[width = 1.0 in]{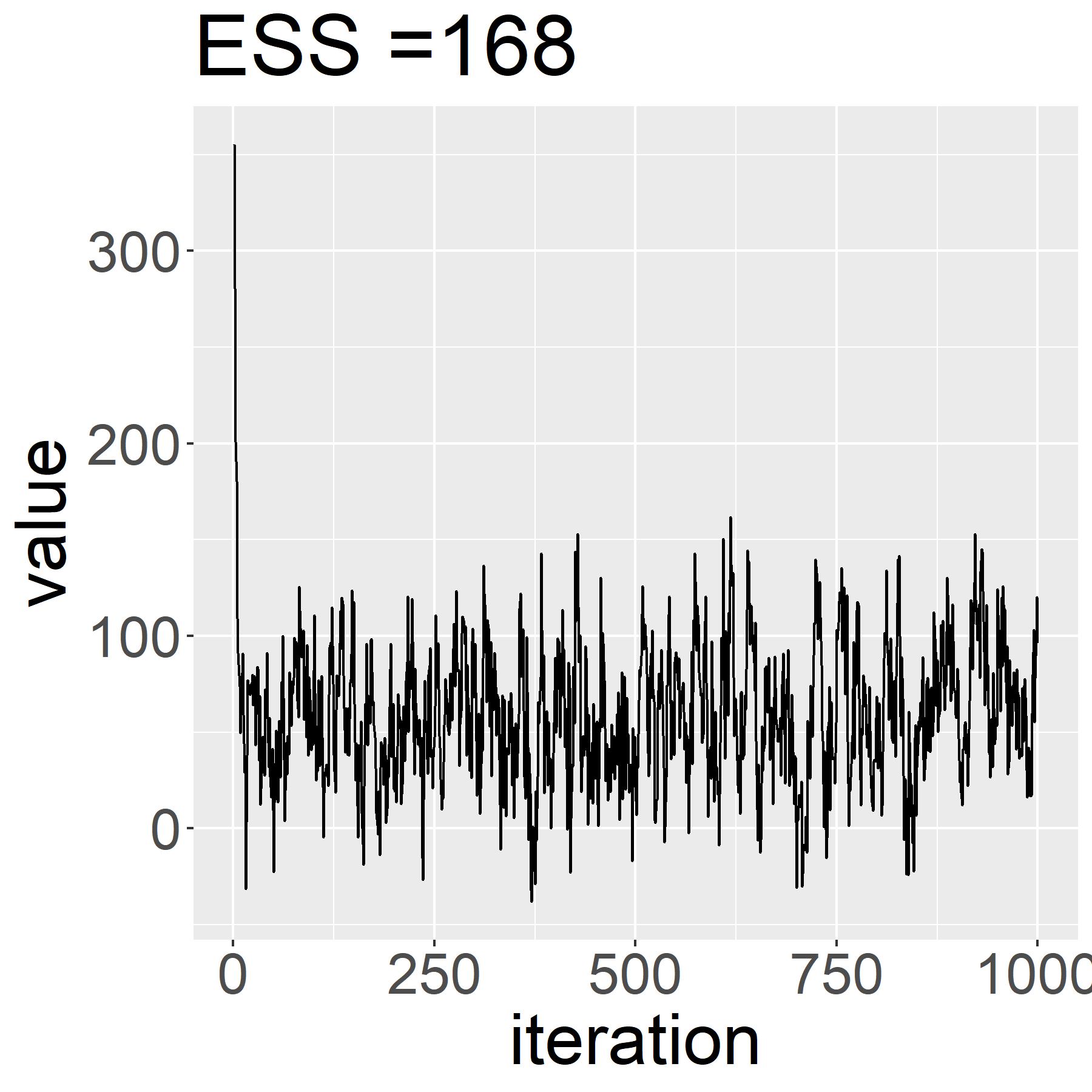} \\
\end{tabular}

$(\lambda_1(t_6),\ldots,\lambda_5(t_6))^\top \Theta_{q}$ for $q = 1,\ldots,4$ \\
 \begin{tabular}{cccc}
\includegraphics[width = 1.0 in]{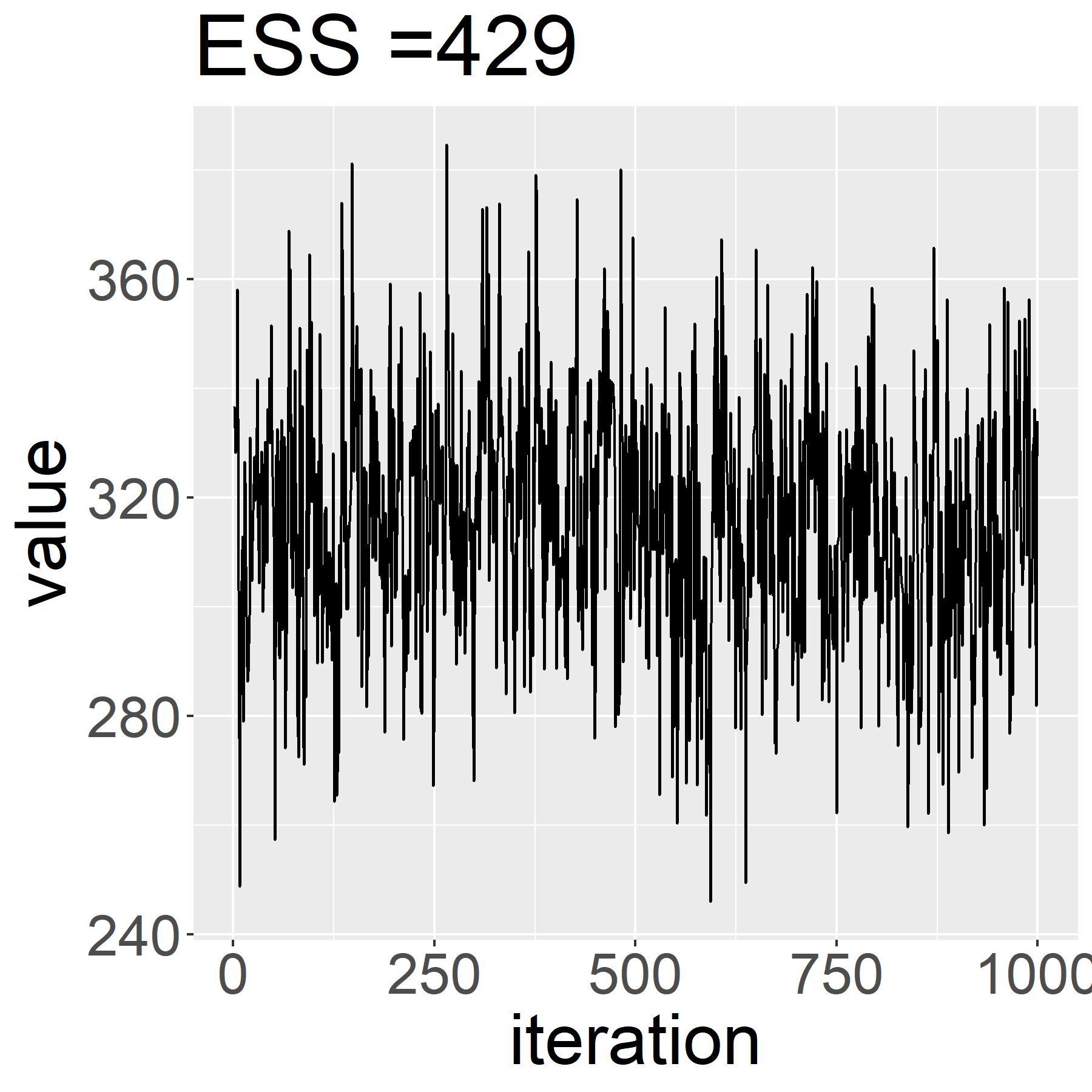} & \includegraphics[width = 1.0 in]{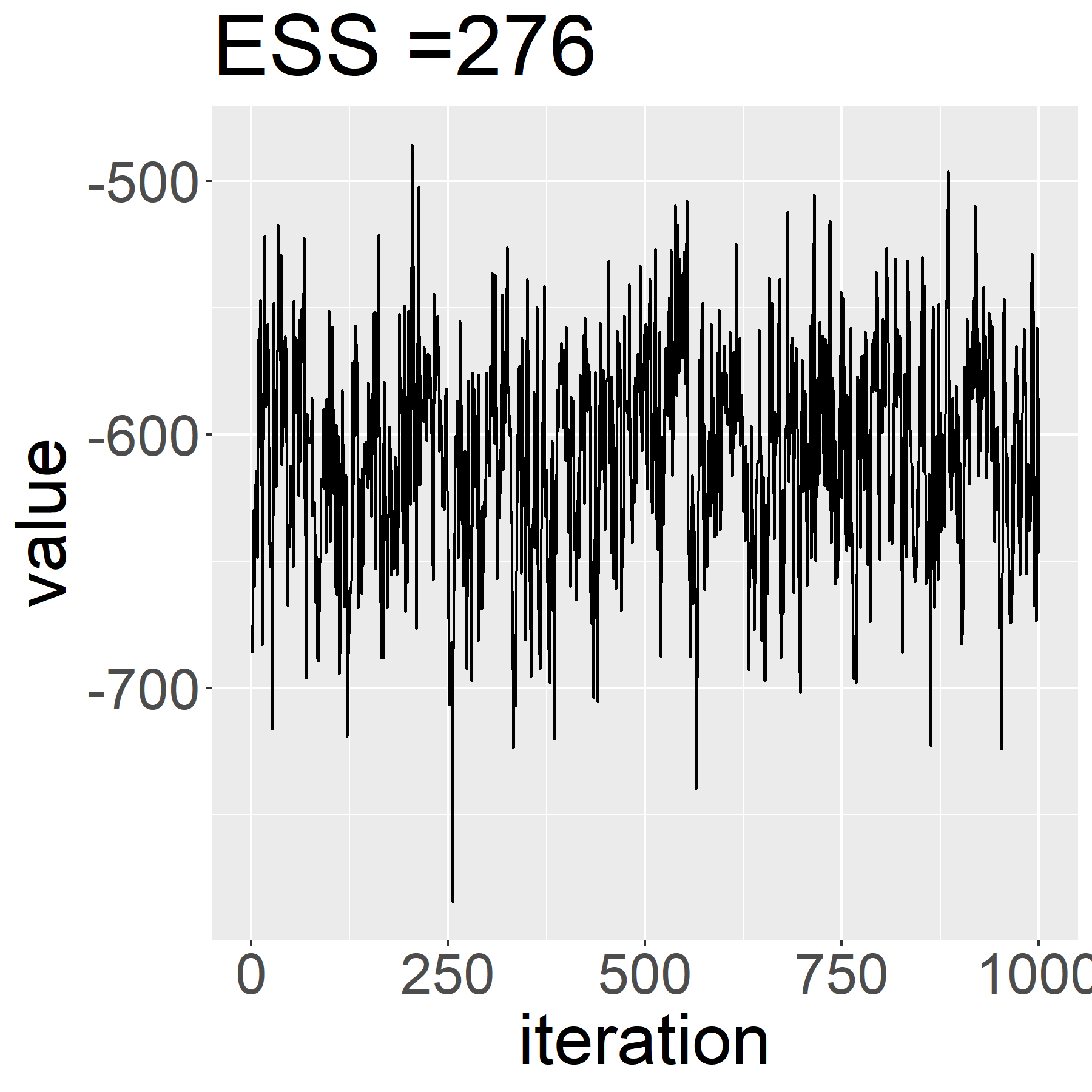} & \includegraphics[width = 1.0 in]{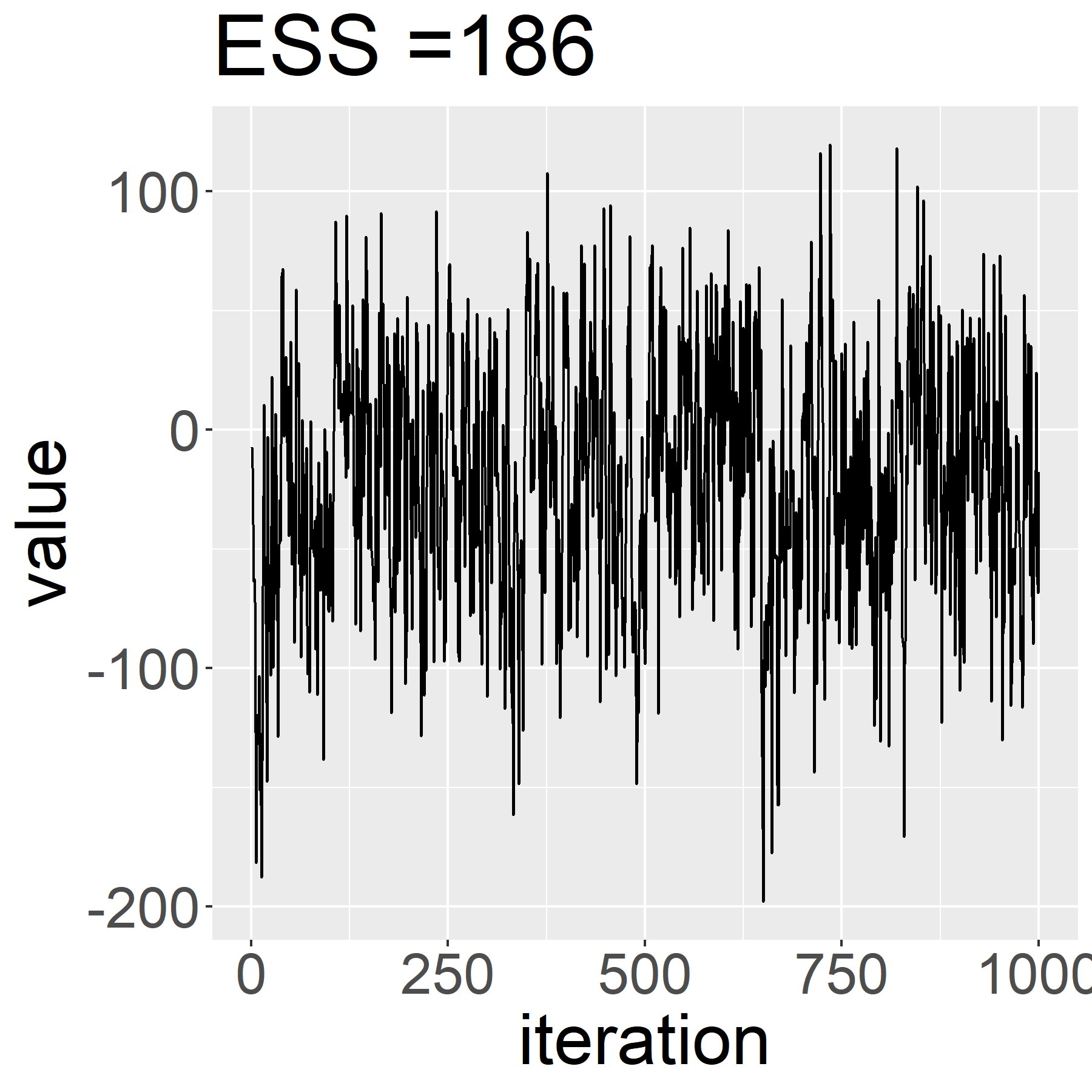}  & \includegraphics[width = 1.0 in]{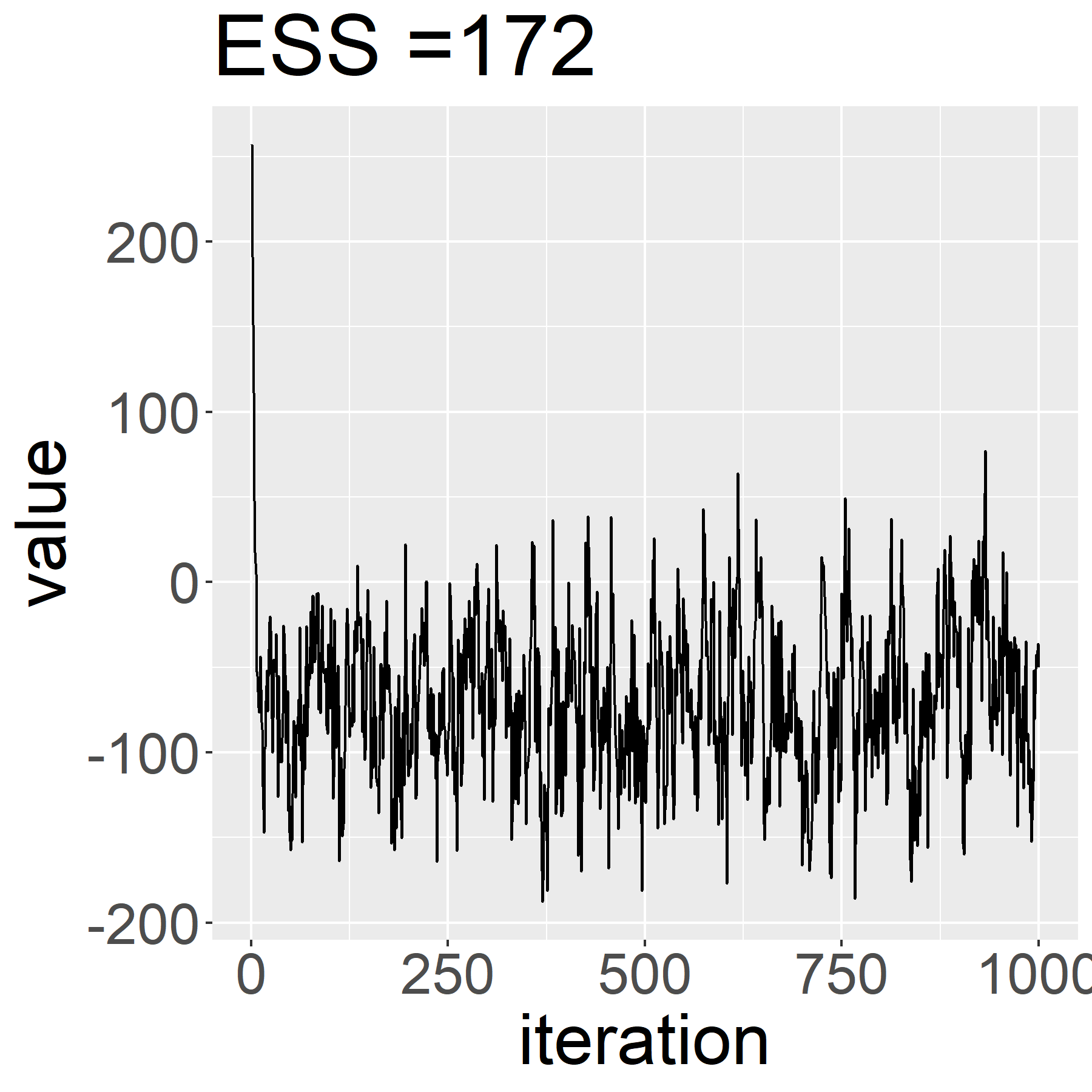} \\
\end{tabular}

$(\lambda_1(t_{13}),\ldots,\lambda_5(t_{13}))^\top \Theta_{q}$ for $q = 1,\ldots,4$ \\
\caption{Trace plots related to latent factor regression onto scalar covariates presented in Figure 5 in the main paper.}\label{fig:trace_fpca_reg}
    \end{center}
\end{figure}

\begin{figure}[t]
\begin{center}
 \begin{tabular}{ccc}
\includegraphics[width = 1.25 in]{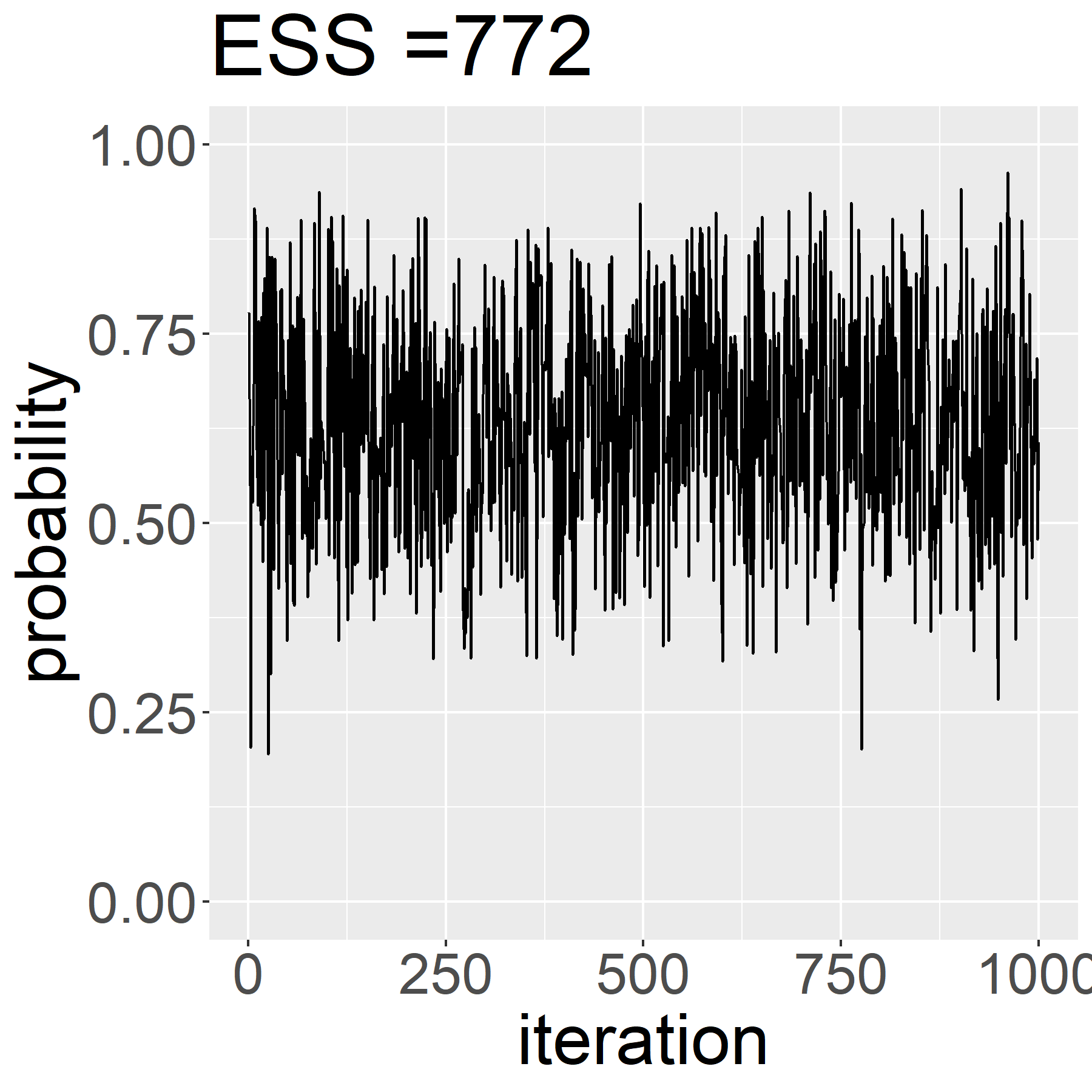} & \includegraphics[width = 1.25 in]{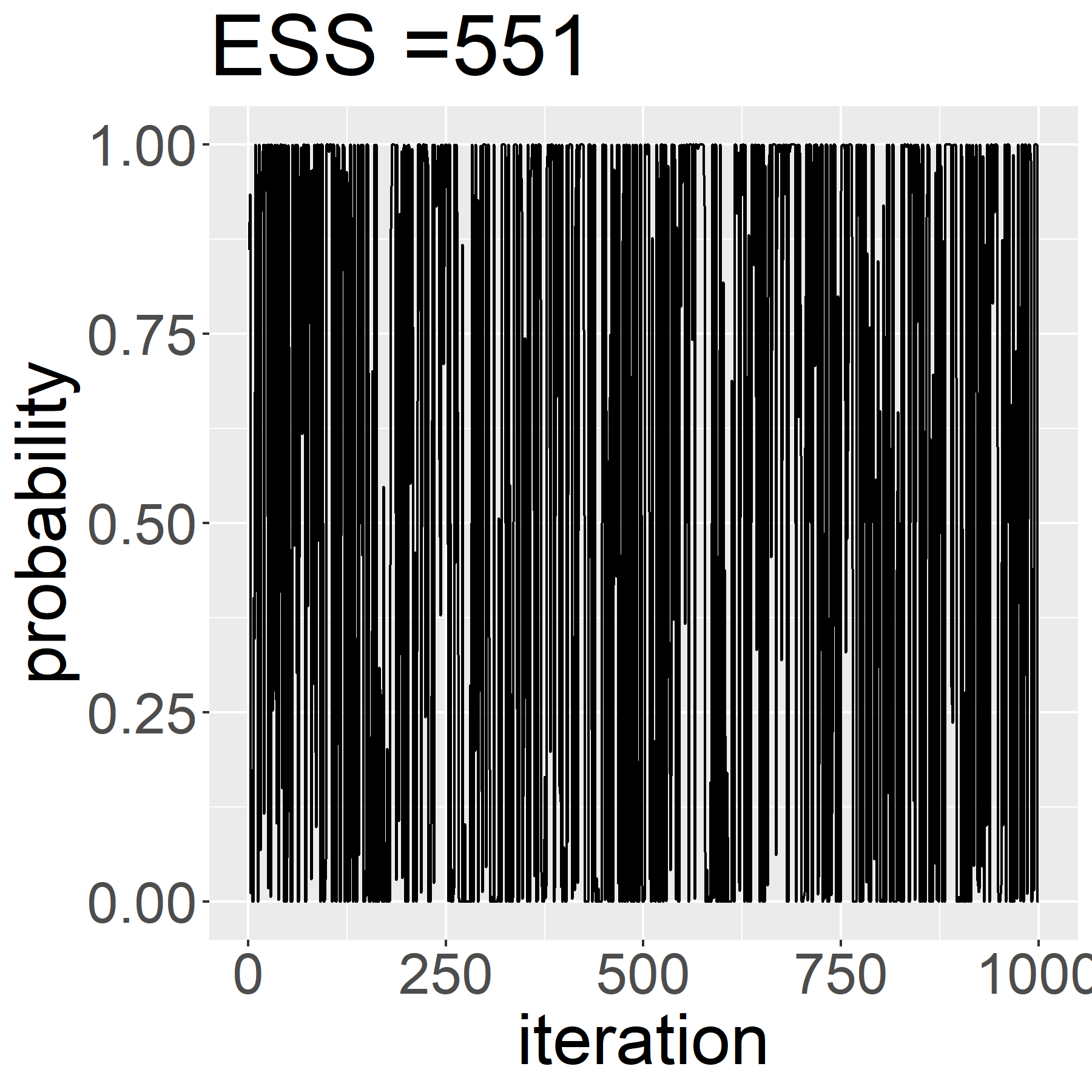} & \includegraphics[width = 1.25 in]{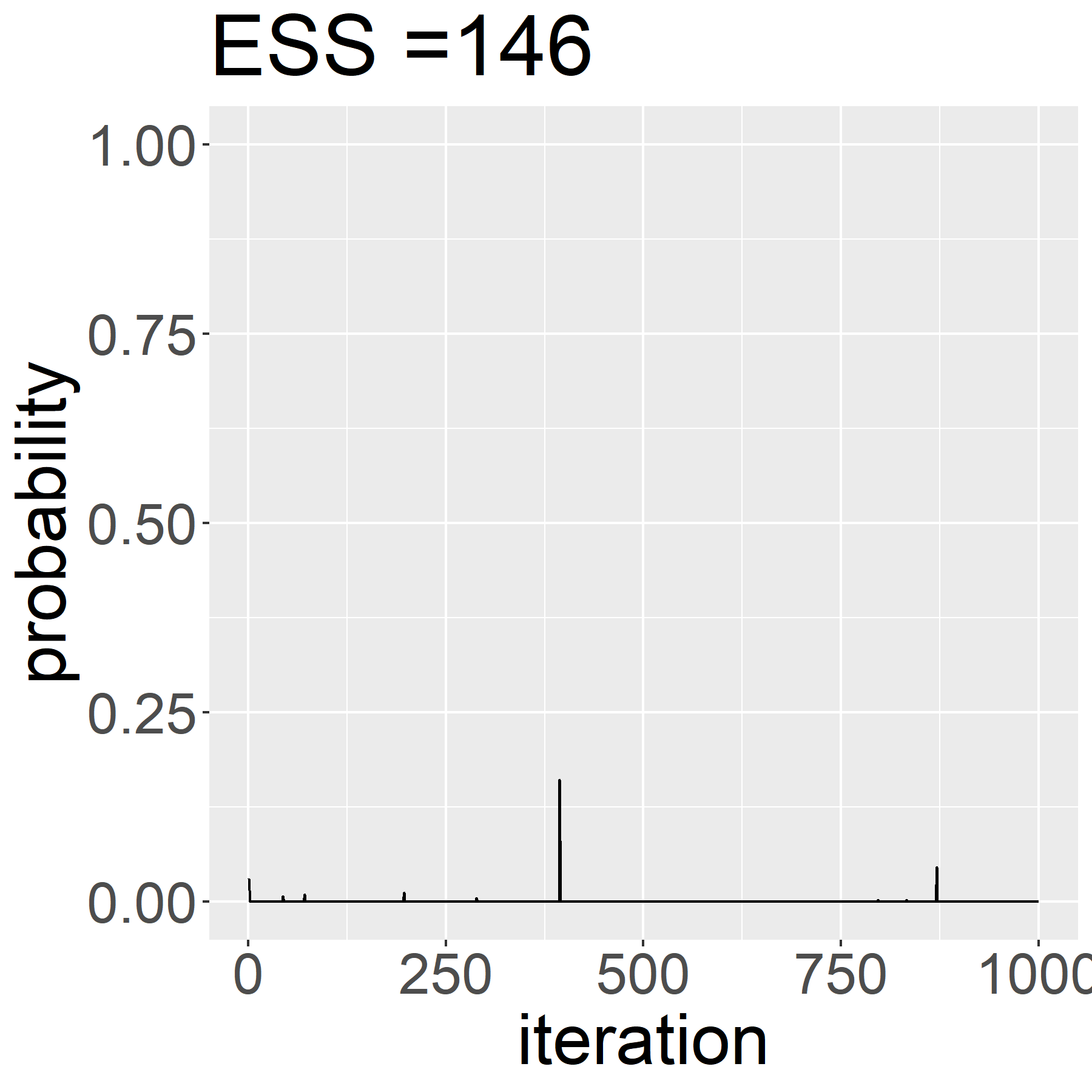} \\ \end{tabular}
$\Phi(\mu^{z_1}(t) + (\lambda^{z}_1(t),\lambda^{z}_2(t),\lambda^{z}_3(t))^\top\eta^z_{2087})$  for $t = 0,1,2$ \\
 \begin{tabular}{ccc}
\includegraphics[width = 1.25 in]{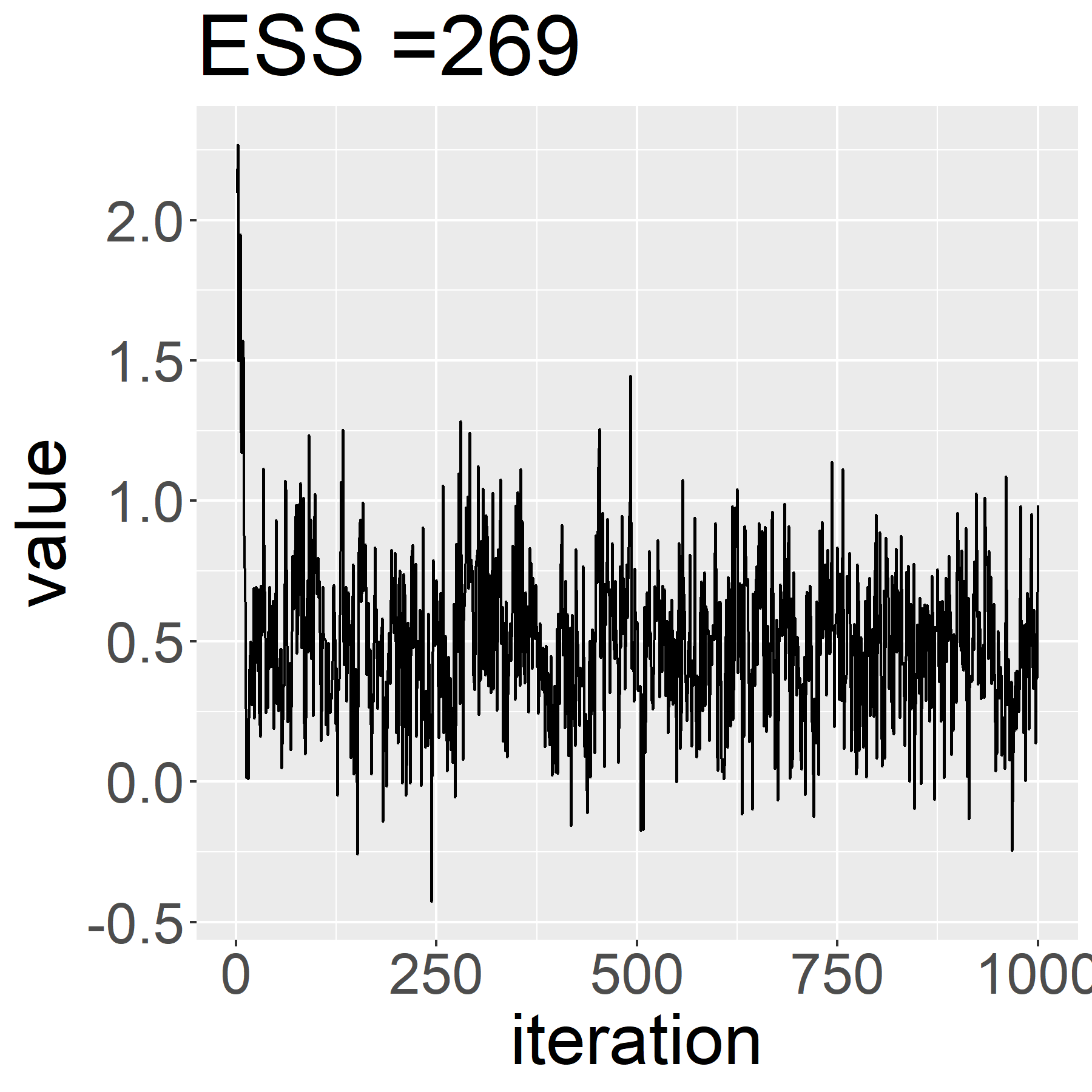} & \includegraphics[width = 1.25 in]{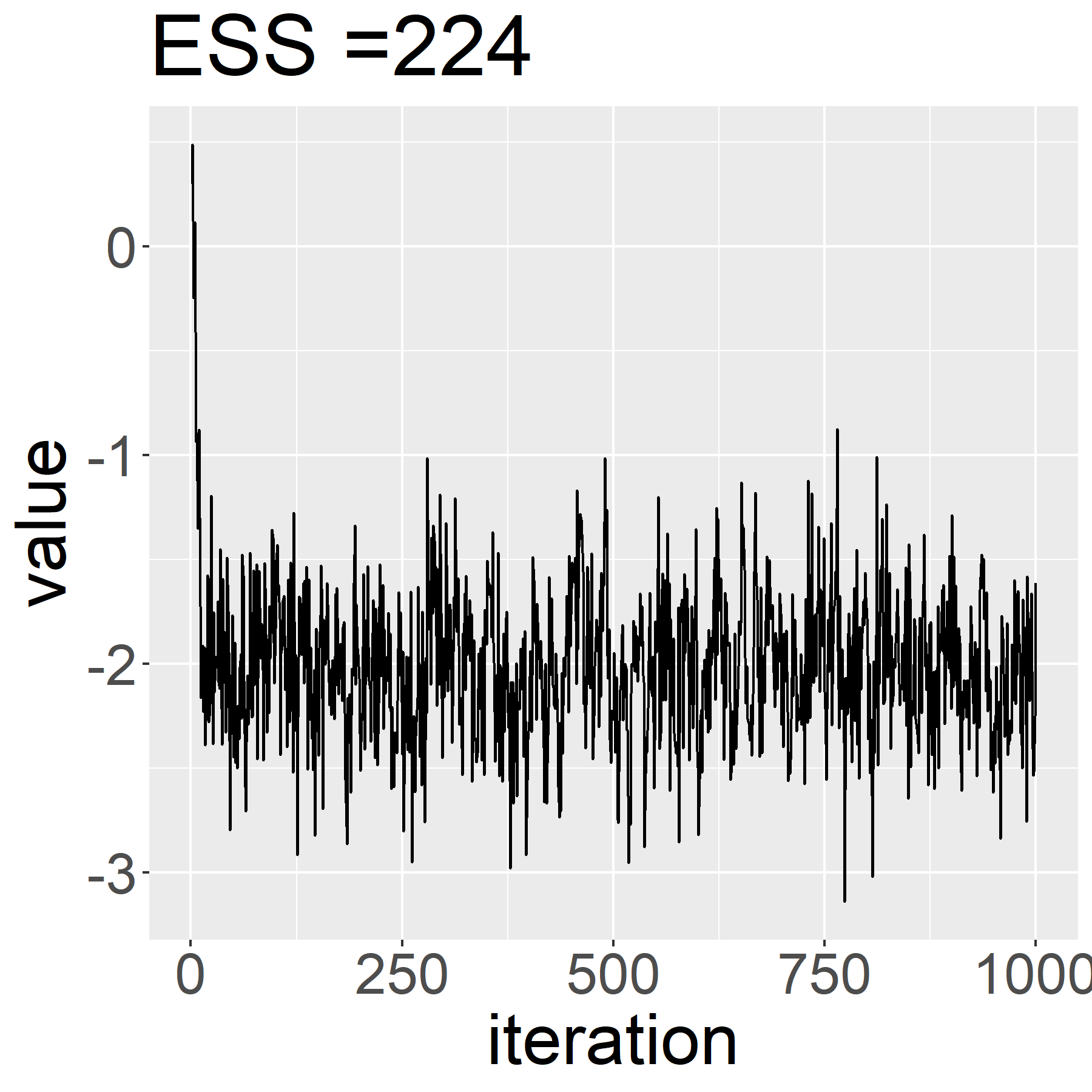} & \includegraphics[width = 1.25 in]{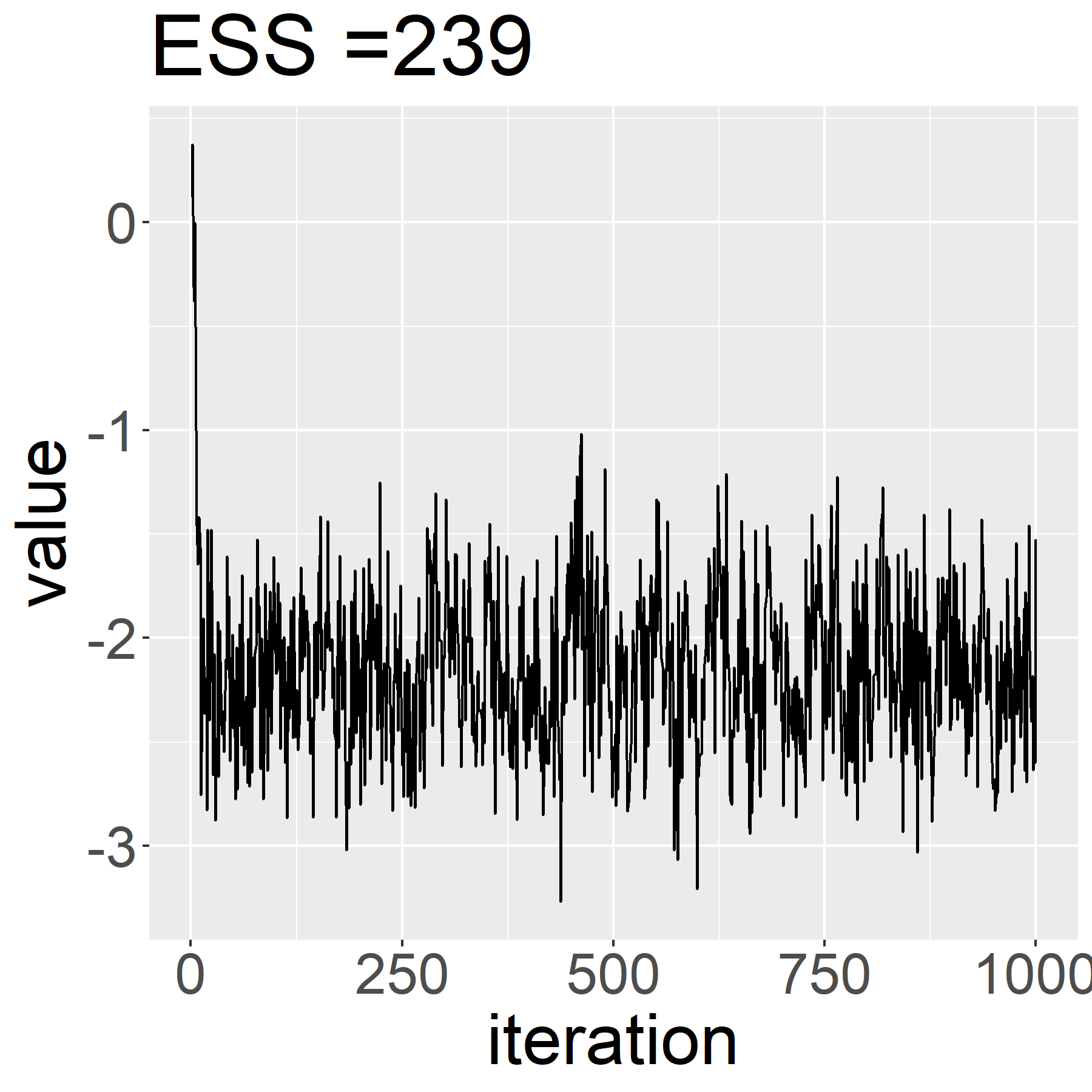} \\
$\rho_{1,1}$ & $\rho_{2,1}$ & $\rho_{3,1}$ \\
\includegraphics[width = 1.25 in]{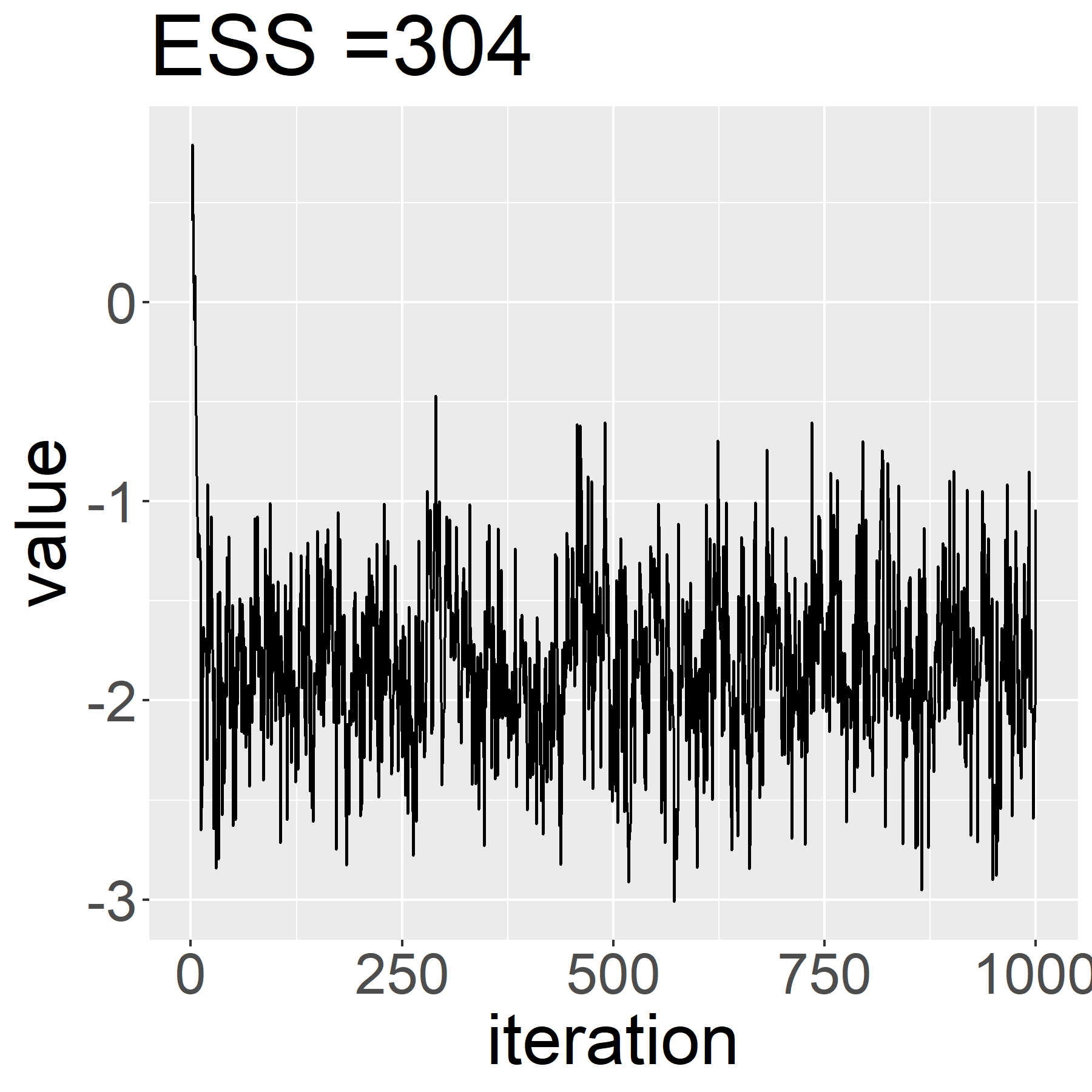} & \includegraphics[width = 1.25 in]{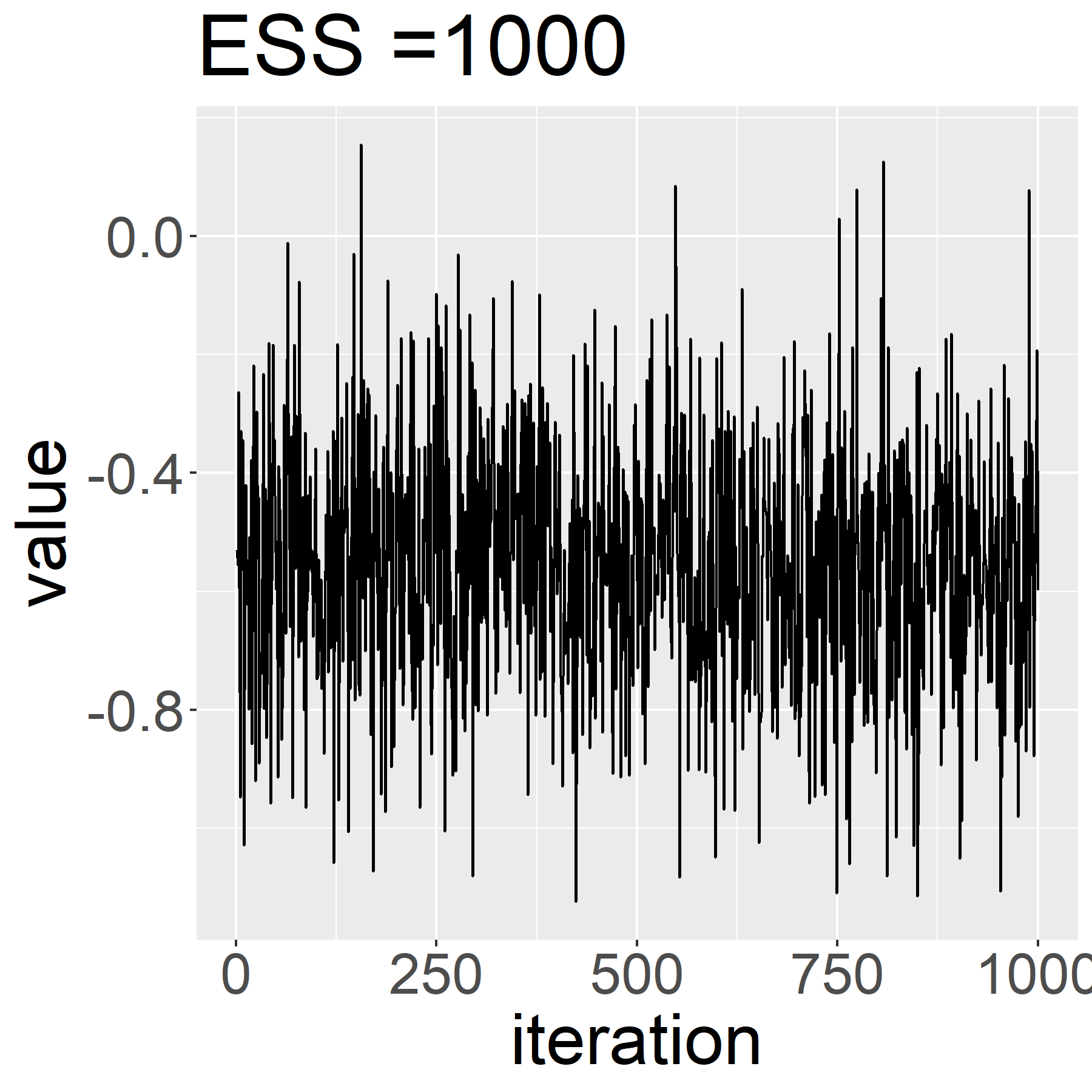} & \includegraphics[width = 1.25 in]{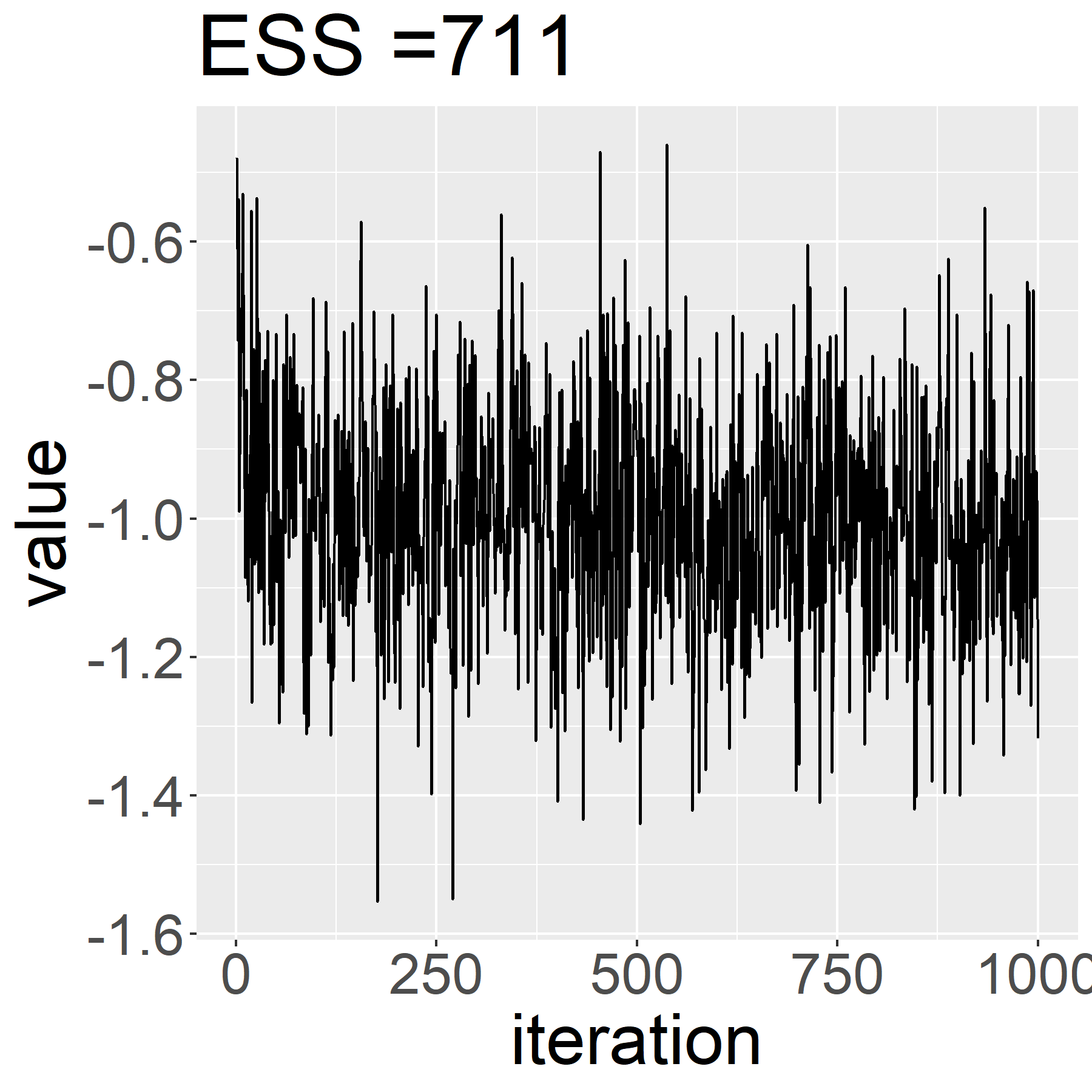} \\
$\rho_{4,1}$ & $\rho_{5,1}$ & $\rho_{6,1}$ \\
\end{tabular}
 \begin{tabular}{cccc}
\includegraphics[width = 1.25 in]{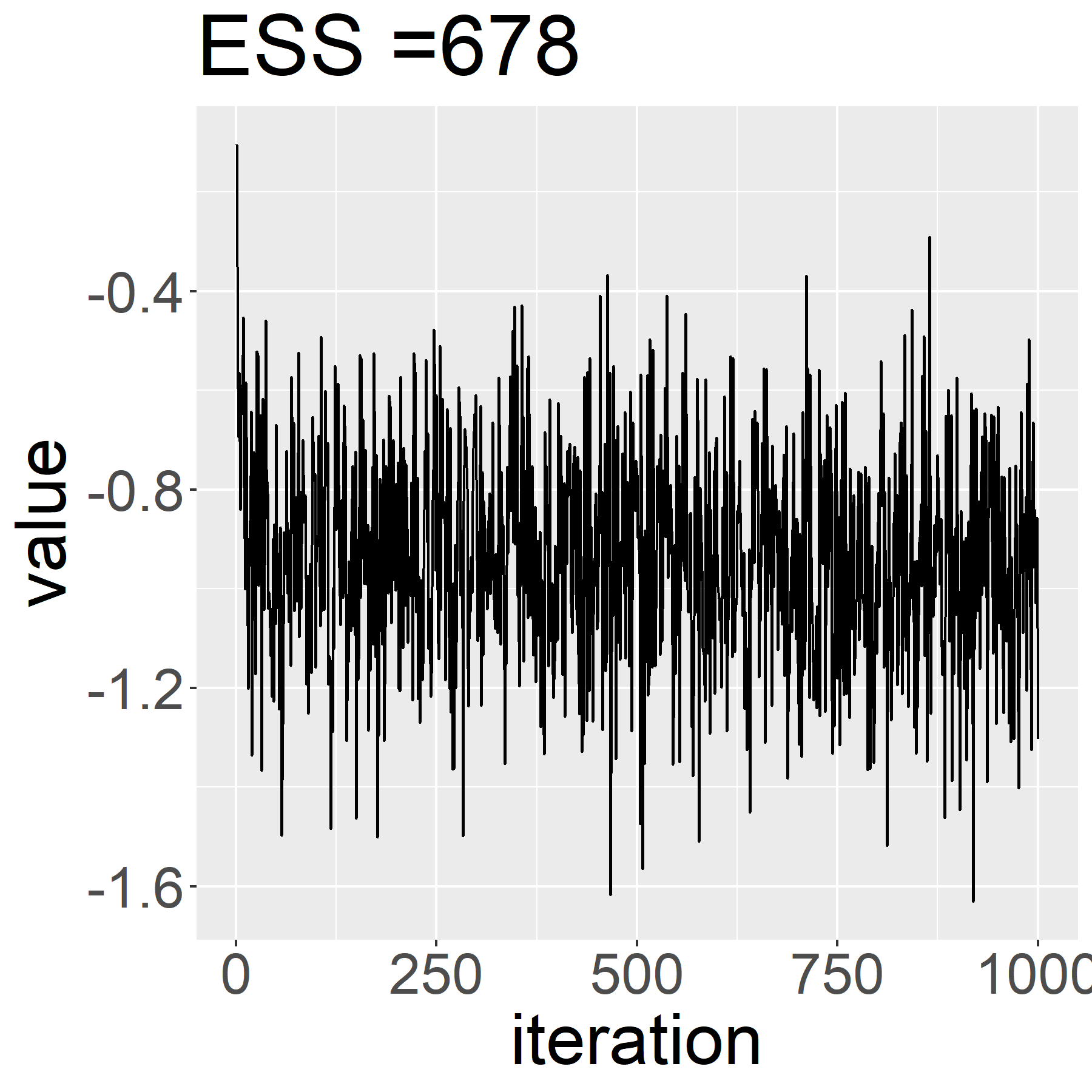} & \includegraphics[width = 1.25 in]{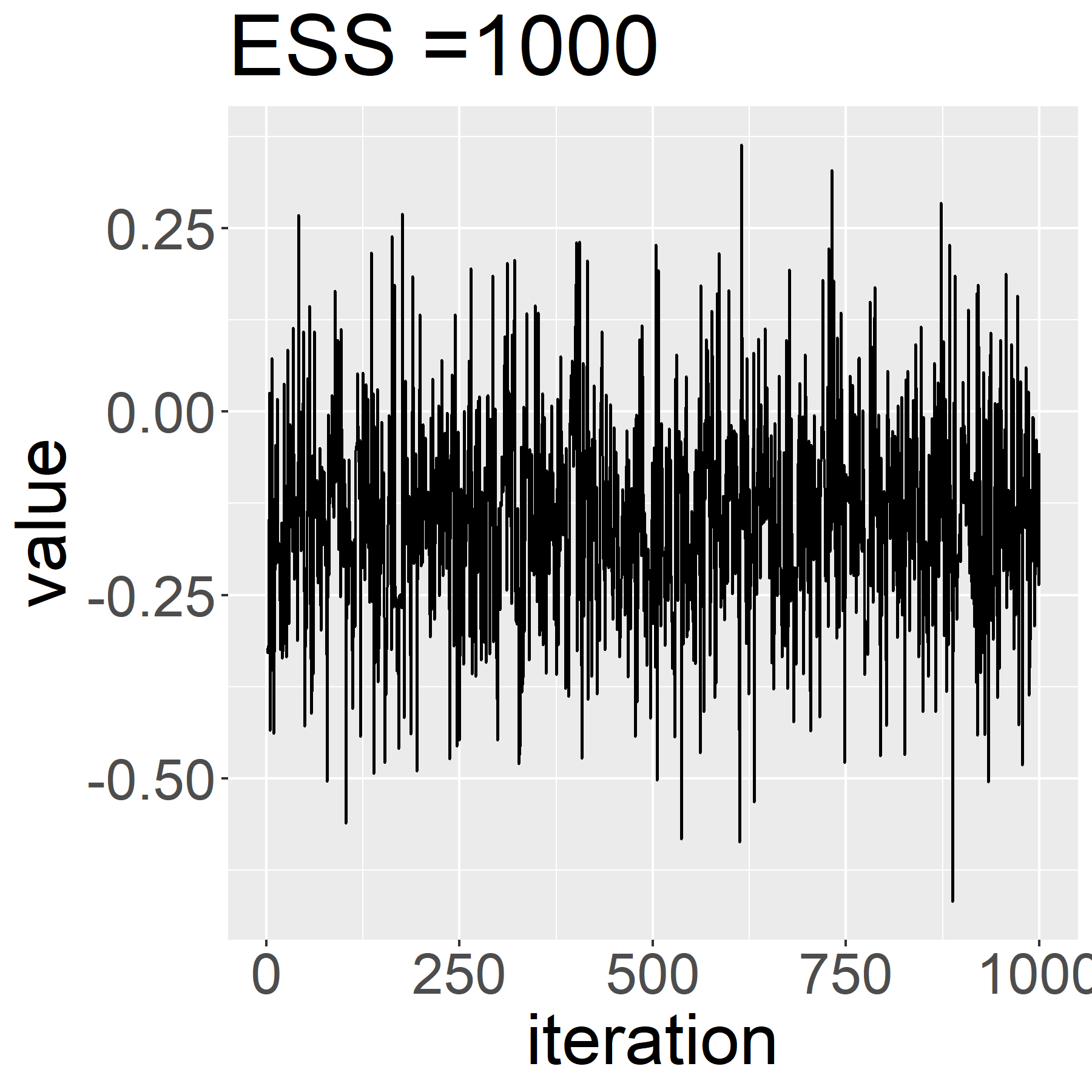} & \includegraphics[width = 1.25 in]{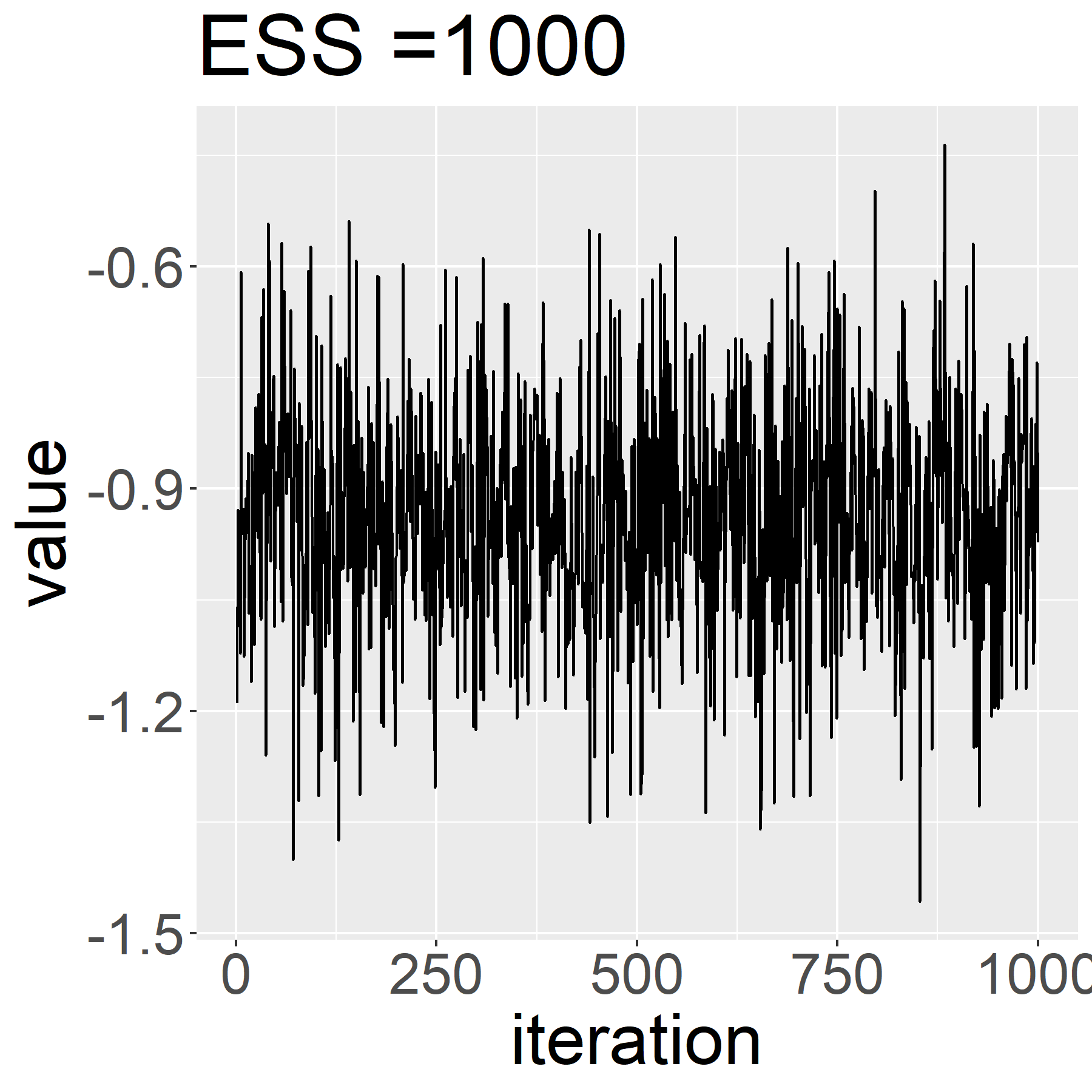} & \includegraphics[width = 1.25 in]{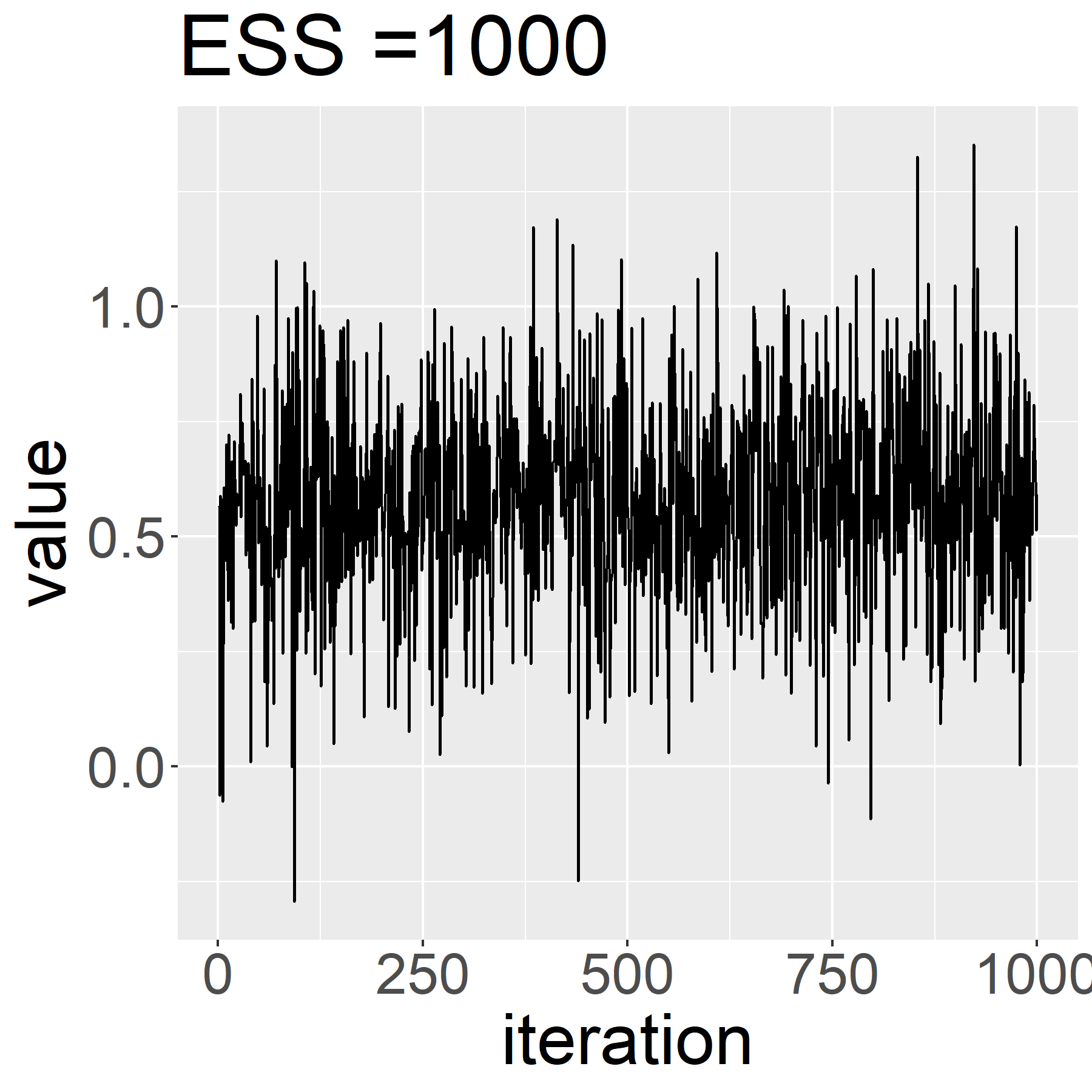} \\
$\rho_{7,1}$ & $\rho_{8,1}$ & $\rho_{9,1}$& $\rho_{10,1}$ \\
\end{tabular}
\caption{Trace plots related to breastfeeding status interpolation and historical regression presented in Figures 6 and 7 in the main paper.}\label{fig:trace_bf}
    \end{center}
\end{figure}

\begin{figure}[t]
\begin{center}
 \begin{tabular}{ccc}
\includegraphics[width = 1.0 in]{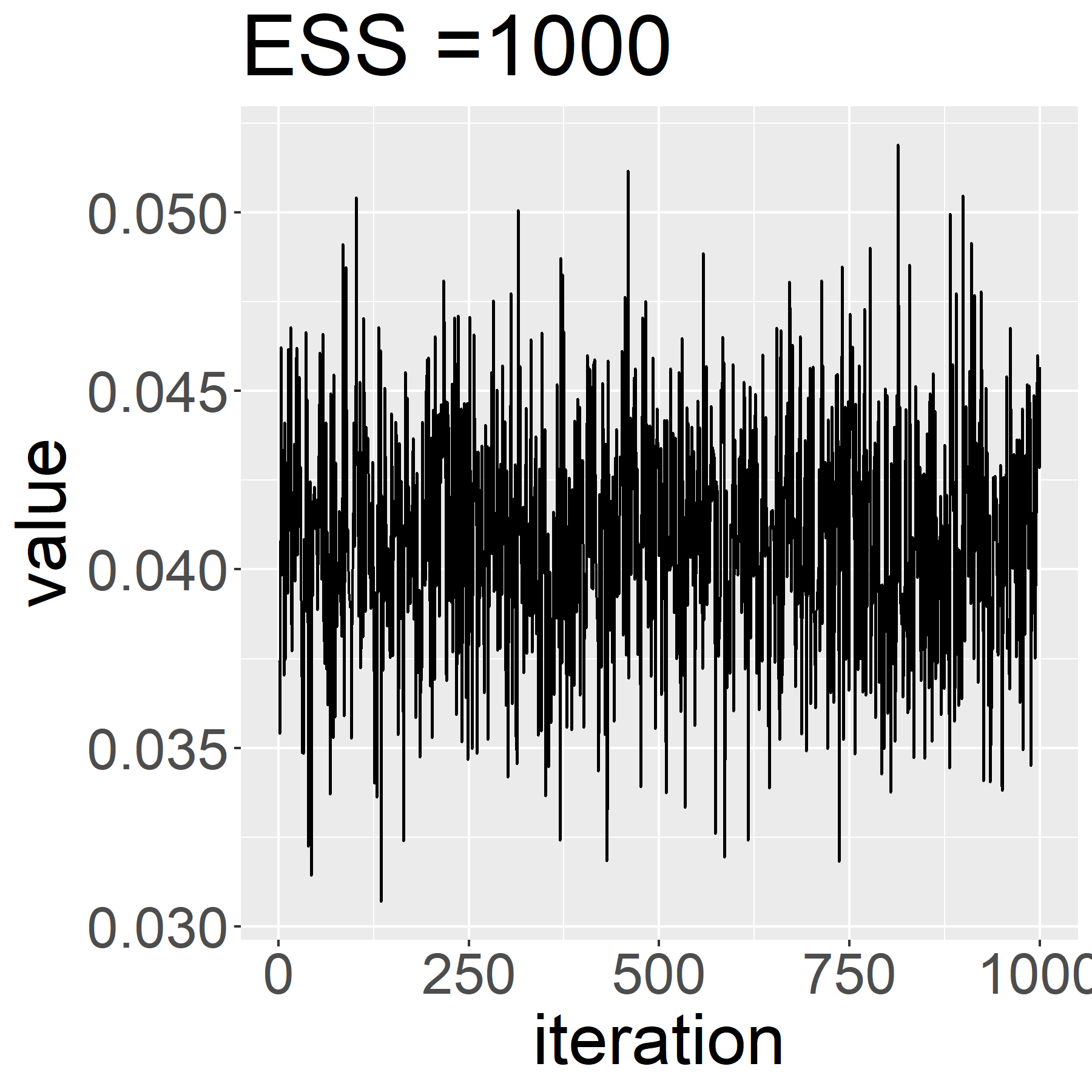} & \includegraphics[width = 1.0 in]{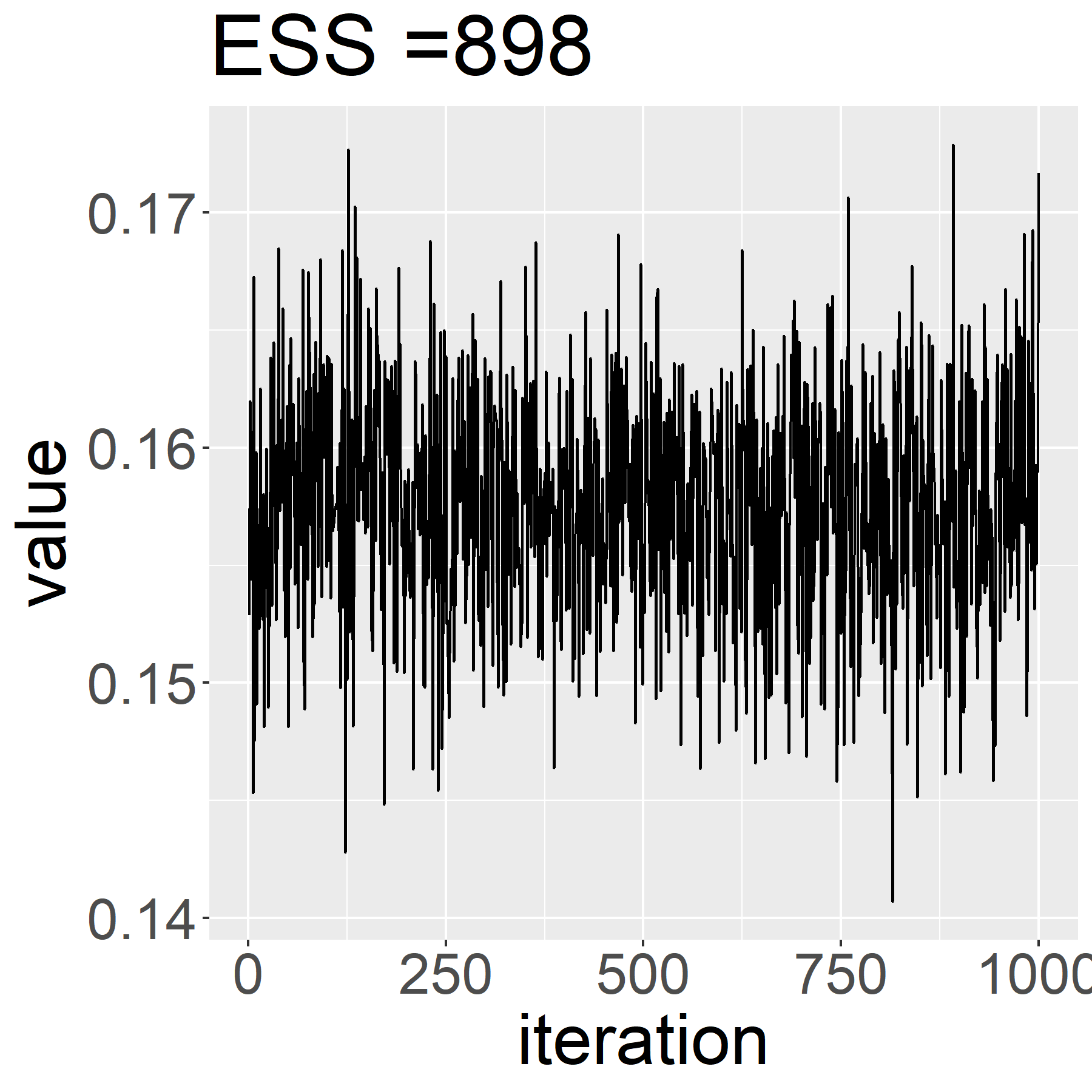} & \includegraphics[width = 1.0 in]{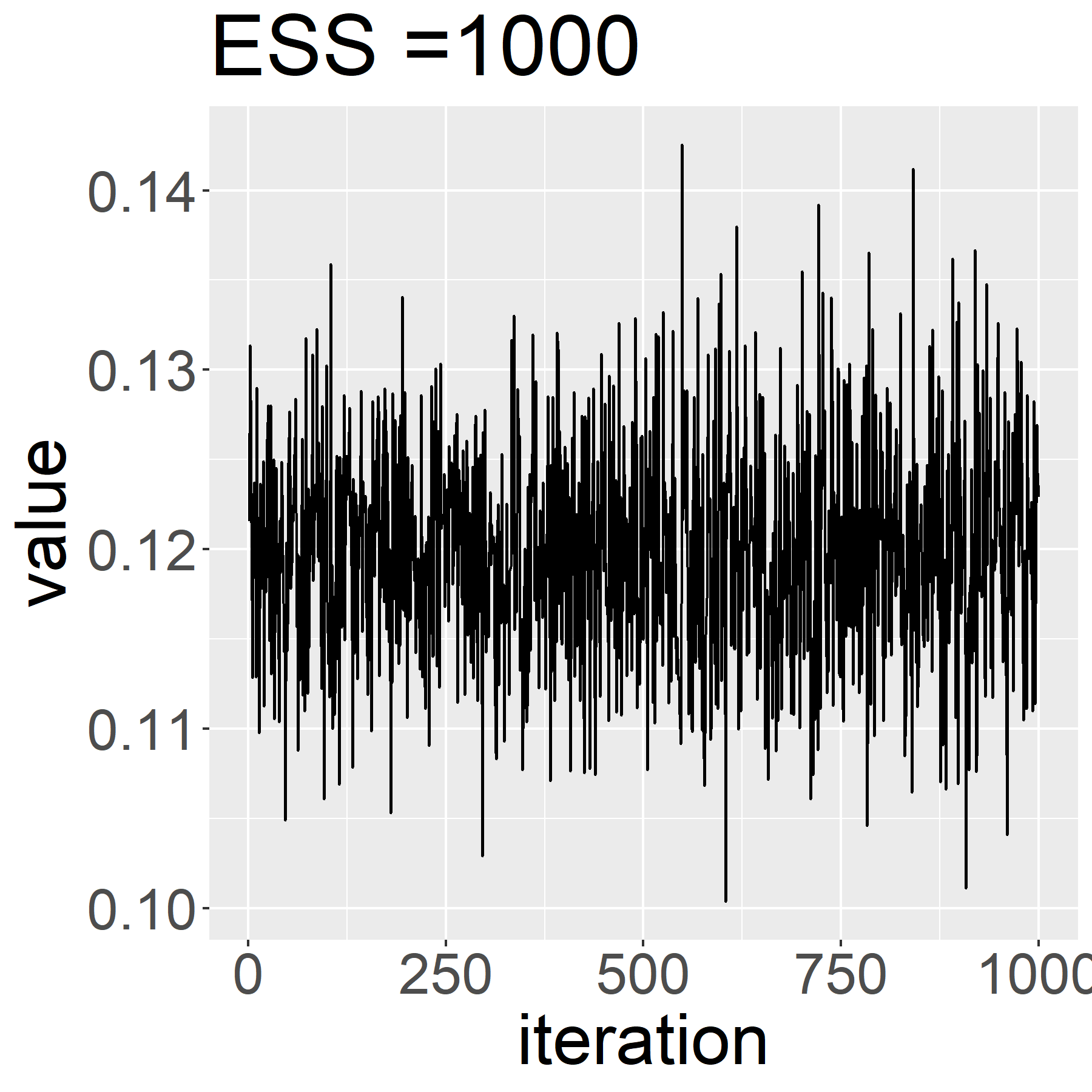} \\
\end{tabular}

$\Phi(\mu^{z_2}(t))$  for $t = 2/24,1,2$ \\
 \begin{tabular}{ccc}
\includegraphics[width = 1.0 in]{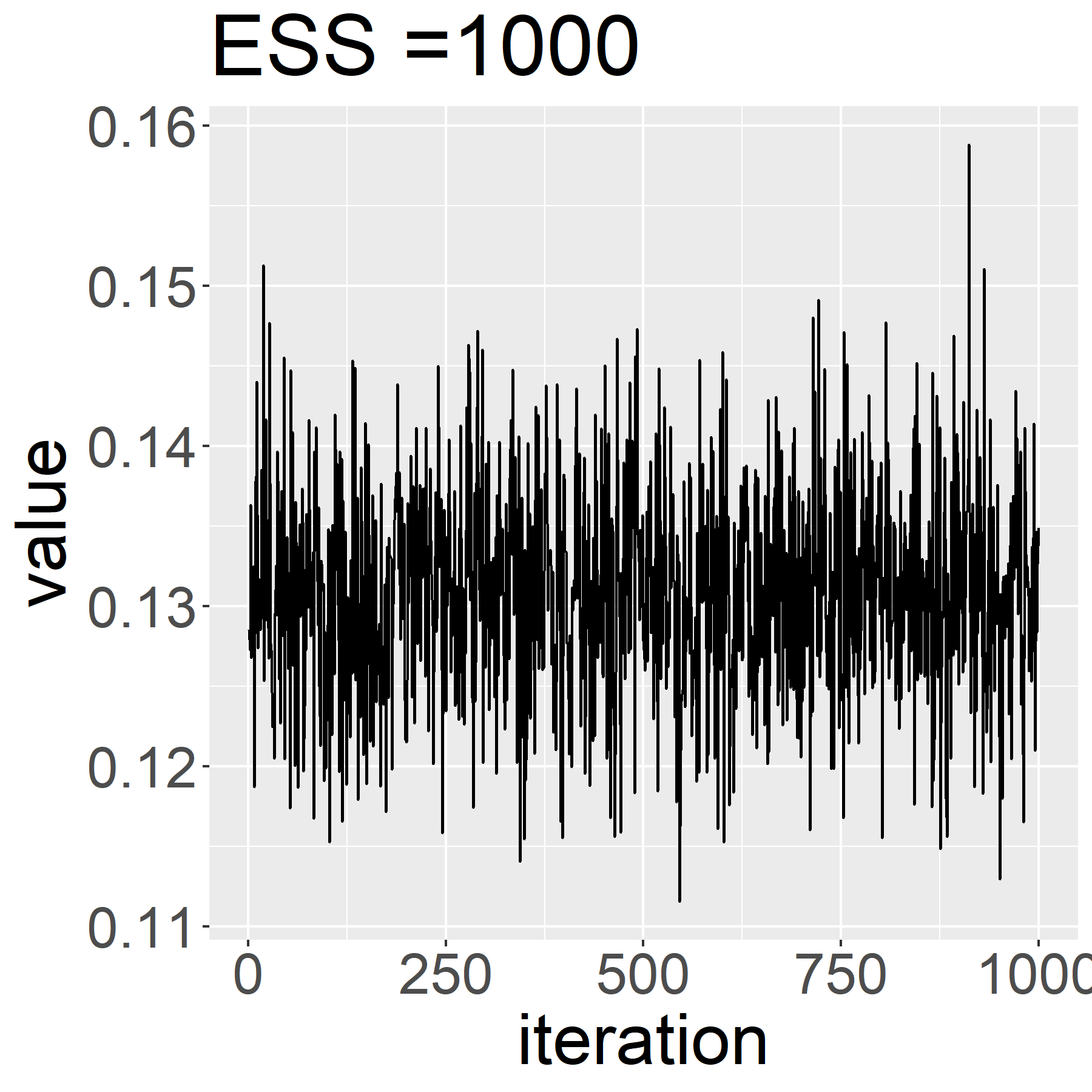} & \includegraphics[width = 1.0 in]{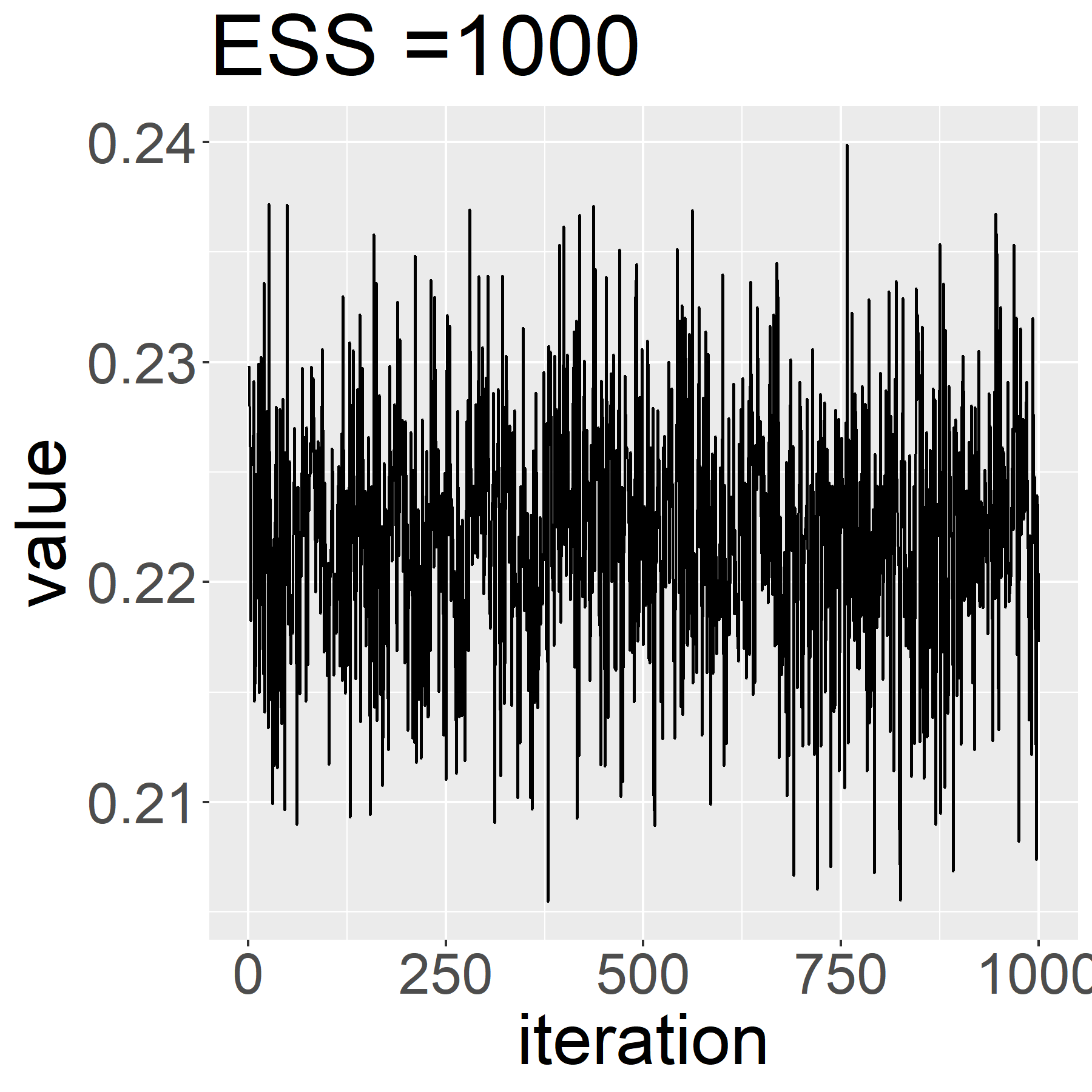} & \includegraphics[width = 1.0 in]{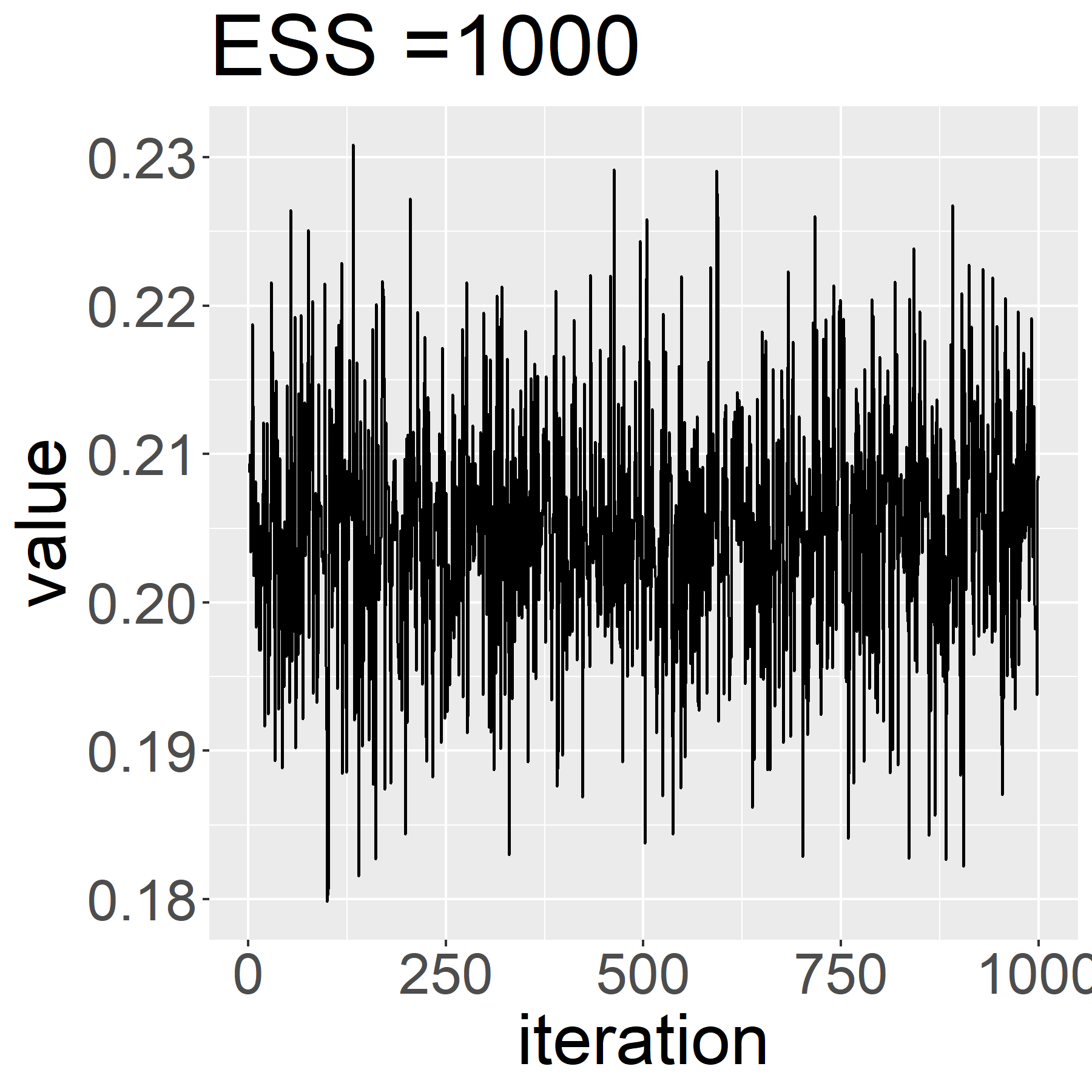} \\
\end{tabular}

$\Phi(\mu^{z_3}(t))$  for $t = 2/24,1,2$ \\
 \begin{tabular}{ccc}
\includegraphics[width = 1.0 in]{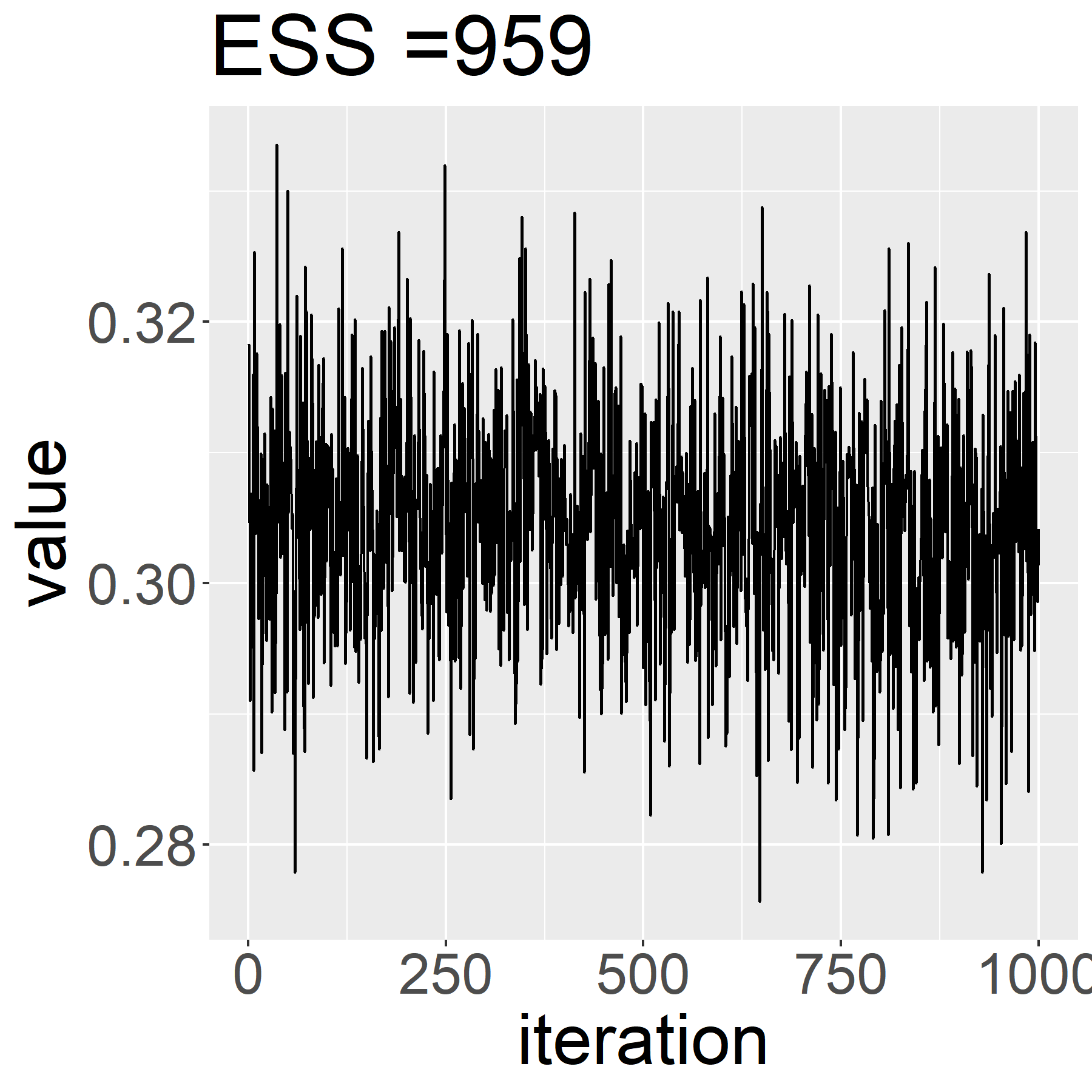} & \includegraphics[width = 1.0 in]{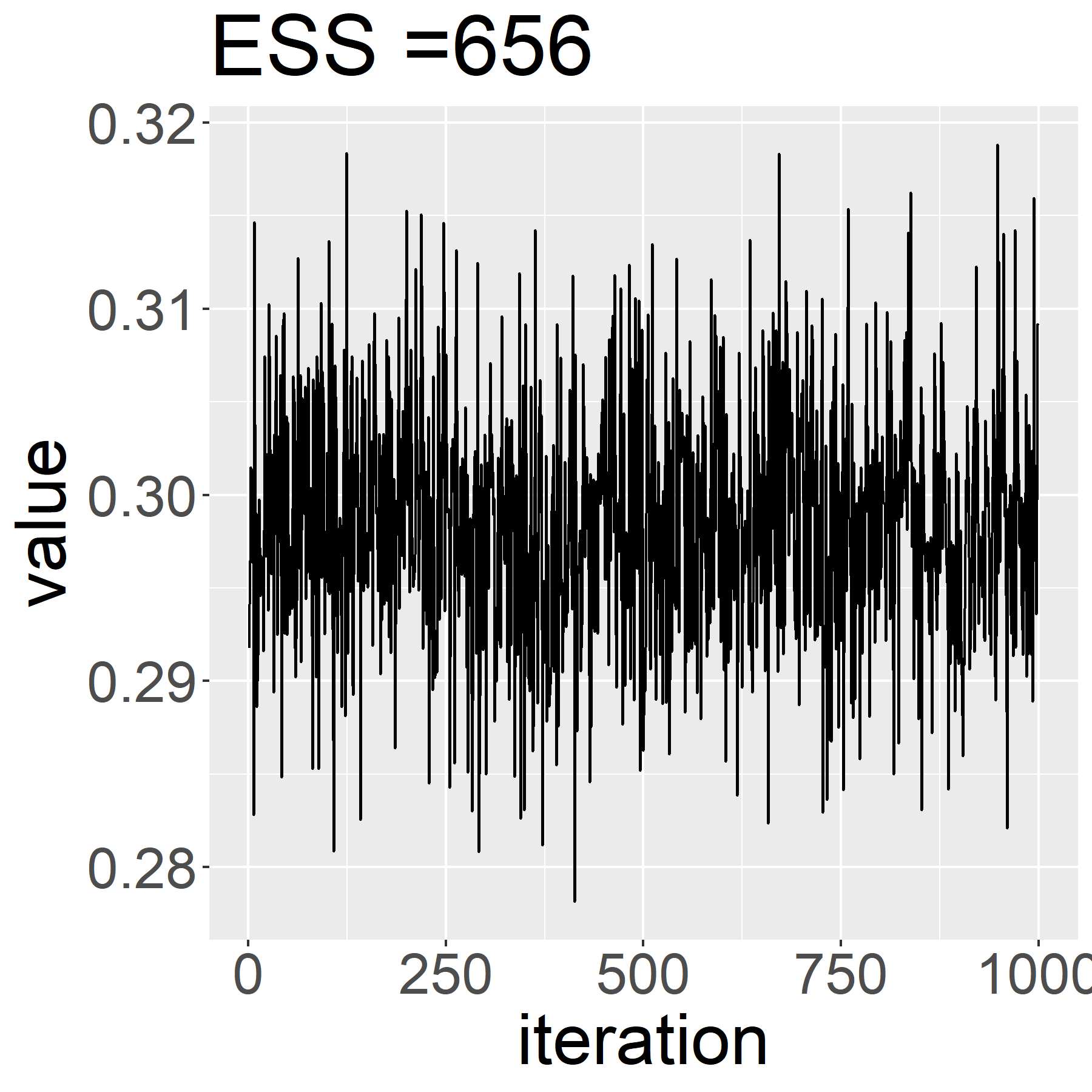} & \includegraphics[width = 1.0 in]{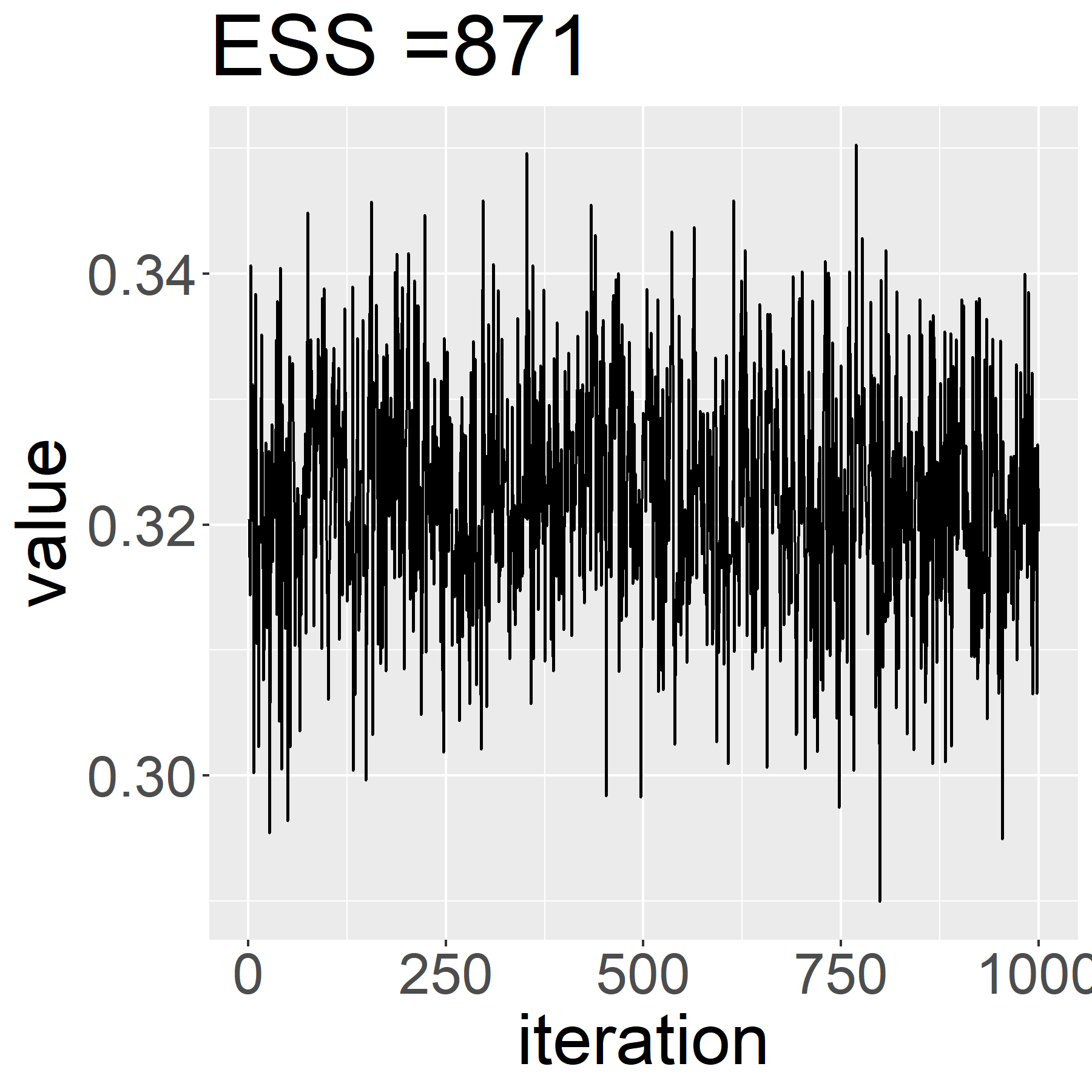} \\
\end{tabular}

$\Phi(\mu^{z_4}(t))$  for $t = 2/24,1,2$ \\

 \begin{tabular}{cccc}  
\includegraphics[width = 1.0 in]{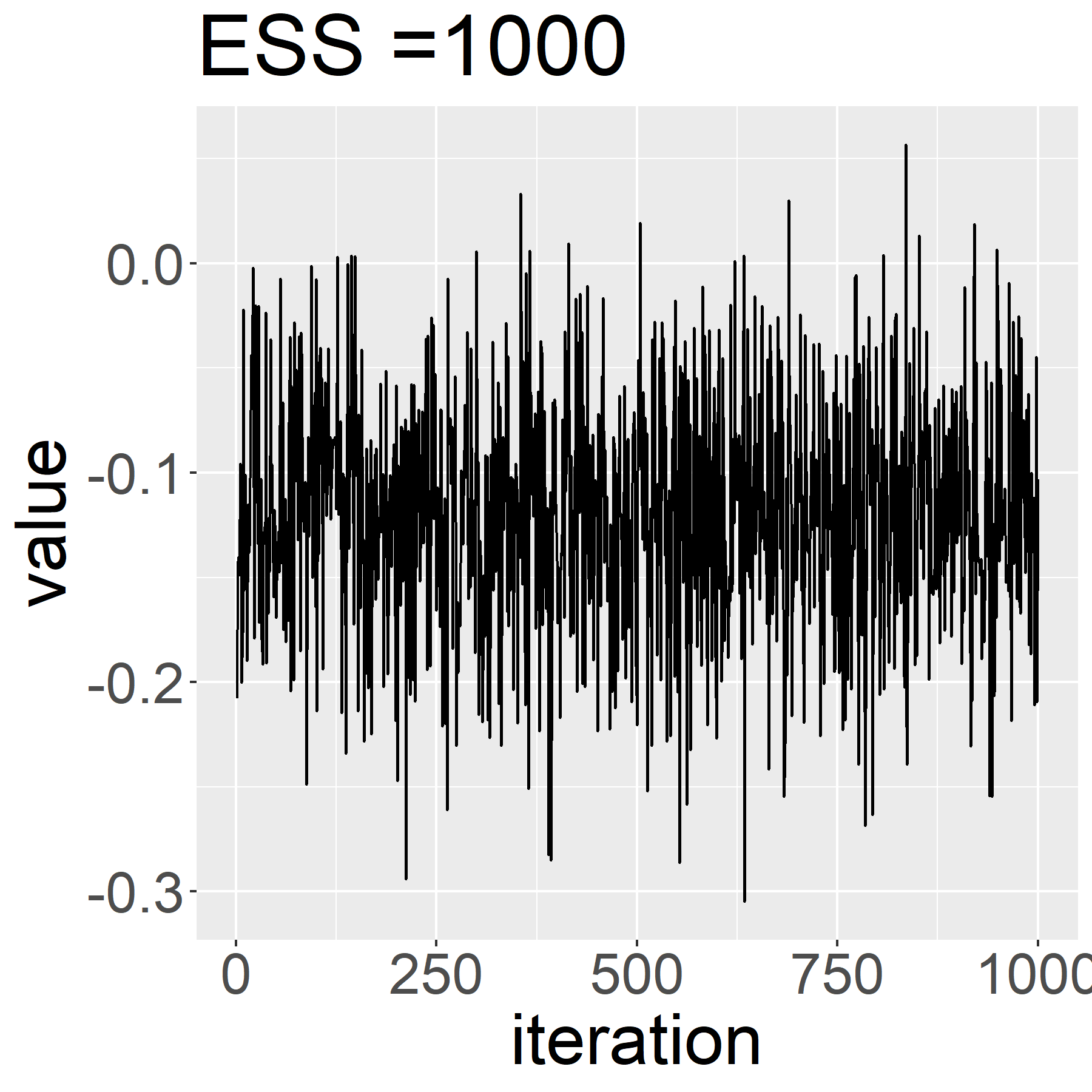} & \includegraphics[width = 1.0 in]{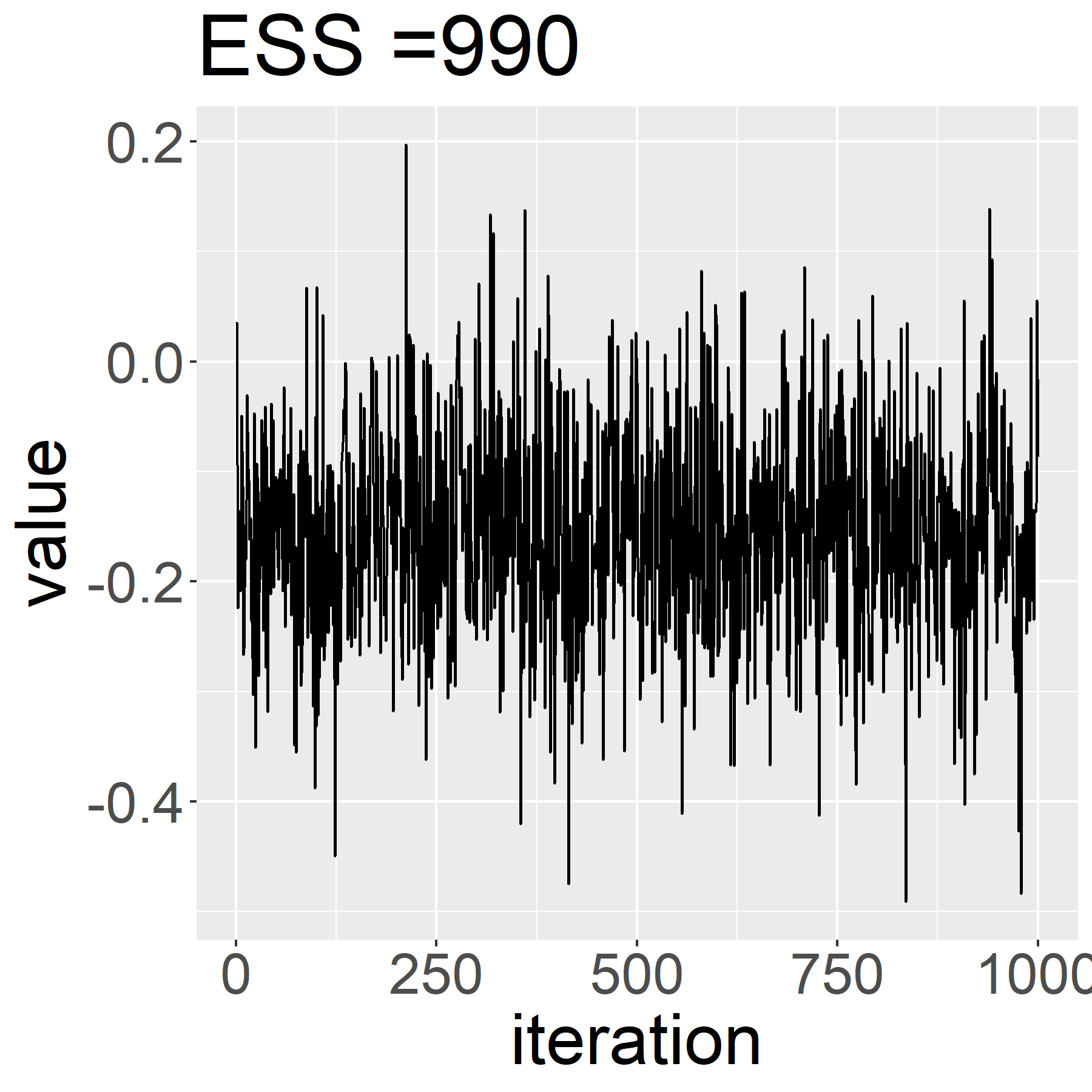} & \includegraphics[width = 1.0 in]{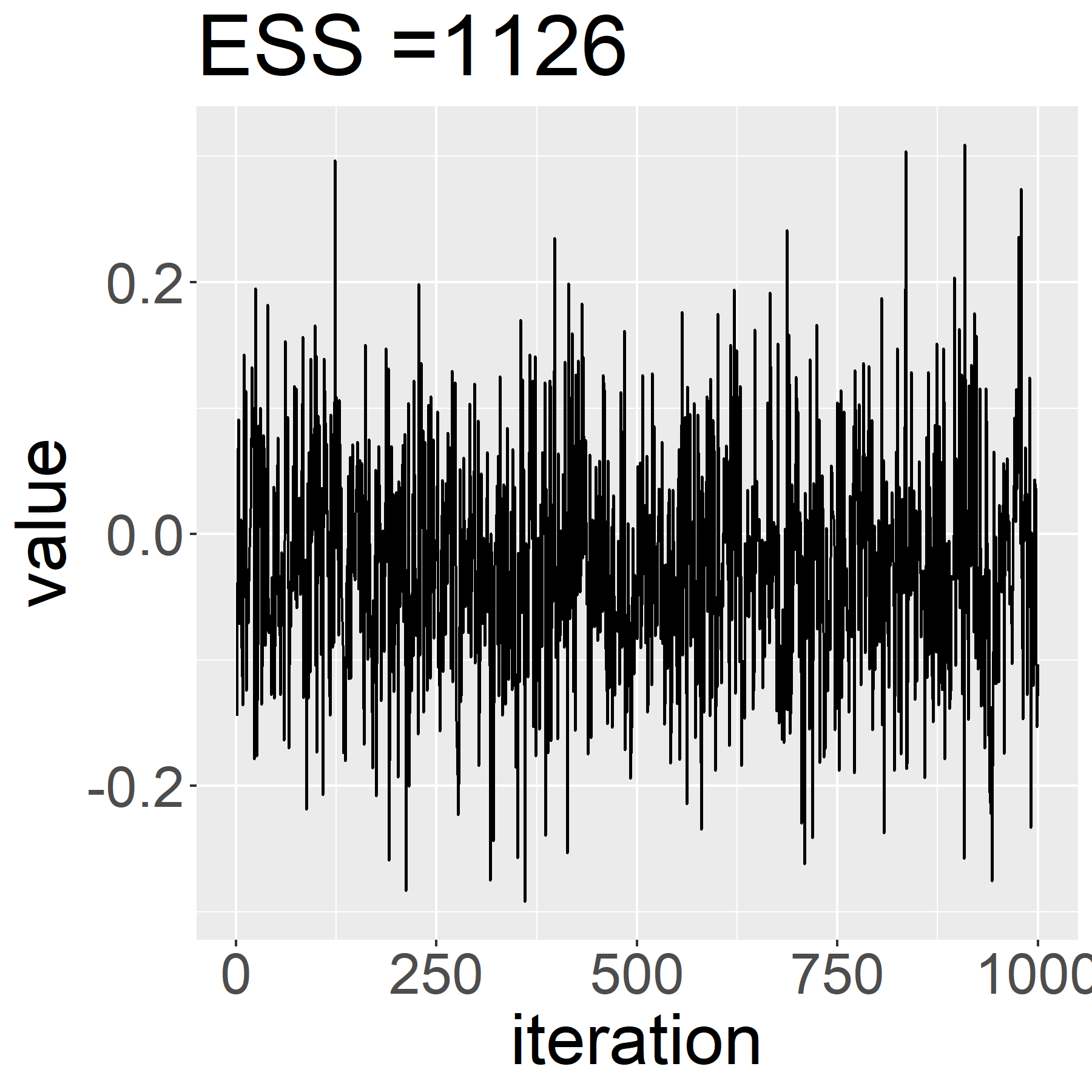} & \includegraphics[width = 1.0 in]{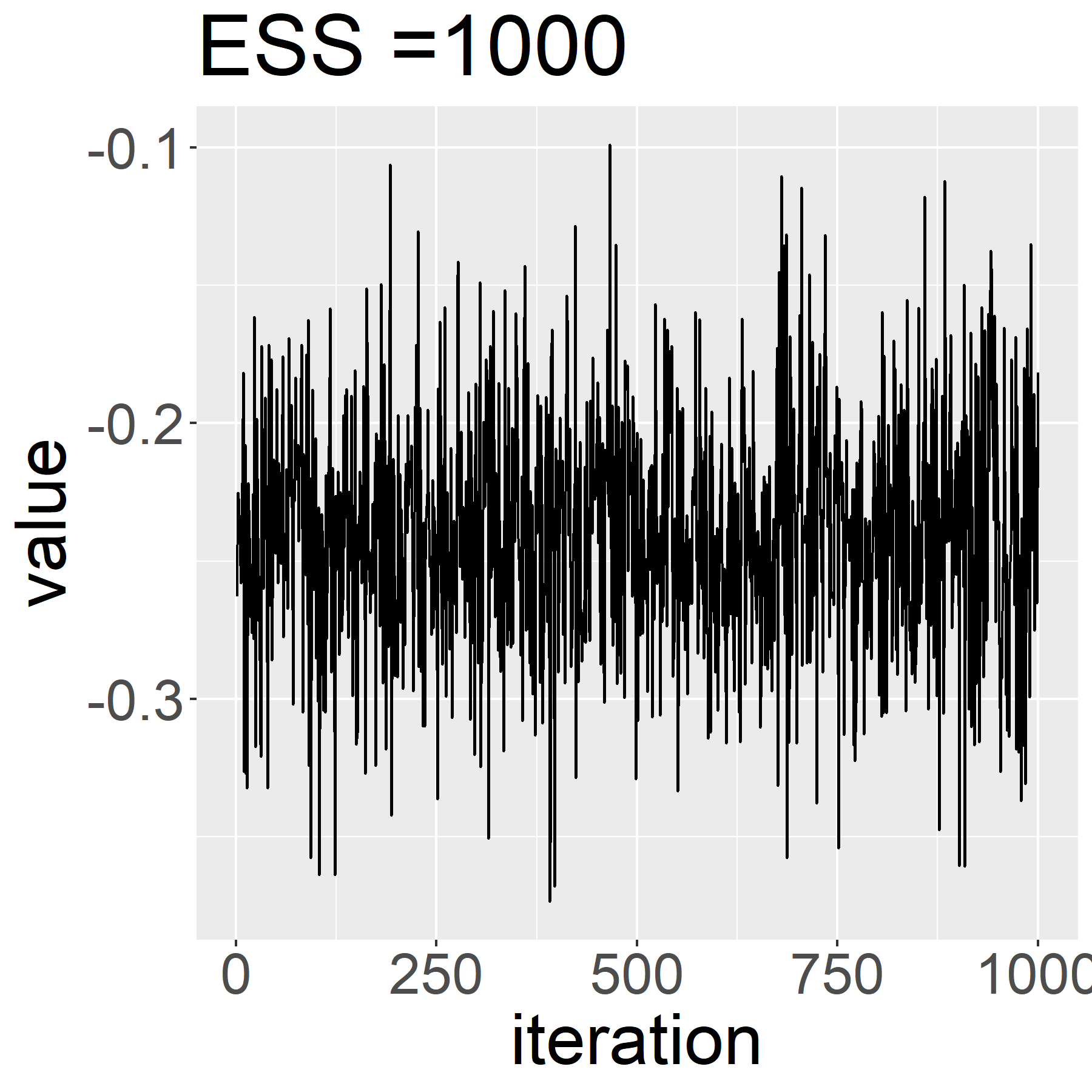} \\
$\rho_{1,2}$ & $\rho_{2,2}$ & $\rho_{3,2}$ & $\rho_{4,2}$ \\
\includegraphics[width = 1.0 in]{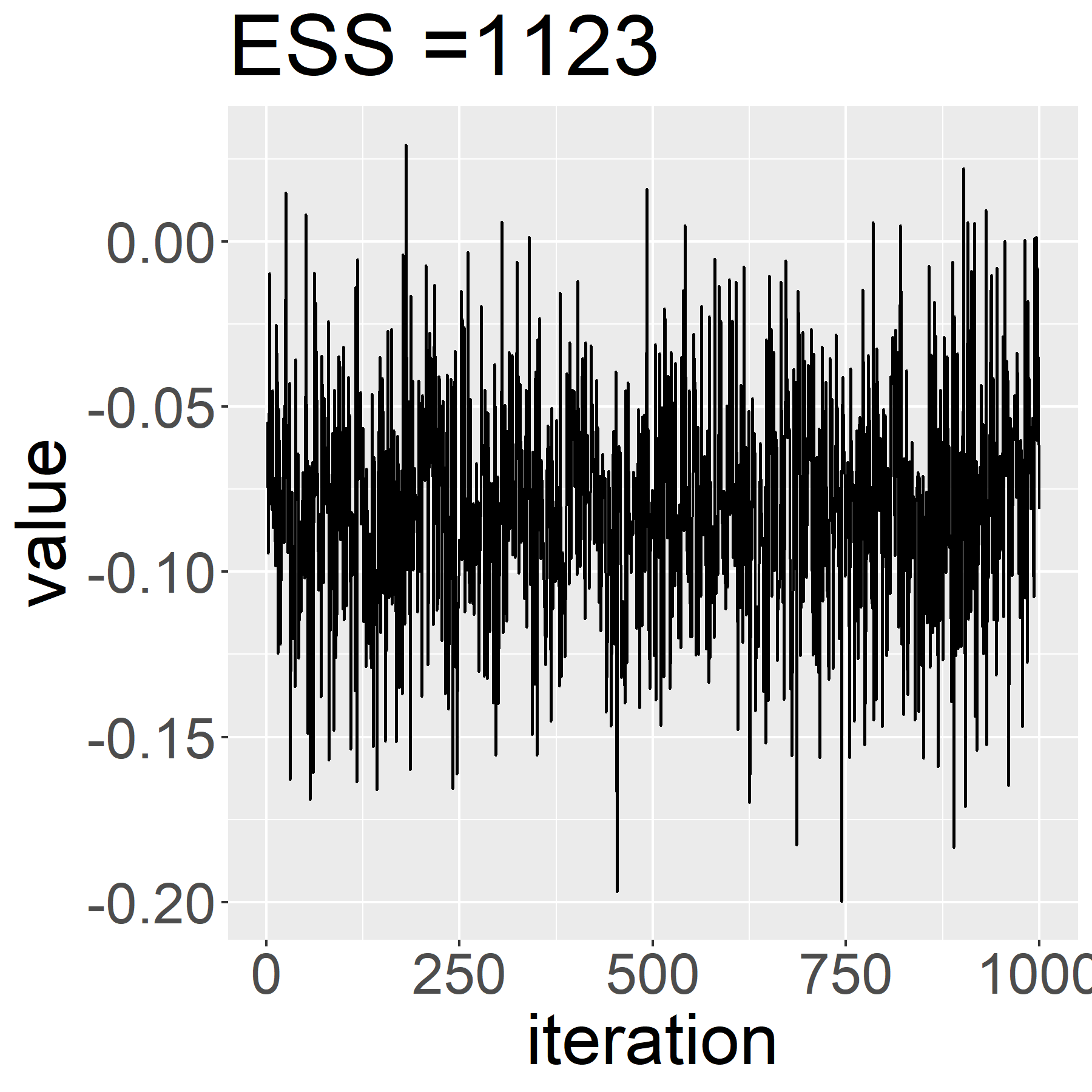} & \includegraphics[width = 1.0 in]{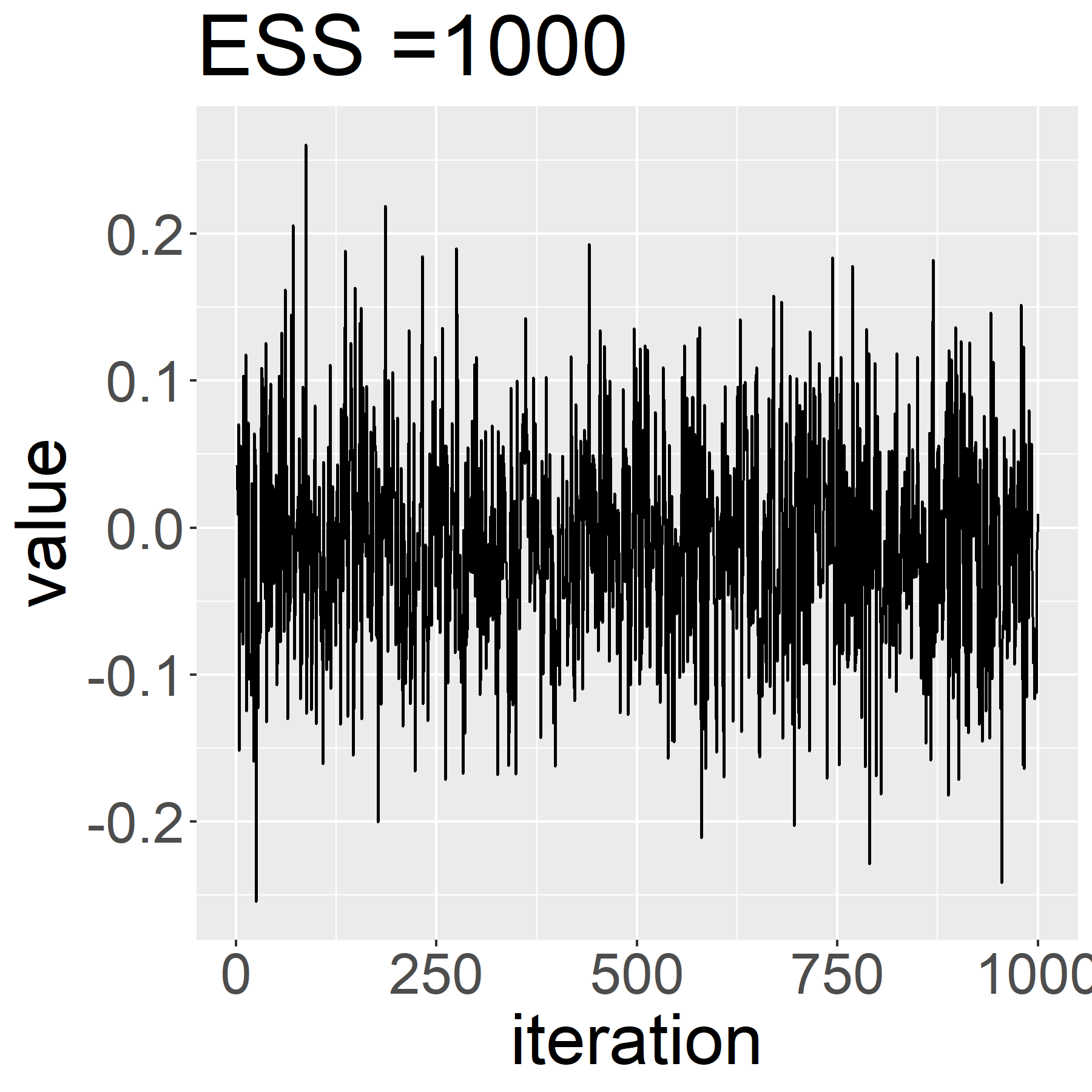} & \includegraphics[width = 1.0 in]{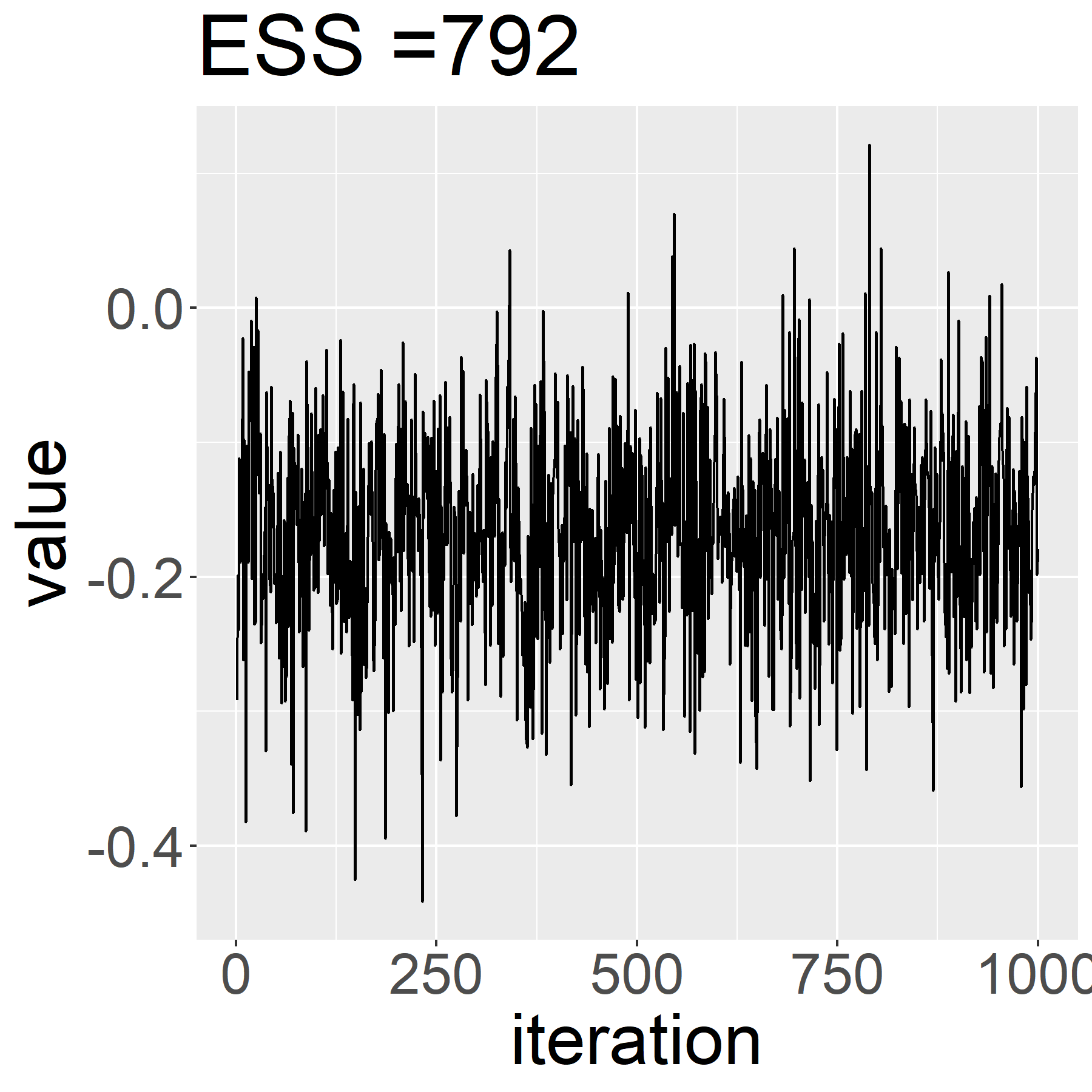} & \includegraphics[width = 1.0 in]{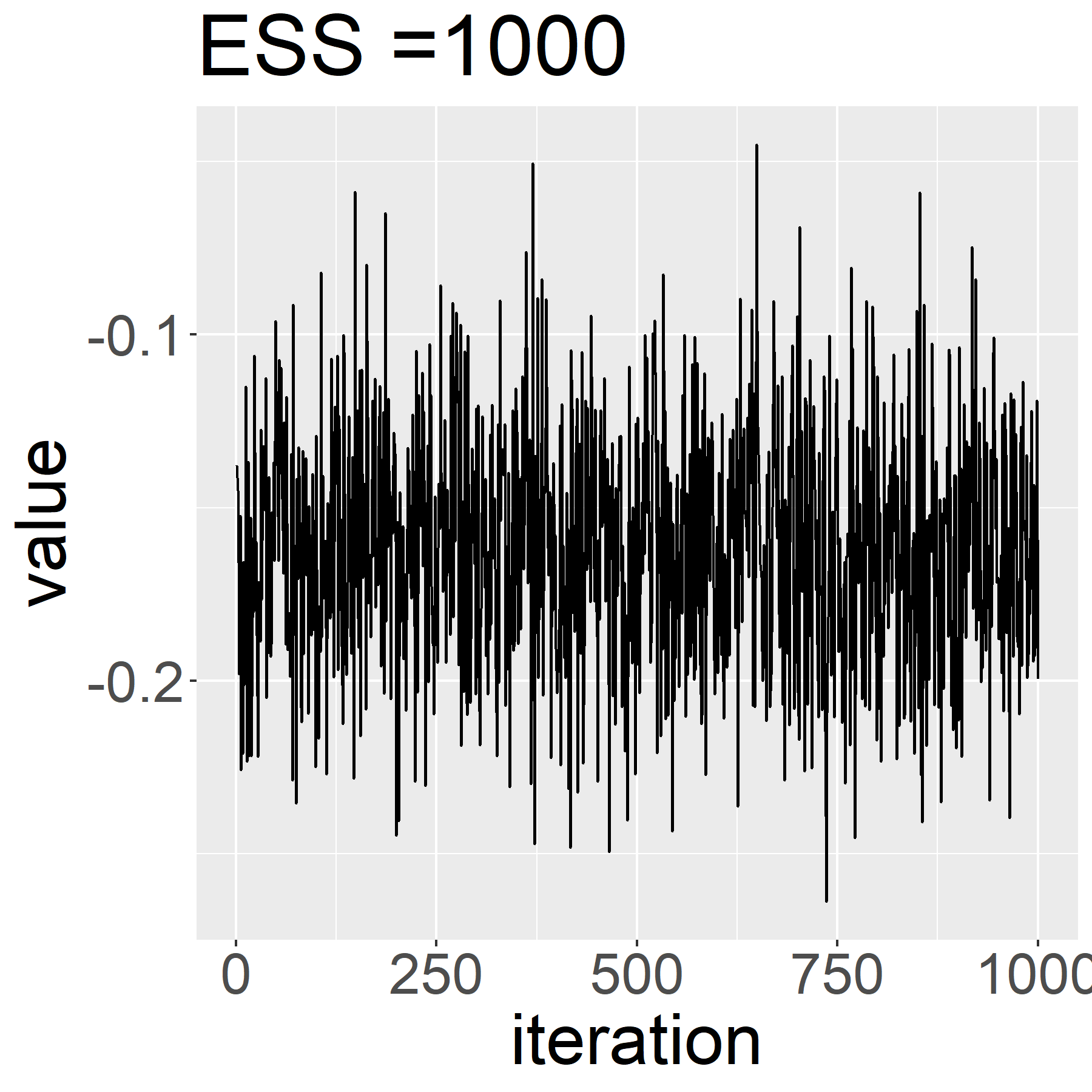} \\
$\rho_{1,3}$ & $\rho_{2,3}$ & $\rho_{3,3}$ & $\rho_{4,3}$ \\
\includegraphics[width = 1.0 in]{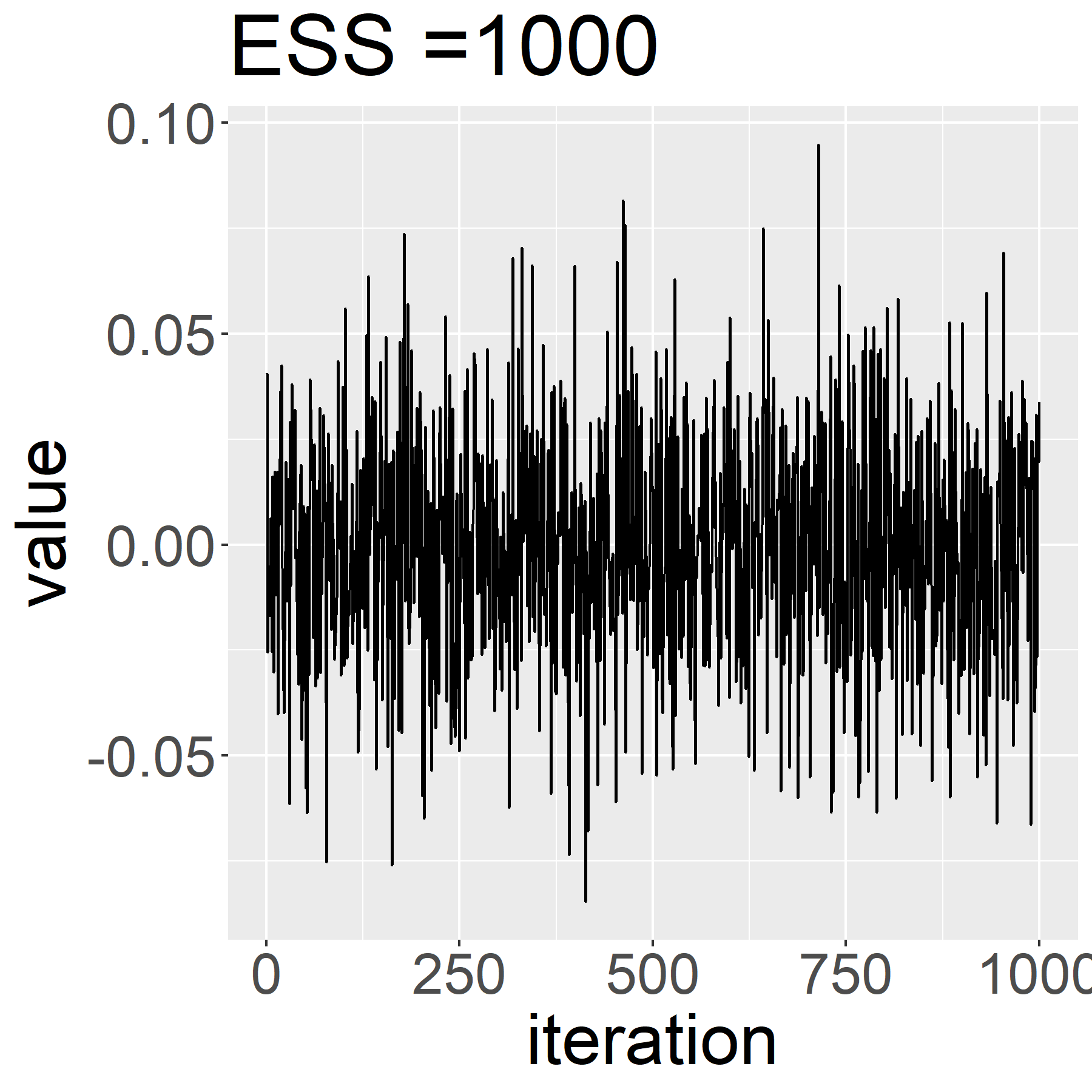} & \includegraphics[width = 1.0 in]{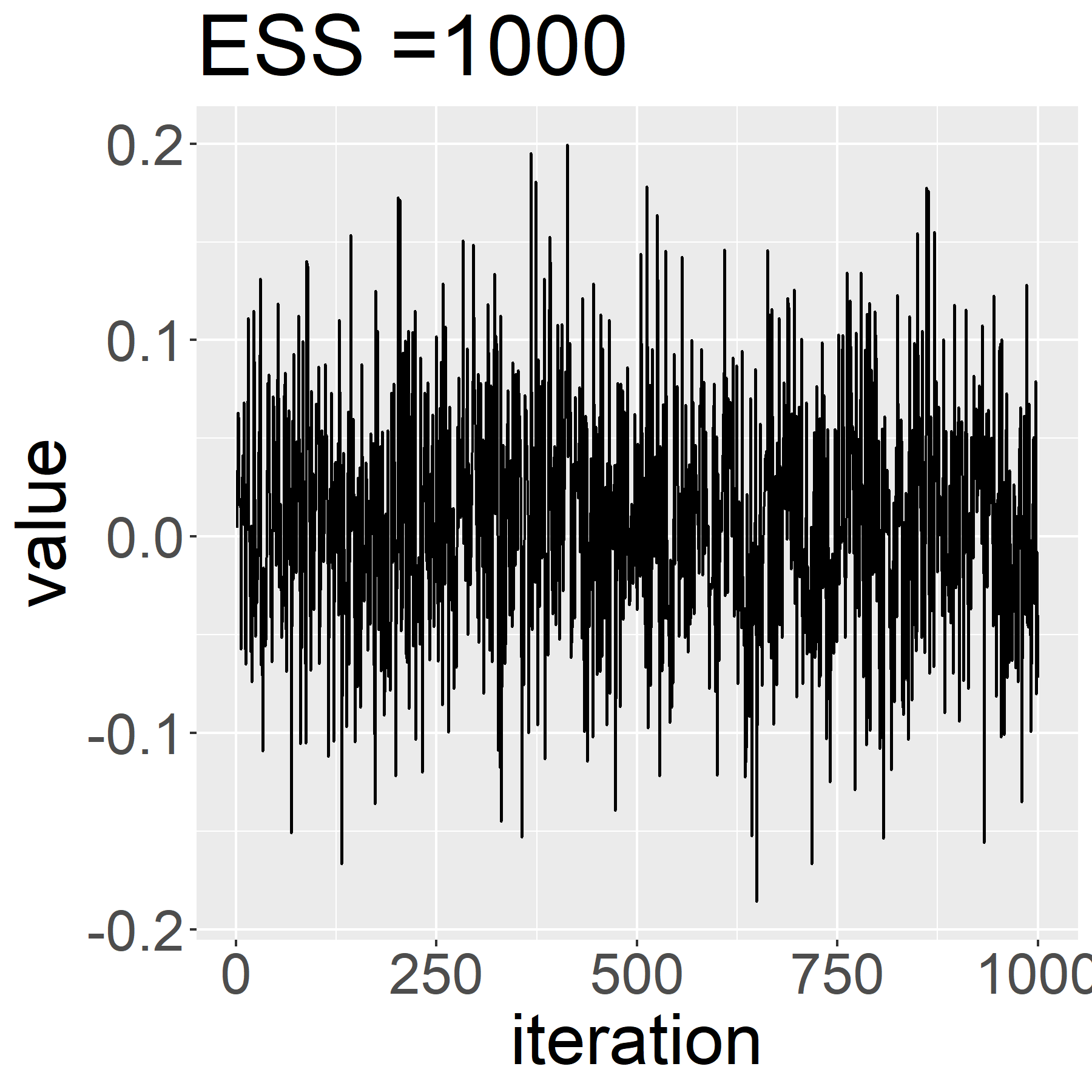} & \includegraphics[width = 1.0 in]{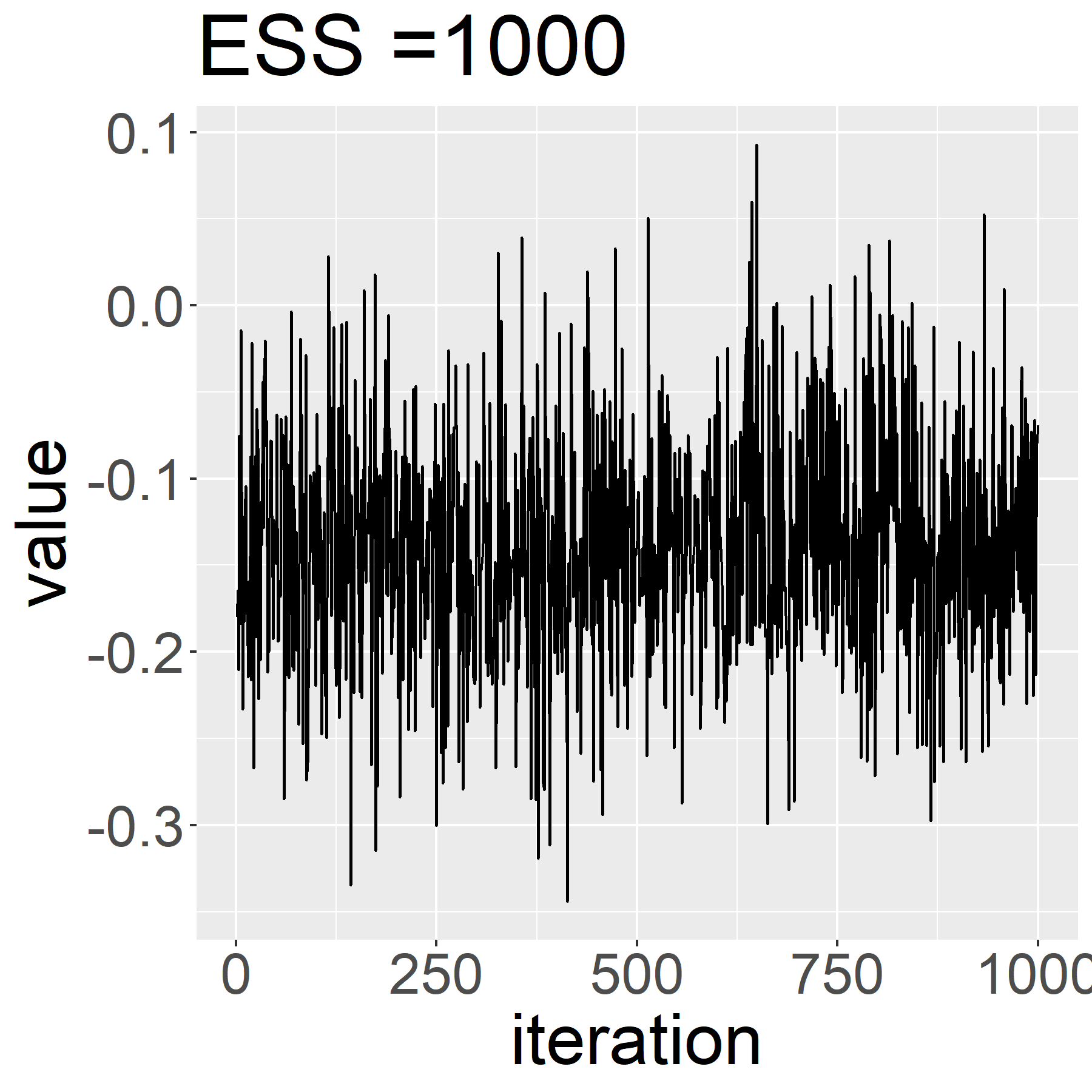} & \includegraphics[width = 1.0 in]{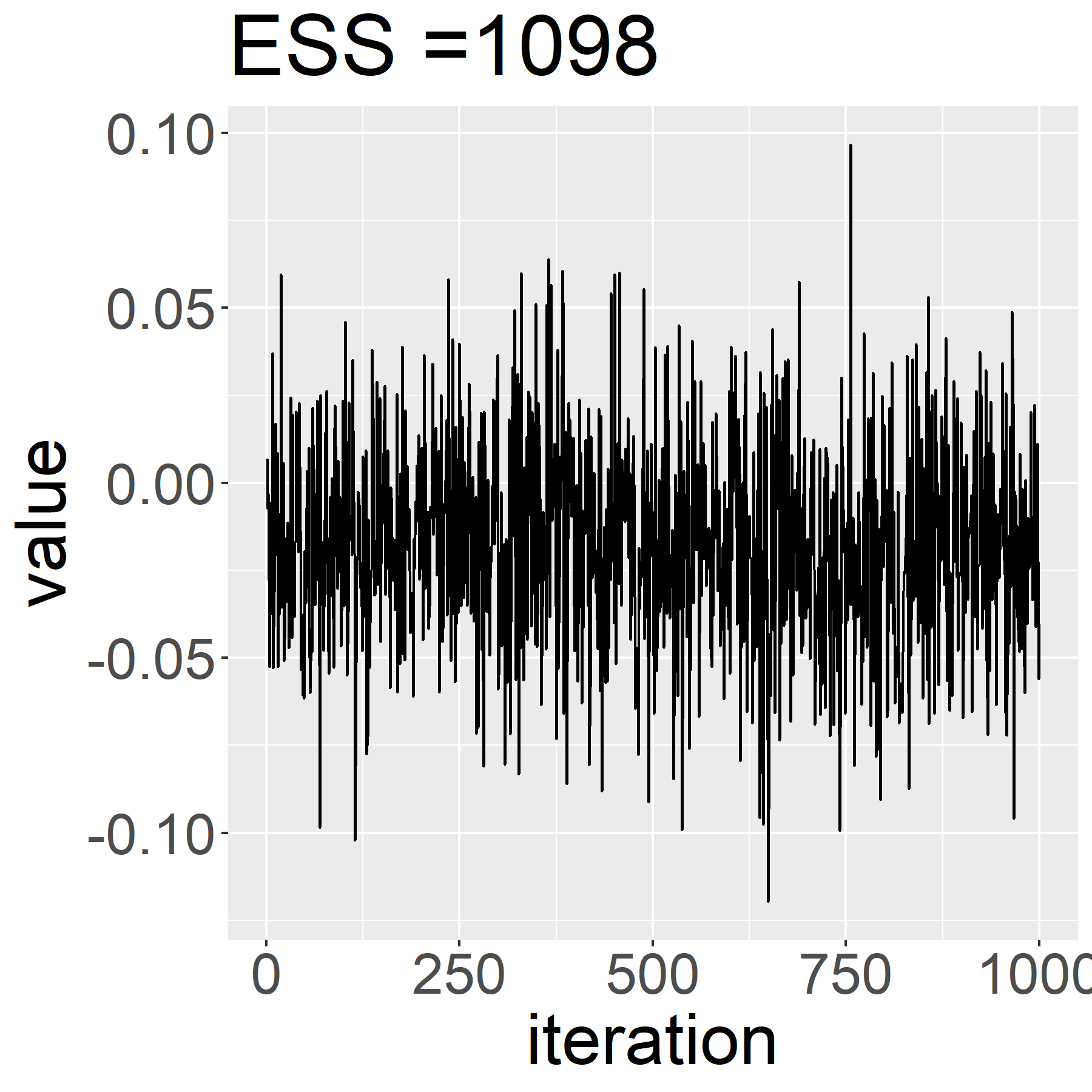} \\
$\rho_{1,4}$ & $\rho_{2,4}$ & $\rho_{3,4}$ & $\rho_{4,4}$ \\
\end{tabular}
\caption{Trace plots related to illness status interpolation and concurrent regression presented in Figure 8 in the main paper.}\label{fig:trace_ill}
    \end{center}
\end{figure}

\bibliographystyle{agsm}
\bibliography{bibliography}

@article{adair2011,
  title={Cohort profile: {T}he {C}ebu longitudinal health and nutrition survey},
  author={Adair, Linda S and Popkin, Barry M and Akin, John S and Guilkey, David K and Gultiano, Socorro and Borja, Judith and Perez, Lorna and Kuzawa, Christopher W and McDade, Thomas and Hindin, Michelle J},
  journal={International Journal of Epidemiology},
  volume={40},
  number={3},
  pages={619--625},
  year={2011},
  publisher={Oxford University Press}
}

@article{albert1993,
  title={Bayesian analysis of binary and polychotomous response data},
  author={Albert, James H and Chib, Siddhartha},
  journal={Journal of the American Statistical Association},
  volume={88},
  number={422},
  pages={669--679},
  year={1993},
  publisher={Taylor \& Francis}
}

@article{andrieu2008,
  title={A tutorial on adaptive {MCMC}},
  author={Andrieu, Christophe and Thoms, Johannes},
  journal={Statistics and Computing},
  volume={18},
  number={4},
  pages={343--373},
  year={2008},
  publisher={Springer}
}

@misc{bayesFPCA2023,
  author = {H. Ruffieux},
  title = {Efficient Bayesian functional principal component analysis of longitudinal data},
  year = {2023},
  publisher = {arXiv},
  url = {https://arxiv.org/abs/2311.05200},
  note = {R package available on GitHub}
}

@article{bhattacharya2011,
  title={Sparse {B}ayesian infinite factor models},
  author={Bhattacharya, Anirban and Dunson, David B},
  journal={Biometrika},
  pages={291--306},
  volume={98},
  number={2},
  year={2011}
}

@article{behseta2005,
  title={Hierarchical models for assessing variability among functions},
  author={Behseta, Sam and Kass, Robert E and Wallstrom, Garrick L},
  journal={Biometrika},
  volume={92},
  number={2},
  pages={419--434},
  year={2005},
  publisher={Oxford University Press}
}

@article{carpenter2017stan,
title={Stan: A probabilistic programming language},
author={Carpenter, Bob and Gelman, Andrew and Hoffman, Matthew D and Lee, Daniel and Goodrich, Ben and Betancourt, Michael and Brubaker, Marcus and Guo, Jiqiang and Li, Peter and Riddell, Allen},
journal={Journal of statistical software},
volume={76},
number={1},
year={2017},
publisher={Columbia Univ., New York, NY (United States); Harvard Univ., Cambridge, MA (United States)}
}

@article{casey1986,
  title={Nutrient intake by breast-fed infants during the first five days after birth},
  author={Casey, Clare E and Neifert, Marianne R and Seacat, Joy M and Neville, Margaret C},
  journal={American Journal of Diseases of Children},
  volume={140},
  number={9},
  pages={933--936},
  year={1986},
  publisher={American Medical Association}
}

@article{dasgupta2021,
  title={Modality-constrained density estimation via deformable templates},
  author={Dasgupta, Sutanoy and Pati, Debdeep and Jermyn, Ian H and Srivastava, Anuj},
  journal={Technometrics},
  volume={63},
  number={4},
  pages={536--547},
  year={2021},
  publisher={Taylor \& Francis}
}

@article{crainiceanu2010,
  title={Bayesian functional data analysis using {WinBUGS}},
  author={Crainiceanu, Ciprian M and Goldsmith, A Jeffrey},
  journal={Journal of Statistical Software},
  volume={32},
  number={11},
  pages={1--33},
  year={2010},
  publisher={NIH Public Access}
}

@article{dewey1998,
  title={Growth characteristics of breast-fed compared to formula-fed infants},
  author={Dewey, Kathryn G},
  journal={Neonatology},
  volume={74},
  number={2},
  pages={94--105},
  year={1998},
  publisher={Karger Publishers}
}

@article{duan2019,
    author = {Duan, Leo L and Young, Alexander L and Nishimura, Akihiko and Dunson, David B},
    title = "{Bayesian constraint relaxation}",
    journal = {Biometrika},
    volume = {107},
    number = {1},
    pages = {191-204},
    year = {2019},
    month = {12},
    issn = {0006-3444},
    doi = {10.1093/biomet/asz069},
    eprint = {https://academic.oup.com/biomet/article-pdf/107/1/191/32450934/asz069.pdf},
}

@article{earls2017,
author = {Cecilia Earls and Giles Hooker},
title = {Combining Functional Data Registration and Factor Analysis},
journal = {Journal of Computational and Graphical Statistics},
volume = {26},
number = {2},
pages = {296-305},
year  = {2017},
publisher = {Taylor & Francis}
}

@book{ferraty2006,
  title={Nonparametric Functional Data Analysis: Theory and Practice},
  author={Ferraty, Fr{\'e}d{\'e}ric and Vieu, Philippe},
  year={2006},
  publisher={Springer}
}

@book{gelman2013,
  title={Bayesian Data Analysis},
  author={Gelman, Andrew and Carlin, John B and Stern, Hal S and Dunson, David B and Vehtari, Aki and Rubin, Donald B},
  year={2013},
  publisher={Chapman and Hall/CRC}
}

@article{ghosh2009,
  title={Default prior distributions and efficient posterior computation in {B}ayesian factor analysis},
  author={Ghosh, Joyee and Dunson, David B},
  journal={Journal of Computational and Graphical Statistics},
  volume={18},
  number={2},
  pages={306--320},
  year={2009},
  publisher={Taylor \& Francis}
}

@article{goldsmith2015,
  title={Generalized multilevel function-on-scalar regression and principal component analysis},
  author={Goldsmith, Jeff and Zipunnikov, Vadim and Schrack, Jennifer},
  journal={Biometrics},
  volume={71},
  number={2},
  pages={344--353},
  year={2015},
  publisher={Wiley Online Library}
}

@article{hall2008,
  title={Modelling sparse generalized longitudinal observations with latent {G}aussian processes},
  author={Hall, Peter and M{\"u}ller, Hans-Georg and Yao, Fang},
  journal={Journal of the Royal Statistical Society: Series B },
  volume={70},
  number={4},
  pages={703--723},
  year={2008},
  publisher={Wiley Online Library}
}

@article{hannah2013,
  title={Multivariate convex regression with adaptive partitioning},
  author={Hannah, Lauren A and Dunson, David B},
  journal={The Journal of Machine Learning Research},
  volume={14},
  number={1},
  pages={3261--3294},
  year={2013},
  publisher={JMLR. org}
}

@book{harville1998,
  title={Matrix Algebra from a Statistician's Perspective},
  author={Harville, David A},
  year={1998},
  publisher={Springer}
}

@article{hastie1993,
  title={Varying-coefficient models},
  author={Hastie, Trevor and Tibshirani, Robert},
  journal={Journal of the Royal Statistical Society: Series B },
  volume={55},
  number={4},
  pages={757--779},
  year={1993},
  publisher={Wiley Online Library}
}

@article{hodges2010,
author = {James S. Hodges and Brian J. Reich},
title = {Adding Spatially-Correlated Errors Can Mess Up the Fixed Effect You Love},
journal = {The American Statistician},
volume = {64},
number = {4},
pages = {325-334},
year  = {2010},
publisher = {Taylor \& Francis},
doi = {10.1198/tast.2010.10052},

}

@article{james2000,
  title={Principal component models for sparse functional data},
  author={James, Gareth M and Hastie, Trevor J and Sugar, Catherine A},
  journal={Biometrika},
  volume={87},
  number={3},
  pages={587--602},
  year={2000},
  publisher={Oxford University Press}
}

@article{jauch2021,
  title={Monte {C}arlo simulation on the {S}tiefel manifold via polar expansion},
  author={Jauch, Michael and Hoff, Peter D and Dunson, David B},
  journal={Journal of Computational and Graphical Statistics},
  volume={30},
  number={3},
  pages={622--631},
  year={2021},
  publisher={Taylor \& Francis}
}

@misc{jiang2020,
  title={BayesTime: {B}ayesian functional principal components for sparse longitudinal data},
  author={Jiang, Lingjing and Zhong, Yuan and Elrod, Chris and Natarajan, Loki and Knight, Rob and Thompson, Wesley K},
  howpublished={arXiv:2012.00579},
  year={2020}
}

@article{khan2022,
author = {Kori Khan and Catherine A. Calder},
title = {Restricted Spatial Regression Methods: Implications for Inference},
journal = {Journal of the American Statistical Association},
volume = {117},
number = {537},
pages = {482-494},
year  = {2022},
publisher = {Taylor & Francis},
doi = {10.1080/01621459.2020.1788949}
}

@article{kneip2008,
  title={Combining registration and fitting for functional models},
  author={Kneip, Alois and Ramsay, James O},
  journal={Journal of the American Statistical Association},
  volume={103},
  number={483},
  pages={1155--1165},
  year={2008},
  publisher={Taylor \& Francis}
}

@article{kowal2017,
author = {Daniel R. Kowal and David S. Matteson and David Ruppert},
title = {A {G}aussian Multivariate Functional Dynamic Linear Model},
journal = {Journal of the American Statistical Association},
volume = {112},
number = {518},
pages = {733-744},
year  = {2017},
publisher = {Taylor & Francis},
}

@article{kowal2021,
author = {Daniel R. Kowal},
title = {Dynamic Regression Models for Time-Ordered Functional Data},
volume = {16},
journal = {Bayesian Analysis},
number = {2},
publisher = {International Society for Bayesian Analysis},
pages = {459 -- 487},
year = {2021},
doi = {10.1214/20-BA1213},
URL = {https://doi.org/10.1214/20-BA1213}
}

@article{kowal2022,
author = {Daniel R. Kowal and Antonio Canale},
title = {{Semiparametric Functional Factor Models with Bayesian Rank Selection}},
volume = {18},
journal = {Bayesian Analysis},
number = {4},
publisher = {International Society for Bayesian Analysis},
pages = {1161 -- 1189},
year = {2023},
doi = {10.1214/23-BA1410},
URL = {https://doi.org/10.1214/23-BA1410}
}

@article{lenk2017,
  title={Bayesian analysis of shape-restricted functions using {G}aussian process priors},
  author={Lenk, Peter J and Choi, Taeryon},
  journal={Statistica Sinica},
  pages={43--69},
  year={2017},
  publisher={JSTOR}
}

@article{lin2014,
  title={Bayesian monotone regression using {G}aussian process projection},
  author={Lin, Lizhen and Dunson, David B},
  journal={Biometrika},
  volume={101},
  number={2},
  pages={303--317},
  year={2014},
  publisher={Oxford University Press}
}

@article{malfait2003,
  title={The historical functional linear model},
  author={Malfait, Nicole and Ramsay, James O},
  journal={Canadian Journal of Statistics},
  volume={31},
  number={2},
  pages={115--128},
  year={2003},
  publisher={Wiley Online Library}
}

@misc{marco2022,
  title={Functional Mixed Membership Models},
  author={Marco, Nicholas and Şentürk, Damala and Jeste, Shafali and DiStefano, Charlotte and Dickinson, Abigail and Telesca, Donatello},
  howpublished={arXiv:2206.12084},
  year={2022}
}

@inproceedings{marcus1972,
  title={Sample behavior of {G}aussian processes},
  author={Marcus, Michael B and Shepp, Lawrence A},
  booktitle={Proceedings of the Sixth Berkeley Symposium on Mathematics, Statisics and Probability},
  volume={2},
  pages={423--421},
  year={1972}
}

@article{marron2015,
  title={Functional data analysis of amplitude and phase variation},
  author={Marron, James Stephen and Ramsay, James O and Sangalli, Laura M and Srivastava, Anuj},
  journal={Statistical Science},
  volume = {30},
  number = {4},
  pages={468--484},
  year={2015},
  publisher={JSTOR}
}

@article{montagna2012,
  title={Bayesian latent factor regression for functional and longitudinal data},
  author={Montagna, Silvia and Tokdar, Surya T and Neelon, Brian and Dunson, David B},
  journal={Biometrics},
  volume={68},
  number={4},
  pages={1064--1073},
  year={2012},
  publisher={Wiley Online Library}
}

@article{moran2021,
  title={Bayesian joint modeling of chemical structure and dose response curves},
  author={Moran, Kelly R and Dunson, David and Wheeler, Matthew W and Herring, Amy H},
  journal={The Annals of Applied Statistics},
  volume={15},
  number={3},
  pages={1405--1430},
  year={2021},
  publisher={Institute of Mathematical Statistics}
}

@article{nolan2021,
author = {Tui H. Nolan and Jeff Goldsmith and David Ruppert},
title = {{Bayesian Functional Principal Components Analysis via Variational Message Passing with Multilevel Extensions}},
volume = {20},
journal = {Bayesian Analysis},
number = {1},
publisher = {International Society for Bayesian Analysis},
pages = {157 -- 183},
year = {2025},
doi = {10.1214/23-BA1393},
URL = {https://doi.org/10.1214/23-BA1393}
}

@article{peng2009,
  title={A geometric approach to maximum likelihood estimation of the functional principal components from sparse longitudinal data},
  author={Peng, Jie and Paul, Debashis},
  journal={Journal of Computational and Graphical Statistics},
  volume={18},
  number={4},
  pages={995--1015},
  year={2009},
  publisher={Taylor \& Francis}
}

@article{plumlee2018,
 ISSN = {10170405, 19968507},
 author = {Matthew Plumlee and V. Roshan Joseph},
 journal = {Statistica Sinica},
 number = {2},
 pages = {601--619},
 publisher = {Institute of Statistical Science, Academia Sinica},
 title = {Orthogonal {G}aussian Process Models},
 volume = {28},
 year = {2018}
}

@article{polson2013bayesian,
  title={Bayesian inference for logistic models using P{\'o}lya--Gamma latent variables},
  author={Polson, Nicholas G and Scott, James G and Windle, Jesse},
  journal={Journal of the American statistical Association},
  volume={108},
  number={504},
  pages={1339--1349},
  year={2013},
  publisher={Taylor \& Francis}
}

@article{poworoznek2021,
author = {Evan Poworoznek and Niccolo Anceschi and Federico Ferrari and David Dunson},
title = {{Efficiently Resolving Rotational Ambiguity in Bayesian Matrix Sampling with Matching}},
journal = {Bayesian Analysis},
publisher = {International Society for Bayesian Analysis},
pages = {1 -- 22},
year = {2025},
doi = {10.1214/25-BA1544},
URL = {https://doi.org/10.1214/25-BA1544}
}

@article{ramsay1991,
  title={Some tools for functional data analysis},
  author={Ramsay, James O and Dalzell, CJ},
  journal={Journal of the Royal Statistical Society: Series B },
  volume={53},
  number={3},
  pages={539--561},
  year={1991},
  publisher={Wiley Online Library}
}

@book{ramsay2005,
  title={Functional Data Analysis},
  author={ Ramsay, J. and Silverman, B.W. },
  isbn={9780387400808},
  lccn={2005923773},
  series={Springer Series in Statistics},
  year={2005},
  publisher={Springer}
}

@book{resnick2019,
  title={A Probability Path},
  author={Resnick, Sidney},
  year={2019},
  publisher={Springer}
}

@book{ruppert2003,
  title={Semiparametric regression},
  author={Ruppert, David and Wand, Matt P and Carroll, Raymond J},
  year={2003},
  publisher={Cambridge University Press}
}

@article{Sartini2026,
author = {Joseph Sartini and Xinkai Zhou and Elizabeth Selvin and Scott Zeger and Ciprian M. Crainiceanu},
title = {Fast Bayesian Functional Principal Components Analysis},
journal = {Journal of Computational and Graphical Statistics},
volume = {0},
number = {0},
pages = {1--12},
year = {2026},
publisher = {Taylor \& Francis},
doi = {10.1080/10618600.2025.2592768},


URL = { 
    
        https://doi.org/10.1080/10618600.2025.2592768
    
    

},
eprint = { 
    
        https://doi.org/10.1080/10618600.2025.2592768
    
    

}

}

@article{Schiavon2022,
    author = {Schiavon, L and Canale, A and Dunson, D B},
    title = "{Generalized infinite factorization models}",
    journal = {Biometrika},
    volume = {109},
    number = {3},
    pages = {817-835},
    year = {2022},
    month = {01},
    issn = {1464-3510},
    doi = {10.1093/biomet/asab056},
    url = {https://doi.org/10.1093/biomet/asab056},
    eprint = {https://academic.oup.com/biomet/article-pdf/109/3/817/45512132/asab056.pdf},
}

@article{saurez2017,
author = {Adam J. Suarez and Subhashis Ghosal},
title = {Bayesian Estimation of Principal Components for Functional Data},
volume = {12},
journal = {Bayesian Analysis},
number = {2},
publisher = {International Society for Bayesian Analysis},
pages = {311 -- 333},
year = {2017},
}

@article{shively2009,
  title={A {G}aussian approach to non-parametric monotone function estimation},
  author={Shively, Thomas S and Sager, Thomas W and Walker, Stephen G},
  journal={Journal of the Royal Statistical Society: Series B},
  volume={71},
  number={1},
  pages={159--175},
  year={2009},
  publisher={Wiley Online Library}
}

@article{shively2011,
  title={Nonparametric function estimation subject to monotonicity, convexity and other shape constraints},
  author={Shively, Thomas S and Walker, Stephen G and Damien, Paul},
  journal={Journal of Econometrics},
  volume={161},
  number={2},
  pages={166--181},
  year={2011},
  publisher={Elsevier}
}

@article{shamshoian2022,
  title={Bayesian analysis of longitudinal and multidimensional functional data},
  author={Shamshoian, John and {\c{S}}ent{\"u}rk, Damla and Jeste, Shafali and Telesca, Donatello},
  journal={Biostatistics},
  volume={23},
  number={2},
  pages={558--573},
  year={2022},
  publisher={Oxford University Press}
}

@article{tipping1999,
  title={Probabilistic principal component analysis},
  author={Tipping, Michael E and Bishop, Christopher M},
  journal={Journal of the Royal Statistical Society: Series B },
  volume={61},
  number={3},
  pages={611--622},
  year={1999},
  publisher={Wiley Online Library}
}

@incollection{van2008,
  title={Reproducing kernel {H}ilbert spaces of {G}aussian priors},
  author={{van Zanten}, JH and {van der Vaart}, AW},
  booktitle={Pushing the Limits of Contemporary Statistics: {C}ontributions in Honor of Jayanta K. Ghosh},
  year={2008},
  publisher={Institute of Mathematical Statistics}
}

@article{vats2019,
    author = {Vats, Dootika and Flegal, James M and Jones, Galin L},
    title = {Multivariate output analysis for Markov chain Monte Carlo},
    journal = {Biometrika},
    volume = {106},
    number = {2},
    pages = {321-337},
    year = {2019},
    month = {06},
    issn = {0006-3444},
    doi = {10.1093/biomet/asz002},
    url = {https://doi.org/10.1093/biomet/asz002},
    eprint = {https://academic.oup.com/biomet/article-pdf/106/2/321/28575440/asz002.pdf},
}

@article{vehtari2017,
  title={Practical {B}ayesian model evaluation using leave-one-out cross-validation and {WAIC}},
  author={Vehtari, Aki and Gelman, Andrew and Gabry, Jonah},
  journal={Statistics and Computing},
  volume={27},
  number={5},
  pages={1413--1432},
  year={2017},
  publisher={Springer}
}

@article{vdLinde2008,
author = {{van der Linde}, Angelika},
year = {2008},
month = {12},
pages = {517-533},
title = {Variational {B}ayesian functional {PCA}},
volume = {53},
journal = {Computational Statistics \& Data Analysis}
}

@article{vdLinde2009,
  title={A {B}ayesian latent variable approach to functional principal components analysis with binary and count data},
  author={{van der Linde}, Angelika},
  journal={AStA Advances in Statistical Analysis},
  volume={93},
  number={3},
  pages={307--333},
  year={2009},
  publisher={Springer}
}

@article{wheeler2017,
  title={Bayesian local extremum splines},
  author={Wheeler, Matthew W and Dunson, David B and Herring, Amy H},
  journal={Biometrika},
  volume={104},
  number={4},
  pages={939--952},
  year={2017},
  publisher={Oxford University Press}
}

@article{wrobel2019,
  title={Registration for exponential family functional data},
  author={Wrobel, Julia and Zipunnikov, Vadim and Schrack, Jennifer and Goldsmith, Jeff},
  journal={Biometrics},
  volume={75},
  number={1},
  pages={48--57},
  year={2019},
  publisher={Wiley Online Library}
}

@article{yao2005,
  title={Functional data analysis for sparse longitudinal data},
  author={Yao, Fang and M{\"u}ller, Hans-Georg and Wang, Jane-Ling},
  journal={Journal of the American Statistical Association},
  volume={100},
  number={470},
  pages={577--590},
  year={2005},
  publisher={Taylor \& Francis}
}

@article{yu2022,
  title={Bayesian inference for stationary points in {G}aussian process regression models for event-related potentials analysis},
  author={Yu, Cheng-Han and Li, Meng and Noe, Colin and Fischer-Baum, Simon and Vannucci, Marina},
  journal={Biometrics},
  year={2022},
  publisher={Wiley Online Library}
}

@misc{zhong2021,
  title={Sparse logistic functional principal component analysis for binary data},
  author={Zhong, Rou and Liu, Shishi and Li, Haocheng and Zhang, Jingxiao},
  howpublished={arXiv:2109.08009},
  year={2021}
}

\end{document}